\documentclass[12pt]{article}
\usepackage{amsmath}
\usepackage{graphicx,psfrag,epsf}
\graphicspath{{./figure/}}
\usepackage{enumerate}
\usepackage{url} % not crucial - just used below for the URL 
\usepackage{float}
\usepackage[table,xcdraw,dvipsnames]{xcolor}

\newcommand{\blind}{0}

\usepackage{algorithm,algorithmic}
\usepackage[algo2e]{algorithm2e}
\usepackage{multirow}
\usepackage{lipsum}
\usepackage{amsfonts}
\usepackage{cite}      % Written by Donald Arseneau
                  
\usepackage{graphicx}   % Written by David Carlisle and Sebastian Rahtz

\usepackage{rotating}% for sidewaysfigure
\usepackage{amsmath}
\usepackage{amsthm}
\usepackage{amsfonts}
\usepackage{amssymb}
\usepackage{graphicx}
\usepackage{subcaption}
\usepackage[skip=0pt]{caption}
\usepackage{transparent}
\usepackage{multirow}
\usepackage{color}
\usepackage{blkarray}
\usepackage[round]{natbib}
\usepackage{booktabs}
\usepackage{xr}
\usepackage[]{hyperref}
\hypersetup{colorlinks=true, 
	linkcolor=red,       
    citecolor=blue,        % color of links to bibliography
    filecolor=magenta,      % color of file links
    urlcolor=cyan           % color of external links
}  

\newcommand{\bm}[1]{\boldsymbol{#1}}
\newcommand{\E}{\mathbb{E}}

\renewcommand{\P}{\mathbb{P}}

\newcommand{\Y}{\bm{Y}}

\newcommand{\wt}[1]{\widetilde{#1}}

\newtheorem{prop}{Proposition}

\newcommand{\1}{\mathbb{I}}
\newcommand{\iid}{\stackrel{iid}{\sim}}

\newcommand{\wh}[1]{\smash{\widehat{#1}}}

\usepackage{array}
\newcolumntype{C}{@{\extracolsep{0.5in}}c@{\extracolsep{0pt}}}

\def\C {\,|\:}

\def\C {\,|\:}

\def\ep{\bm{\epsilon}}

\def\b{\bm{\beta}}
\def\W{\bm{W}}
\def\m{\bm{\mu}}
\def\Th{\bm{\Theta}}
\def\Y{\bm{Y}}

\def\pen{\text{pen}}

\def\X{\bm{X}}
\def\x{\bm{x}}
\def\y{\bm{y}}

\def\bg{\bm{\gamma}}

\def\z{\bm{\zeta}}

\def\b{\bm{\beta}}

\newcommand{\e}{\mathrm{e}}

\newcommand{\R}{\mathbb{R}}

\newtheorem{definition}{Definition}[section]
\newtheorem{lemma}{Lemma}[section]

\newtheorem{theorem}{Theorem}[section]

\newtheorem{remark}{Remark}[section]
\newtheorem{corollary}{Corollary}[section]

 \theoremstyle{assumption}

\begin{document}

\def\spacingset#1{\renewcommand{\baselinestretch}%
{#1}\small\normalsize} \spacingset{1}

%%%%%%%%%%%%%%%%%%%%%%%%%%%%%%%%%%%%%%%%%%%%%%%%%%%%%%%%%%%%%%%%%%%%%%%%%%%%%%

\if0\blind
{
	\title{\sf Bayesian Bootstrap Spike-and-Slab LASSO}
	\author{
		Lizhen Nie\footnote{
			4$^{th}$ year PhD Student at the  {\sl  \small Department of  Statistics, University of Chicago}}\,\, and    Veronika Ro\v{c}kov\'{a}\footnote{  Associate Professor in Econometrics and Statistics and James S. Kemper Faculty Scholar at the  {\sl \small Booth School of Business, University of Chicago}. \newline
			The authors gratefully acknowledge the support from the James S. Kemper Faculty Fund at the Booth School of Business and 
			the National Science Foundation (Grant No. NSF DMS-1944740).} 
	} 
	\date{\vspace{-5ex}}
	\maketitle
} \fi

\if1\blind
{
	\bigskip
	\bigskip
	\bigskip
	\begin{center}
		{\LARGE\sf Bayesian Bootstrap Spike-and-Slab LASSO}
	\end{center}
	\medskip
} \fi

\bigskip
\begin{abstract}

The impracticality of posterior sampling  has prevented the widespread adoption of spike-and-slab priors in high-dimensional applications.
To alleviate the computational burden, optimization strategies have been proposed that quickly find local posterior modes. Trading off uncertainty quantification for computational speed, these strategies have enabled spike-and-slab deployments at scales that would be previously unfeasible. We build on one recent development in this strand of work: the  Spike-and-Slab LASSO procedure of \cite{rovckova2018spike}. Instead of optimization, however, we explore  multiple avenues for posterior sampling, some traditional and some new. Intrigued by the speed of  Spike-and-Slab LASSO mode detection, we explore the possibility of  sampling from an approximate posterior by performing MAP optimization on many independently perturbed datasets. To this end, we explore Bayesian bootstrap ideas and introduce a new class of jittered Spike-and-Slab LASSO priors with random shrinkage targets. These priors are a key constituent of  the {\em Bayesian Bootstrap Spike-and-Slab LASSO} (BB-SSL) method proposed here. BB-SSL turns fast optimization into approximate posterior sampling. Beyond its scalability, we show that BB-SSL has a strong theoretical support. Indeed, we find that the induced pseudo-posteriors contract around the truth at a near-optimal rate   in sparse normal-means and in high-dimensional regression.
We compare our algorithm to the traditional Stochastic Search Variable Selection (under Laplace priors) as well as many state-of-the-art methods for shrinkage priors. We show, both in simulations and on real data, that our method fares very well  in these comparisons, often providing substantial computational gains. 

%We explore optimization methods for approximate posterior sampling using Bayesian bootstrap ideas. 
%We show that with suitable perturbation of the data and the prior, MAP estimation can be used to efficiently sample from the posterior.

\end{abstract}

\noindent%
{\bf Keywords:} {\em Bayesian Bootstrap, Posterior Contraction, Spike-and-Slab LASSO, Weighted Likelihood Bootstrap.}
\vfill

\newpage
\spacingset{1.45} % DON'T change the spacing!

% Define document title and author

% Write abstract here

\section{Posterior Sampling under Shrinkage Priors}
Variable selection is arguably one of the most widely used dimension reduction techniques in modern statistics. 
The default Bayesian  approach to variable selection assigns a probabilistic blanket over models via spike-and-slab priors \cite{mitchell1988bayesian,george1993variable}. 
The major conceptual appeal of the spike-and-slab approach is the availability of {\em uncertainty quantification} for both model parameters {\sl as well as} models themselves (\cite{madigan1994model}). However,  practical costs of posterior sampling  can be formidable given the immense scope of modern analyses.  The main thrust of this work is to extend the reach of  existing posterior sampling algorithms in new faster directions.

This paper focuses on the canonical linear regression model, where a vector of responses $\Y=(Y_1,\dots,Y_n)^T$ is stochastically linked to fixed predictors $\x_i\in\R^p$  through 
\begin{equation}\label{eq:model}
Y_i=\x_i^T\bm \beta_0+\epsilon_i\quad\text{with}\quad \epsilon_i\overset{\text{i.i.d}}{\sim} \mathcal{N}(0, \sigma^2)\quad \text{for}\quad  1\leq i\leq n,
\end{equation}
where $\sigma^2>0$  and where $\b_0\in\R^p$ is a possibly sparse vector of regression coefficients. In this work, we assume that $\sigma^2$ is   known and we refer to  \cite{moran2018variance} for elaborations with an unknown variance. 
 We  assume that the vector $\Y$ and the regressors $\X=[\X_1,\dots,\X_p]$ have been centered and, thereby, we omit  the intercept.
 In the presence of uncertainty about which subset of $\b_0$ is in fact nonzero, one can assign a prior distribution over the regression coefficients $\b=(\beta_1,\dots,\beta_p)^T$ as well as the pattern of nonzeroes $\bg=(\gamma_1,\dots,\gamma_p)^T$ where $\gamma_j\in\{0,1\}$ for whether or not the effect $\beta_j$ is active. This formalism can be condensed into the usual spike-and-slab prior form
\begin{equation}\label{eq:ss}
\pi(\b\C\bg)=\prod_{j=1}^p[\gamma_j\psi_1(\beta_j)+(1-\gamma_j)\psi_0(\beta_j)],\quad \P(\gamma_j=1\C\theta)=\theta,\quad \theta\sim\mathrm{Beta}(a,b),
\end{equation}
where $a,b>0$ are scale parameters and where $\psi_0(\cdot)$ is a highly concentrated prior density around zero (the spike) and $\psi_1(\cdot)$ is a diffuse density (the slab).
The dual purpose of the spike-and-slab prior  is to (a) shrink  small signals towards zero and (b) keep large signals intact.  The most popular  incarnations of the
spike-and-slab prior  include:  the point-mass spike (\cite{mitchell1988bayesian}), the non-local slab priors (\cite{johnson2012bayesian}), the  Gaussian mixture (\cite{george1993variable}), the Student mixture (\cite{ishwaran2005spike}). 
More recently, \cite{rovckova2018bayesian} proposed the Spike-and-Slab LASSO  (SSL) prior, a mixture of two Laplace distributions 
$\psi_0(\beta)=\frac{\lambda_0}{2}\e^{-|\beta|\lambda_0}$ and  $\psi_1(\beta)=\frac{\lambda_1}{2}\e^{-|\beta|\lambda_1}$ where $\lambda_0\gg\lambda_1$,
which forms a continuum between the point-mass mixture prior and the LASSO prior (\cite{park2008bayesian}).

Posterior sampling under spike-and-slab priors is notoriously difficult. Dating back to at least 1993 \citep{george1993variable},
multiple advances have been made to speed up spike-and-slab  posterior simulation (\cite{george1997approaches}, \cite{bottolo2010evolutionary}, \cite{clyde2011bayesian}, \cite{hans2009bayesian}, \cite{johndrow2017bayes}, \cite{welling2011bayesian}, \cite{xu2014distributed}). %\citep{GM97,bottolo,dellaportas,Clyde,hans_lasso,hans}. 
More recently, several clever computational tricks have been suggested that avoid costly matrix inversions by using linear solvers \citep{bhattacharya2016fast} or by disregarding correlations between active and inactive coefficients  \citep{narisetty2019skinny}. 
Neuronized priors have been proposed (\cite{shin2018neuronized}) that offer computational benefits by using  close approximations to spike-and-slab priors without latent binary indicators.  Modern applications have nevertheless challenged MCMC algorithms and new computational strategies are desperately needed to keep pace with big data.

 Optimization strategies have shown great promise and enabled deployment of spike-and-slab priors at scales that would be previously unfeasible (\cite{rovckova2014emvs}, \cite{rovckova2018spike}, \cite{carbonetto2012scalable}). Fast posterior mode  detection is effective  in  structure discovery and data exploration, a little less so for inference. In this paper, we review and propose new strategies for posterior sampling under the Spike-and-Slab LASSO priors, filling the gap between  exploratory data analysis and proper statistical inference.

We capitalize on the latest MAP optimization and MCMC developments to provide several  posterior sampling implementations for the Spike-and-Slab LASSO method of \cite{rovckova2018spike}. The first one  (presented in Algorithm \ref{SSVS}) is  exact and conventional, following in the footsteps of Stochastic Search Variable Selection 
\citep{george1993variable}). 
%For the inherent matrix inversion, we use the trick of \cite{bhattacharya2016fast}. 
The second one is approximate and new. The cornerstone of this strategy is the Weighted Likelihood Bootstrap (WLB) of \cite{newton1994approximate} which was recently resurrected  in the context of posterior sampling with sparsity priors by \cite{newton2020weighted}, \cite{fong2019scalable} {\color{black}and \cite{ng2020random}}.  The main idea behind WLB is to perform approximate sampling by independently  optimizing randomly perturbed likelihood functions. We extend the WLB framework by incorporating perturbations {\em both} in the likelihood and in the prior. The main contributions of this work are two-fold. First, we introduce BB-SSL ({\em Bayesian Bootstrap Spike-and-Slab LASSO}), a  novel algorithm for approximate posterior sampling in high-dimensional regression under Spike-and-Slab LASSO priors. Second, we show that suitable ``perturbations" lead to  approximate posteriors that contract around the truth {\em at the same speed} (rate) as the actual posterior. These theoretical results have nontrivial practical implications as they offer guidance on the choice of the distribution for perturbing weights. Up until now, theoretical properties of  WLB have  %been confined to 
{largely concentrated on} consistency statements in low dimensions for iid data  (\cite{newton1994approximate}). {More recently,   \cite{ng2020random} established  conditional consistency (asymptotic normality)  in the context of LASSO regression  for a fixed dimensionality and model selection consistency for a growing dimensionality. Our theoretical results  also allow the dimensionality to increase with the sample size and go beyond mere consistency by showing that BB-SSL  leads to rate-optimal estimation in sparse normal-means and  high-dimensional regression under standard assumptions.}  Last but not least, we make thorough comparisons with the gold standard  (i.e. exact MCMC sampling) on multiple simulated and real datasets, concluding that the proposed algorithm is scalable and reliable  in practice. {BB-SSL  is  (a)  unapologetically parallelisable,  and (b) it does not require costly matrix inversions (due to its coordinate-wise optimization nature), thereby having the potential to meet the demands of large datasets. }

The structure of this paper is as follows. Section \ref{sec:notation} introduces the notation. Section \ref{sec:SSL_revisit} revisits Spike-and-Slab LASSO and presents a traditional algorithm for posterior sampling. Section \ref{sec:BB} investigates performance of weighted Bayesian bootstrap, the building block of this work, in high dimensions. In Section \ref{sec:introducing_BB-SSL}, we introduce BB-SSL and present our theoretical study showing rate-optimality as well as its connection with other bootstrap methods. Section \ref{sec:simulations} shows simulated examples and Section \ref{sec:real_data} shows performance on real data. We conclude the paper with a discussion in Section \ref{sec:discussion}.

%To this end, spike-and-slab  priors  have been regarding as default for Bayesian model selection.
%The main thrust of this work is to push the frontiers of  spike-and-slab posterior sampling in new faster directions so that these sacrifices are less dramatic.

%Spike-and-slab variable selection continues to have a decided impact on Bayesian high-dimensional statistics and is widely regarded  as the default Bayesian model selection tool.

%Many algorithms have been suggested for these various priors, some based on posterior sampling (George and McCulloch 1997, Hans et al. shotgun, Botollo and Richardson (2010), (2010)) and some based on optimization (Stephens  and Guo (2016), Rockova and George (2014), Rockova and George (2018)).
\subsection{Notation}\label{sec:notation}

With $\phi(y;\mu;\sigma^2)$ we denote the Gaussian density with a mean $\mu$ and a variance $\sigma^2$. 
%Throughout the paper, we assume the following spike distribution   $\psi_0(\cdot)=\frac{\lambda_0}{2}\e^{-|\beta|\lambda_0}$ and slab distribution  $\psi_1(\cdot)=\frac{\lambda_1}{2}\e^{-|\beta|\lambda_1}$. 
We use $\xrightarrow{\text{d}}$ to denote convergence in distribution. We write $a_n=O_p(b_n)$ if  { for any $\epsilon>0$, there exist finite $M>0$ and $N>0$ such that $\P(|a_n/b_n|>M)<\epsilon$ for any $n>N$.} We write $a_n=o_p(b_n)$ if for any $\epsilon>0$, $\lim_{n\rightarrow \infty}\P(|a_n/b_n|>\epsilon)=0$.  We also write $a_n=O(b_n)$ as $a_n\lesssim b_n$.
We use $a\asymp b$ if $a\lesssim b$ and $b\lesssim a$. We   use $a_n\gg b_n$ to denote $b_n=o(a_n)$ and $a_n\ll b_n$ to denote $a_n=o(b_n)$.
We denote with $\X_{A}$ a sub-matrix consisting of   columns of $\X$'s indexed by a subset $A\subset\left\{1,\cdots,p\right\}$ and with 
$\bm P_{A}$ the orthogonal projection to the range of $\X_{A}$ \citep{zhang2012general}, i.e., $\bm P_A=\X_A\X_A^+$ where $\X_A^+$ is the Moore-Penrose  inverse of $\X_{A}$. We denote with $\|\X\|$  the matrix operator norm of $\X$.

 \section{Spike-and-Slab LASSO Revisited}\label{sec:SSL_revisit}
 The Spike-and-Slab LASSO (SSL) procedure of \cite{rovckova2018spike} recently emerged as one of the more successful non-convex penalized likelihood methods. Various SSL incarnations  have spawned since its introduction,  including a version for group shrinkage (\cite{bai2020spike}, \cite{tang2018group}), survival analysis (\cite{tang2017spike}), varying coefficient models (\cite{bai2020spike}) and/or Gaussian graphical models (\cite{deshpande2019simultaneous}, \cite{li2019bayesian}). The original procedure proposed for Gaussian regression targets a posterior mode
\begin{equation}\label{SSL}
\wh\b=\arg\max_{\b\in\R^p}\left\{\prod_{i=1}^n\phi(Y_i;\x_i^T\b;\sigma^2)\times \int_\theta \prod_{j=1}^p\pi(\beta_j\C\theta)d\pi(\theta)\right\},
\end{equation}
where $\pi(\beta_j\C\theta)=\theta\psi_1(\beta_j)+(1-\theta)\psi_0(\beta_j)$ is obtained from \eqref{eq:ss} by integrating out the missing indicator $\gamma_j$ and by deploying $\psi_1(\beta_j)=\lambda_1/2\e^{-|\beta_j|\lambda_1}$ and $\psi_0(\beta_j)=\lambda_0/2\e^{-|\beta_j|\lambda_0}$
with $\lambda_0\gg\lambda_1$.    \cite{rovckova2018spike} develop a coordinate-ascent strategy which targets $\wh\b$ and which quickly finds (at least a local) mode of the posterior landscape. 
This strategy (summarized in Theorem 3.1 of \cite{rovckova2018spike}) iteratively updates each $\wh\beta_j$ using an implicit equation\footnote{Here we are not necessarily assuming that $\|\X_j\|_2^2=n$ and the above formula is hence slightly different from Theorem 3.1 of \cite{rovckova2018spike}. 
%Also we do not necessarily have $\sigma^2=1$ so we incorporate $\sigma^2$ as in \cite{moran2018variance}.
}
\begin{align}\label{beta:implicit}
		\wh\beta_j=&
		%\tilde S(z_j^{(k)},\lambda^*(\wh\beta_j), \Delta_j)=
		\frac{1}{\|\X_j\|_2^2}\left(|z_j|-\sigma^2\lambda^*_{\wh\theta_j}(\wh\beta_j)\right)_+\mathrm{sign}(z_j)\times \1(|z_j|>\Delta_j)		
\end{align}
where   $\wh\theta_j=\E[\theta\C\wh\b_{\backslash j}]$, $z_j=\X_j^T(\Y-\X_{\backslash j}\wh\b_{\backslash j})$ and 
{
%\begin{equation}\label{eq:Delta_def}
$\Delta_j=\inf_{t>0}\left(\|\X_j\|^2t/2-\sigma^2\rho(t\C \wh\theta_j)/t\right)$ with
%\end{equation}with
%\begin{equation}\label{dfn:rho}
$	\rho(t\C \theta)=-\lambda_1|t|+\log[p^\star_\theta(0)/p^\star_\theta(t)],$	
%\end{equation}
where
	\begin{equation}\label{eq: p_star}\centering
	p^\star_\theta(t)=\frac{\theta\psi_1(t)}{\theta\psi_1(t)+(1-\theta)\psi_0(t)}\quad\text{and}\quad \lambda^*_\theta(t)=\lambda_1p^\star_\theta(t)+\lambda_0(1-p^\star_\theta(t)).
	\end{equation}
}
\cite{rovckova2018spike} also provide fast updating schemes for $\Delta_j$ and $\wh\theta_j$. In this work, we are interested in {\em sampling from the posterior} as opposed to mode hunting.

One immediate strategy for sampling from the Spike-and-Slab LASSO posterior is the Stochastic Search Variable Selection (SSVS) algorithm  of \cite{george1993variable}. One can 
regard the Laplace distribution (with a penalty $\lambda>0$) as a scale mixture of Gaussians  with an exponential mixing distribution  (with a rate $\lambda^2/2$   as in \cite{park2008bayesian})  and rewrite the SSL prior using the following  hierarchical form:
\begin{align*}
\bm \beta \C \bm \tau&\sim \mathcal{N}(\textbf{0},D_{\bm\tau})\quad\text{with}\quad   D_{\bm\tau}=\mathrm{Diag}(1/\tau_1^2,1/\tau_2^2,\cdots,1/\tau_p^2),\\
\bm\tau^{-1} \C \bg &\sim \prod_{j=1}^{p}\frac{\lambda_{j}^2}{2}\e^{-\lambda_{j}^2/2\tau_j^2},\quad\text{where}\quad \lambda_j=\gamma_j\lambda_1+(1-\gamma_j)\lambda_0,\\
\gamma_j \C \theta &\sim \mathrm{Bernoulli}(\theta)\quad\text{with}\quad \theta \sim \mathrm{Beta}(a,b),
\end{align*}
where  $\bm\tau^{-1}=(1/\tau_1^2,\dots,1/\tau_p^2)^T$ is the vector of variances.
The conditional conjugacy of the SSL prior enables direct  Gibbs sampling for $\b$ (see {Algorithm }\ref{SSVS} below).
However, as with any other Gibbs sampler for Bayesian shrinkage models \citep{bhattacharya2015dirichlet}, this algorithm involves costly matrix inversions and can be quite slow  when both $n$ and $p$ are large. In order to improve the MCMC computational efficiency when $p>n$,  \cite{bhattacharya2016fast} proposed a clever trick. By recasting the sampling step as  a solution to a linear system, one  can circumvent a  Cholesky factorization which would otherwise have   a complexity $O(n^2p)$ per iteration.
 %This reduces the computational cost to $O(n^2p)$ or $O(np^2)$ when $\DD$\footnote{what is $\DD$?} is not diagonal. 
 Building on  this development, \cite{johndrow2017bayes} developed a blocked Metropolis-within-Gibbs algorithm to sample from  horseshoe posteriors \citep{carvalho2010horseshoe} and designed an approximate algorithm which thresholds  small effects based on the sparse structure of the target. The exact method has a per-step complexity $O(n^2p)$ while the approximate one has only $O(np)$. In similar vein, the Skinny Gibbs MCMC method of \cite{narisetty2019skinny} also bypasses large matrix inversions by independently sampling from  active  and inactive $\beta_i$'s. While the method is only approximate, it has a rather favorable computational complexity $O(np)$.   
 {A referee suggested another Gibbs sampler implementation with a complexity $O(np)$ which can be obtained by updating $(\beta_j,\gamma_j)$ one at a time while conditioning on the remaining $(\beta_j,\gamma_j)$'s \citep{geweke1991efficient}. While this implementation is very fast for point-mass spikes,  the Spike-and-Slab LASSO prior requires sampling from a half-normal distribution which can be inefficient in practice. One-site Gibbs samplers also generally lead to slower mixing due to increased autocorrelation. In simulations, we find the performance of this method to be comparable with SSVS using  \cite{bhattacharya2016fast}'s  trick. The detailed description of this algorithm is included in the Appendix (Section \ref{sec_appendix:MCMC_np}).    }

The impressive speed of the Spike-and-Slab LASSO mode detection  makes one wonder whether performing many independent optimizations on randomly perturbed datasets will lead to posterior simulation that is more economical. Moreover, one may wonder whether the induced approximate posterior is sufficiently close to the actual posterior $\pi(\b\C\Y)$ and/or whether it can be used for meaningful estimation/uncertainty quantification. We attempt to address  these intriguing questions in the next sections.
%For the Spike-and-Slab LASSO posterior sampling, we  can  adapt  the Weighted Bayesian Bootstrap (WBB) method of \cite{newton2020weighted} from previous section, with both a fixed or random weight $\wt w^t$.

{\scriptsize
	\begin{algorithm}[t]\small
		\spacingset{1.1}
%		\KwData{ Data ($Y_i$, $\x_i$) for $1\leq i\leq n, \x_i\in \R ^p$}
%		\KwResult{$({\bm\beta}^1,\dots,\b^T)$}
		\textbf{Set}: $\lambda_0\gg\lambda_1,$  $a,b>0$, $T$ (number of MCMC iterations), $B$ (number of samples to discard as burn-in).
		
		\textbf{Initialize}: $\bm \beta^0$ (e.g. LASSO solution after 10-fold cross validation) and $\bm\tau^0$.
%		; $\bm \gamma$: binary vector with $i$-th component having $E\gamma_j=\frac{\psi_1(\beta_j)}{\psi_0(\beta_j)+\psi_1(\beta_j)}$; $\frac{1}{(\tau_j)^2}\sim$ Inverse-Gaussian$(\mu'_j,\lambda'_j)$, where $\mu'_j = \frac{\gamma_j^{t-1}}{|\beta_j^t|},\lambda'_j=\lambda_{\gamma_j}^2$; $\theta\sim Beta(\sum_{j=1}^{p} \gamma_j+a,p-\sum_{j=1}^{p} \gamma_j+b)$.
		\\
		\For{$t=1,2,\cdots, T$}{
			{\vspace{-0.3cm} \hspace{5cm}\color{white} \tiny ahoj}\\
			\textbf{(a)} Sample $\bm \beta^t$  $\sim \mathcal{N}\left(\mu_{\bg}, \Sigma_{\bg}\right)$, where $ \Sigma_{\bg}=(\X^T\X/\sigma^2+D_{\bm\tau^{t-1}}^{-1})^{-1}$ and $\mu_{\bg}= \Sigma_{\bg}\X^T{\Y}/\sigma^2$.\\
			\textbf{(b)} Sample $(\tau_j^{t})^2\sim \mathrm{Inv}$-$\mathrm{Gaus}(\mu'_j,(\lambda'_j)^2)$ for $j=1,2,\dots,p$ , where 
			$$\mu'_j = \frac{|\lambda'_j|}{|\beta_j^t|}\quad\text{and}\quad (\lambda'_j)^2=\gamma_j^{t-1}\lambda_{1}^2+(1-\gamma_j^{t-1})\lambda_0^2$$
			\textbf{(c)} Sample $\gamma_j^t \sim \mathrm{Bernoulli}\left(\frac{\pi_1}{\pi_1+\pi_0}\right)$, where 
			$$
\pi_1=\theta^{t-1}\lambda_1^2\e^{-\lambda_1^2/2(\tau_j^t)^2}/2\quad\text{and}\quad\pi_0=(1-\theta^{t-1})\lambda_0^2\e^{-\lambda_0^2/2(\tau_j^t)^2}/2.
$$ 
			\textbf{(d)} Sample $\theta^t\sim\mathrm{Beta}(\sum_{j=1}^{p} \gamma_j^t+a,p-\sum_{j=1}^{p} \gamma_j^t+b)$.\\		}
		\textbf{Return}: $\b^t,\bm\gamma^t,\theta^t$ where $t=B+1,B+2,\cdots,T$.
		\caption{\bf : SSVS}\label{SSVS}
\end{algorithm}}

 \section{Likelihood Reweighting and Bayesian Bootstrap}\label{sec:BB}
%We begin our development by disambiguating terminology involving Bayesian computation and bootstrap. 
The jumping-off point of our methodology is the weighted likelihood bootstrap (WLB) method introduced by \cite{newton1994approximate}. The premise of  WLB  is to draw approximate samples from the posterior by independently  maximizing  randomly reweighted likelihood functions.   Such a sampling strategy is computationally beneficial  when, for instance, maximization is easier than Gibbs sampling from conditionals.   
%Moreover, under suitable weights distribution, these maxima will be approximate samples from a posterior distribution. 

In the context of linear regression \eqref{eq:model}, the WLB method of \cite{newton1994approximate} will produce a series of draws $\wt{\b}_t$ by first   sampling 
random weights $\bm{w}_t=(w_1^t,w_2^t,\cdots,w_n^t)^T$ from some weight distribution $\pi(\bm w)$ and then  maximizing a reweighted likelihood
\begin{equation}\label{eq:beta_tilde}
\wt{\bm \beta}_t=\arg\,\max_{\bm \beta} \wt L^{\bm w_t}(\b,\sigma^2; \X^{(n)},\Y^{(n)})
\end{equation}
where
$$
\wt L^{\bm w_t}(\b,\sigma^2; \X^{(n)},\Y^{(n)})=\prod_{i=1}^n\phi(Y_i;\x_i^T\b;\sigma^2)^{w_i^t}.
$$
\cite{newton1994approximate} argue that for certain weight distributions $\pi(\bm w)$, the conditional distribution of $\wt\b_t$'s given the data can provide a good approximation to the posterior distribution of $\b$. Moreover, WLB was shown to have nice theoretical guarantees when the number of parameters does not grow. Namely, under uniform Dirichlet weights (more below) and iid data samples, WLB is consistent (i.e. concentrating on any arbitrarily small neighborhood around an MLE estimator) and asymptotically first-order correct (normal with the same centering)  for almost every realization of the data. The WLB method, however, is only approximate and it does not naturally accommodate a prior.
Uniform Dirichlet weights %, i.e $\sqrt{nI(\wh\b^{ML})}(\wt\b-\wh\b^{ML}})$ asymptotically follows the standard normal distribution, where $\wh\b^{ML}}$ is maximum likelihood estimator and $I(\b)$ is the Fisher information calculated at $\b$. 
 provide a higher-order asymptotic equivalence when one chooses the squared Jeffrey's prior. However, for more general prior distributions (such as shrinkage priors considered here), the correspondence between the prior $\pi(\b)$ and $\pi(\bm w)$ is  unknown. 
\cite{newton1994approximate} suggest post-processing  the posterior samples with importance sampling to leverage  prior information. 
This pertains to \cite{efron2012bayesian}, %who develops the connection between parametric Bootstrap and posterior sampling through reweighting in exponential family models. 
who proposes a posterior sampling method for  exponential family models with importance sampling on parametric bootstrap distributions.

%The WLB method was revived recently in the context of model misspecification by \cite{lyddon2018nonparametric}, who propose Loss-Likelihood Bootstrap (LLB) for non-parametric inference by assigning a Dirichlet Process (DP) prior distribution on the underlying model. We discuss the connections to our method and LLB in Section \ref{sec:connection}. This method, while not originally posited as an approximate Bayesian posterior sapling strategy, is functionally equivalent to WLB under log-likelihood  loss function and when $\pi(\bm w)$ is the uniform Dirichlet distribution. However,  it does not directly allow for the inclusion of the prior (only through the Dirichlet process prior distribution or the loss function). Setting the concentration parameter in the DP prior to zero, one recovers the Bayesian Bootstrap (BB) of \cite{rubin1981bayesian}. While equivalent under specific conditions, WLB (LB) and BB are conceptually different. The Bayesian bootstrap is obtained by sampling uniform Dirichlet weights which represent a posterior sample to the unknown parameters of a discrete distribution under specific choice of the prior and model assumptions. The WLB, on the other hand, treats the weights as perturbations of each observation's contribution to the likelihood.

Alternatively, \cite{newton2020weighted} suggested blending the prior directly  into WLB by including  a   weighted prior term, i.e. replacing \eqref{eq:beta_tilde} with
$$
\wt\b_t=\arg\,\max_{\bm \beta} \wt L^{\bm w_t}(\b,\sigma^2; \X^{(n)},\Y^{(n)})\pi(\bm \beta)^{\wt w^t},
$$
where $w_i^t\overset{\text{i.i.d}}{\sim} Exp(1)$\footnote{Note that if $w_1,w_2,\cdots,w_n\overset{\text{i.i.d}}{\sim} Exp(1)$, then $\frac{w_1}{\sum w_i},\frac{w_2}{\sum w_i },\cdots,\frac{w_n}{\sum w_i }\sim Dir(1,1,\cdots,1)$ which brings us back to the uniform Dirichlet distribution.}. 
This so called Weighted Bayesian Bootstrap (WBB) method treats the prior weight $\wt w^t$ as either fixed (and equal to one) or as one of the random data weights arising from the exponential distribution. We explore these two strategies in the next section within the context of the Spike-and-Slab LASSO where  $\pi(\bm\beta)$ is the SSL shrinkage prior implied by \eqref{eq:ss}.

\subsection{WBB meets  Spike-and-Slab LASSO}\label{sec:motivation_noncenter}

Since SSL is a thresholding procedure (see \eqref{beta:implicit}), WBB will ultimately create samples from pseudo-posteriors that have a point mass at zero. This is misleading since the posterior under the Gaussian likelihood  and a  single Laplace prior  is half-normal (\cite{hans2009bayesian}, \cite{park2008bayesian}). Deploying the WBB method thus does not guarantee that uncertainty  be propertly captured for the zero (negligible) effects since their posterior samples may  very often be exactly zero. 
We formalize this intuition below. We want to understand the extent to which the WBB (or WLB) pseudo-posteriors correspond to the actual posteriors. To this end, we focus on the 
canonical  Gaussian sequence model
\begin{equation}\label{eq:Gaussian_sequence}
y_{i}=\beta_{i}^{0}+\epsilon_{i}/\sqrt{n}\quad\text{for}\quad i=1,2,\cdots n.
\end{equation}
Under the separable SSL prior (i.e. $\theta$ fixed), the true posterior is a mixture
\begin{equation}\label{eq:true_posterior}
\begin{split}
\pi\left(\beta_{i}\C y_i\right) & =w_1\pi\left( \beta_i\C y_i,\gamma_i=1  \right) + w_0\pi\left( \beta_i\C y_i,\gamma_i=0  \right)
% =\pi\left( \gamma_i=1\C y_i  \right)\pi\left( \beta_i\C y_i,\gamma_i=1  \right) + \pi\left( \gamma_i=0\C y_i  \right)\pi\left( \beta_i\C y_i,\gamma_i=0  \right)
\end{split}
\end{equation}
where 
$w_1=\pi\left( \gamma_i=1\C y_i  \right)$ and $w_0=\pi\left( \gamma_i=0\C y_i  \right)$.
From \cite{hans2009bayesian}, we know that $\pi\left( \beta_i\C y_i,\gamma_i=1  \right)$ and $\pi\left( \beta_i\C y_i,\gamma_i=0  \right)$ are orthant truncated Gaussians and thus $\pi\left(\beta_{i}\C y_i\right)$ is a mixture of orthant truncated Gaussians.

\iffalse
\begin{align*}
c_1^{(-)}&=\theta\lambda_1e^{-y_i\lambda_1+\lambda_1^2/2n},  &c_1^{(+)}&=\theta\lambda_1e^{y_i\lambda_1+\lambda_1^2/2n},\\
c_0^{(-)}&=(1-\theta)\lambda_0e^{-y_i\lambda_0+\lambda_0^2/2n},  &c_0^{(+)}&=(1-\theta)\lambda_0e^{y_i\lambda_0+\lambda_0^2/2n}. 
\end{align*}
Next, define
\begin{align*}
\phi_1^{(-)}(x)&=\phi(x;y_i-\lambda_1/n,1/n), 
&\phi_1^{(+)}(x)&=\phi(x;y_i+\lambda_1/n,1/n), \\
\phi_0^{(-)}(x)&=\phi(x;y_i-\lambda_0/n,1/n), 
&\phi_0^{(+)}(x)&=\phi(x;y_i+\lambda_0/n,1/n), 
\end{align*}

\begin{equation}
\begin{split}
\pi\left(\beta_{i}\C y_i\right) & =\frac{\1(\beta_{i}\geq0)\left[c_1^{(-)}\phi_1^{(-)}(\beta_i)+c_0^{(-)}\phi_0^{(-)}(\beta_i)\right] +\1(\beta_{i}<0)\left[c_1^{(+)}\phi_1^{(+)}(\beta_i)+c_0^{(+)}\phi_0^{(+)}(\beta_i)\right]}
{\int_{0}^{\infty}c_1^{(-)}\phi_1^{(-)}(\beta_i)d\beta_i+\int_{0}^{\infty}c_0^{(-)}\phi_0^{(-)}(\beta_i)d\beta_i +\int_{-\infty}^{0}c_1^{(+)}\phi_1^{(+)}(\beta_i)d\beta_i+\int_{-\infty}^{0}c_0^{(+)}\phi_0^{(+)}(\beta_i)d\beta_i}
\end{split}
\end{equation}
Definitions of $c_1^{(-)},\phi_1^{(-)}(\cdot),c_1^{(+)},\phi_1^{(+)}(\cdot)$ and all detailed proofs of this section are in the Appendix.
\fi

We start by examining the posterior distribution of {\em active coordinates} such that $|y_i|>|\beta_i^0|/2>0$  (this event happens with  high probability when $n$ is sufficiently large).  For the true posterior, we show in the Appendix (Proposition \ref{prop:1}) that 
$w_0 \rightarrow 0$ and $w_1 \rightarrow 1$.
The true posterior $\pi\left(\beta_{i}\C y_i\right)$ is hence dominated by the component
$\pi\left( \beta_i\C y_i,\gamma_i=1  \right)$, which takes the following form
\begin{equation*}
\begin{split}
\pi\left( \beta_i\C y_i,\gamma_i=1  \right)&=
\frac{\1(\beta_{i}\geq0)c_1^{(-)}\phi_1^{(-)}(\beta_i) +\1(\beta_{i}<0)c_1^{(+)}\phi_1^{(+)}(\beta_i)}
{\int_{0}^{\infty}c_1^{(-)}\phi_1^{(-)}(\beta_i)d\beta_i+ \int_{-\infty}^{0}c_1^{(+)}\phi_1^{(+)}(\beta_i)d\beta_i}
\end{split}
\end{equation*}
where 
\begin{align}
c_1^{(-)}=\theta\lambda_1e^{-y_i\lambda_1+\lambda_1^2/2n}\quad&\text{and}\quad c_1^{(+)}=\theta\lambda_1e^{y_i\lambda_1+\lambda_1^2/2n},\\
\phi_1^{(-)}(x)=\phi\left(x;y_i-\frac{\lambda_1}{n},\frac{1}{n}\right)\quad&\text{and}\quad \phi_1^{(+)}(x)=\phi\left(x;y_i+\frac{\lambda_1}{n},\frac{1}{n}\right).
\end{align}

\begin{figure}
	\begin{subfigure}{0.33\textwidth}\centering
		\includegraphics[width=.9\linewidth]{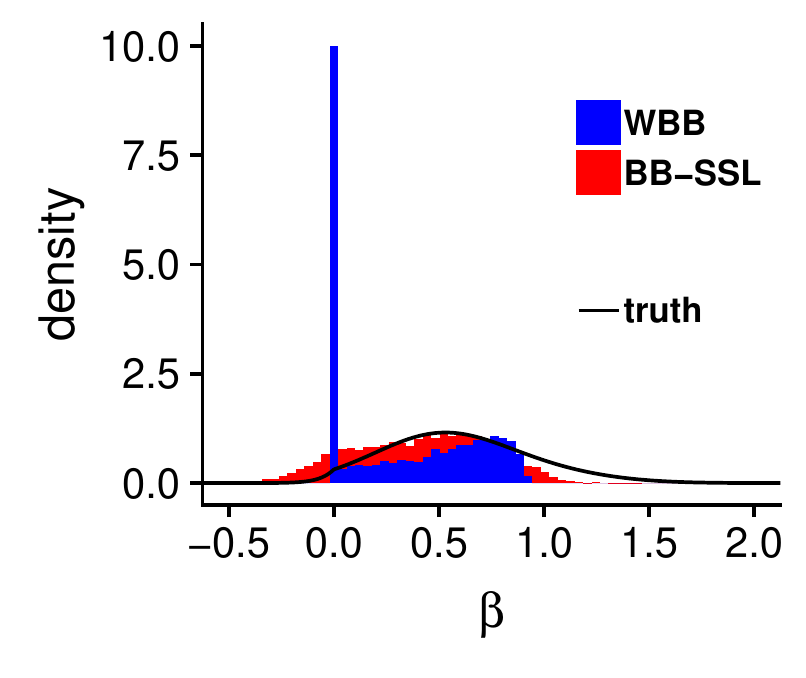}
		\caption{\small $y_i=1$}\label{fig1b}
	\end{subfigure}
	\begin{subfigure}{0.33\textwidth}\centering
		\includegraphics[width=.9\linewidth]{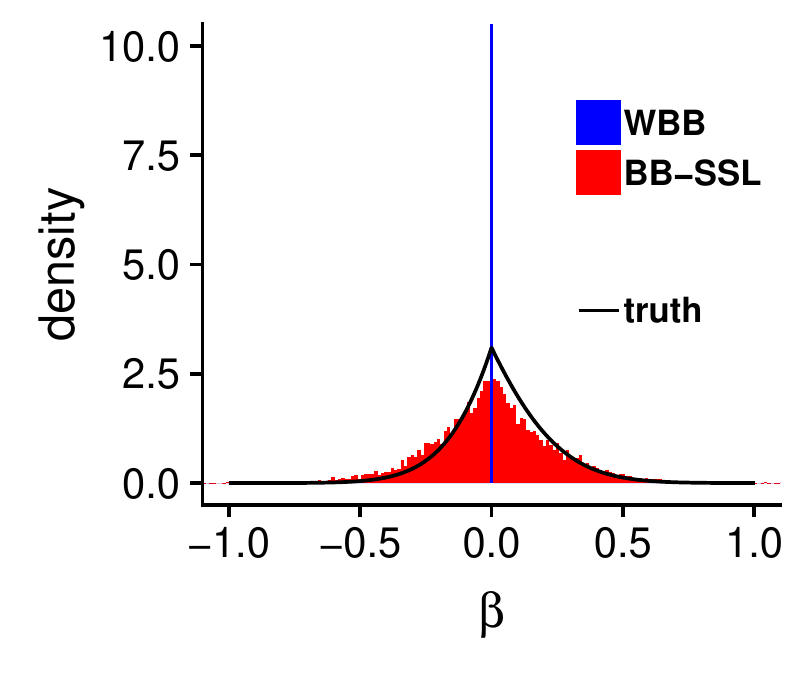}
		\caption{\small $y_i=0.1$}\label{fig1a}
	\end{subfigure}%
	%\hfill <-- it is superfluous 
	%\begin{subfigure}{0.33\textwidth}\centering
	%	\includegraphics[width=\linewidth]{Motivation_rescale_y4.pdf}
	%	\caption{\scriptsize $y_i=4$}\label{fig1c}
	%\end{subfigure}
	\begin{subfigure}{.33\textwidth}\centering
		\includegraphics[width = .92\linewidth]{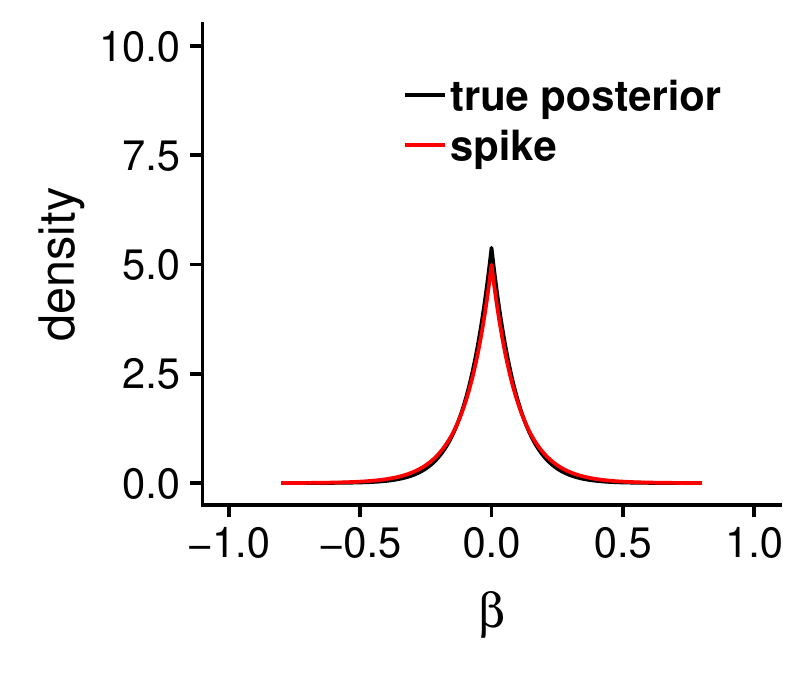}
		\caption{{\small True posterior versus spike. 
		}}
		\label{fig: mu_motivation}
	\end{subfigure}
	\caption{\small {True and WBB approximated posterior distribution $\pi(\beta_i\C y_i)$ under the separable SSL prior with $\lambda_0=5,\lambda_1=0.1$ and  
	$\theta=0.2$ and a Gaussian sequence model $y_i=x_i+\epsilon_i/\sqrt{n}$  with $n=10$. Red bins represent BB-SSL pseudo-posterior, blue bins represent WBB pseudo-posterior (with a random prior weight), black line is the true posterior and  $\alpha=2.5$. The plot  (c) is for the same setting except with $\lambda_0=10$.}}
	\label{fig: gaussian_seq_compare}
\end{figure}

Intuitively, $\lambda_1/n$ vanishes when $n$ is large, so both $\phi_1^{(-)}(\beta_i)$ and $\phi_1^{(+)}(\beta_i)$ will be close to $\phi(\beta_i;y_i,\frac{1}{n})$. This intuition is proved rigorously in the Appendix (Section \ref{sec:appendix_motivation_active}), where  we show  that  the density of the transformed variable $\sqrt{n}(\beta_i-y_i)$ converges pointwise to the standard normal density and thereby the posterior $\pi\left( \sqrt{n}(\beta_i-y_i)\C y_i,\gamma_i=1  \right)$ converges to $N(0,1)$ in total variation (\cite{scheffe1947useful}).

We now investigate the limiting shape of the pseudo-distribution obtained from WBB. For a given weight $w_i>0$, the WBB estimator $\wh\beta_i$ equals 
\begin{equation}\label{eq: WBB_solution}
\wh{\beta_i}=\begin{cases}
0, & \text{if $|y_i|\leq {\Delta_{w_i}}$}.\\
\left[|y_i|-\frac{1}{w_in}\lambda^*(\wh{\beta_i})\right]_+\mathrm{sign}(\sqrt{w_in}y_i), & \text{otherwise}.
\end{cases}
\end{equation}
{where 
	$\Delta_{w_i}=\inf_{t>0}[t/2-\rho(t\C\theta)/(nw_it)]$
is the analogue of $\Delta_j$ defined below \eqref{beta:implicit} for the regression model  and  where  $\rho(t\C\theta)$ was also defined below \eqref{beta:implicit}.}
When $\wh\beta_i\neq 0$, we show  in the Appendix (Section \ref{sec:appendix_motivation_active}) that
$
n\left(\wh{\beta_i}-y_i\right)
\rightarrow -\frac{1}{w_i}\lambda_1
$.
Under the condition $|y_i|>\frac{|\beta_i^0|}{2}>0$, it can be shown (Section \ref{sec:appendix_motivation_active}) that
$
\P_{w_i}\left( \wh\beta_i=0  \C y_i\right)\rightarrow 0
$.
For active coordinates, the distribution of the WBB samples $\wh\beta_i$ is thus purely determined by that of $-\frac{1}{w_i}$. The shape of this posterior can be very different from the standard normal one, as can be seen from Figure \ref{fig: gaussian_seq_compare}. In particular,  Figure \ref{fig1b} shows how WBB (a) assigns a non-negligible prior mass to zero (in spite of evidence of signal)   (b) incurs bias in estimation and  (c) underestimates variance with a skewed  misrepresentation of the posterior distribution. This last aspect is particularly pronounced when the signal is even stronger.\footnote{
In fact, under the  uniform Dirichlet distribution,  the marginal distribution becomes $w_i\sim n\times \text{Beta}(1,n-1)$. Since  $n\times \text{Beta}(1,n-1)\xrightarrow{\text{d}}\text{Gamma}(1,1)$, the distribution of $-\frac{1}{w_i}$ converges to Inverse-Gamma(1,1) which exhibits a skewed shape, which is in sharp contrast to the symmetric Gaussian distribution of the true posterior.}

The  approximability of WBB does not get any better for {\em  inactive coordinates} such that $\beta_i^0=0$ and thereby $|y_i|=O_p(\frac{1}{\sqrt{n}})$ from \eqref{eq:Gaussian_sequence}.  The following arguments will be under the assumption  $|y_i|\asymp\frac{1}{\sqrt{n}}$.  One can show (Section \ref{sec:appendix_motivation_inactive} in the Appendix) that the true posterior
$\pi\left(\beta_{i}\C y_i\right)$ is dominated by the component 
$\pi\left( \beta_i\C y_i,\gamma_i=0  \right)$ since
$w_0 \rightarrow 1$ and $w_1 \rightarrow 0$. When $n$ is sufficiently large, one can then approximate this distribution with the Laplace spike, indeed
$\pi\left(\lambda_0\beta_{i}\C y_i,\gamma_i=0\right) 
$
converges to
$
\frac{1}{2}e^{-|\lambda_0\beta_i|}
$
in total variation (Section \ref{sec:appendix_motivation_inactive} in the Appendix).
When the signal is weak, the posterior thus closely resembles the spike Laplace distribution, as can be seen from Figure \ref{fig: mu_motivation}.  For the fixed (and also random) WBB pseudo-posteriors, we show (in Section \ref{sec:appendix_motivation_inactive} in the Appendix) that % {\color{red} delete: whenever $\E e^{tw_i}$ (or $\E e^{tw_i/w_p}$ for random WBB) is bounded for some $t>0$,}
 the posterior converges to a point mass at 0, i.e.
$
\P_{w_i}\left( \wh\beta_i=0  \C y_i\right)\rightarrow 1.
$
This is a misleading approximation of the actual posterior (Figure \ref{fig1a}). To conclude, since SSL is always shrinking the estimates towards $0$,  WBB samples will often be zero. The true posterior, however, follows roughly a spike  Laplace distribution when the signal is weak.  Motivated by \cite{papandreou2010gaussian}, one possible solution is to introduce randomness in the shrinkage target of the prior.

\iffalse
To conclude, WBB can underestimate variance for inactive coordinates and, since SSL is always shrinking the estimates towards $0$,  WBB samples will be often be zero. The true posterior, however, follows roughly a spike  Laplace distribution when the signal is weak. Motivated by \cite{papandreou2010gaussian}, one possible solution is to introduce some randomness in the shrinkage center in the prior, and thus we introduce the idea of jittered SSL priors.
\fi

%In order to correct for this undesirable behavior we introduce random perturbation also in the prior, but not in terms of weighting but rather in terms of location.

\section{Introducing  BB-SSL}\label{sec:introducing_BB-SSL}
Similarly as \cite{newton2020weighted}, we argue that the random perturbation should affect  both the prior and the data. 
%One intuition is that If we perturb the prior and data suitably, we might be able to sample from the posterior by calculating the maximum a posteri under each perturbation, whose implementation is readily available from [SSL] for the nonseparable spike and slab lasso using coordinate-descent algorithm and its computational complexity grows linearly in $p$. 
%Similar to [WLB], we reweight log-likelihood of each data point, i.e.,  $$\tilde p(x|\bm \beta):=\prod_{i=1}^{n}f(x_i|\bm \beta)^{w_{i}}$$ 
%where the weight vector $$\bm w=(w_{1}, w_{2}, \cdots, w_{n})\sim Dir(\frac{1}{n}, \frac{1}{n}, \cdots,\frac{1}{n})$$  
Instead of inflating the  prior  contribution by a fixed or random weight, we {\sl perturb the prior mean for each coordinate}. This creates a random shift in the centering of the prior so that the posterior can shrink to a random location as opposed to zero. Instead of the prior $\pi(\b\C\bg)$ in \eqref{eq:ss} which is centered around zero, we consider a variant that uses hierarchical jittered Laplace distributions.

\begin{definition} 
For $\lambda_0\gg\lambda_1$,  a location shift vector $\bm \mu=(\mu_1,\mu_2,\cdots,\mu_p)^T\in\R^p$ and a prior inclusion weight $\theta\in(0,1)$,  the {\em jittered Spike-and-Slab LASSO prior} is defined  as
\begin{equation}\label{eq:bbl_prior}
\wt\pi(\b\C\bm\mu,\theta)=\prod_{i=1}^{p}[\theta \psi_1(\beta_i;\mu_i)+(1-\theta)\psi_0(\beta_i;\mu_i)],\quad\text{where} \quad \mu_i\iid  \psi_0(\mu)
\end{equation}
and where  $\psi_1(\beta;\mu)=\lambda_1/2\e^{-|\beta-\mu|\lambda_1}$, $\psi_0(\beta;\mu)=\lambda_0/2\e^{-|\beta-\mu|\lambda_0}$ and $\psi_0(\beta)=\psi_0(\beta;0)$.
%The non-separable variant of the prior is obtained by margining out $\theta$ as$ \wt\pi(\b\C\mu)=\int_{0}^{1} \wt\pi(\b\C\bm\mu,\theta)\d\pi(\theta).$
\end{definition}

{\scriptsize
	\begin{algorithm}[t]\small
		\spacingset{1.1}
%		\KwData{ Data ($Y_i$, $\x_i$) for $1\leq i\leq n, x_i\in \R ^p$}
%		\KwResult{$\tilde{\bm \beta}^t, t=1,2,\cdots,T$}
	\textbf{Set}: $\lambda_0\gg\lambda_1,$  $a,b>0$  and $T$ (number of  iterations).
			\\
		\For{$t=1,2,\cdots, T$}{
			{\vspace{-0.3cm} \hspace{5cm}\color{white} \tiny ahoj}\\
			\textbf{(a)} Sample $\bm w_t\sim\pi(\bm w)$.\\
			\textbf{(b)} Sample $\bm \mu_t$ from $\mu^t_j\iid \psi_0(\mu)$.\\
			\textbf{(c)} Calculate $\wt\b_t$ from  \eqref{BBL}.}
		\caption{\bf : BB-SSL Sampling}\label{BBLASSO}
\end{algorithm}}

The Bayesian Bootstrap Spike-and-Slab LASSO (which we abbreviate as BB-SSL) is obtained by maximizing a pseudo-posterior obtained by reweighting the likelihood and recentering the prior. Namely,
one first samples weights $w_i^t$, one for each observation, from $\bm w_t=(w_1^t,\dots,w_n^t)^T\sim n\,\mathrm{Dirichlet}(\alpha,\cdots,\alpha)$ for some $\alpha>0$ (we discuss choices in the next section). 
Second, one samples the location shifts $\bm \mu_t=(\mu_1^t,\dots,\mu_p^t)^T$ from the spike distribution as in \eqref{eq:bbl_prior}. Lastly, a draw $\wt\b_t$ from the BB-SSL posterior is obtained as a pseudo-MAP estimator 
\begin{equation}\label{BBL}
\wt\b_t=\arg\max_{\b\in\R^p}\left\{  \wt L^{\bm w_t}(\b,\sigma^2; \X^{(n)},\Y^{(n)})\times \int_\theta \wt\pi(\b\C\bm\mu_t,\theta)d\pi(\theta)\right\}.
\end{equation}
BB-SSL can be implemented by directly applying the SSL algorithm we described in the previous section on randomly perturbed data. 
In particular, denote with $\b^*=\b-\bm\mu_t,\, Y_i^*=\sqrt{w_i^t}\left(Y_i-\x_i^T\bm\mu_t\right), \, \x_i^*=\sqrt{w_i^t}\x_i$. We can first calculate
\begin{equation*}
\wh{\b_t^*}=\arg\max_{\b\in\R^p}\left\{  \wt L^{\bm w_t=(1,1,\cdots, 1)}(\b^*,\sigma^2; \X^*,\Y^*)\times \int_\theta \wt\pi(\b\C\bm 0,\theta)d\pi(\theta)\right\},
\end{equation*}
and then get  $\wt\b_t$ through a post-processing step 
$
\wt\b_t=\wh{\b^*_t}+\bm\mu_t.
$

\subsection{Theory for BB-SSL}
For the uniform Dirichlet weights, \cite{newton1994approximate} (Theorem 1 and 2) show first-order correctness, i.e. consistency and asymptotic normality,  of WLB in low dimensional settings (a fixed number of parameters) and iid observations. 
Their result can be generalized to WBB \citep{newton2020weighted, ng2020random}  as well as BB-SSL.
While the uniform Dirichlet weight distribution is a natural choice,  \cite{newton1994approximate}  point out that it is doubtful that  such weights would yield good higher-order approximation properties.
% Indeed, their study reveals that the uniform Dirichlet weights correspond to some specific effective prior which is not a prior in the usual sense and depends also on the data.
The authors  leave open the question of relating the weighting distribution to the model itself and  to a more general prior.  
{ A more recent theoretical development in this direction is the work of \cite{ng2020random} who find that WBB first-order correctness holds for a wide class of random weight distributions  in  low-dimensional LASSO regressions. They also theoretically assess the influence of assigning random weights to the penalty term. } 
Here, we address this question by looking into asymptotics for guidance about the weight distribution. We focus on high-dimensional scenarios where the number of parameters ultimately increases with the sample size.

In particular, we provide  sufficient conditions for the weight distribution $\pi(\bm w)$ so that the pseudo-posterior concentrates at the same rate as the actual posterior under the same prior settings.
After stating  the result for general weight distributions, we particularize our considerations to  Dirichlet and gamma distributions and  provide specific guidance for  implementation.
Our first result is obtained for the  canonical high-dimensional normal-means problem, where  $\Y^{(n)}=(Y_1,\dots, Y_n)^T$ is observed as a noisy version of a sparse mean vector $\b_0=(\beta_1^0,\dots,\beta_n^0)^T$, i.e. 
\begin{equation}\label{eq: normal_means_model}
Y_i=\beta^0_i+\epsilon_i,\quad \text{where}\quad \epsilon_i\iid \mathcal{N}(0,{\sigma^2})\quad \text{for $1\leq i\leq n$}.
\end{equation}

\begin{theorem}[Normal Means]
	\label{normal_means_weight}
	Consider the normal means model \eqref{eq: normal_means_model} with $q=\|\b_0\|_0$ such that   $q=o(n)$ as $n\rightarrow\infty$. Assume the SSL prior with 
	$0<\lambda_1<\frac{1}{e^2}$ and $\theta\asymp (\frac{q}{n})^\eta,\lambda_0\asymp(\frac{n}{q})^\gamma$ with $\eta,\gamma>0$ such that $\eta+\gamma>1$. 
	%Define\footnote{this depends on $j$ if we have different $\|\X_j\|$, do we assume $\|X_j\|=1$?} $\Delta=\inf_{t>0}\left[\|\X_j\|_2^2/2-\rho(t\C\theta)/t\right]=\inf_{t>0}\left[1/2-\rho(t\C\theta)/t\right]$. 
	Assume that $\bm w=(w_1,\dots, w_n)^T$ are non-negative and arise from $\pi(\bm w)$ such that
	\begin{enumerate}
		\item[(1)] $\E\,  w_i =1$ for each $1\leq i\leq n$,
		\item[(2)] $\exists\, C_1, C_2>0$ such that $\E\left(\frac{1}{w_i}\right)\leq C_1$ and $\E\left(\frac{1}{w_i^2}\right)\leq C_2$ for each $1\leq i\leq n$,
		\item[(3)] $\exists\, C_3>0$ such that for each $1\leq i\leq n$ 
		$$\P(w_i>\eta+\gamma)
		\leq C_3\,\frac{q}{n}\,\sqrt{\log\left(\frac{n}{q}\right)}.
		$$
	\end{enumerate}
Then, for any $M_n\rightarrow\infty$, the BB-SSL posterior concentrates at the minimax rate, i.e.
\begin{equation}\label{eq:post_conc}
\lim_{n\rightarrow\infty}\E_{\b_0}\P_{\bm w,\bm\mu}\left[\|\wt\b_{\bm w}^{\bm\mu}-\b_0\|_2^2>M_n\,q\log\left(\frac{n}{q}\right)\C \Y^{(n)}\right]=0.
\end{equation}
\end{theorem}
\proof See Section \ref{sec:appendix_normal_means_proof} in the Appendix.

In Theorem \ref{normal_means_weight},  $\wt\b_{\bm w}^{\bm\mu}$ denotes the BB-SSL sample whose distribution, for each given $\Y^{(n)}$, is induced  by  random weights $\bm w$ arising from $\pi(\bm w)$ and random recentering  $\bm \mu$ arising from $\psi_0(\cdot)$. 
Despite the approximate nature of BB-SSL,  the concentration rate \eqref{eq:post_conc} is {\em minimax optimal} and it is the same rate achieved by the {\em actual} posterior distribution under the same prior assumptions (\cite{rovckova2018bayesian}).
 Condition (1) in Theorem \ref{normal_means_weight} is not surprising and aligns with considerations in \cite{newton2020weighted}. Conditions (2) and (3) can be viewed as regularizing the tail behavior  of $w_i$'s (left and right, respectively).  While \cite{newton2020weighted} only showed consistency for iid models in finite-dimensional settings, Theorem \ref{normal_means_weight} is far stronger as it shows optimal  convergence rate in a high-dimensional scenario. The following Corollary discusses specific choices of $\pi(\bm w)$.

\iffalse
\begin{figure}
	\begin{subfigure}{.5\textwidth}\centering
		\includegraphics[width = \textwidth]{theorem_normal_means.pdf}
		\caption{\scriptsize $t=5$}
	\end{subfigure}%
	\begin{subfigure}{.5\textwidth}\centering
		\includegraphics[width = \textwidth]{theorem_normal_means_t.pdf}
		\caption{\scriptsize $n=1000$}
	\end{subfigure}%
	\caption{\small Plot of the right hand size of condition 3 in \ref{normal_means_weight} under $\gamma=2.5,\eta=2.5,q=\log n$.}
\end{figure}
\fi

\begin{corollary}\label{normal_mean_dirichlet}
	Assume the same model and prior as in Theorem \ref{normal_means_weight}. Next, when
	$\bm w=(w_1,w_2,\cdots,w_n)^T\sim n\times\mathrm{Dir}(\alpha,\alpha,\cdots,\alpha)$ with $\alpha \gtrsim \sigma^2\log[\frac{(1-\theta)\lambda_0}{\theta\lambda_1}]$ %and $\alpha \geq 2$
	or $w_i\overset{\text{i.i.d}}{\sim} \frac{1}{\alpha}Gamma(\alpha,1)$ with $\alpha \gtrsim \sigma^2\log\left[\frac{(1-\theta)\lambda_0}{\theta\lambda_1}\right]$, 
	%and $\alpha \geq 2$, 
	the BB-SSL posterior satisfies \eqref{eq:post_conc}.
\end{corollary}
\proof See the Appendix (Section \ref{sec:appendix:normal_means_corollary_proof}).

Theorem \ref{normal_means_weight} and Corollary \ref{normal_mean_dirichlet} give  insights into what weight distributions are appropriate for sparse normal means.  In parametric models, the uniform Dirichlet distribution would be enough to achieve consistency \citep{newton1994approximate}. It is interesting to note, however, that in the non-parametric normal means model,  the assumption $\bm w\sim n\times \mathrm{Dir}(\alpha,\cdots,\alpha)$ for  $\alpha<2$ yields risk (for active coordinates) that can be arbitrarily large (as we show in Section \ref{sec:appendix_normal_means_remark} in the Appendix). 
The requirement $\alpha\geq 2$ is thus {\em necessary} for controlling the risk of active coordinates and the plain uniform Dirichlet prior (with $\alpha=1$) would not be appropriate.

In the following theorems, we study the high-dimensional regression model \eqref{eq:model} {with rescaled columns $\|\X_j\|_2=\sqrt{n}$ for all $j=1,2,\cdots,p$.} %\footnote{The definition of null consistency condition here is slightly different from classic definitions, as we incorporate perturbations $\m$ and weights $\W$.} 

{\color{black} 
	\begin{theorem}[Regression Model Size]
		\label{thm: regression_weight_selected}
		Consider the regression model \eqref{eq:model}
		with {$p>n$,} $q=\|\b\|_0$ (unknown). Assume the SSL prior with $(1-\theta)/\theta\asymp p^\eta$ and 
		$\lambda_0\asymp p^\gamma$ where $\eta,\gamma\geq 1$.
		Assume that $\bm w=(w_1,\cdots,w_n)^T$ are non-negative and arise from $\pi(\bm w)$ such that
		\begin{enumerate}
			\item[(1)]$\E\,  w_i =1$ for each $1\leq i\leq n$,
			\item[(2)]$\exists\, m\in(0,1)\,\,\text{s.t.}\,\,\lim_{n\rightarrow \infty}\P(\min_{i}w_i>m)=1$,
			\item[(3)]$\exists\, M>1\,\,\text{s.t.}\,\,\lim_{n\rightarrow \infty}\P(\max_{i}w_i<M)= 1$,
			%\item[(5)]Identifiability: there exists constant $\wt\eta\in(0,m)$ such that the $\wt\eta$-NC condition
			\item[(4)]$Var(w_i)\lesssim \frac{1}{\log n},\, Cov(w_i,w_j)=C_0\lesssim \frac{1}{n\log n}$ for any $1\leq i,j\leq n$,
			\item[(5)] %Set $(1+\xi_0)\frac{\sigma}{\wt\eta}\sqrt{n}(1+\sqrt{\frac{5}{2}\log p})\leq \lambda_1\leq 4\sqrt{n\log p}$ where $\wt\eta\in(0,m)$ and $\xi_0>0$ satisfies
			{$\max_{i\neq j}|\x_i^T\x_j|\lesssim \lambda_0^2/n$, and $\xi_0>0$ satisfies }
			\begin{equation*}
			\begin{split}
			\max\{&\lambda_{\text{max}}^{1/2}\left(\X_B^T\bm{P}_A\X_B/n\right):B\cap A=\emptyset, |A|=rank(P_A)=|B|=k, \,
			\\&k(1+\xi_0)^2\left(1+\sqrt{2.5\log p}\right)^2\leq 2n\}\leq\xi_0.
			\end{split}
			\end{equation*}
			\item[(6)]  %$D=\frac{\left(M\eta^*\sqrt{\eta+\gamma}+2\sqrt{2}\right)^2}{mc^2(\eta+\gamma-1)}<1-\delta$ 
			{$D=\frac{M}{m}\left(\frac{\eta^*}{c}+\frac{d\lambda_1}{\sqrt{n\log p}}\right)^2<1-\delta$ }
			for some $\delta>0$, where $\eta^*=\max\left\{\wt\eta+C_n\frac{\|\X\|}{\lambda_1},\frac{\wt\eta}{m}\right\}\in(0,1)$, {$\widetilde \eta>\frac{\sqrt{5}}{2\sqrt{\eta+\gamma}}(1+\xi_0)$}, $C_n$ is a sequence s.t. $C_n\rightarrow \infty$, $d=\frac{\sigma}{c\sqrt{2(\eta+\gamma-1)}}$ and $c=c(\eta^*;\b)$.
		\end{enumerate}
		Then the BB-SSL posterior satisfies
		\[\lim_{n\rightarrow \infty}\E_{\b_0}\P_{\m,\,\bm w}\left(\|\wt\b_{\bm w}^{\m}-\m\|_0\leq q(1+K)\C \Y^{(n)}\right)=1\]
		where $K=2\frac{D}{1-D}$. The definition of $c(\eta^*;\b)$ is in the Appendix, Section \ref{sec:appendix_regression_proof}.
	\end{theorem}
}

\proof  Section \ref{sec:appendix_thm_regression_selected_proof} in the Appendix.

{\color{black} 
	\begin{theorem}[Regression model]
		\label{thm:regression_weight}
		Under the same conditions as in Theorem \ref{thm: regression_weight_selected}, the BB-SSL posterior concentrates at the near-minimax rate, i.e.,
		\begin{equation}\label{eq:thm_regression_weight}
		\lim_{n\rightarrow \infty}\E_{\b_0}\P_{\bm w,\,\m}\left(\|{\wt\b}_{\bm w}^{\m}-\b_0\|_2^2> \frac{C_5(\eta^*)^2{M}}{m\,\phi^2\, c^2}q(1+K)\frac{\log p}{n}\C\Y^{(n)}\right)=0 
		\end{equation}
		where $c=c(\eta^*;\b),\phi=\phi\left(C(\eta^*;\b_0)\right)$, whose definition are in the Appendix \ref{sec:appendix_regression_proof}.
	\end{theorem}
}
\proof 	Section \ref{sec:appendix_thm_regression_proof} in the Appendix.

It follows from \cite{rovckova2018spike} and from \eqref{eq:thm_regression_weight}  that the BB-SSL posterior achieves the same rate of posterior concentration  as the \emph{actual} posterior. In Theorem \ref{thm: regression_weight_selected},  Conditions (1)-(4) regulate  the distribution $\pi(\bm w)$ while Conditions (5) and (6) impose requirements on $\X$, $\lambda_0$ and $\lambda_1$. 
Conditions (2) and (3) are counterparts of  (2) and (3) in Theorem \ref{normal_means_weight}  and control the left and right tail of $w_i$'s, respectively. The larger $M$  (or the smaller $m$) is, the larger $D$ and $K$ will become and the larger the bound on $\|{\wt\b}_{\bm w}^{\m}-\b_0\|_2^2$ will be. Compared with the normal means model,  we have one additional Condition (4) which requires that each $w_i$ becomes more and more concentrated around its mean and that $w_i$'s are asymptotically uncorrelated. It is interesting to note that distributions $\pi(\bm w)$ in Corollary \ref{normal_mean_dirichlet} both satisfy Condition (4). Moreover, the Dirichlet distribution in Corollary \ref{normal_mean_dirichlet} achieves both upper bounds tightly. Finally,  Condition (5)  ensures that identifiability holds with high probability (\cite{zhang2012general}) and Condition (6) ensures that our bound is meaningful.\footnote{ In order for $\eta^*$ to be a bounded real number smaller than 1, we would need $\lambda_1/ \|\X\|\rightarrow \infty$. For example when $p\asymp n$, for random matrix where each element is generated independently by Gaussian distribution, we have $\|\X\|=O_p(\sqrt{n}+\sqrt{p})$ (\cite{vivo2007large}). So in order for such a sequence $C_n$ (s.t. $C_n\rightarrow \infty$) to exist, we need $\lambda_1/ (\sqrt{n}+\sqrt{p})\rightarrow\infty$. We can choose $\lambda_1=(\sqrt{n}+\sqrt{p})\sqrt{\log p}$ and $C_1=\log \frac{\lambda_1}{\|\X\|}$ under such settings.}
In practice, many distributions will satisfy Conditions (1)-(4), e.g. bounded distributions with a proper covariance structure or distributions from Corollary \ref{normal_mean_dirichlet}.

\begin{remark}\label{regression_dirichlet}
	In the regression model, when $\bm w$ arises from the same distribution as in Corollary \ref{normal_mean_dirichlet}, Conditions (1)-(4) in Theorem \ref{thm:regression_weight} are satisfied  by setting $m=\frac{1}{e}$ and $M=\frac{2}{3}(\eta+\gamma)$. The detailed proof is in Section \ref{sec:appendix_regression_corollary_proof} in the Appendix.
\end{remark}

{\scriptsize
	\begin{algorithm}[t]\small
		\spacingset{1.1}
		\KwData{ Data ($Y_i$, $\x_i$) for $1\leq i\leq n, x_i\in \R ^p$, truncation limit $m$}
		\KwResult{$\tilde{\bm \beta}^t, t=1,2,\cdots,T$}
		\For{$t=1,2,\cdots, T$}{
			{\vspace{-0.3cm} \hspace{5cm}\color{white} \tiny ahoj}\\
			\textbf{(a) Draw prior pseudo-samples $\tilde \x_{1:m}, \tilde y_{1:m}\sim F_\pi$}.\\
			\textbf{(b) Draw $(w_{1:n},\tilde w_{1:m})$} $\sim\text{Dir}(1,1,\cdots,1,c/m,c/m,\cdots,c/m)$.\\
			\textbf{(c) Calculate} $\wt\b^t=\arg\max_{\b\in\R^p}\left\{\sum_{j=1}^{n}w_jl(\x_j,y_j,\b)+\sum_{k=1}^{m}\tilde w_kl(\tilde\x_k,\tilde y_k,\b)\right\}$.
		}
		\caption{\bf : Posterior Bootstrap Sampling}\label{alg: Fong}
\end{algorithm}}

\subsection{Connections to Other Bootstrap Approaches}\label{sec:connection}
Our approach bears a resemblance to Bayesian non-parametric learning (NPL) introduced by \cite{lyddon2018nonparametric} and \cite{fong2019scalable} which generates exact posterior samples  under a Bayesian non-parametric model that assumes less about the underlying model  structure. Under a prior on the sampling distribution function $F_\pi$, one can   use WBB (and also WLB)  to draw  samples from a posterior of $F_\pi$ by optimizing a randomly weighted loss function $l(\cdot)$ based on an enlarged sample (observed plus pseudo-samples) with weights following a Dirichlet distribution (see Algorithm \ref{alg: Fong} which follows from \cite{fong2019scalable}).   Despite the fact that these two procedures have different objectives, there are many interesting connections. In particular, the idea of randomly perturbing the prior has an effect similar to adding pseudo-samples $\wt \x_{1:m},\wt y_{1:m}$ from the prior $F_\pi(\x, y)$  defined through%\footnote{Why not $\wt y_k\C\wt \x_k=Y_k+\x_k^T\bm\mu$? Also, let's remind what $\bm\mu$ is (i.e. the spike)}
$$
\wt\x_k\sim  F_n(\x)=\frac{1}{n}\sum_{i=1}^{n}\delta(\x_i),\quad \wt y_k\C\wt \x_k=\wh y_k+\wt \x_k^T\bm\mu
$$
where $\delta(\cdot)$ is the Dirac measure, $\m$ is the Spike, and $\wh y_k=y_i$ where $i$ satisfies $\wt \x_k=\x_i$. A motivation for this prior is derived in the Appendix (Section \ref{sec:appendix_connection_explanation}). Under this prior, the NPL posterior samples $\wt\b^t$ generated by Algorithm \ref{alg: Fong} approximately follow the distribution  (see Section \ref{sec:appendix_connection_explanation} in the Appendix)
\begin{equation}\label{eq: fong}
\wt\b^t\text{\ensuremath{\overset{\text{d}}{\approx}}}\arg\max_{\wt\b\in\R^p}\left\{-\frac{1}{2}\sum_{i=1}^nw_i^*(Y_i-\x_i^T\tilde\b)^2+ \log\left[\int \prod_{j=1}^p\pi\left(\tilde\beta_j-\frac{c}{c+n}\mu_j^*\C\theta\right)d\pi(\theta)\right]
\right\}-\frac{c}{c+n}\bm\mu^*
\end{equation}
where  $(w_1^*,w_2^*,\cdots,w_n^*)^T\sim n\times \text{Dir}(1+c/n, \cdots,1+c/n)$, {each coordinate of $\bm\mu^*$ independently follows the spike distribution, and where} $c$ represents the strength of our belief in $F_\pi$ and can be interpreted as the effective sample size from $F_\pi$. 
In comparison with the BB-SSL estimate
\begin{equation}\label{eq:BB_SSL}
\wt\b^t=\arg\max_{\wt\b\in\R^p}\left\{-\frac{1}{2}\sum_{i=1}^nw_i(Y_i-\x_i^T\wt\b)^2+ \log \left[\int_\theta\prod_{j=1}^p\pi\left(\tilde\beta_j-\mu_j\C\theta\right)d\pi(\theta)\right]
\right\}
\end{equation}
where $(w_1,w_2,\cdots,w_n)^T\sim n\times\text{Dir}(\alpha, \cdots,\alpha)$,
both \eqref{eq: fong} and \eqref{eq:BB_SSL} are shrinking towards a random location and both are using Dirichlet weights. The main difference is  in the choice of  the concentration parameter $c$. When $c=0$,  \eqref{eq: fong} reduces to WBB (with a fixed weight on the prior) which reflects less confidence in the prior $F_\pi$ and thus less prior perturbation (location shift).  
When $c$ is large, \eqref{eq: fong} becomes more similar to \eqref{eq:BB_SSL} where the prior $F_\pi$ is stronger and thereby more prior perturbation is induced.  
Another difference is that  \eqref{eq: fong}, although shrinking towards a random location $\frac{c}{c+n}\mu_j^*$, adds the location back  which results in less variance (see Figures in Section  \ref{sec:appendix_connection_explanation} in the Appendix). 
%(2) the role of concentration paramerters are different. In Equation (\ref{eq: fong}), we can control the concentration parameter of Dirichlet distribution by tuning $c$. The larger $c$ is, the more weight we put on the prior and also the more close to 1 our weights $w_i^*$ will become. In BB-SSL, we put constant weight on prior and the tuning parameter $\alpha$ only controls the degree of concentration of our weights.

\begin{table}[!t]
	\small
	\centering
	\scalebox{0.7}{\begin{tabular}{l|l}
			\hline\hline
			\bf Algorithm       & \bf Complexity \\
			\hline \bf SSVS1          & $O(p^2n)$ %($O(p^2n)$ if using Woodbury matrix identity)   
			\\
			\hline \bf SSVS2  & $O(n^2\max(p,n))$ \\
			\hline \bf Skinny Gibbs & $O(np)$ \\
			\hline \bf WLB           &  $O(np^2)$ when $n\leq p$, not applicable when $p>n$    \\
			\hline \bf BB-SSL & $O\left(\min\left(\text{maxiter}\times p(n+\frac{p}{c_1}),\, (n+\text{maxiter})\times p^2\right)\right)$ for a single value $\lambda_{0}$\\ %and $O(L\times\text{maxiter}\times p(n+\frac{p}{c_1}))$ for sequence of $\lambda_{0}$\\
			\hline\hline
	\end{tabular}}
	\caption{\small A computational complexity analysis (per sample) of each algorithm. \texttt{Maxiter} is the user-specified maximum number of iterations with a default value  $500$. For BB-SSL, $c_1$ is the pre-specified number of iterations after which  $\theta$  is updated with a default value $10$. By setting $c_1\propto p$ we have BB-SSL complexity $O(np)$. }
	\label{table: computation}
\end{table}

\section{Simulations}\label{sec:simulations}

We compare the empirical performance of our BB-SSL with several existing posterior sampling methods including WBB (\cite{newton2020weighted}), SSVS (\cite{george1993variable}), and Skinny Gibbs (\cite{narisetty2019skinny}). We implement two versions of WBB: WBB1 (with a fixed prior weight) and WBB2 (with a random prior weight). We also implement the original SSVS  algorithm (Algorithm 1 further referred to as SSVS1) and compare its complexity and running times with its faster version (further referred to as SSVS2) which uses the trick from   \cite{bhattacharya2016fast}.   
Comparisons are based on the marginal posterior distributions for $\beta_i$'s, marginal inclusion probabilities (MIP) $\P(\gamma_i=1\C \Y^{(n)})$ as well as the joint posterior distribution $\pi(\bg\C\Y^{(n)})$. As the benchmark gold standard for comparisons, we run SSVS initialized at the truth for a sufficiently large number of  iterations $T$  and discard the first $B$ samples as a burn-in. We use the same $T$ and $B$ for Skinny Gibbs except that we initialize $\b$ at the origin. For BB-SSL, we draw weights $\bm w\sim n\times\text{Dir}(\alpha,\cdots,\alpha)$ where $\alpha$ depends on $(n,p,\sigma^2)^T$. {When solving the optimization problem \eqref{BBL} using coordinate-ascent, the default initialization for $\bm\beta$ in the \texttt{SSLASSO} R package \citep{rockova2019package} is at the origin. In high-dimensional correlated settings when $\lambda_0>>\lambda_1$, however,  the performance of BB-SSL can be further enhanced by using a warm start re-initialization strategy for a sequence of increasing $\lambda_0$'s  where the last value is the target $\lambda_0$ value (as recommended by \cite{rovckova2018bayesian} and  \cite{rockova2019package}). It is computationally more economical  to perform such annealing   only once on the original data and then use the output (for the target value $\lambda_0$) for each BB-SSL iteration. We apply this  strategy  using an output  obtained from the R package \texttt{SSLASSO} using an equispaced sequence of $\lambda_0$'s of length $50$, starting at $\lambda_1$ and ending at $\lambda_0$.}
We then run  WWB1, WBB2 and BB-SSL  for $T$ iterations. Throughout the simulations we set $\sigma^2=1$ and assume the prior $\theta \sim B(1,p)$.  
{ Computational complexity of each algorithm is summarized in Table \ref{table: computation} with actual running times  reported (for varying $p$ and $n$)   in Figure \ref{fig:comp_times}.}

\begin{figure}[!t]\centering
	\begin{subfigure}{.4\textwidth}\centering
		\includegraphics[width=.97\textwidth]{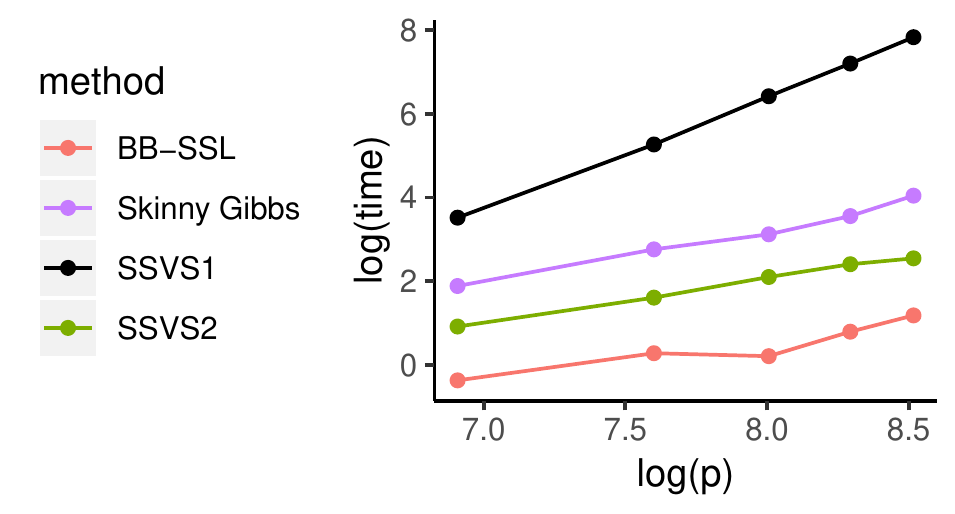}
		\caption{\small $n=100$. }
	\end{subfigure}
\begin{subfigure}{.4\textwidth}\centering
	\includegraphics[width=.7\textwidth]{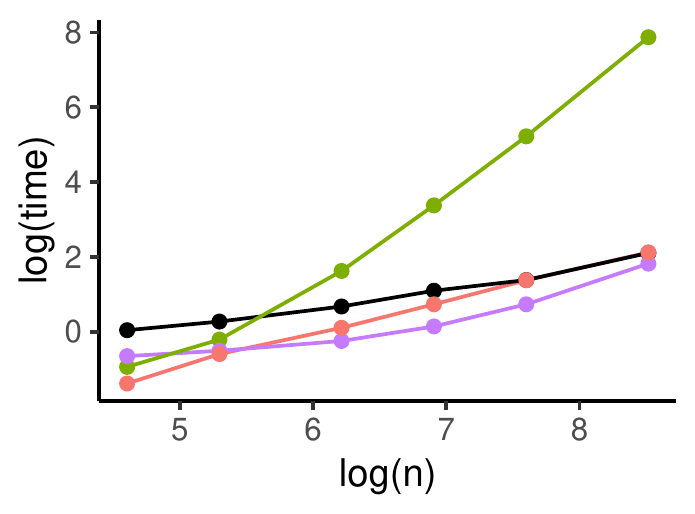}
	\caption{\small $p=100$.   }
\end{subfigure}
	\caption{\small Average running times (in seconds on a log scale) for $100$ iterations of each algorithm based $10$ independent runs.  Covariates are all correlated with a correlation coefficient $\rho=0.6$. Signals are $(2,3,-3,4)$ and we set $\lambda_0=200, \lambda_1=0.05, a=1,b=p$. BB-SSL is initialized at the SS-LASSO solution obtained with a sequence of $\lambda_0$'s (an equi-spaced series of length $50$ starting at $\lambda_1=0.05$ and ending at 
	$\lambda_0=200$). 
	}\label{fig:comp_times}
\end{figure}

\subsection{The Low-dimensional Case}
Similarly to the experimental setting in \cite{rovckova2018particle}, we generate $n=50$ observations on $p=12$ predictors with $\b_0 = (1.3,0,0,1.3,0,0,1.3,0,0,1.3,0,0)^T$, where the predictors have been grouped into $4$ blocks.  Within each block, predictors have an equal correlation {$\rho$} and there is only one active predictor. All the other correlations are set to 0. We choose a single value  for $\lambda_0\propto p$ and generate Dirichlet weights assuming $\alpha=1$ (for WBB1, WBB2 and BB-SSL). 

\paragraph{Uncorrelated Designs}
Assuming $\rho=0, \lambda_0=12$ and $\lambda_1=0.05$ we run   SSVS1 and Skinny Gibbs for $T=10\,000$ iterations with a burn-in $B=5\,000$. For WBB1, WBB2 and BB-SSL we use $T= 5\,000$ iterations. 
All methods perform very well under various metrics in this setting. We refer the reader to Section \ref{sec:simu_low_ind} in the Appendix for details. 

\paragraph{Correlated Designs}
 \begin{sidewaysfigure}
	\begin{subfigure}{0.5\hsize}\centering
		\includegraphics[width=\hsize, height=5in]{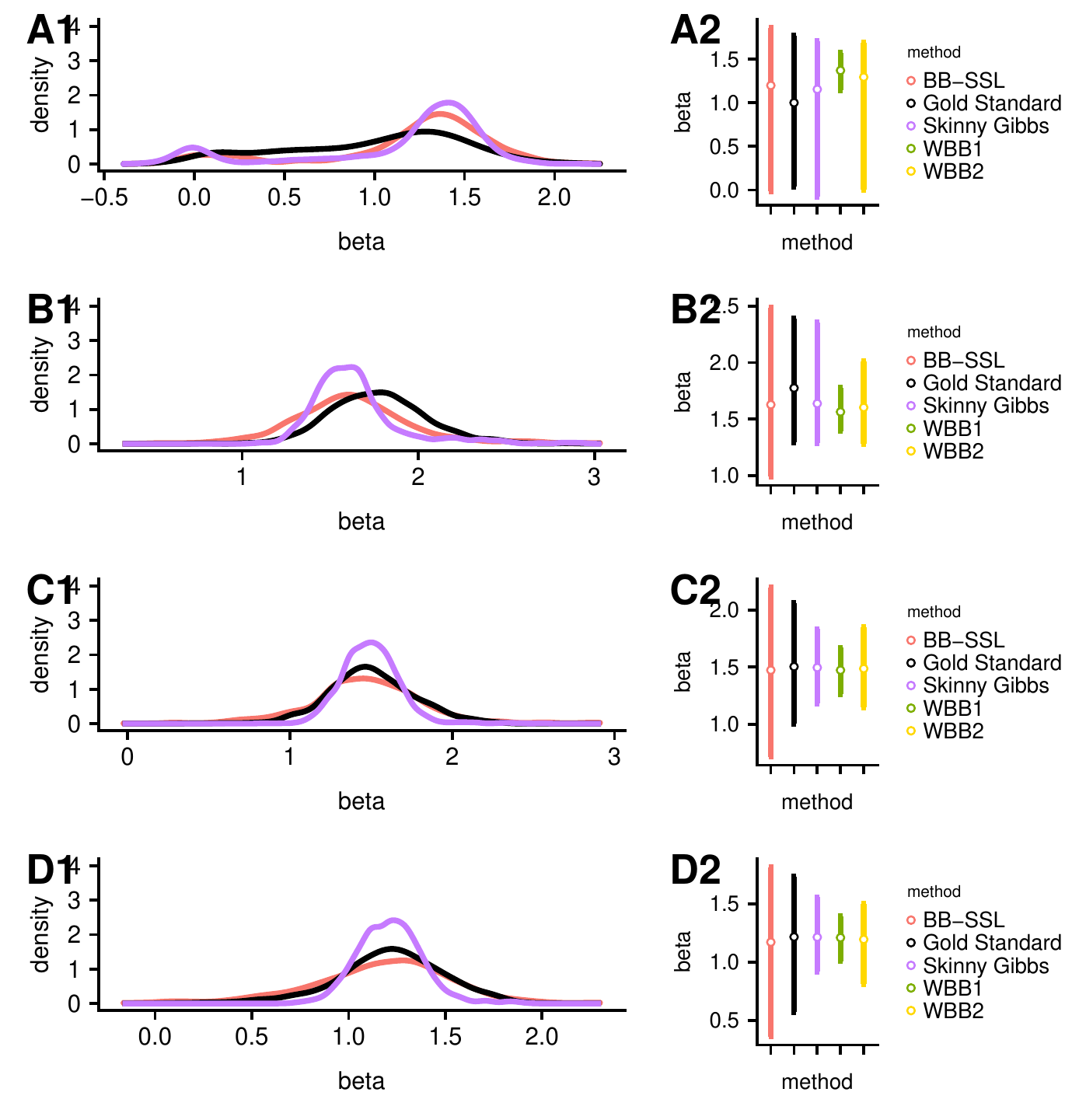}
		\caption{\small Active predictors, from top to bottom: $\beta_1,\beta_{4},\beta_{7},\beta_{10}$}
	\end{subfigure}%
	%\hfill <-- it is superfluous 
	\begin{subfigure}{0.5\hsize}\centering
		\includegraphics[width=\hsize, height=5in]{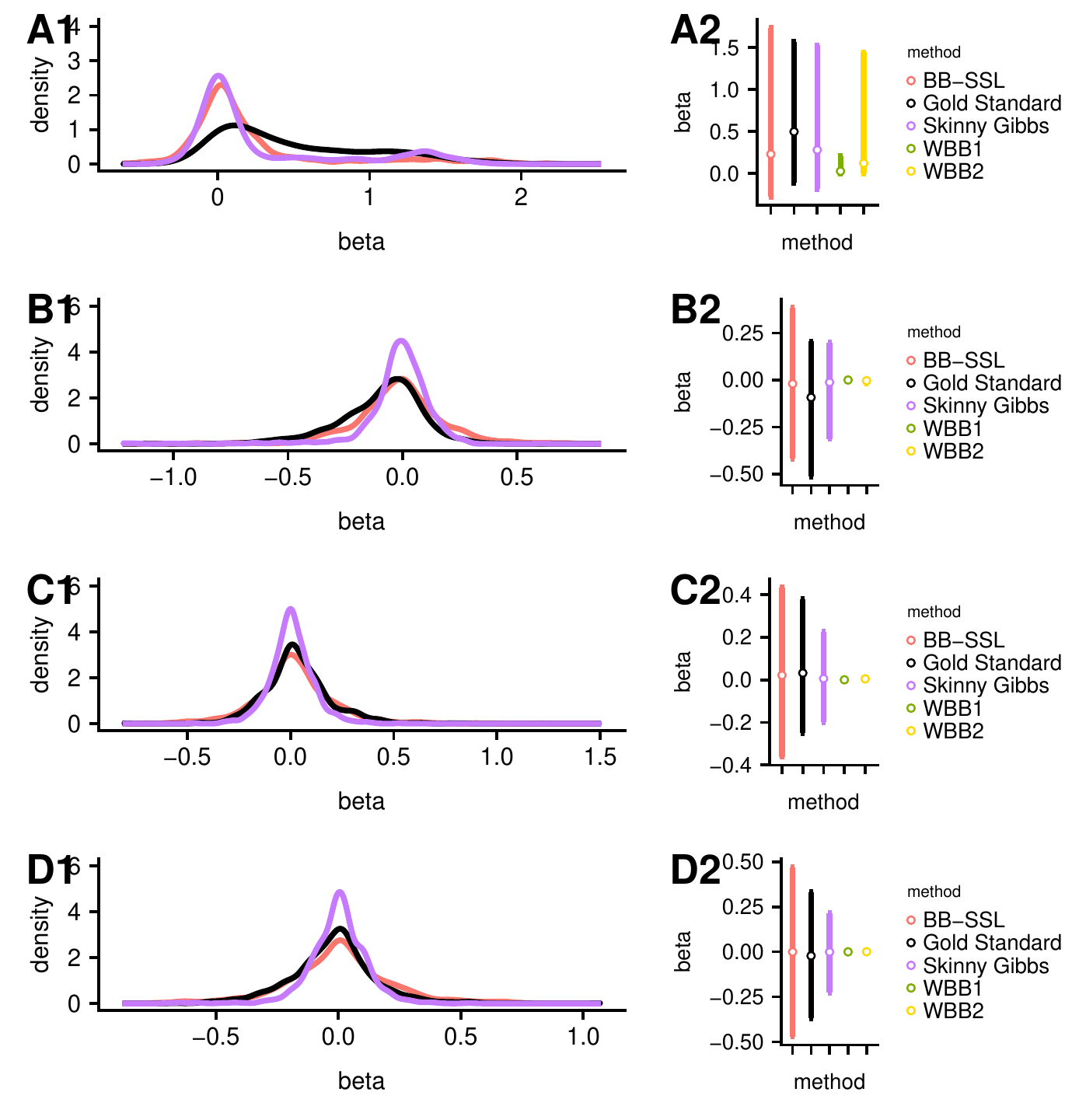}
		\caption{\small Inactive predictors, from top to bottom: $\beta_2,\beta_{5},\beta_{8},\beta_{11}$}
	\end{subfigure}
	\caption{\small Estimated posterior density (left panel) and credible intervals (right panel) of $\beta_i$'s in the low-dimensional correlated case. We have $n=50, p=12, \beta_{active}=(1.3,1.3,1.3,1.3)',\lambda_0=7,\lambda_1=0.15, \rho=0.9$. Each method has %\footnote{do we need thinning for BB-SSL?} 
	$5\,000$ sample points (after thinning for SSVS and Skinny Gibbs). BB-SSL is fitted using a single value $\lambda_0=7$. Since WBB1 and WBB2 produce a point mass at zero, we exclude them from density comparisons. }
	\label{fig:low_cor_beta}
\end{sidewaysfigure}

\begin{figure}
\begin{subfigure}{\textwidth}\centering
	\includegraphics[width=.85\textwidth,height=.15\textheight]{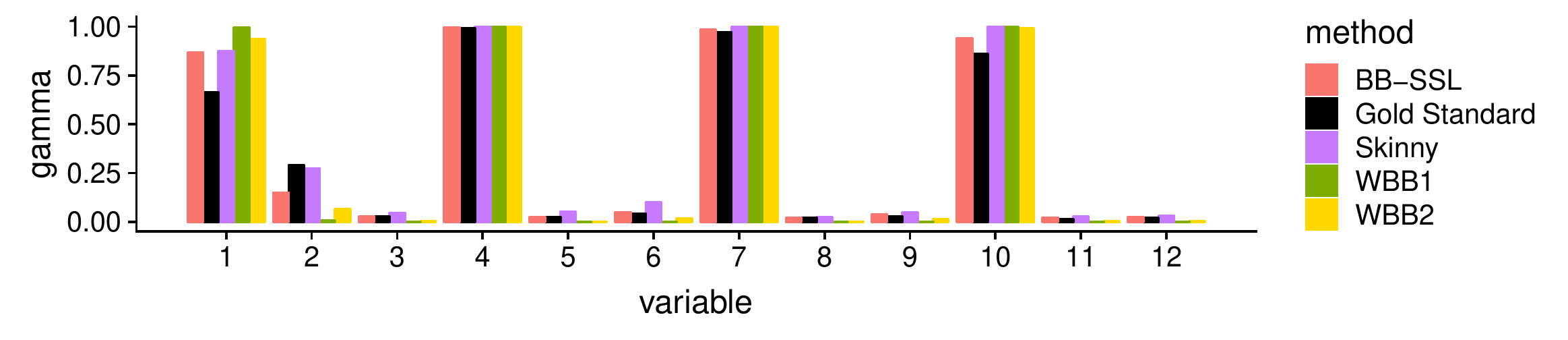}
	\caption{\small Marginal inclusion probability. }
	\label{fig:low_cor_gamma}
\end{subfigure}

\begin{subfigure}{\textwidth}
	\includegraphics[width=\textwidth, height=.18\textheight]{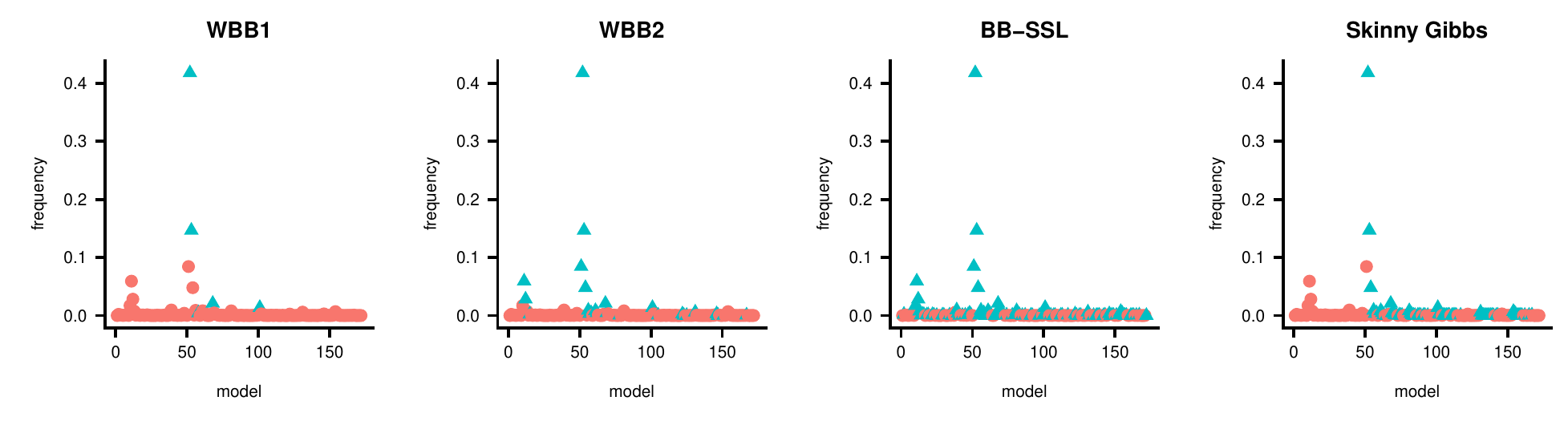}
	\caption{\small Posterior exploration plot. Blue (red) triangles are visited (unvisited) models.}
	\label{fig:low_cor_PE}
\end{subfigure}
	\caption{\small The low-dimensional correlated case with $n=50, p=12, \beta_{active}=(1.3,1.3,1.3,1.3)'$ where predictors are grouped into $4$ correlated blocks with $\rho=0.9$. %Each method has 20,000 sample points. %SSLASSO is fitted using a single value $\lambda_0=7$. 
	We choose $\lambda_{0}=7,\lambda_1=0.15$. %Each method has 20,000 sample points. SSLASSO is fitted using a single value $\lambda_0=7$. 
	}
\end{figure}
Correlated designs are far more interesting for comparisons. 
We choose $\rho=0.9, \lambda_0=7$ and $\lambda_1=0.15$ to deliberately encourage multimodality in the model posterior (see Figure \ref{fig:low_cor_PE}).  For SSVS1 and Skinny Gibbs, we set $T=100\,000$ and $B=5\,000$. For WBB1, WBB2 and BB-SSL we set $T=95\,000$. %We use SSVS initialized at true $\b^0$ with 100,000 iterations and 5,000 burn-in as gold standard. We run fixed WBB, unfixed WBB, WLB, Skinny Gibbs and BB-SSL for 100,000 iterations with SSL fitted using a single value of $\lambda_{0}$. We run \cite{bhattacharya2016fast} for 100,000 iterations initialized at LASSO solution and discard the first 5000 as burn-in. 

In terms of the marginal densities of $\beta_i$'s, Figure \ref{fig:low_cor_beta} shows that BB-SSL tracks SSVS1 very closely. All methods can cope with multi-collinearity where   BB-SSL tends to have slightly longer credible intervals with the opposite being true for Skinny Gibbs, WBB1 and WWB2.
%This might result from the fact that Skinny Gibbs uses an orthogonal approximation when doing sampling, which introduces bias in this extreme setting. %But it is worth mentioning that Skinny Gibbs does well in detecting multi-modality under this setting. 
%In terms of the marginal density of $\beta_i$, we see from figure \ref{fig:low_cor_beta} that results from \cite{bhattacharya2016fast} is in general close to SSVS (because it is essentially the same as SSVS), but we do observe some difference from SSVS in $\beta_{10}$ and $\beta_{11}$, which might be due to the different initializations we use. In general BB-SSL is also close to SSVS and can detect multi-modality, but tends to have slightly larger credible interval covering two consecutive modes. We see in the independent setting that Skinny Gibbs performs pretty well, but in this highly correlated setting $\rho=0.9$, we observe that Skinny Gibbs tends to be biased for active coordinates. This might result from the fact that Skinny Gibbs uses an orthogonal approximation when doing sampling, which introduces bias in this extreme setting. But it is worth mentioning that Skinny Gibbs does well in detecting 
In terms of the marginal means of $\gamma_i$ (Figure \ref{fig:low_cor_gamma}) all methods  perform well, where the median probability model rule (truncating the marginal means at 0.5) yields the true model.
In terms of the overall posterior $\pi(\bg\C\Y^{(n)})$, we identify over 60 unique models using SSVS1 where the true model accounts for most of the posterior mass. In Figure \ref{fig:low_cor_PE}, we show the visited (blue triangle) and not visited (red dots) among these models, where $y$-axis represents the estimated posterior probability for each model (calculated from SSVS1). All methods  can detect the dominating model. BB-SSL tracked down 99\% of the posterior probability, followed by WBB1 (92\%), WBB2 (91\%) and  Skinny Gibbs (73\%). %Altogether, the models detected by fixed WBB accounts for 77.94$\%$ of posterior frequency, unfixed WBB accounts for 93.55$\%$, BB-SSL 96.46$\%$, MCMC with \cite{bhattacharya2016fast} 96.39$\%$ and Skinny Gibbs 0.92$\%$\footnote.
{The average times (reported in seconds and ordered from fastest to slowest) spent on generating $1\,000$ effective samples for $\beta_j$'s are WBB2 (0.68 s) $<$ WBB1 (0.72 s)  $<$ BB-SSL (0.74 s) $<$ SSVS2 (0.82 s) $<$ SSVS1 (0.85 s)  $<$ Skinny Gibbs (1.19 s). }
%{\color{blue} Lizhen, do you think we could produce Table 2 for the low-dimensional setting and put it in Appendix E.1? The referee wanted to see more approximation metrics.}

%In summary, under this correlated low-dimensional setting, BB-SSL can efficiently discover multi-modality while tends to produce slightly more variance.  Skinny Gibbs does well in terms of marginal inclusion probability but produces bias for marginal distribution of $\beta_i$ and detects few models except the dominating one. \cite{bhattacharya2016fast} performs well. WLB posterior is unimodal and has large bias. Fixed and unfixed WBB underestimates the variance for inactive predictors.

\subsection{The High-dimensional Case}

%We set $n=100$ and $p=1\,000$, where the active predictors have regression coefficients $(1,2,-2,3)^T$. Predictors are grouped into blocks of size $10$, where each group has exactly one active coordinate and where predictors have a within-group correlation $\rho$. We set $\lambda_0=50,\lambda_1=0.05$ and  $\alpha=2$ for BB-SSL. For SSVS1 and Skinny Gibbs we set $T=15\,000$ and $B=5\,000$ while for WBB1, WBB2 and BB-SSL we set $T=1\,000$.

We now consider a higher-dimensional case with $n=100$ and $p=1\,000$, assuming  $\lambda_0=50,\lambda_1=0.05$ and  $\alpha=2$ for BB-SSL. For SSVS1 and Skinny Gibbs we set $T=15\,000$ and $B=5\,000$ while for WBB1, WBB2 and BB-SSL we set $T=1\,000$.
%\paragraph{Independent case $\rho=0$}
\begin{figure}
	\centering
	\begin{subfigure}{.35\textwidth}
		\centering
		\includegraphics[width=.85\linewidth,height=1.65in]{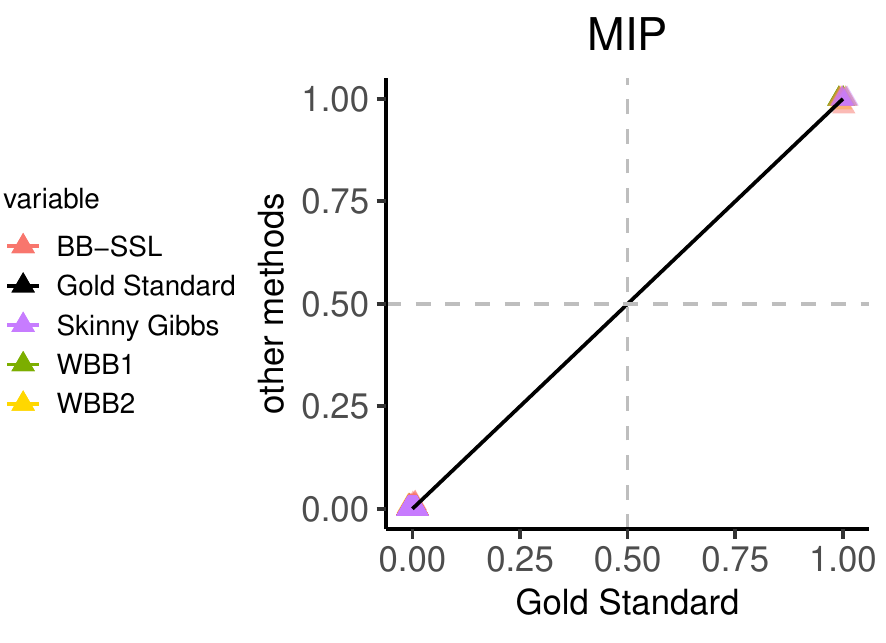}
		\caption{\small $\rho=0$}
		\label{fig:high_dim_gamma_ind}
	\end{subfigure}%
\iffalse
	\begin{subfigure}{.3\textwidth}
		\centering
		\includegraphics[width=.85\linewidth,height=1.65in]{High_dim_cor_pointsix_gamma_updated_1.pdf}
		\caption{\small $\rho=0.6$}
		\label{fig:high_dim_gamma_6}
	\end{subfigure}%
\fi
	\begin{subfigure}{.28\textwidth}
		\centering
		\includegraphics[width=.85\linewidth,height=1.65in]{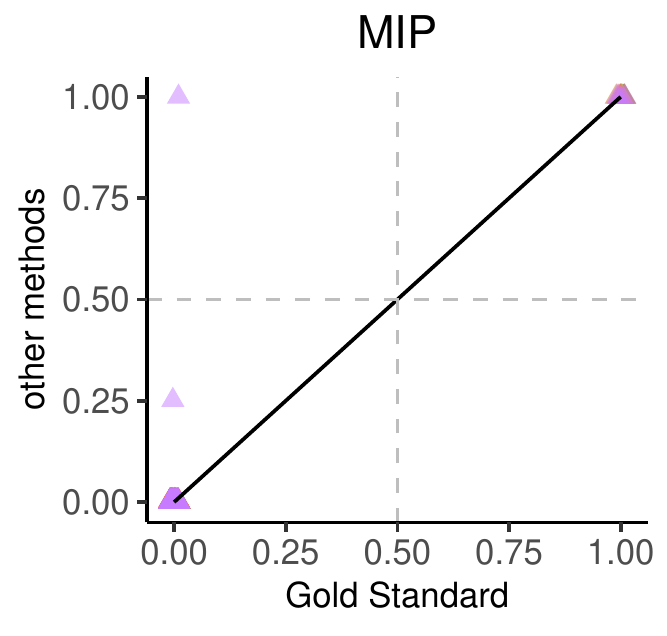}
		\caption{\small block-wise $\rho=0.9$}
		\label{fig:high_dim_gamma_9}
	\end{subfigure}
\begin{subfigure}{.28\textwidth}
	\centering
	\includegraphics[width=.85\linewidth,height=1.65in]{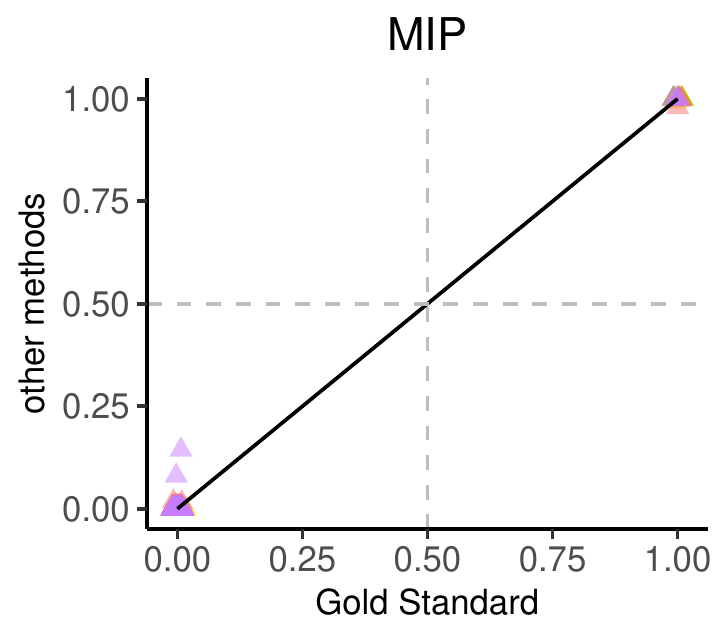}
	\caption{\small equi-correlation $\rho=0.6$}
	\label{fig:allcorrelated_gamma_6}
\end{subfigure}
	\caption{\small Posterior means of $\gamma_i$'s (i.e. a marginal inclusion probabilities) in high-dimensional settings with $n=100, p=1000$. We set $\lambda_0=50,\lambda_1=0.05$. }
	\label{fig:high_dim_gamma}
\end{figure}
{
We consider two correlation structures: (a) block-wise correlation, and (b) equi-correlation. 
In the setting (a), the  active predictors have regression coefficients $(1,2,-2,3)^T$ and all predictors are grouped into blocks of size $10$, where each group has exactly one active coordinate and where predictors have a within-group correlation $\rho$. We consider $\rho\in\{0,0.6,0.9\}$ and an extreme case  $\rho=0.99$ with a larger signal $(2,4,-4,6)^T$. 
For the equi-correlation setting (with a correlation coefficient $\rho$),  active predictors have regression coefficients $(2,3,-3,4)^T$. We consider $\rho\in\{0.6,0.9\}$.
}

{For brevity, we only show results for $\rho=0.6$  in the equi-correlation setting with the rest postponed until the Appendix (Section \ref{sec_appendix:simulation}). } 
{In the setting (b)  with $\rho=0.6$}, in terms of the marginal density of $\beta_i$'s (shown in Figure \ref{fig:r6}), Skinny Gibbs tends to underestimate the variance for active coordinates and WBB1 and WBB2 produce a point mass at $0$ for inactive coordinates. BB-SSL, on the other hand,  fares very well. Figure \ref{fig:high_dim_gamma} shows that BB-SSL, WBB1 and WBB2 accurately reproduce the MIPs, while Skinny Gibbs tends to slightly overestimate the MIP as $\rho$ increases.
%{\color{red}The same observations also hold in the more challenging equi-correlation structure with $\rho=0.6$ (Figure \ref{fig:high_dim_gamma} and Figure \ref{fig:r6}).}

\iffalse
\begin{sidewaysfigure}
	\begin{subfigure}{0.5\hsize}\centering
		%\includegraphics[width=0.9\hsize, height=5in]{High_dim_cor_pointsix_active_updated_all_1}
		\includegraphics[width=0.9\hsize, height=5in]{High_dim_cor_pointnine_active_updated_1}
		\caption{\small Active predictors, from top to bottom: $\beta_1,\beta_{11},\beta_{21},\beta_{31}$}
	\end{subfigure}%
	%\hfill <-- it is superfluous 
	\begin{subfigure}{0.5\hsize}\centering
		%\includegraphics[width=0.9\hsize, height=5in]{High_dim_cor_pointsix_inactive_updated_all_1}
		\includegraphics[width=0.9\hsize, height=5in]{High_dim_cor_pointnine_inactive_updated_1}
		\caption{\small Inactive predictors, from top to bottom: $\beta_2,\beta_{12},\beta_{22},\beta_{32}$}
	\end{subfigure}
	\caption{\small  Estimated posterior density (left panel) and credible intervals (right panel) of $\beta_i$'s in the high-dimensional, block-wise correlated case ($\rho=0.9$). We have $n=100, p=1000, \beta_{active}=(1,2,-2,3)',\lambda_0=50,\lambda_1=0.05$. Each method has %\footnote{do we need thinning for BB-SSL?} 
	$5\,000$ sample points (after thinning for SSVS and Skinny Gibbs). BB-SSL is fitted using a single $\lambda_0$ and initialized at \texttt{SSLASSO} solution on the original $\X,\y$. Since WBB1 and WBB2 produce a point mass at zero, we exclude them from density comparisons.}
	\label{fig:high_moderate_beta}
\end{sidewaysfigure}
\fi

\begin{sidewaysfigure}
	\begin{subfigure}{0.5\hsize}\centering
		\includegraphics[width=\hsize, height=5in]{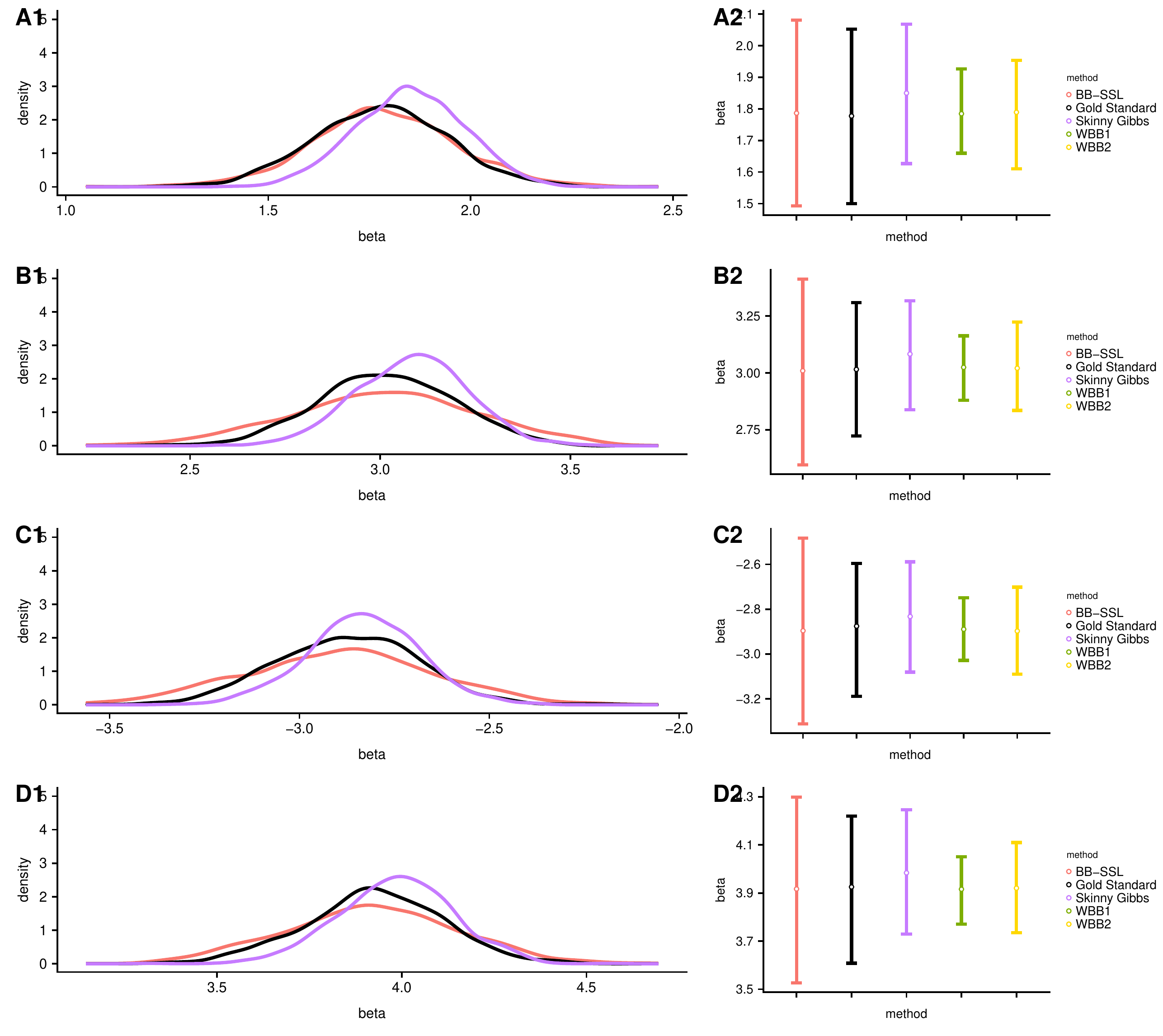}
		\caption{\small Active predictors, from top to bottom: $\beta_1,\beta_{2},\beta_{3},\beta_{4}$}
	\end{subfigure}%
	%\hfill <-- it is superfluous 
	\begin{subfigure}{0.5\hsize}\centering
		\includegraphics[width=\hsize, height=5in]{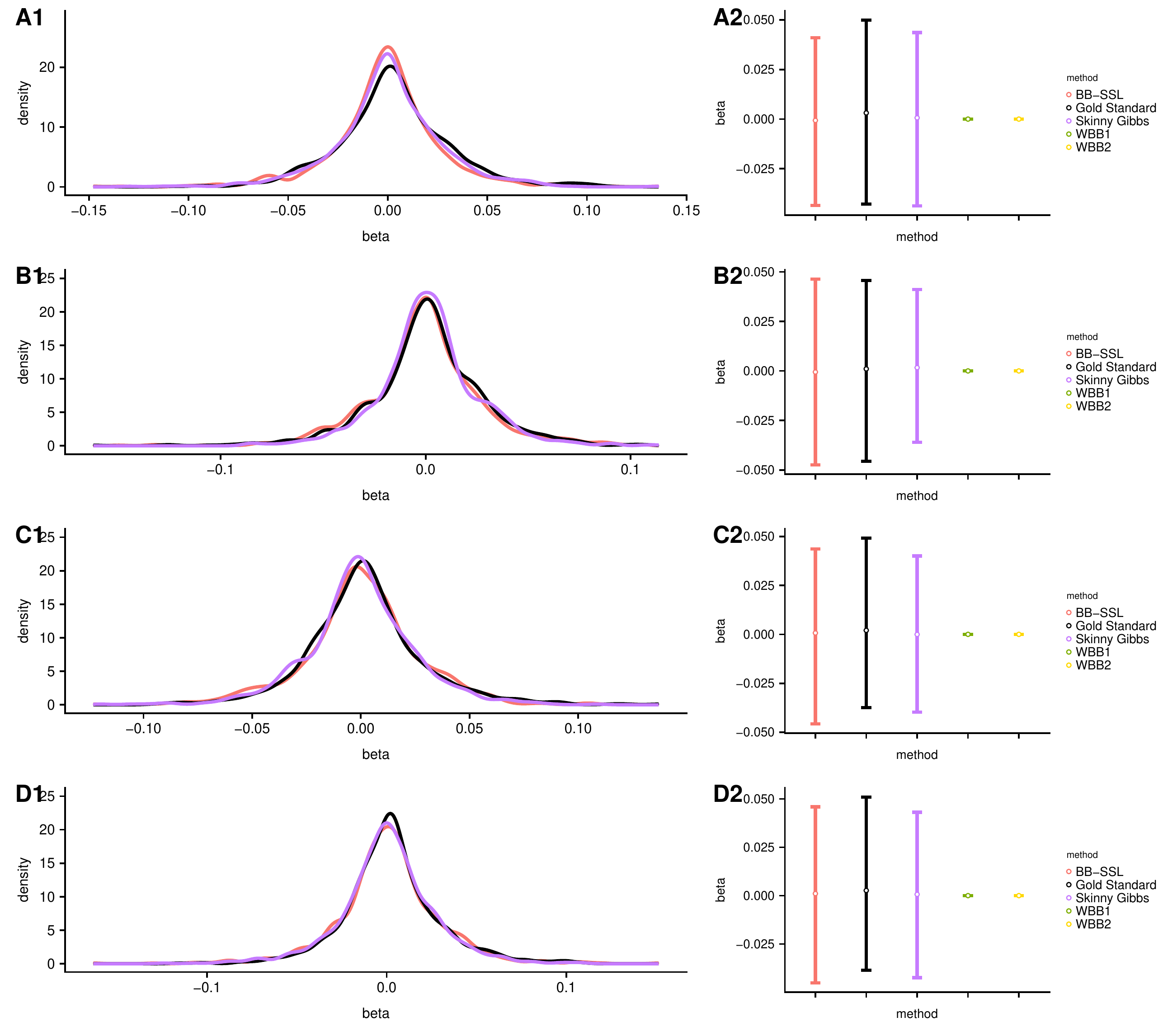}
		\caption{\small Inactive predictors, from top to bottom: $\beta_5,\beta_{6},\beta_{7},\beta_{8}$}
	\end{subfigure}
	\caption{\small  Estimated posterior density (left panel) and $90\%$ credible intervals (right panel) of $\beta_i$'s when all covariatess are correlated with $\rho=0.6$. We have $n=100, p=1000, \beta_{active}=(2,3,-3,4)',\lambda_0=50,\lambda_1=0.05$. BB-SSL is fitted using a single $\lambda_0$ and initialized at \texttt{SSLASSO} solution on the original $\X,\y$. Since WBB1 and WBB2 produce a point mass at zero, we exclude them from density comparisons.
	\label{fig:r6}
	}
\end{sidewaysfigure}

\begin{table}
\begin{subtable}[!t]{\textwidth}
	\centering
	\scalebox{0.58}{
		\begin{tabular}{l|l|l|l|l|l|l|l|l|l||l|l|l|l|l|l|l|l|l} 
			\hline\hline
			\bf \large Setting&\multicolumn{9}{l||}{  \cellcolor[gray]{0.8}\large \bf Block-wise $\rho=0.9$, $\beta_{active}=(1,2,-2,3)'$}   &\multicolumn{9}{l}{  \large\cellcolor[gray]{0.8} \bf Block-wise $\rho=0.99$, $\beta_{active}=(2,4,-4,6)'$}   \\\hline
			\multirow{3}{*}{Metric}& \multicolumn{6}{l|}{\quad\quad\quad\quad\quad\quad\quad\quad$\beta_j$'s}                                                                                         & \multicolumn{2}{l|}{\quad$\gamma_j$'s }              & Model 
			& \multicolumn{6}{l|}{\quad\quad\quad\quad\quad\quad\quad\quad$\beta_j$'s}                                                                                         & \multicolumn{2}{l|}{\quad$\gamma_j$'s }              & Model 
			\\ \cline{2-19}
			Metric& \multicolumn{2}{l|}{KL} & \multicolumn{2}{l|}{JD of 90\% CI} & \multicolumn{2}{l|}{`Bias'} & \multicolumn{2}{l|}{`Bias'}    & HD &
			\multicolumn{2}{l|}{KL } & \multicolumn{2}{l|}{JD of 90\% CI} & \multicolumn{2}{l|}{`Bias'} & \multicolumn{2}{l|}{`Bias'}    & HD 
			\\ \cline{2-19}
			& $+$ & $-$           & $+$ & $-$         & $+$ & $-$       & $+$ & $-$        & all    & $+$ & $-$   & $+$ & $-$       & $+$ & $-$      & $+$ & $-$         & all    \\ 
			\hline
			Skinny Gibbs 	& 0.19   & 0.009                      & 0.30   & \bf{0.10}                                    & 0.04  & \bf{0.003}            & *&*& \bf{0}   
			& 2.00   & 0.02                      & 0.62   & \bf{0.10}                                    & \bf{0.68}  & \bf{0.005}            & \bf{0.15}&0.002& 2.2    \\ 
			WBB1         & 0.41   & 3.09                      & 0.45   & 1                                       & \bf{0.03}  & \bf{0.003}            & *&* & \bf{0}    
			& 1.89  & 3.09                      & 0.51   & 1                                       & {0.74}  & {0.006}            & 0.25&\bf{0.001} & \bf{2}      \\ 
			WBB2       & 0.21   & 3.09                      & 0.37   & 1                                       & \bf{0.03}  & \bf{0.003}            &*&*& \bf{0}  
			& 1.89  & 3.09                      & 0.52   & 1                                       & {0.74}  & {0.006}            &0.25&\bf{0.001}& \bf{2}   \\ 
			BB-SSL     & \bf{0.02}   & \bf{0.003}                     & \bf{0.14}   & \bf{0.10}                                    & 0.04  & \bf{0.003}            &*&*& \bf{0}   
			 & \bf{1.73}   & \bf{0.01}                     & \bf{0.37}   & \bf{0.10}                                    & {0.74}  & {0.006}            &0.25&\bf{0.001}& \bf{2}   
			\\
			\hline\hline
			\bf \large Setting&\multicolumn{9}{l||}{  \cellcolor[gray]{0.8} \bf \large Equi-correlation $\rho=0.6$, $\beta_{active}=(2,3,-3,4)'$}   &\multicolumn{9}{l}{  \large \cellcolor[gray]{0.8} \bf Equi-correlation $\rho=0.9$, $\beta_{active}=(2,3,-3,4)'$}   \\\hline
			\multirow{3}{*}{Metric}& \multicolumn{6}{l|}{\quad\quad\quad\quad\quad\quad\quad\quad$\beta_j$'s}                                                                                         & \multicolumn{2}{l|}{\quad$\gamma_j$'s }              & Model 
			& \multicolumn{6}{l|}{\quad\quad\quad\quad\quad\quad\quad\quad$\beta_j$'s}                                                                                         & \multicolumn{2}{l|}{\quad$\gamma_j$'s }              & Model 
			\\ \cline{2-19}
			Metric& \multicolumn{2}{l|}{KL} & \multicolumn{2}{l|}{JD of 90\% CI} & \multicolumn{2}{l|}{`Bias'} & \multicolumn{2}{l|}{`Bias'}    & HD &
			\multicolumn{2}{l|}{KL } & \multicolumn{2}{l|}{JD of 90\% CI} & \multicolumn{2}{l|}{`Bias'} & \multicolumn{2}{l|}{`Bias'}    & HD 
			\\ \cline{2-19}
			& $+$ & $-$           & $+$ & $-$         & $+$ & $-$       & $+$ & $-$        & all    & $+$ & $-$   & $+$ & $-$       & $+$ & $-$      & $+$ & $-$         & all    \\ 
			\hline
			Skinny Gibbs & 0.13   & 0.01                      & 0.23   & \bf{0.11}                    & {0.05} & \bf{0.003}            & 0.0008&*& 1  
			& 0.30   & 0.02                      & 0.33   & 0.09                     & 0.23 & {0.004}            & 0.001&*& 2  \\ 
			WBB1         & 0.12   & 3.09                      & 0.30   & 1                                       & \bf{0.03}  & \bf{0.003}            &  *&*& \bf{0}    
			& {0.14}   & 3.08                      & \bf{0.23}   & 1                                       & {0.13}  & \bf{0.002}            &*&*& \bf{0}    \\ 
			WBB2         & 0.13   & 3.09                      & {0.30}   & 1                                       & \bf{0.03}  & \bf{0.003}            & * &*& \bf{0}   
			& {0.15}   & 3.08                      & \bf{0.23}   & 1                                       & {0.13}  & \bf{0.002}            & *&*& \bf{0}  \\ 
			BB-SSL       & \bf{0.06}   & \bf{0.003}                     & \bf{0.20}   & $\bm{0.11} $                                   & \bf{0.03}  & \bf{0.003}            &*&*& \bf{0}      
			& \bf{0.11}   & \bf{-0.003}                     & \bf{0.23}   & \bf{0.08}                                    & \bf{0.12}  & \bf{0.002}            &*&*& \bf{0}   
			\\\hline\hline 
	\end{tabular}}
	%\caption{\small $n=100,p=1000,\rho=0.6$ equi-correlation structure with $\beta_{active}=(2,3,-3,4)'$.}
\end{subtable}

\caption{\small Evaluation of approximation properties (relative to SSVS) in the high-dimensional setting with $n=100$ and $p=1\,000$ based on 10 independent runs. The best performance is marked in bold font. %More results on approximation performance in different settings are in the Appendix (Section \ref{sec_appendix:simulation}). 
KL is the Kullback-Leibler divergence, JD is the Jaccard distance  of credible intervals (CI), HD is the Hamming distance of the median models.  `Bias'   refers to the $l_1$ distance of estimated posterior means. We denote with $*$ all numbers smaller than 0.0001, with $+$ an average over active coordinates, and with $-$ an average over inactive coordinates. }
\label{table:exp_high}
\end{table}

To better quantify the performance of each method, we gauge the  quality of the posterior approximation using various metrics in Table \ref{table:exp_high}. The  KL divergence is calculated using an R package ``FNN'' (\cite{beygelzimer2013fnn}), where all parameters are set to their default values. {We also report the Jaccard distance of $90\%$ credible intervals relative to the SSVS  benchmark. The Jaccard distance \citep{jaccard1912distribution} of two intervals $A$ and $B$ is defined as $d_J(A,B)=1-J(A,B)$ where $J(A,B)=\frac{|A\cap B|}{|A\cup B|}$ and $|\cdot|$ denotes the length. The Hamming distance is calculated using an R package ``e1071'' (\cite{meyer2019package}). We also compare the $\ell_1$ norm of posterior means (i.e. ``bias" relative to the SSVS standard) for $\beta_i$'s as well as $\gamma_i$'s.} All methods do well in terms of MIP and the selected model (based on the median probability model rule). For $\beta_i$'s,  all methods estimate the mean accurately. Taking into account the shape of the posterior for $\beta_i$'s, the performance is divided among coordinates and methods. For all methods, the approximability of active coordinates is less accurate than for the inactive ones. In the settings we tried, {we rank the performance of}   various methods as follows:  $\text{BB-SSL}\, >\, \text{Skinny Gibbs}\, >\, \text{WBB1}\, \approx\, \text{WBB2} $. {The average times (in seconds (s) when $\rho=0.6$ in the equi-correlated design) spent on generating $100$ effective samples for each $\beta_j$ are: BB-SSL (0.69 s) $<$ SSVS2 (2.58 s) $<$ Skinny Gibbs (6.61 s) $<$ WBB2 (13.53 s) $<$ WBB1 (17.25 s)  $<$ SSVS1 (34.67 s).}

\paragraph{Conclusion}
We found that BB-SSL is a reliable approximate method for posterior sampling that achieves a close-to-exact (SSVS) performance but is computationally cheaper. Additional speedups can be obtained with  parallelization. {The most expensive step in BB-SSL is solving the optimization problem \eqref{BBL} at each iteration. This could be potentially circumvented by using  the Generative Bootstrap Sampler (GBS) \citep{shin2020scalable} which constructs a generator function that can transform weights into samples from the posterior distribution. This strategy could be particularly beneficial when both $n$ and $p$ are large and when many posterior samples are needed.}   While MCMC-based methods are sensitive to the initialization and can fall into a local trap (e.g. when predictors are highly correlated),  we have seen BB-SSL  to be  less susceptible to this problem. {
%In most cases where MCMC-based methods are trapped in local mode, BB-SSL is still likely to find multiple modes. 
BB-SSL, in some sense,  relies on the optimization procedure {\em not} finding the global mode at all times. Indeed, we want to provide a representation of the entire posterior distribution consisting of {\em both} local and global modes. However, we anticipate that the global mode will be found more often, correctly reflecting the amount of posterior mass assigned to it.
			This issue was also  discussed in Section 2.5.1 in \cite{fong2019scalable}, who point out that not necessarily finding the global mode will result in assigning more posterior density to local modes. We have found BB-SSL (initialized at the SS-LASSO solution after annealing) perform similarly as SSVS initialized at the truth in very highly correlated cases.}
%WBB1 and WBB2 severely underestimate the variance of active $\beta_i$'s, are slightly biased towards more extreme values for MIP, and find only a limited number of models. Skinny Gibbs underestimates the variance for active $\beta_i$'s, and misses some models when multi-modality exists, but does well in estimating MIP. 

%In simulations, we found that SSVS2 and Skinny Gibbs are sensitive to intialization, and the sensitivity becomes more severe when correlation between predictors becomes stronger. In contrast, WBB1, WBB2 and BB-SSL are free of this problem.

%In terms of computational complexity, as seen from table \ref{table: computation}, BB-SSL and Skinny Gibbs are computationally advantageous when both $n$ and $p$ are large, especially when predictors are independent or only weakly correlated (in that case, we can fit SSL using only a single value of $\lambda_0$). 

%The advantage of perturbed WBB can be best summarized in two aspects: 
%\begin{itemize}
%	\item[(1)] High computational efficiency: per iteration computation is much faster, and typically MCMC requires a burn in and thinning, which can be avoided by our method since different runs are independent. The computation gain is most significant when SSLASSO can perform well using a single value of $\lambda_{0}$.  
%	\item[(2)] Ability to discover multi-modality in posterior and less likely to suffer from local trap than classic MCMC under highly correlated high dimensional settings.
%\end{itemize}

\section{Data Analysis}\label{sec:real_data}

\subsection{Life Cycle Savings Data}

\begin{figure}\centering
	\includegraphics[width=\textwidth,height=0.28\textheight]{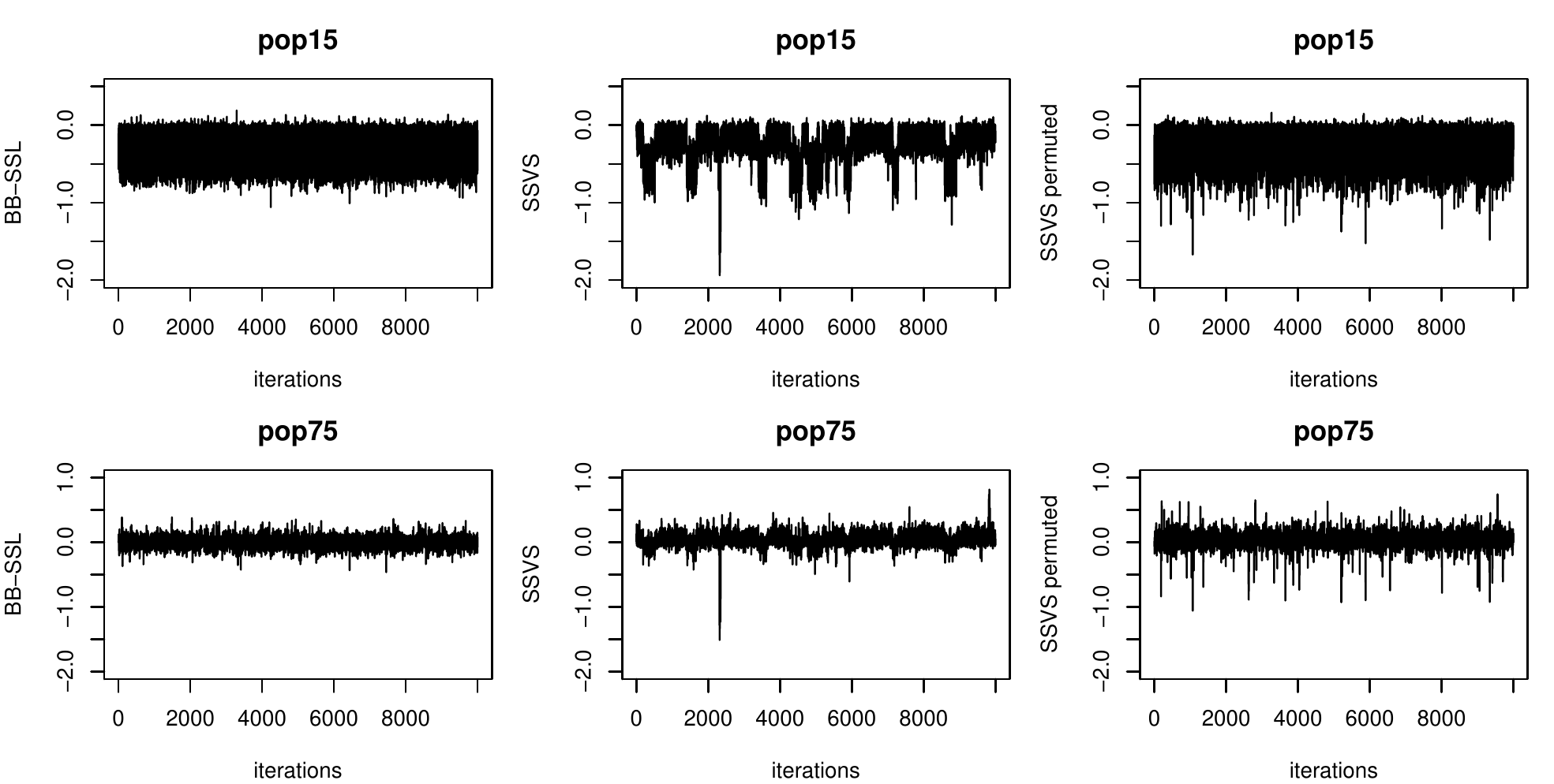}
	\includegraphics[width=\textwidth,height=0.28\textheight]{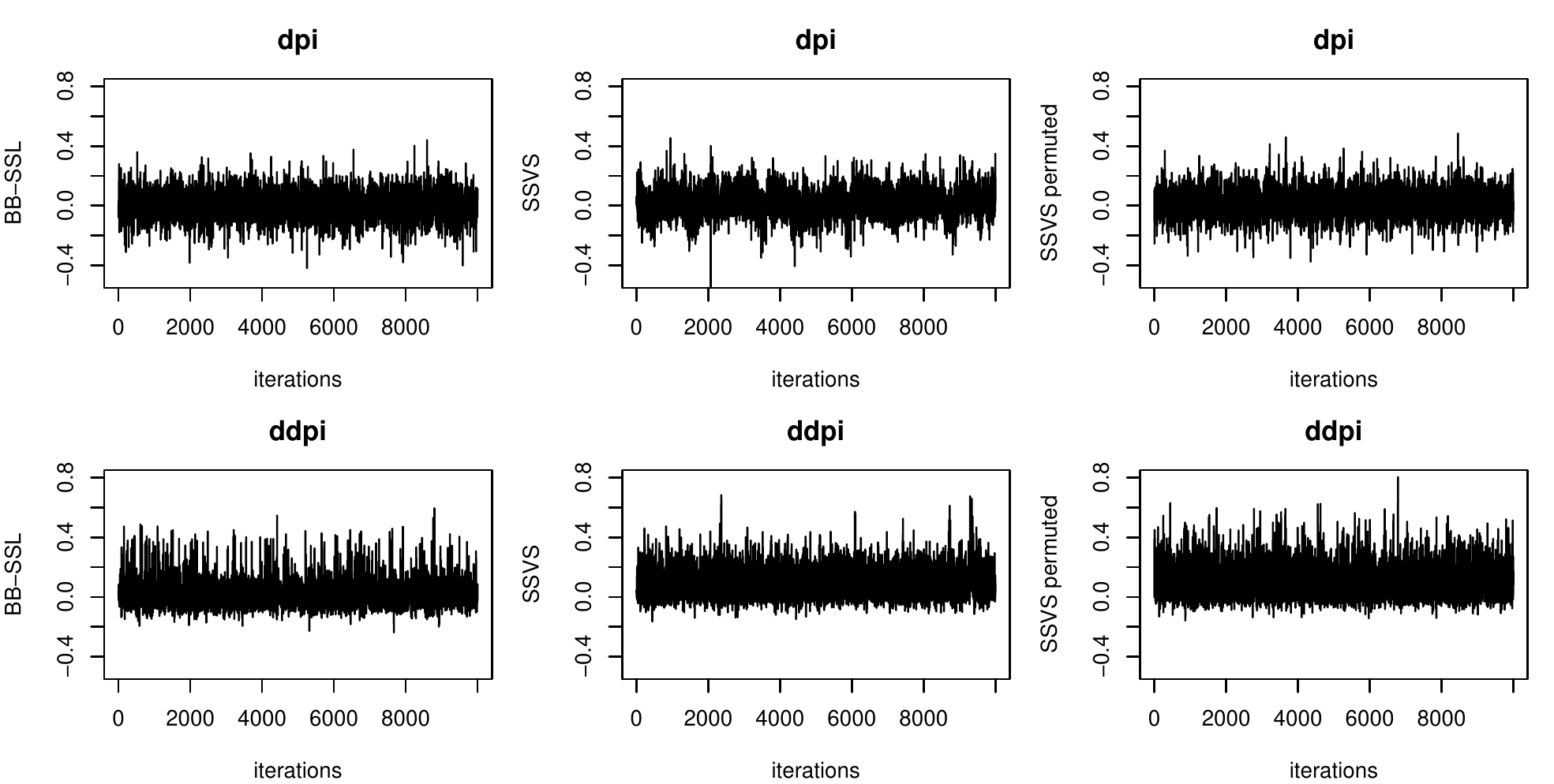}
	\caption{\small Trace plots for the Life Cycle Savings data. We choose $\lambda_0=20,\lambda_1=0.05$.  The first column is the BB-SSL traceplot with weight distribution $\alpha=2\log \frac{(1-\theta)\lambda_0}{\theta\lambda_1}=14$, the second column is thinned SSVS traceplot chain with a LASSO initialization (regularization parameter chosen by cross validation). The third column is the same SSVS chain only with samples permuted.}
	\label{fig:LCS_trace}
\end{figure}

\begin{figure}
	\begin{subfigure}{.6\textwidth}\centering
		\includegraphics[width=\textwidth, height=.35\textheight]{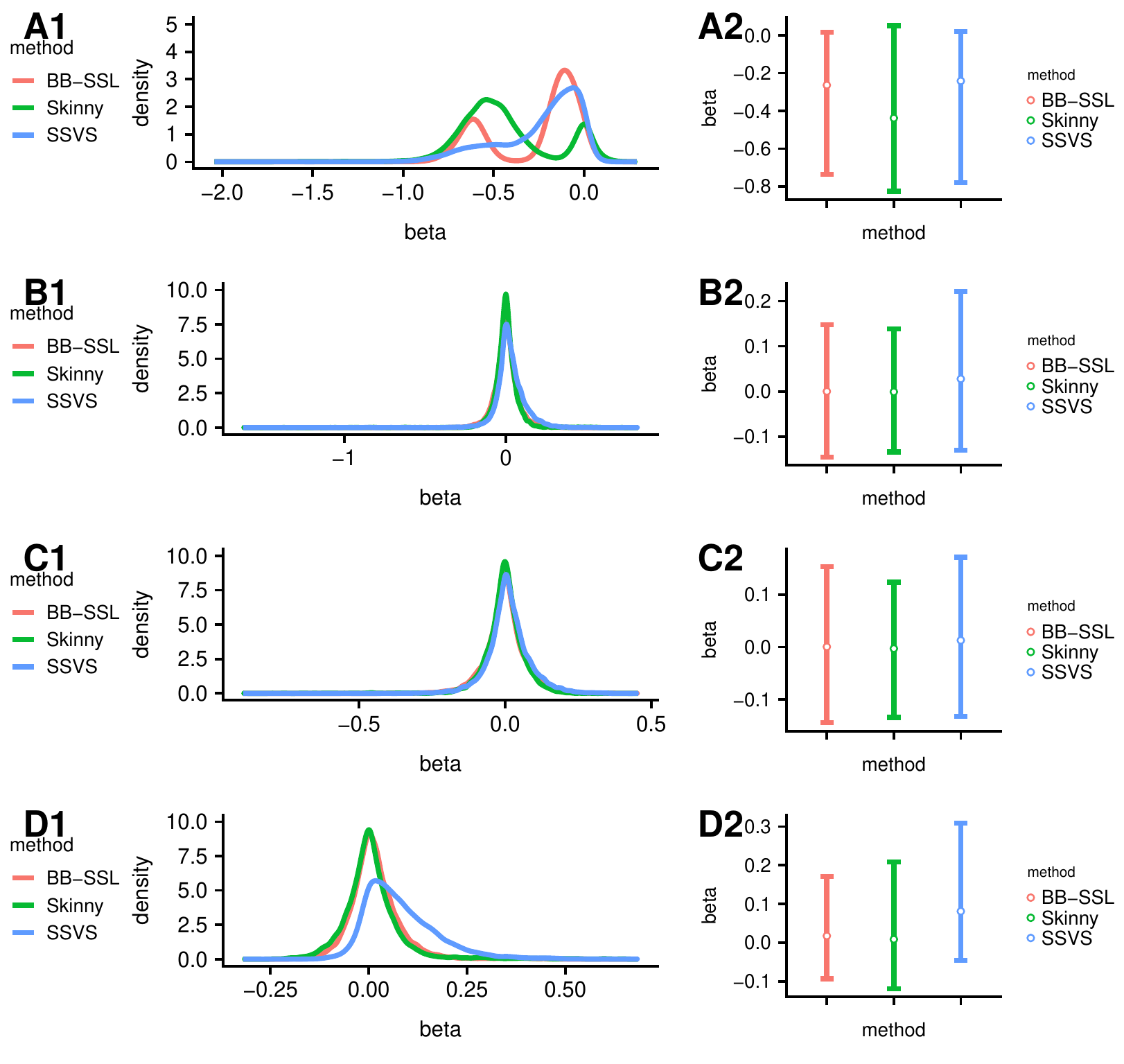}
		\caption{\small Density of $\beta_i$'s.. }
		\label{fig:LCS_density}
	\end{subfigure}
	\begin{subfigure}{.35\textwidth}\centering
		\includegraphics[width=\textwidth]{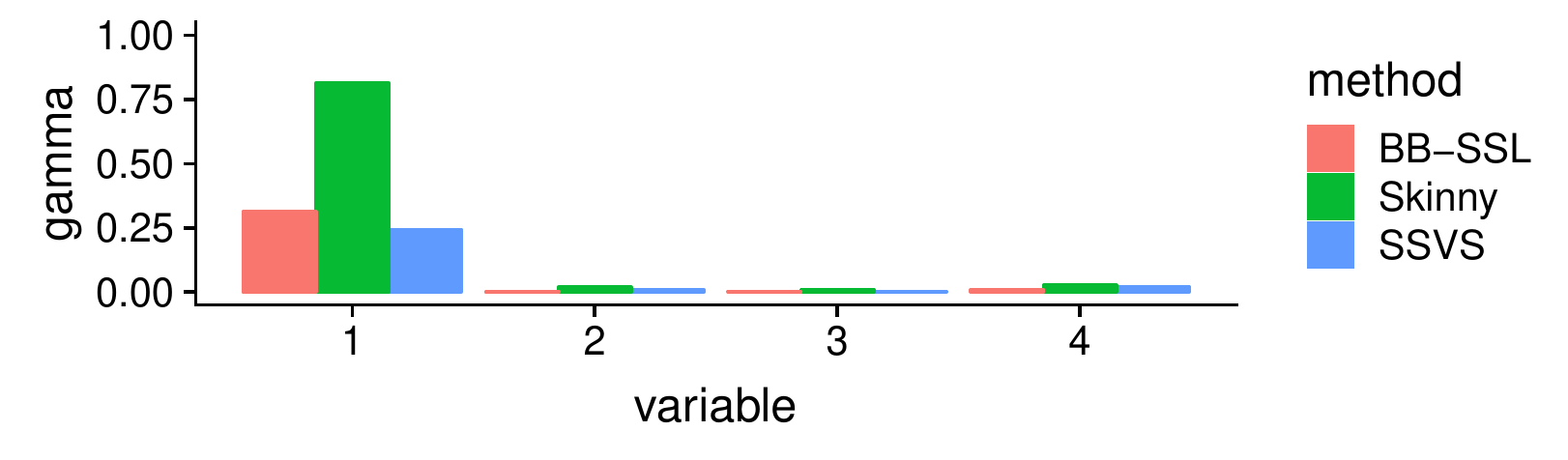}
		\caption{\small Comparison of mean of $\gamma_i$'s. }
		\label{fig:LCS_gamma}
	\end{subfigure}
	\caption{\small Plots for Life Cycle Savings Data. We choose $\lambda_0=20,\lambda_1=0.05$. SSVS is initialized at the LASSO solution (with the regularization parameter chosen by cross-validation). The weight distribution for BB-SSL uses $\alpha=2\log \frac{(1-\theta)\lambda_0}{\theta\lambda_1}=14$.}
\end{figure}

The Life Cycle Savings data \citep{belsley2005regression} consists of $n=50$ observations on  $p=4$ highly correlated predictors: ``pop15'' (percentage of population under 15 years old), ``pop75'' (percentage of population over 75 years old), ``dpi'' (per-capita disposable income), ``ddpi'' (percentage of growth rate of dpi).
According to the life-cycle savings hypothesis proposed by  \cite{ando1963life}, the savings ratio ($y$) can be explained by these four predictors and  a linear model can be used to model their relationship.  

We preprocess the data in the following way. First, we standardize predictors so that each column of $\X$ is centered and rescaled so that $||\X_j||_2=\sqrt{n}$. Next, we estimate the noise variance $\sigma^2$ using an ordinary least squares regression. We then divide $y$ by the estimated noise standard deviation and  estimate $\theta$ by fitting SSL with $\lambda_0=20,\lambda_1=0.05$. For BB-SSL we set $\alpha=2\log\frac{(1-\theta)\lambda_0}{\theta \lambda_1}\approx 14$ and set $a=1$, $b=4$. We run SSVS1 and Skinny Gibbs for $T=100\,000,\,B=5\,000$ and  BB-SSL for $T=10\,000$. 
%From simulations we see that WBB1, WBB2 and VB does not work well so we exclude them in real analysis.%We use different initializations and check that convergence has reached using Gelman-Rubin diagnostics.  

Figure \ref{fig:LCS_trace} shows the trace plots on the four predictors. BB-SSL (first column)  has the same mean and spread as SSVS (third column). We also observe that raw samples from SSVS (second column) are correlated, so more iterations are needed in order to fully explore the posterior. In contrast, each sample from BB-SSL is independent  and thereby fewer samples will be needed in practice. See Table \ref{tab:ess} for effective sample size comparisons. Figure \ref{fig:LCS_density} shows the marginal density of $\beta_i$'s and \ref{fig:LCS_gamma} shows the marginal mean of $\gamma_i$'s. In both figures BB-SSL achieves good performance.

\iffalse
\begin{figure}\centering
	\begin{subfigure}[b]{.48\textwidth}\centering
		\includegraphics[width=\textwidth,height=.45\textheight]{LCS_100}
		\caption{\scriptsize First 100 iterations after burn-in. }
	\end{subfigure}
	\begin{subfigure}[b]{.48\textwidth}\centering
	\includegraphics[width=\textwidth,height=.45\textheight]{LCS_1000}
	\caption{\scriptsize First 1,000 iterations after burn-in. }
	\end{subfigure}

	\begin{subfigure}[b]{.48\textwidth}
		\includegraphics[width=\textwidth,height=.45\textheight]{LCS_10000}
		\caption{\scriptsize First 10,000 iterations after burn-in. }
	\end{subfigure}
	\begin{subfigure}[b]{.48\textwidth}
		\includegraphics[width=\textwidth,height=.45\textheight]{LCS_all}
		\caption{\scriptsize First 190,000 iterations after burn-in.}
	\end{subfigure}
\caption{\scriptsize Density for the Life Cycle Savings data based on different number of iterations after burn-in. We choose $\lambda_0=20,\lambda_1=0.05$. MCMC is initialized at LASSO solution (regularization parameter chosen by cross validation). Weight distribution for BBLASSO is $\alpha=14$.}
\end{figure}
\fi

%\begin{figure}\centering
%	\includegraphics[width=.6\textwidth,height=.4\textheight]{LCS_gelman_plot}
%	\caption{\scriptsize Gelman-Rubin plot for the Life Cycle Savings data.}
%\end{figure}

\begin{table}
	\centering\scalebox{.8}{
		\begin{tabular}{l|l|l|l} 
			\hline\hline
			& SSVS & Skinny Gibbs & BB-SSL  \\ 
			\hline
			Effective sample size  & 2716 &11188        & 15000    \\ 
			\hline\hline
		\end{tabular}
	}
	\caption{\small Average effective sample size (out of $15\,000$ samples) for Life Cycle Saving Data. Effective sample size is calculated using R package coda (\cite{plummer2006coda}). }\label{tab:ess}
\end{table}

\subsection{Durable Goods Marketing Data Set} \label{sec:marketing_intro}
Our second application examines a cross-sectional dataset from \cite{ni2012database} (ISMS Durable Goods Dataset 2) consisting of durable goods sales data from a major anonymous  U.S. consumer electronics retailer. The dataset features the results of a direct-mail promotion campaign in November 2003 where
roughly half of the $n=176\,961$ households received a promotional mailer with $10\$$ off their purchase during the promotion time period (December 4-15). The treatment assignment ($tr_i=\1\left(\text{promotional mailer}_i\right)$) was random. The data contains $146$ descriptors of all customers including prior purchase history, purchase of warranties etc. We will investigate the effect of the promotional campaign (as well as other covariates) on December sales. In addition, we will interact the promotion mail indicator   with customer characteristics to identify the ``mail-deal-prone" customers. 
To be more specific, we adopt the following model
\begin{equation}\label{eq:real_data_model}
	Y_i=\alpha\times  tr_i + \b^T \x_i + \bm{\gamma} \times tr_i\times \x_i+ \epsilon_i
\end{equation}
where $\x_i$ refers to the 146 covariates, $tr_i $ is the treatment assignment, and the noise $\epsilon_i$ is  iid normally distributed. And our aim is to (1) estimate the coefficients $\alpha,\b,\bm{\gamma}$; (2) identify those customers with $\E\left[Y_i\C \x_i, tr_i=1\right] > \E \left[Y_i\C \x_i, tr_i=0\right]$. 

For preprocessing, we first remove all variables that contain missing values or that are all $0$'s. We also create new predictors by interacting the treatment effect with the descriptor variables. After that the total number of predictors  becomes $p=273$. We standardize $\X$ such that each column has a zero mean   and a standard deviation $\sqrt{n}$ and we use the maximum likelihood estimate of the standard deviation to rescale the outcome.
We run BB-SSL for $T=1\,000$ iterations and SSVS1 for $T=20\,000$ iterations with a  $B=1\,000$ burnin period,  initializing MCMC at the origin. We set $\lambda_0=100,\lambda_1=0.05,a=1,b=p$. Estimating $\hat \theta=\frac{\text{\# of selected variables}}{p}$ by fitting the Spike-and-Slab LASSO, we then set $\alpha=2\log \frac{(1-\hat \theta)\theta_0}{\hat \theta \lambda_1}$.

Figure \ref{fig:causal_interaction_density} depicts estimated posterior density of selected coefficients in the model \eqref{eq:real_data_model}, showing that BB-SSL estimation is very close to the gold standard (SSVS). Further, BB-SSL identified $67.3\%$ of customers as ``mail-deal-prone'', reaching accuracy $98.2\%$ and a false positive rate $2.1\%$ (treating SSVS estimation as the truth).  Despite the comparable performance to SSVS, BB-SSL is advantageous in terms of computational efficiency. As shown in Figure \ref{fig:RealData_time_comparison}, within the same amount of time, BB-SSL obtains more effective samples compared with SSVS and its advantage becomes even more significant as time increases. This experiment confirms our hypothesis that BB-SSL has a great potential as an approximate method for large datasets.

\begin{figure}\centering
	\includegraphics[width=.4\textwidth]{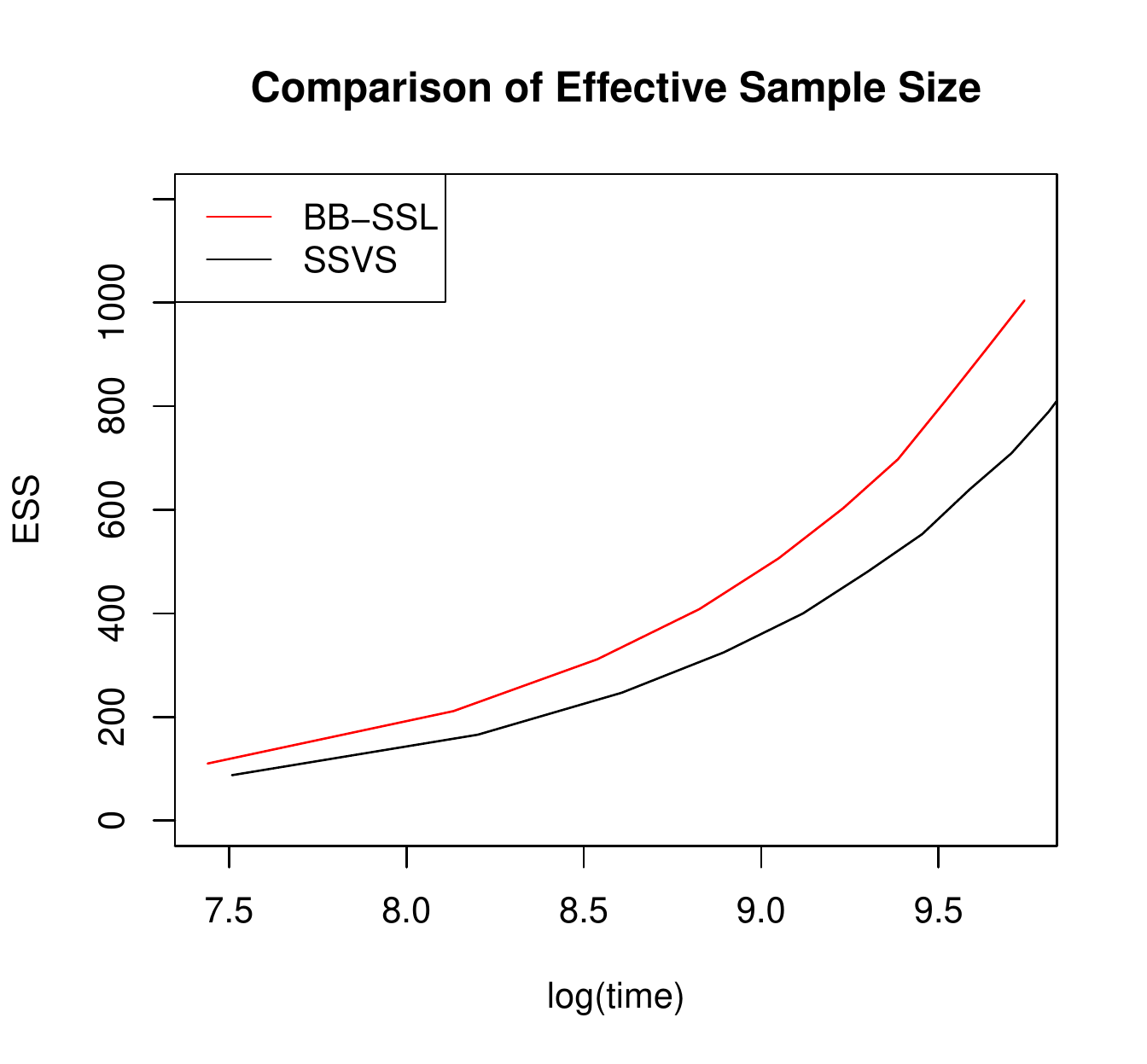}
	\caption{\small Effective sample size comparison for ISMS Durable Goods Dataset 2. We choose $\lambda_0=100,\lambda_1=0.05$.  Red line is BB-SSL with $\alpha=2\log \frac{(1-\theta)\lambda_0}{\theta\lambda_1}\approx 15$ and black line is SSVS initialized at origin.}
	\label{fig:RealData_time_comparison}
\end{figure}

\begin{figure}[!t]\centering
		\includegraphics[width=.5\textwidth,height=0.4\textheight]{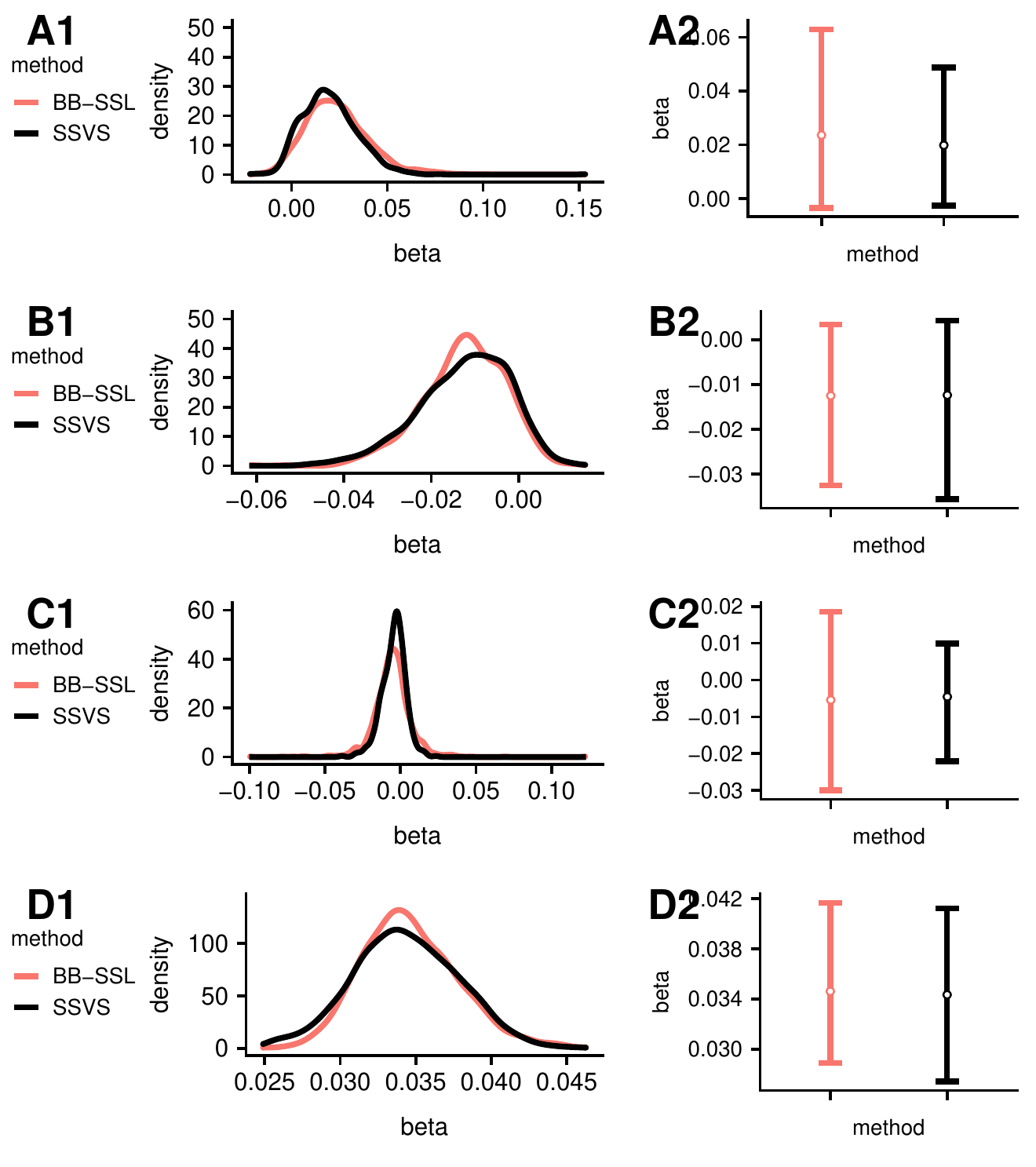}%\caption{\scriptsize $\beta_i,i=1,2,3,4$}%\end{subfigure}%%\begin{subfigure}{.5\textwidth}\centering
		\includegraphics[width=.5\textwidth,height=0.4\textheight]{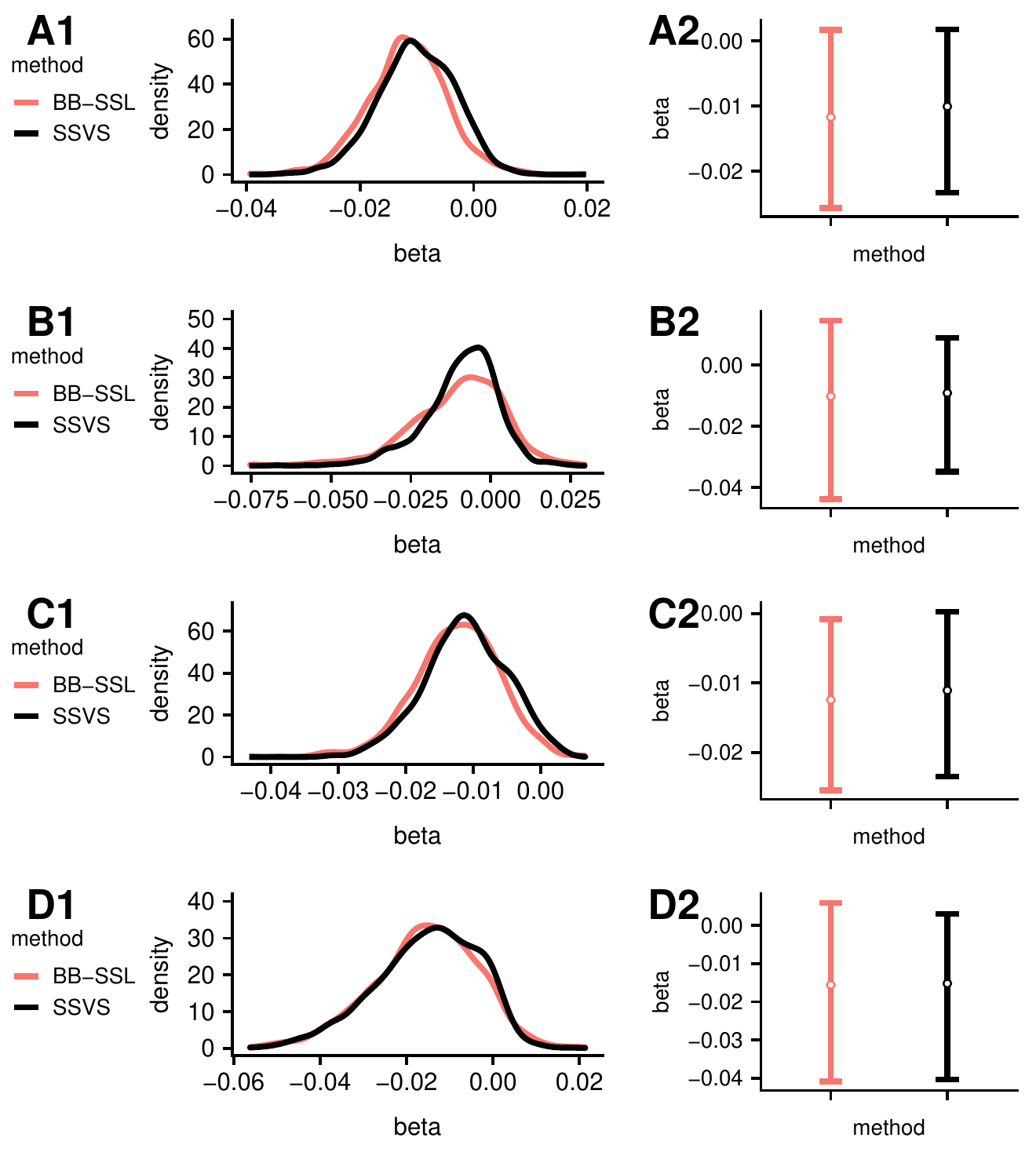}
%\caption{\small Selected results on causal data for interaction model.}
\caption{\small Posterior density and credible intervals for the selected $\beta_i$'s. From left to right, top to bottom they correspond to ``S-SAL-TOT60M'', ``S-U-CLS-NBR-12MO'', ``PH-HOLIDAY-MAILER-RESP-SA'', ``PROMO-NOV-SALES'',  ``S-SAL-FALL-24MO'', ``S-TOT-CAT $\times$ treatment'', ``C-ESP-RECT $\times$ treatment'', ``S-CNT-TOT24M $\times$ treatment''.  We set $\lambda_0=100,\lambda_1=0.05,a=1,b=273$.  We set $\alpha=2\log\frac{(1-\theta)\lambda_0}{\theta\lambda_1}\approx 15$.}
\label{fig:causal_interaction_density}
\end{figure}

%\begin{figure}\centering	\includegraphics[width=.7\textwidth,height=0.12\textheight]{DurData_interaction_gamma.eps}
%\caption{Comparison of $\gamma_i,i=1,2,\cdots,12$ for each chosen variable in interaction model when choosing $\alpha=20$.}
%\end{figure}

\iffalse
\begin{table}
	\centering\scalebox{.6}{
		\begin{tabular}{l|l|l}
			\hline\hline
			Metric            & MCMC        & BBLASSO                                              \\\hline
			Running time: $n=176,961$, without pre-selection  & 25.92 & 432  \\\hline
			Running time: $n=20,000$, without pre-selection  & 4.32 & 10.8\\\hline
			Running time: $n=176961$, with pre-selection & 0.45  & 72                                                            \\\hline
			Running time: $n=60,000$, with pre-selection        &             0.02                        & 3.96                      \\\hline
			Running time: $n=40,000$, with pre-selection        &             0.012                        & 2.4                     \\\hline\hline    
		\end{tabular}
	}\caption{\small Comparison of running time (per iteration, in seconds) between SSVS1 and BB-SSL on causal dataset. SSL is fitted with $\lambda_0$ being an equal difference sequence of length 10 starting at $\lambda_1=0.05$ and ends at $100$.}
\end{table}
\fi 

\section{Discussion}\label{sec:discussion}
In this paper we developed BB-SSL, a computational approach for approximate posterior sampling under Spike-and-Slab LASSO priors based on  Bayesian bootstrap ideas.
The fundamental premise of BB-SSL is the following: replace sampling from conditionals (which can be costly when either $n$ or $p$ are large) with fast optimization of randomly perturbed (reweigthed) posterior densities. We have explored various ways of performing the perturbation and looked into asymptotics for guidance about perturbing (weighting) distributions. We have concluded that with suitable conditions on the weights distribution, the pseudo-posterior distribution attains the same rate as the actual posterior in high-dimensional estimation problems (sparse normal means and high-dimensional regression). These theoretical results are reassuring and significantly extend existing knowledge about Weighted Likelihood Bootstrap (\cite{newton2020weighted}), which was shown to be consistent for iid data in finite-dimensional problems. We have shown in simulations and on real data that BB-SSL can approximate the true posterior well and can be computationally beneficial.  {The GBS method of \cite{shin2020scalable} could potentially greatly improve the scalability of BB-SSL. We leave this direction for future research.}

\bibliography{main} 
\bibliographystyle{apa}

\appendix
\section{Proofs}
%\subsection{Theoretical Proof}
This section presents proofs of the main theoretical statements: Section \ref{sec_appendix:normal_means} shows the proof of Theorem \ref{normal_means_weight}. Section \ref{sec:appendix:normal_means_corollary_proof} presents the proof of Corollary \ref{normal_mean_dirichlet}. Section \ref{sec:appendix_normal_means_remark} shows an example for the remark under Theorem \ref{normal_means_weight}. Section \ref{sec_appendix:regression} proves Theorem \ref{thm: regression_weight_selected}  and Theorem \ref{thm:regression_weight}. Section \ref{sec:appendix_regression_corollary_proof} proves Corollary \ref{regression_dirichlet}. Section \ref{sec_appendix:motivation_noncenter_proof} introduces a motivation for Section \ref{sec:motivation_noncenter}.

\subsection{Proof of Theorem \ref{normal_means_weight}}\label{sec_appendix:normal_means}
\subsubsection{Definitions and Lemmas used in the proof of Theorem \ref{normal_means_weight}}\label{sec_appendix:normal_means_notations}

Denote constants 
$$
c_{+}^{\lambda_{0},\lambda_{1}}=0.5(1+\sqrt{1-4/(\lambda_{0}-\lambda_{1})^{2}})
$$
and 
$$
\delta_{c+}^{\lambda_{0},\lambda_{1}}=\frac{1}{\lambda_{0}-\lambda_{1}}\log\left[\frac{1-\theta}{\theta}\frac{\lambda_{0}}{\lambda_{1}}\frac{c_{+}}{1-c_{+}}\right].
$$
With slight abuse of notation, we denote the marginal prior (after integrating out $\gamma_i$) by
$$
\pi(\beta_{i}\mid\theta,\lambda_{0},\lambda_{1})=\theta\lambda_{1}/2e^{-|\beta_{i}|\lambda_{1}}+(1-\theta)\lambda_{0}/2e^{-|\beta_{i}|\lambda_{0}}.
$$
Similarly as in  \cite{rovckova2018bayesian}, the conditional mixing weight between the spike and the slab will be denoted with
\begin{equation}\label{eq:pstar}
p^{*}(t;\lambda_{0},\lambda_{1})=\frac{\theta\lambda_{1}e^{-|t|\lambda_{1}}}{\theta\lambda_{1}e^{-|t|\lambda_{1}}+(1-\theta)\lambda_{0}e^{-|t|\lambda_{0}}}
\end{equation}
and the weighted penalty with 
$$
\lambda^{*}(t;\lambda_{0},\lambda_{1})=\lambda_{1}p^{*}(t;\lambda_{0},\lambda_{1})+\lambda_{0}(1-p^{*}(t;\lambda_{0},\lambda_{1})).
$$
The quantities $p^{*}(t;\lambda_{0},\lambda_{1})$ and $\lambda^{*}(t;\lambda_{0},\lambda_{1})$ are analogues of $p^\star_\theta$ and $\lambda^\star_\theta$ defined in (5) (main manuscript). Here, we have emphasized the dependence on $\lambda_1$ and $\lambda_0$ and suppressed the dependence on $\theta$ since this parameter is treated as fixed in our theory. 
Finally, we denote
$$
\rho(t\C\theta,\lambda_{0},\lambda_{1})=-\lambda_{1}|t|+\log[p^{*}(0;\lambda_{0},\lambda_{1})/p^{*}(t;\lambda_{0},\lambda_{1})]
$$
and (similarly as in  \cite{rovckova2018bayesian})
$$
g(t;\lambda_{0},\lambda_{1})=[\lambda^{*}(t;\lambda_{0},\lambda_{1})-\lambda_{1}]^{2}+2\log p^{*}(t;\lambda_{0},\lambda_{1}).
$$

The following Lemma (a version of Theorem 3.1 in \citet{rovckova2018bayesian}) will be useful for characterizing the BB-SSL posterior properties.

\begin{lemma}\label{lemma_appendix:delta_original_bound}
	Assume $Y_{i}$'s arise from the normal means model \eqref{eq: normal_means_model} with $\sigma=1$ for $1\leq i\leq n$.
	Under the SSL prior \eqref{eq:ss}, the mode $\wh\b=(\wh\beta_{1},\cdots,\wh\beta_{p})^{T}$ maximizing $L(\bm\beta,\Y)=-\frac{1}{2}\sum_{i=1}^n(Y_{i}-\beta_i)^{2}+\sum_{i=1}^n\rho(\beta_i\C\theta,\lambda_{0},\lambda_{1})$ 
	satisfies 
	\begin{equation}
	\wh\beta_{i}=\begin{cases}
	0, & \text{if \ensuremath{|Y_{i}|\leq\Delta(\lambda_{0},\lambda_{1})}}.\\{}
	[|Y_{i}|-\lambda^{*}(\wh{\beta_{i}};\lambda_{0},\lambda_{1})]_{+}\text{sign}(Y_{i}), & \text{otherwise}.
	\end{cases}\label{eq:iterative}
	\end{equation}
	%	where $\Delta(\lambda_{0},\lambda_{1})$ is the selection
	%	threshold which depends on $\lambda_{0},\lambda_{1}$ and is defined
	%	as 
	where
	\[
	\Delta(\lambda_0,\lambda_1)=\inf_{t>0}[t/2-\rho(t\mid\theta,\lambda_0,\lambda_1)/t].
	\]
\end{lemma}

Notice that $\Delta(\lambda_0,\lambda_1)$ is the  normal-means analogue of $\Delta_j$ (defined below \eqref{beta:implicit} in the main manuscript)  where we have added $(\lambda_0,\lambda_1)$ to emphasize its dependence on $\lambda_0$ and $\lambda_1$. %Further, we have $w_i^{-1/2}\Delta(\lambda_0w_i^{-1/2},\lambda_1w_i^{-1/2})=\Delta_{w_i}$ with $\Delta_{w_i}$ defined under Equation (12) in the revised main manuscript.
Next, we will need an upper and lower bound for $\Delta(\lambda_{0},\lambda_{1})$. When $\lambda_{0}-\lambda_{1}>2$
and $g(0;\lambda_{0},\lambda_{1})>0$,  $-\frac{1}{2}(Y_{i}-\beta)^{2}+\rho(\beta\C\theta,\lambda_{0},\lambda_{1})$ has two modes and \cite{rovckova2018bayesian} shows that $\Delta(\lambda_0,\lambda_1)$ can also be written as
\[
\Delta(\lambda_{0},\lambda_{1})=\sqrt{2\log[1/p^{*}(0;\lambda_{0},\lambda_{1})]+g(\wt\beta;\lambda_{0},\lambda_{1})}+\lambda_{1},
\]	
where $\wt\beta$ is the local mode (not the global mode) %\footnote{is it always different from $\hat\beta$?} 
of $-\frac{1}{2}(Y_{i}-\beta)^{2}+\rho(\beta\C\theta,\lambda_{0},\lambda_{1})$.	 It follows from  Theorem 3.1 in \citet{rovckova2018bayesian} and Lemma 1.2 in the Appendix of \cite{rovckova2018bayesian} that
\begin{equation}
\Delta^{L}(\lambda_{0},\lambda_{1})<\Delta(\lambda_{0},\lambda_{1})\leq\Delta^{U}(\lambda_{0},\lambda_{1}),
\end{equation}
where 
\[
\Delta^{L}(\lambda_{0},\lambda_{1})=\sqrt{2\log[1/p^{*}(0;\lambda_{0},\lambda_{1})]-2}+\lambda_{1}\quad\text{and}\quad\Delta^{U}(\lambda_{0},\lambda_{1})=\sqrt{2\log[1/p^{*}(0;\lambda_{0},\lambda_{1})]}+\lambda_{1}.
\]
%and $0<d=-g(\delta_{c+}^{\lambda_{0},\lambda_{1}};\lambda_{0},\lambda_{1})$.
Based on these facts, we then obtain the following simple Lemma.	

\begin{lemma}\label{lemma_appendix:Delta_bound} 
	Under assumptions
	in Theorem \ref{normal_means_weight}, for any $w_{i}>0$ which satisfies 
	$$
	g(0;\lambda_{0}w_{i}^{-1/2},\lambda_{1}w_{i}^{-1/2})>0\quad\text{and}\quad\lambda_{0}w_{i}^{-1/2}-\lambda_{1}w_{i}^{-1/2}>2
	$$
	we have
	\begin{align}
	\Delta^{L}(\lambda_{0},\lambda_{1})+\lambda_{1}(w_{i}^{-1/2}-1)
	\le \Delta(\lambda_{0}w_{i}^{-1/2},\lambda_{1}w_{i}^{-1/2})
	\le\Delta^{U}(\lambda_{0},\lambda_{1})+\lambda_{1}(w_{i}^{-1/2}-1).
	\label{eq_appendix:Delta_bound}
	\end{align}
\end{lemma} 
\begin{proof}
	Since $g(0;\lambda_{0}w_{i}^{-1/2},\lambda_{1}w_{i}^{-1/2})>0$ and $\lambda_{0}w_{i}^{-1/2}-\lambda_{1}w_{i}^{-1/2}>2$, from Lemma
	\ref{lemma_appendix:delta_original_bound} (the discussion below) and the fact that $p^{*}(0;\lambda_{0},\lambda_{1})=p^{*}(0;\lambda_{0}w_{i}^{-1/2},\lambda_{1}w_{i}^{-1/2})$,
	we have 
	\begin{align*}
	& \Delta(\lambda_{0}w_{i}^{-1/2},\lambda_{1}w_{i}^{-1/2})\le\Delta^{U}(\lambda_{0}w_{i}^{-1/2},\lambda_{1}w_{i}^{-1/2})
	=\Delta^{U}(\lambda_{0},\lambda_{1})+\lambda_{1}(w_{i}^{-1/2}-1).
	%<\Delta^{U}(\lambda_{0},\lambda_{1})+\lambda_{1}w_{i}^{-1/2}.
	\end{align*}
	%{\color{blue} Why do we need the 2 above?}
	The other inequality is obtained similarly.\qedhere
	
\end{proof}

\begin{lemma}\label{lemma_appendix:condtional_normal} For any fixed
	constant $C>0$, $\mu\in\mathbb{R}$ and a random variable $Y\sim N(0,1)$,
	we have 
	\begin{align*}
	& \E_{Y}\left[|Y-\mu|^{2}\1(|Y-\mu|\ge C)\right]\\
	& \quad =  1+\mu^{2}-\Phi(\mu+C)+\Phi(\mu-C)-(\mu-C)\phi(\mu-C)+(\mu+C)\phi(\mu+C),
	\end{align*}
	where $\Phi$ and $\phi$ are the standard normal distribution and density functions.
\end{lemma} \begin{proof} Denote a truncated normal random variable with
	$Z=[Y\mid Y\in[a,b]]$ and recall that  
	\begin{align*}
	\E\left[Z\right]=\frac{\phi(a)-\phi(b)}{\Phi(b)-\Phi(a)}\quad\text{and}\quad \E\left[Z^2\right]=1+\frac{a\phi(a)-b\phi(b)}{\Phi(b)-\Phi(a)}.
	\end{align*}
	Then we have 
	\begin{align*}
	& \E_{Y}\left[|Y-\mu|^{2}\1(|Y-\mu|\ge C)\right]\\
	&= \E_{Y}\left[|Y-\mu|^{2}\right]-\E_{Y}\left[|Y-\mu|^{2}\1(|Y-\mu|<C)\right]\\
	&= 1+\mu^{2}-\left[\Phi(\mu+C)-\Phi(\mu-C)\right]\\
	& \quad\times\left[\E\left[Y^{2}\mid Y\in[\mu-C,\mu+C]\right]-2\mu\E\left[Y\mid Y\in[\mu-C,\mu+C]\right]+\mu^{2}\right]\\
	&= 1+\mu^{2}-\Phi(\mu+C)+\Phi(\mu-C)-(\mu-C)\phi(\mu-C)+(\mu+C)\phi(\mu+C).\qedhere
	\end{align*}
\end{proof}

\subsubsection{Proof of Theorem \ref{normal_means_weight}}

\label{sec:appendix_normal_means_proof}

First assume that $\sigma=1$. Throughout this section,  we will simply denote $\Y=\Y^{(n)}$. We want to obtain the expression for BB-SSL estimate $\wt\b$. Notice that the objective function
\eqref{BBL} for BB-SSL can be written as 
\begin{equation}
\begin{split}
\wt\b & =\arg\max_{\b\in\R^{p}}\left\{ -\sum_{i=1}^{n}\frac{w_{i}}{2}(Y_{i}-\beta_{i})^{2}+\sum_{i=1}^{n}\log\pi(\beta_{i}-\mu_{i}\C\theta)\right\} \\
& =\arg\max_{\b\in\R^{p}}\left\{ -\sum_{i=1}^{n}\frac{1}{2}(Y_{i}^{*}-\beta_{i}^{*})^{2}+\sum_{i=1}^{n}\log\pi_{i}^{*}(\beta_{i}^{*}\C\theta)\right\},
\end{split}
\label{eq_appendix:bb-ssl_optimization}
\end{equation}
where $Y_{i}^{*}=w_{i}^{1/2}(Y_{i}-\mu_{i})$, $\beta_{i}^{*}=w_{i}^{1/2}(\beta_{i}-\mu_{i})$,
and  $\pi_{i}^{*}(\beta_{i}^{*}\C\theta)=\pi(\beta_{i}^{*}\C\theta,\lambda_{0}w_{i}^{-1/2},\lambda_{1}w_{i}^{-1/2})$.
%\theta\psi_{1,j}^*(\cdot)+(1-\theta)\psi_{0,j}^*(\cdot)$ with $\psi_{1,j}^*(\cdot)=\lambda_{1,j}^*/2e^{-|\cdot|\lambda_{1,j}^*}$, $\psi_{0,j}^*(\cdot)=\lambda_{0,j}^*/2e^{-|\cdot|\lambda_{0,j}^*}$ and  $\lambda_{0,j}^*=\lambda_{0}/\sqrt{w_j}$, $\lambda_{1,j}^*=\lambda_1/\sqrt{w_j}$
So from \eqref{eq:iterative}, the mode estimator $\wh{\beta_{i}^{*}}$
for $\beta_{i}^{*}$ satisfies 
\[
\wh{\beta_{i}^{*}}=\begin{cases}
0, & \text{if \ensuremath{|Y_{i}^{*}|\leq\Delta(\lambda_{0}w_{i}^{-1/2},\lambda_{1}w_{i}^{-1/2})}}.\\{}
[|Y_{i}^{*}|-\lambda^{*}(\wh{\beta_{i}^{*}};\lambda_{0}w_{i}^{-1/2},\lambda_{1}w_{i}^{-1/2})]_{+}\text{sign}(Y_{i}^{*}), & \text{otherwise}.
\end{cases}
\]
The BB-SSL estimate $\bm{\wt\b}=(\wt\beta_{1},\cdots,\wt\beta_{p})^{T}$
can then be calculated via $\wt\beta_{i}=\wh{\beta_{i}^{*}}w_{i}^{-1/2}+\mu_{i}$.
Equation (3.2) in \citet{rovckova2018bayesian} says that 
\[
\lambda^{*}(\beta_{i};\lambda_{0},\lambda_{1})=-\frac{\partial\rho(\beta_{i}\C\theta,\lambda_{0},\lambda_{1})}{\partial|\beta_{i}|}.
\]
From the chain rule and the fact that 
\[
\rho\left(w_{i}^{1/2}(\widetilde{\beta_{i}}-\mu_{i});\lambda_{0}w_{i}^{-1/2},\lambda_{1}w_{i}^{-1/2}\right)=\rho(\widetilde{\beta_{i}}-\mu_{i};\lambda_{0},\lambda_{1}),
\]
we know 
\begin{align*}
& \lambda^{*}(\widehat{\beta_{i}^{*}};\lambda_{0}w_{i}^{-1/2},\lambda_{1}w_{i}^{-1/2})=\lambda^{*}\left(w_{i}^{1/2}\left(\widetilde{\beta_{i}}-\mu_{i}\right);\lambda_{0}w_{i}^{-1/2},\lambda_{1}w_{i}^{-1/2}\right)\\
& =-\frac{\partial\rho\left(w_{i}^{1/2}(\widetilde{\beta_{i}}-\mu_{i});\lambda_{0}w_{i}^{-1/2},\lambda_{1}w_{i}^{-1/2}\right)}{\partial\left|w_{i}^{1/2}(\widetilde{\beta_{i}}-\mu_{i})\right|}\\
& =-\frac{\partial\rho(\widetilde{\beta_{i}}-\mu_{i};\lambda_{0},\lambda_{1})}{\partial\left|\widetilde{\beta_{i}}-\mu_{i}\right|}\frac{\partial\left|\widetilde{\beta_{i}}-\mu_{i}\right|}{\partial\left|w_{i}^{1/2}\left(\widetilde{\beta_{i}}-\mu_{i}\right)\right|}=w_{i}^{-1/2}\lambda^{*}\left(\widetilde{\beta_{i}}-\mu_{i};\lambda_{0},\lambda_{1}\right).
\end{align*}
Thus, 
\begin{equation}
\widetilde{\beta}_{i}=\begin{cases}
& \mu_{i},\quad\text{if \ensuremath{|w_{i}^{1/2}(Y_{i}-\mu_{i})|\leq\Delta(\lambda_{0}w_{i}^{-1/2},\lambda_{1}w_{i}^{-1/2})}}.\\
& \mu_{i}+[|Y_{i}-\mu_{i}|-w_{i}^{-1}\lambda^{*}(\wt\beta_{i}-\mu_{i};\lambda_{0},\lambda_{1})]_{+}\text{sign}(Y_{i}-\mu_{i}),\quad\text{otherwise}.
\end{cases}\label{eq_appendix:NN_solution}
\end{equation}
From the Markov's inequality, we know that for any $M_{n}>0$, %{\color{blue} is the conditioning on $Y$ or $\bm Y$? Please fix in the entire appendix. Also, sometimes we use $\|\wt\b-\b^{0}\|_{2}^{2}$ and sometimes $\|\wt\b-\b^{0}\|^{2}$. Please use the first one throughout.}
\[
\E_{\Y}\P_{\bm{\mu},\bm{w}}\left(\|\wt\b-\b^{0}\|_{2}^{2}>M_{n}q\log\left(\frac{n}{q}\right)\C \Y\right)
\leq
\frac{\E_{\Y}\E_{\bm{\mu},\bm{w}}\left[\|\wt\b-\b^{0}\|_2^{2}\C \Y\right]}{M_{n}q\log\frac{n}{q}}.
\]
Thus, in order to prove the desired statement, we only need to bound
$
\E_{\Y}\left[\E_{\bm{\mu},\bm{w}}\left[\|\wt\b-\b^{0}\|_2^{2}\C \Y\right]\right]
$.
Notice that for the separable SSL prior (i.e. with a fixed value of $\theta$), the posterior is separable and we obtain
\[
\E_{\Y}\left[\E_{\bm{\mu},\bm{w}}\left[\|\wt\b-\b^{0}\|_2^{2}\C \Y\right]\right]
=
\sum_{i=1}^{n}\E_{Y_{i}}\left[\E_{\mu_{i},w_{i}}\left[(\widetilde{\beta_{i}}-\beta_{i}^{0})^{2}\mid Y_{i}\right]\right].
\]
We will bound the risk $\E_{Y_{i}}\left[\E_{\mu_{i},w_{i}}\left[(\widetilde{\beta_{i}}-\beta_{i}^{0})^{2}\mid Y_{i}\right]\right]$
separately for active and inactive coordinates.

\paragraph{Active coordinates}

For active coordinates, we have 
\begin{equation}
\begin{split} & \E_{Y_{i}}\left[\E_{\mu_{i},w_{i}}\left[(\widetilde{\beta_{i}}-\beta_{i}^{0})^{2}\mid Y_{i}\right]\right]=\E_{Y_{i}}\left[\E_{\mu_{i},w_{i}}\left[(\widetilde{\beta_{i}}-Y_{i}+Y_{i}-\beta_{i}^{0})^{2}\mid Y_{i}\right]\right]\\
\le\,\, & \E_{Y_{i}}\left[2\E_{\mu_{i},w_{i}}\left[(\widetilde{\beta_{i}}-Y_{i})^{2}\mid Y_{i}\right]+2\E_{\mu_{i},w_{i}}\left[(Y_{i}-\beta_{i}^{0})^{2}\mid Y_{i}\right]\right]\\
\le\,\, & 2\E_{Y_{i}}\left[\E_{\mu_{i},w_{i}}\left[(\widetilde{\beta_{i}}-Y_{i})^{2}\mid Y_{i}\right]\right]+2\\
=\,\, & 2\E_{Y_{i}}\left[\E_{\mu_{i},w_{i}}\left[(\widetilde{\beta_{i}}-Y_{i})^{2}\1\left(w_{i}^{1/2}|Y_{i}-\mu_{i}|\leq\Delta(\lambda_{0}w_{i}^{-1/2},\lambda_{1}w_{i}^{-1/2})\right)\mid Y_{i}\right]\right]\\
\,\, & +2\E_{Y_{i}}\left[\E_{\mu_{i},w_{i}}\left[(\widetilde{\beta_{i}}-Y_{i})^{2}\1\left(w_{i}^{1/2}|Y_{i}-\mu_{i}|>\Delta(\lambda_{0}w_{i}^{-1/2},\lambda_{1}w_{i}^{-1/2})\right)\mid Y_{i}\right]\right]+2\\
=\,\, & 2U_{1}+2U_{2}+2,
\end{split}
\label{eq:appendix_normalmeans_active1}
\end{equation}
where 
\begin{align*}
& U_{1}=\E_{Y_{i}}\left[\E_{\mu_{i},w_{i}}\left[(\widetilde{\beta_{i}}-Y_{i})^{2}\1\left(w_{i}^{1/2}|Y_{i}-\mu_{i}|\leq\Delta(\lambda_{0}w_{i}^{-1/2},\lambda_{1}w_{i}^{-1/2})\right)\mid Y_{i}\right]\right],\\
& U_{2}=\E_{Y_{i}}\left[\E_{\mu_{i},w_{i}}\left[(\widetilde{\beta_{i}}-Y_{i})^{2}\1\left(w_{i}^{1/2}|Y_{i}-\mu_{i}|>\Delta(\lambda_{0}w_{i}^{-1/2},\lambda_{1}w_{i}^{-1/2})\right)\mid Y_{i}\right]\right].
\end{align*}
Now we bound $U_{1}$ and $U_{2}$ separately. First, 
\begin{equation}
\begin{split}U_{1} & =\E_{Y_{i}}\left[\E_{\mu_{i},w_{i}}\left[(\widetilde{\beta_{i}}-Y_{i})^{2}\1\left(w_{i}^{1/2}|Y_{i}-\mu_{i}|\leq\Delta(\lambda_{0}w_{i}^{-1/2},\lambda_{1}w_{i}^{-1/2})\right)\mid Y_{i}\right]\right]\\
& \stackrel{(a)}{\le}\E_{Y_{i}}\left[\E_{\mu_{i},w_{i}}\left[(\mu_{i}-Y_{i})^{2}\1\left(|Y_{i}-\mu_{i}|\leq w_{i}^{-1/2}\Delta(\lambda_{0}w_{i}^{-1/2},\lambda_{1}w_{i}^{-1/2})\right)\mid Y_{i}\right]\right]\\
& \leq\E_{w_{i}}\left[\Delta^{2}(\lambda_{0}w_{i}^{-1/2},\lambda_{1}w_{i}^{-1/2})\frac{1}{w_{i}}\right],
\end{split}
\label{eq:appendix_normalmeans_active20}
\end{equation}
where $(a)$ utilizes the observation $|\widetilde{\beta_{i}}-Y_{i}|\le|\mu_{i}-Y_{i}|$
from Equation \eqref{eq_appendix:NN_solution}.
Deploying Lemma \ref{lemma_appendix:Delta_bound} in equation \eqref{eq:appendix_normalmeans_active20},
we get (when $n$ is sufficiently large)
\begin{equation}
\begin{split}
U_{1} & \leq\E_{w_{i}}\left[\Delta^{2}(\lambda_{0}w_{i}^{-1/2},\lambda_{1}w_{i}^{-1/2})\frac{1}{w_{i}}
\1\left(g(0;\lambda_{0}w_{i}^{-1/2},\lambda_{1}w_{i}^{-1/2})>0\,\,\text{and}\,\,\lambda_{0}w_{i}^{-1/2}-\lambda_{1}w_{i}^{-1/2}>2\right)\right]
\\&\quad+\E_{w_{i}}\left[\Delta^{2}(\lambda_{0}w_{i}^{-1/2},\lambda_{1}w_{i}^{-1/2})\frac{1}{w_{i}}
\1\left(g(0;\lambda_{0}w_{i}^{-1/2},\lambda_{1}w_{i}^{-1/2})\le0\,\,\text{or}\,\, \lambda_{0}w_{i}^{-1/2}-\lambda_{1}w_{i}^{-1/2}\le2\right)\right]\\
& \le\E_{w_{i}}\left[\frac{1}{w_{i}}\left(\Delta^{U}(\lambda_{0},\lambda_{1})+\lambda_{1}w_{i}^{-1/2}-\lambda_{1}\right)^{2}\right]\\&\quad+\E_{w_{i}}\left[\Delta^{2}(\lambda_{0}w_{i}^{-1/2},\lambda_{1}w_{i}^{-1/2})\frac{1}{w_{i}}
\1\left(w_{i}\ge\frac{(\lambda_{0}-\lambda_{1})^{2}(1-p^{*}(0;\lambda_{0},\lambda_{1}))^{2}}{2\log(1/p^{*}(0;\lambda_{0},\lambda_{1}))}\,\,\text{or}\,\,
w_i\ge\frac{(\lambda_0-\lambda_1)^2}{4}
\right)\right]\\
& \stackrel{(a)}{\le}
2\left(\Delta^{U}(\lambda_{0},\lambda_{1})-\lambda_{1}\right)^{2}\E_{w_{i}}\frac{1}{w_{i}}
+2\lambda_1^2\E_{w_{i}}\frac{1}{w_{i}^{2}}
+\E_{w_{i}}\left[\frac{\lambda_{0}^{2}}{w_{i}}\frac{1}{w_{i}}\1\left(w_{i}\ge C(\lambda_0,\lambda_1)\right)\right]\\
& \le2\left(\Delta^{U}(\lambda_{0},\lambda_{1})-\lambda_1\right)^{2}\E_{w_{i}}\frac{1}{w_{i}}+2\E_{w_{i}}\frac{1}{w_{i}^{2}}
+\frac{\lambda_0^2}{C^2(\lambda_0,\lambda_1)}\P\left(w_{i}\ge C(\lambda_0,\lambda_1)\right)\\
& \stackrel{(b)}{\le}2\left(\Delta^{U}(\lambda_{0},\lambda_{1})-\lambda_1\right)^{2}\E_{w_{i}}\frac{1}{w_{i}}+2\E_{w_{i}}\frac{1}{w_{i}^{2}}+\frac{\lambda_0^2}{C^3(\lambda_0,\lambda_1)}\\
& \le2\left(\Delta^{U}(\lambda_{0},\lambda_{1})-\lambda_1\right)^{2}\E_{w_{i}}\frac{1}{w_{i}}+2\E_{w_{i}}\frac{1}{w_{i}^{2}}+1,
\end{split}
\label{eq:appendix_normalmeans_active2}
\end{equation}
where $(a)$ uses the fact that $(a+b)^2\le 2a^2+2b^2$, $\Delta(\lambda_{0}w_{i}^{-1/2},\lambda_{1}w_{i}^{-1/2})\le\lambda_{0}w_{i}^{-1/2}$, and we denote $C(\lambda_0,
\lambda_1)=\min\left\{\frac{(\lambda_{0}-\lambda_{1})^{2}(1-p^{*}(0;\lambda_{0},\lambda_{1}))^{2}}{2\log(1/p^{*}(0;\lambda_{0},\lambda_{1}))},\frac{(\lambda_0-\lambda_1)^2}{4}\right\}$,
and $(b)$ follows from the Markov's inequality.

For $U_{2}$, we have 
\begin{equation}
\begin{split}
U_{2}= & \E_{Y_{i}}\left[\E_{\mu_{i},w_{i}}\left[(\widetilde{\beta_{i}}-Y_{i})^{2}\1\left(w_{i}^{1/2}(Y_{i}-\mu_{i})>\Delta(\lambda_{0}w_{i}^{-1/2},\lambda_{1}w_{i}^{-1/2})\right)\mid Y_{i}\right]\right]\\
\text{\ensuremath{}}\stackrel{(a)}{\le} & \E_{Y_{i}}\left[\E_{\mu_{i},w_{i}}\left[\left(w_{i}^{-1}\lambda^{*}(\widetilde{\beta_{i}}-\mu_{i})\right)^{2}\1\left(w_{i}^{1/2}|Y_{i}-\mu_{i}|>\Delta(\lambda_{0}w_{i}^{-1/2},\lambda_{1}w_{i}^{-1/2})\right)\mid Y_{i}\right]\right]\\
\le & \E_{Y_{i}}\left[\E_{\mu_{i},w_{i}}\left[\left(w_{i}^{-1}\lambda^{*}(\widetilde{\beta_{i}}-\mu_{i})\right)^{2}
\1\left(w_{i}^{1/2}|Y_{i}-\mu_{i}|>\Delta(\lambda_{0}w_{i}^{-1/2},\lambda_{1}w_{i}^{-1/2})\right)\right.\right.
\\&\quad\quad\left.\left.\1\left(\lambda_{0}w_i^{-1/2}-\lambda_{1}w_i^{-1/2}>2\,\,\text{and}\,\,g(0;\lambda_{0}w_{i}^{-1/2},\lambda_{1}w_{i}^{-1/2})>0\right)\mid Y_{i}\right]\right]\\
& +\E_{Y_{i}}\left[\E_{\mu_{i},w_{i}}\left[\left(w_{i}^{-1}\lambda^{*}(\widetilde{\beta_{i}}-\mu_{i})\right)^{2}
\1\left(w_{i}^{1/2}|Y_{i}-\mu_{i}|>\Delta(\lambda_{0}w_{i}^{-1/2},\lambda_{1}w_{i}^{-1/2})\right)\right.\right.
\\&\quad\quad\quad\left.\left.\1\left(\lambda_{0}w_i^{-1/2}-\lambda_{1}w_i^{-1/2}\le2\right)\mid Y_{i}\right]\right]\\
& +\E_{Y_{i}}\left[\E_{\mu_{i},w_{i}}\left[\left(w_{i}^{-1}\lambda^{*}(\widetilde{\beta_{i}}-\mu_{i})\right)^{2}
\1\left(w_{i}^{1/2}|Y_{i}-\mu_{i}|>\Delta(\lambda_{0}w_{i}^{-1/2},\lambda_{1}w_{i}^{-1/2})\right)\right.\right.
\\&\quad\quad\quad\left.\left.\1\left(g(0;\lambda_{0}w_{i}^{-1/2},\lambda_{1}w_{i}^{-1/2})\le0\right)\mid Y_{i}\right]\right]\\
= & U_{3}+U_{4}+U_{5},
\end{split}
\label{eq:appendix_normalmeans_active3}
\end{equation}
where $(a)$ utilizes the fact that $|\widetilde{\beta_{i}}-Y_{i}|\le w_{i}^{-1}\lambda^{*}(\widetilde{\beta_{i}}-\mu_{i})$
from Equation \eqref{eq_appendix:NN_solution}, and we abbreviate
$\lambda^{*}(\widetilde{\beta_{i}}-\mu_{i})=\lambda^{*}(\widetilde{\beta_{i}}-\mu_{i};\lambda_{0},\lambda_{1})$.
In order to bound \eqref{eq:appendix_normalmeans_active3}, we bound
$U_{3},U_{4},U_{5}$ separately.
\begin{itemize}
	\item To bound $U_{3}$: When $\lambda_{0}w_i^{-1/2}-\lambda_{1}w_i^{-1/2}>2$ and
	$g(0;\lambda_{0}w_{i}^{-1/2},\lambda_{1}w_{i}^{-1/2})>0$, following
	the proof of Theorem 4.1 in \citet{rovckova2018bayesian}, we know
	that $|\widehat{\beta_{i}^{*}}|>\delta_{c+}^{\lambda_{0}/\sqrt{w_{i}},\lambda_{1}/\sqrt{w_{i}}}$
	and thus $p^{*}(\widehat{\beta_{i}^{*}};\lambda_{0}w_{i}^{-1/2},\lambda_{1}w_{i}^{-1/2})>c_{+}^{\lambda_{0}/\sqrt{w_{i}},\lambda_{1}/\sqrt{w_{i}}}$.
	This implies $p^{*}(\widetilde{\beta_{i}}-\mu_{i};\lambda_{0},\lambda_{1})>c_{+}^{\lambda_{0}/\sqrt{w_{i}},\lambda_{1}/\sqrt{w_{i}}}$, so we have
	\begin{align*}
	& \lambda^{*}(\widetilde{\beta_{i}}-\mu_{i};\lambda_{0}w_{i}^{-1/2},\lambda_{1}w_{i}^{-1/2})<c_{+}^{\lambda_{0}/\sqrt{w_{i}},\lambda_{1}/\sqrt{w_{i}}}(\lambda_{1}w_{i}^{-1/2}-\lambda_{0}w_{i}^{-1/2})+\lambda_{0}w_{i}^{-1/2}\\
	= & (1-c_{+}^{\lambda_{0}/\sqrt{w_{i}},\lambda_{1}/\sqrt{w_{i}}})(\lambda_{0}w_{i}^{-1/2}-\lambda_{1}w_{i}^{-1/2})+\lambda_{1}w_{i}^{-1/2}
	\stackrel{(a)}{<}\frac{2w_{i}^{1/2}}{\lambda_{0}-\lambda_{1}}+\lambda_{1}w_{i}^{-1/2}<1+\lambda_{1}w_{i}^{-1/2},
	\end{align*}
	where $(a)$ follows from the fact that
	\[
	c_{+}^{\lambda_{0}/\sqrt{w_{i}},\lambda_{1}/\sqrt{w_{i}}}(1-c_{+}^{\lambda_{0}/\sqrt{w_{i}},\lambda_{1}/\sqrt{w_{i}}})=\frac{1}{(\lambda_{0}w_{i}^{-1/2}-\lambda_{1}w_{i}^{-1/2})^{2}}=\frac{w_{i}}{(\lambda_{0}-\lambda_{1})^{2}},
	\]
	and 
	\[
	c_{+}^{\lambda_{0}/\sqrt{w_{i}},\lambda_{1}/\sqrt{w_{i}}}>0.5.
	\]
	In view that $\lambda^{*}(\widetilde{\beta_{i}}-\mu_{i};\lambda_{0}w_{i}^{-1/2},\lambda_{1}w_{i}^{-1/2})=w_{i}^{-1/2}\lambda^{*}(\widetilde{\beta_{i}}-\mu_{i};\lambda_{0},\lambda_{1})$,
	we have 
	\[
	U_{3}\le\E_{w_{i}}\left(\frac{w_{i}^{1/2}+\lambda_1}{w_{i}}\right)^{2}
	\le2\E_{w_{i}}\frac{1}{w_{i}}+2\lambda_1^2\E_{w_{i}}\frac{1}{w_{i}^{2}}.
	\]
	\item To bound $U_{4}$: When $\lambda_{0}w_i^{-1/2}-\lambda_{1}w_i^{-1/2}\le2$,
	we have 
	\[
	\lambda^{*}(\widetilde{\beta_{i}}-\mu_{i};\lambda_{0}w_{i}^{-1/2},\lambda_{1}w_{i}^{-1/2})\le\lambda_{0}w_{i}^{-1/2}\le\lambda_{1}w_{i}^{-1/2}+2.
	\]
	So in view that $\lambda^{*}(\widetilde{\beta_{i}}-\mu_{i};\lambda_{0}w_{i}^{-1/2},\lambda_{1}w_{i}^{-1/2})=w_{i}^{-1/2}\lambda^{*}(\widetilde{\beta_{i}}-\mu_{i};\lambda_{0},\lambda_{1})$,
	we have 
	\[
	U_{4}
	\le \E_{w_{i}}\left(\frac{\lambda_1+2w_{i}^{1/2}}{w_{i}}\right)^{2}
	\le2\lambda_1^2\E_{w_{i}}\frac{1}{w_{i}^{2}}+8\E_{w_{i}}\frac{1}{w_{i}}.
	\]
	\item To bound $U_{5}$: When $g(0;\lambda_{0}w_{i}^{-1/2},\lambda_{1}w_{i}^{-1/2})\le0$,
	we have $w_{i}\ge\frac{(\lambda_{0}-\lambda_{1})^{2}(1-p^{*}(0;\lambda_{0},\lambda_{1}))^{2}}{2\log(1/p^{*}(0;\lambda_{0},\lambda_{1}))}$,
	thus when $n$ is sufficiently large,
	\[
	U_{5}\le\E_{w_{i}}\left(\frac{\lambda_{0}}{w_{i}}\1\left(w_{i}\ge\frac{(\lambda_{0}-\lambda_{1})^{2}(1-p^{*}(0;\lambda_{0},\lambda_{1}))^{2}}{2\log(1/p^{*}(0;\lambda_{0},\lambda_{1}))}\right)\right)^{2}
	\le\frac{4[\log(1/p^{*}(0;\lambda_{0},\lambda_{1}))]^{2}\lambda_{0}^{2}}{(\lambda_{0}-\lambda_{1})^{4}(1-p^{*}(0;\lambda_{0},\lambda_{1}))^{4}}\le\frac{1}{\lambda_0}.
	\]
\end{itemize}
So plugging the above bounds into Equation \eqref{eq:appendix_normalmeans_active3},
we obtain 
\begin{align}
U_{2}
\le U_{3}+U_{4}+U_{5}
\le 10\E_{w_{i}}\frac{1}{w_{i}^{2}}+4\lambda_1^2\E_{w_{i}}\frac{1}{w_{i}}+\frac{1}{\lambda_0}.
\label{eq:appendix_normalmeans_active4}
\end{align}

Thus, from \eqref{eq:appendix_normalmeans_active1}, \eqref{eq:appendix_normalmeans_active2},
\eqref{eq:appendix_normalmeans_active4} and condition (2), we know
that for active coordinates, when $n$ is sufficiently large, 
\begin{align*}
& \E_{Y_{i}}\E_{\mu_{i},w_{i}}[(\wt\beta_{i}-\beta_{i}^{0})^{2}\C Y]
\leq 6+\left[4\left(\Delta^{U}(\lambda_{0},\lambda_{1})-\lambda_1\right)^{2}+8\lambda_1^2\right]\E_{w_{i}}\frac{1}{w_{i}}+24\E_{w_{i}}\frac{1}{w_{i}^{2}}\\
& \leq C_{3}\left(\Delta^{U}(\lambda_{0},\lambda_{1})\right)^{2}.
\end{align*}

\paragraph{Inactive coordinates}

For inactive coordinates, we have 
\begin{equation}
\begin{split} & \E_{Y_{i}}\left[\E_{\mu_{i},w_{i}}[(\widetilde{\beta_{i}}-\beta_{i}^{0})^{2}\C Y_{i}]\right]=\E_{Y_{i}}\left[\E_{\mu_{i},w_{i}}[(\widetilde{\beta_{i}})^{2}\C Y_{i}]\right]\\
= & \E_{Y_{i}}\left[\E_{\mu_{i},w_{i}}\left[(\widetilde{\beta_{i}}-\mu_{i}+\mu_{i})^{2}\1(\sqrt{w_{i}}|Y_{i}-\mu_{i}|\geq\Delta(\lambda_{0}w_{i}^{-1/2},\lambda_{1}w_{i}^{-1/2}))\C Y_{i}\right]\right]\\
& +\E_{Y_{i}}\left[\E_{\mu_{i},w_{i}}\left[(\widetilde{\beta_{i}})^{2}\1(\sqrt{w_{i}}|Y_{i}-\mu_{i}|<\Delta(\lambda_{0}w_{i}^{-1/2},\lambda_{1}w_{i}^{-1/2}))\C Y_{i}\right]\right]\\
\stackrel{(a)}{\le} & \E_{Y_{i}}\left[\E_{\mu_{i},w_{i}}\left[\left(|Y_{i}-\mu_{i}|+|\mu_{i}|\right)^{2}\1(\sqrt{w_{i}}|Y_{i}-\mu_{i}|\geq\Delta(\lambda_{0}w_{i}^{-1/2},\lambda_{1}w_{i}^{-1/2}))\C Y_{i}\right]\right]+\E_{\mu_{i}}\mu_{i}^{2}\\
\leq & 2\E_{Y_{i}}\left[\E_{\mu_{i},w_{i}}\left[\left(|Y_{i}-\mu_{i}|^{2}+|\mu_{i}|^{2}\right)\1(\sqrt{w_{i}}|Y_{i}-\mu_{i}|\geq\Delta(\lambda_{0}w_{i}^{-1/2},\lambda_{1}w_{i}^{-1/2}))\right.\right.\\
& \quad\quad\quad\left.\left.\1(w_{i}\leq\eta+\gamma)\1(|\mu_{i}|\leq\frac{1}{\lambda_{0}}+\frac{1}{\sqrt{\lambda_{0}}})\C Y_{i}\right]\right]\\
& +2\E_{Y_{i}}\left[\E_{\mu_{i},w_{i}}\left[\left(|Y_{i}-\mu_{i}|^{2}+|\mu_{i}|^{2}\right)\1(\sqrt{w_{i}}|Y_{i}-\mu_{i}|\geq\Delta(\lambda_{0}w_{i}^{-1/2},\lambda_{1}w_{i}^{-1/2}))\right.\right.\\
& \quad\quad\quad\quad\left.\left.\1(w_{i}\leq\eta+\gamma)\1(|\mu_{i}|>\frac{1}{\lambda_{0}}+\frac{1}{\sqrt{\lambda_{0}}})\C Y_{i}\right]\right]\\
& +2\E_{Y_{i}}\left[\E_{\mu_{i},w_{i}}\left[\left(|Y_{i}-\mu_{i}|^{2}+|\mu_{i}|^{2}\right)\1(\sqrt{w_{i}}|Y_{i}-\mu_{i}|\geq\Delta(\lambda_{0}w_{i}^{-1/2},\lambda_{1}w_{i}^{-1/2}))\right.\right.\\
& \quad\quad\quad\quad\left.\left.\1(w_{i}>\eta+\gamma)\C Y_{i}\right]\right]+\frac{2}{\lambda_{0}^{2}}\\
= & 2U_{1}+2U_{2}+2U_{3}+\frac{2}{\lambda_{0}^{2}},
\end{split}
\label{eq:appendix_normalmeans_inactive}
\end{equation}
where $(a)$ uses the fact that $|\widetilde{\beta}_{i}-\mu_{i}|\le|Y_{i}-\mu_{i}|$ when $\sqrt{w_{i}}|Y_{i}-\mu_{i}|>\Delta(\lambda_{0}w_{i}^{-1/2},\lambda_{1}w_{i}^{-1/2})$ and $\widetilde\beta_i=\mu_i$ when $\sqrt{w_{i}}|Y_{i}-\mu_{i}|\leq\Delta(\lambda_{0}w_{i}^{-1/2},\lambda_{1}w_{i}^{-1/2}))$
in view of Equation \eqref{eq_appendix:NN_solution}, and we denote
\begin{align*}
& U_{1}=\E_{Y_{i}}\left[\E_{\mu_{i},w_{i}}\left[\left(|Y_{i}-\mu_{i}|^{2}+|\mu_{i}|^{2}\right)\1(\sqrt{w_{i}}|Y_{i}-\mu_{i}|\geq\Delta(\lambda_{0}w_{i}^{-1/2},\lambda_{1}w_{i}^{-1/2}))\right.\right.\\
& \quad\quad\quad\left.\left.\1(w_{i}\leq\eta+\gamma)\1(|\mu_{i}|\leq\frac{1}{\lambda_{0}}+\frac{1}{\sqrt{\lambda_{0}}})\C Y_{i}\right]\right],\\
& U_{2}=\E_{Y_{i}}\left[\E_{\mu_{i},w_{i}}\left[\left(|Y_{i}-\mu_{i}|^{2}+|\mu_{i}|^{2}\right)\1(\sqrt{w_{i}}|Y_{i}-\mu_{i}|\geq\Delta(\lambda_{0}w_{i}^{-1/2},\lambda_{1}w_{i}^{-1/2}))\right.\right.\\
& \quad\quad\quad\quad\left.\left.\1(w_{i}\leq\eta+\gamma)\1(|\mu_{i}|>\frac{1}{\lambda_{0}}+\frac{1}{\sqrt{\lambda_{0}}})\C Y_{i}\right]\right],\\
& U_{3}=\E_{Y_{i}}\left[\E_{\mu_{i},w_{i}}\left[\left(|Y_{i}-\mu_{i}|^{2}+|\mu_{i}|^{2}\right)\1(\sqrt{w_{i}}|Y_{i}-\mu_{i}|\geq\Delta(\lambda_{0}w_{i}^{-1/2},\lambda_{1}w_{i}^{-1/2}))\right.\right.\\
& \quad\quad\quad\quad\left.\left.\1(w_{i}>\eta+\gamma)\C Y_{i}\right]\right].
\end{align*}
Now we bound $U_{1},U_{2},U_{3}$ separately.

For the first term $U_{1}$ in \eqref{eq:appendix_normalmeans_inactive},
when $n$ is sufficiently large, we have 
\begin{equation}
\begin{split}
2U_{1}\stackrel{(a)}{\leq} & 2\E_{Y_{i}}\left[\E_{\mu_{i},w_i}\left[|Y_{i}-\mu_{i}|^{2}\1\left(|Y_{i}-\mu_{i}|\geq\frac{\Delta^L(\lambda_0,\lambda_1)+\lambda_1(w_i^{-1/2}-1)}{\sqrt{\eta+\gamma}}\right)\1\left(|\mu_{i}|\leq\frac{1}{\lambda_{0}}+\frac{1}{\sqrt{\lambda_{0}}}\right)\C Y_{i}\right]\right]\\
& +2\E\mu_{i}^{2}\\
\le & 2\E_{\mu_{i}}\left[\E_{Y_{i}}\left[|Y_{i}-\mu_{i}|^{2}\1\left(|Y_{i}-\mu_{i}|\geq\frac{\Delta^L(\lambda_0,\lambda_1)-\lambda_1}{\sqrt{\eta+\gamma}}\right)\C\mu_{i}\right]\1\left(|\mu_{i}|\leq\frac{1}{\lambda_{0}}+\frac{1}{\sqrt{\lambda_{0}}}\right)\right]
+4/\lambda_{0}^{2}\\
\stackrel{(b)}{=} & 2\E_{\mu_{i}}\left[\1\left(|\mu_{i}|\leq 1/\lambda_0+1/\sqrt{\lambda_0}\right)\left[\left(1-\Phi\left(\mu_{i}+\frac{\Delta^L(\lambda_0,\lambda_1)-\lambda_1}{\sqrt{\eta+\gamma}}\right)\right)+\mu_{i}^{2}\right.\right.\\
& \quad\quad+\Phi\left(\mu_{i}-\frac{\Delta^L(\lambda_0,\lambda_1)-\lambda_1}{\sqrt{\eta+\gamma}}\right)-\left(\mu_{i}-\frac{\Delta^L(\lambda_0,\lambda_1)-\lambda_1}{\sqrt{\eta+\gamma}}\right)\phi\left(\mu_{i}-\frac{\Delta^L(\lambda_0,\lambda_1)-\lambda_1}{\sqrt{\eta+\gamma}}\right)\\
& \left.\left.\quad\quad+\left(\mu_{i}+\frac{\Delta^L(\lambda_0,\lambda_1)-\lambda_1}{\sqrt{\eta+\gamma}}\right)\phi\left(\mu_{i}+\frac{\Delta^L(\lambda_0,\lambda_1)-\lambda_1}{\sqrt{\eta+\gamma}}\right)\right]\right]+4/\lambda_{0}^{2}\\
\le & 8/\lambda_{0}^{2}
+4\left[1-\Phi\left(-1/\lambda_{0}-1/\sqrt{\lambda_{0}}+\frac{\Delta^L(\lambda_0,\lambda_1)-\lambda_1}{\sqrt{\eta+\gamma}}\right)\right]\\
& +2\left[\left(\frac{\Delta^L(\lambda_0,\lambda_1)-\lambda_1}{\sqrt{\eta+\gamma}}+1/\lambda_0+1/\sqrt{\lambda_0}\right)
\phi\left(1/\lambda_{0}+1/\sqrt{\lambda_{0}}-\frac{\Delta^L(\lambda_0,\lambda_1)-\lambda_1}{\sqrt{\eta+\gamma}}\right)\right]\\
\stackrel{(c)}{\le} & 8/\lambda_{0}^{2}+4\left[\frac{1}{-1/\lambda_{0}-1/\sqrt{\lambda_{0}}+\left(\Delta^{L}(\lambda_{0},\lambda_{1})-\lambda_1\right)/\sqrt{\eta+\gamma}}\right.\\
& \quad\quad\quad\quad\quad\left.\phi\left(-1/\lambda_{0}-1/\sqrt{\lambda_{0}}+\left(\Delta^{L}(\lambda_{0},\lambda_{1})-\lambda_1\right)/\sqrt{\eta+\gamma}\right)\right]\\
& +4\left[\left(\left(\Delta^{L}(\lambda_{0},\lambda_{1})-\lambda_1\right)/\sqrt{\eta+\gamma}\right)\phi\left(1/\lambda_{0}+1/\sqrt{\lambda_{0}}-\left(\Delta^{L}(\lambda_{0},\lambda_{1})-\lambda_1\right)/\sqrt{\eta+\gamma}\right)\right]\\
\le & 5\frac{\Delta^{L}(\lambda_{0},\lambda_{1})-\lambda_1}{\sqrt{\eta+\gamma}}\left[\phi\left(1/\lambda_{0}+1/\sqrt{\lambda_{0}}-\left(\Delta^{L}(\lambda_{0},\lambda_{1})-\lambda_1\right)/\sqrt{\eta+\gamma}\right)\right],
\end{split}
\label{eq:appendix_normalmeans_inactive_term1}
\end{equation}
where $(a)$ follows from the fact that when $n$ is sufficiently
large, $w_{i}\le\eta+\gamma\Rightarrow g(0;\lambda_{0}w_{i}^{-1/2},\lambda_{1}w_{i}^{-1/2})>0,\lambda_{0}w_{i}^{-1/2}-\lambda_{1}w_{i}^{-1/2}>2$, which ensures that Lemma \ref{lemma_appendix:Delta_bound} holds, $(b)$ uses Lemma \ref{lemma_appendix:condtional_normal},
$(c)$ follows from the fact that $1-\Phi(x)\le \phi(x)/x$ for all $x>0$, and here $x=-1/\lambda_0-1/\sqrt{\lambda_0}+(\Delta^L(\lambda_0,\lambda_1)-\lambda_1)/\sqrt{\eta+\gamma}>0$ always holds when $n$ is sufficiently large.

For the second term in \eqref{eq:appendix_normalmeans_inactive},
\begin{align*}
2U_{2} & =2\E_{Y_{i}}\left[\E_{\mu_{i},w_{i}}\left[\left(|Y_{i}-\mu_{i}|^{2}+|\mu_{i}|^{2}\right)\1(\sqrt{w_{i}}|Y_{i}-\mu_{i}|\geq\Delta(\lambda_{0}/\sqrt{w_{i}},\lambda_{1}/\sqrt{w_{i}}))\right.\right.\\
& \left.\quad\quad\quad\quad\left.\1(w_{i}\leq\eta+\gamma)\1\left(|\mu_{i}|>\frac{1}{\lambda_{0}}+\frac{1}{\sqrt{\lambda_{0}}}\right)\mid Y_{i}\right]\right]\\
& \leq2\E_{Y_{i}}\left[\E_{\mu_{i},w_{i}}\left[\left(|Y_{i}-\mu_{i}|^{2}+|\mu_{i}|^{2}\right)\1\left(|\mu_{i}|>\frac{1}{\lambda_{0}}+\frac{1}{\sqrt{\lambda_{0}}}\right)\mid Y_{i}\right]\right]\\
& \leq2\E_{Y_{i}}\left[\E_{\mu_{i},w_{i}}\left[(2Y_{i}^{2}+3\mu_{i}^{2})\1\left(|\mu_{i}|>\frac{1}{\lambda_{0}}+\frac{1}{\sqrt{\lambda_{0}}}\right)\mid Y_{i}\right]\right]\\
& =4\P(|\mu_{i}|>\frac{1}{\lambda_{0}}+\frac{1}{\sqrt{\lambda_{0}}})+6\E_{\mu_{i}}\mu_{i}^{2}\1\left(|\mu_{i}|>\frac{1}{\lambda_{0}}+\frac{1}{\sqrt{\lambda_{0}}}\right)\\
& \stackrel{(a)}{=}4e^{-1-\sqrt{\lambda_{0}}}+6\left[\left(\frac{1}{\lambda_{0}}+\frac{1}{\sqrt{\lambda_{0}}}\right)^{2}+\frac{2}{\lambda_{0}}\left(\frac{2}{\lambda_{0}}+\frac{1}{\sqrt{\lambda_{0}}}\right)\right]e^{-1-\sqrt{\lambda_{0}}}\\
& <\frac{\Delta^{L}(\lambda_{0},\lambda_{1})-\lambda_1}{\sqrt{\eta+\gamma}}\left[\phi\left(1/\lambda_{0}+1/\sqrt{\lambda_{0}}-\left(\Delta^{L}(\lambda_{0},\lambda_{1})-\lambda_1\right)/\sqrt{\eta+\gamma}\right)\right],
\end{align*}
where $(a)$ follows from integration by parts.

For the third term in \eqref{eq:appendix_normalmeans_inactive}, utilizing
condition (3), when $n$ is sufficiently large, we have 
\begin{align*}
2U_{3} & =2\E_{Y_{i}}\left[\E_{\mu_{i},w_{i}}\left[\left(|Y_{i}-\mu_{i}|^{2}+|\mu_{i}|^{2}\right)\1(\sqrt{w_{i}}|Y_{i}-\mu_{i}|\geq\Delta(\lambda_{0}/\sqrt{w_{i}},\lambda_{1}/\sqrt{w_{i}}))\1(w_{i}>\eta+\gamma)\mid Y_{i}\right]\right]\\
& \leq2\E_{Y_{i}}\left[\E_{\mu_{i},w_{i}}\left[\left(2|Y_{i}|^{2}+3|\mu_{i}|^{2}\right)\1(w_{i}>\eta+\gamma)\mid Y_{i}\right]\right]\\
& \le\left(4+\frac{12}{\lambda_{0}^{2}}\right)\mathbb{P}\left(w_{i}>\eta+\gamma\right)\le\widetilde{C_{4}}\frac{q}{n}\sqrt{\log\left(\frac{n}{q}\right)}\\
& \le C_{4}\frac{\Delta^{L}(\lambda_{0},\lambda_{1})-\lambda_1}{\sqrt{\eta+\gamma}}\left[\phi\left(1/\lambda_{0}+1/\sqrt{\lambda_{0}}-\left(\Delta^{L}(\lambda_{0},\lambda_{1})-\lambda_1\right)/\sqrt{\eta+\gamma}\right)\right].
\end{align*}
In \eqref{eq:appendix_normalmeans_inactive}, combining the bound
on $U_{1},U_{2},U_{3}$, the risk for inactive coordinates will be
bounded 
\[
\begin{split} & \E_{Y_{i}}\left[\E_{\mu_{i},w_{i}}[(\tilde{\beta}_{i}-\beta_{i}^{0})^{2}\C Y_{i}]\right]=2U_{1}+2U_{2}+2U_{3}+2/\lambda_{0}^{2}\\
& \leq C_{5}\frac{\Delta^{L}(\lambda_{0},\lambda_{1})-\lambda_1}{\sqrt{\eta+\gamma}}\left[\phi\left(1/\lambda_{0}+1/\sqrt{\lambda_{0}}-\left(\Delta^{L}(\lambda_{0},\lambda_{1})-\lambda_1\right)/\sqrt{\eta+\gamma}\right)\right].
\end{split}
\]

Combining the risk for active and inactive coordinates, we obtain% {\color{blue} typo for $(\|\wt\b-\b^{0}\|_{2}^{2}$}
\[
\begin{split} 
& \E_{\Y}\E_{\bm{\mu},\bm{w}}[\|\widetilde{\b}-\b^{0}\|_2^{2}\C \Y]\\
\le & qC_{3}\left[\Delta^{U}(\lambda_{0},\lambda_{1})\right]^{2}+(n-q)C_{5}\frac{\Delta^{L}(\lambda_{0},\lambda_{1})-\lambda_1}{\sqrt{\eta+\gamma}}\left[\phi\left(1/\lambda_{0}+1/\sqrt{\lambda_{0}}-\left(\Delta^{L}(\lambda_{0},\lambda_{1})-\lambda_1\right)/\sqrt{\eta+\gamma}\right)\right]\\
= & qC_{3}\left[\Delta^{U}(\lambda_{0},\lambda_{1})\right]^{2}+(n-q)C_{6}\frac{\Delta^{L}(\lambda_{0},\lambda_{1})}{\sqrt{\eta+\gamma}}\exp\left\{ -\frac{\left(1/\lambda_{0}+1/\sqrt{\lambda_{0}}-\left(\Delta^{L}(\lambda_{0},\lambda_{1})-\lambda_1\right)/\sqrt{\eta+\gamma}\right)^{2}}{2}\right\} \\
\le & qC_{3}\left[\Delta^{U}(\lambda_{0},\lambda_{1})\right]^{2}+(n-q)C_{7}\Delta^{L}(\lambda_{0},\lambda_{1})\exp\left\{ -\frac{\left(\Delta^{L}(\lambda_{0},\lambda_{1})\right)^{2}}{2(\eta+\gamma)}\right\} \\
\le & qC_{3}\left[\Delta^{U}(\lambda_{0},\lambda_{1})\right]^{2}+(n-q)C_{7}\Delta^{L}(\lambda_{0},\lambda_{1})\frac{q}{n}\\
\le & qC_{8}\left[\Delta^{U}(\lambda_{0},\lambda_{1})\right]^{2}.
\end{split}
\]

Then from the Markov's inequality, for any $M_{n}\rightarrow\infty$, 
\[
\E_{\Y}\P_{\bm{\mu},\bm{w}}\left(\|\wt\b-\b^{0}\|_{2}^{2}>M_{n}q\log\left(\frac{n}{q}\right)\C \Y\right)
\leq\E_{\Y}\frac{\E_{\bm{\mu},\bm{w}}\left[\|\wt\b-\b^{0}\|_2^{2}\C \Y\right]}{M_{n}q\log\frac{n}{q}}
\leq\frac{qC_{8}\left[\Delta^{U}(\lambda_{0},\lambda_{1})\right]^{2}}{M_{n}q\log\frac{n}{q}},
\]
where $\frac{qC_{3}\left[\Delta^{U}(\lambda_{0},\lambda_{1})\right]^{2}}{M_{n}q\log\frac{n}{q}}\rightarrow0$.
This means when $\sigma=1$, for any $M_{n}\rightarrow\infty$, 
\[
\E_{\Y}\P_{\bm{\mu},\bm{w}}\left(\|\wt\b-\b^{0}\|_{2}^{2}>M_{n}q\log\left(\frac{n}{q}\right)\C \Y\right)\rightarrow0.
\]

\smallskip

For a general fixed value $\sigma>0$, notice that we can always rescale the model
so that it becomes
\[
Y_{i}/\sigma=\beta_{i}^{0}/\sigma+N(0,1).
\]
%{\color{blue}how does $\sigma^2$ affect the threshold $\Delta$? What if $\sigma$ depends on n? Let's discuss.}
And thus all previous analysis holds for the rescaled model. Thus, for any $M_{n}\rightarrow\infty$, 
\begin{align*}
& \E_{\Y}\P_{\bm{\mu},\bm{w}}\left(\|\wt\b-\b^{0}\|_{2}^{2}>M_{n}q\log\left(\frac{n}{q}\right)\C \Y\right)
=\E_{\Y}\P_{\bm{\mu},\bm{w}}\left(\|\wt\b/\sigma-\b^{0}/\sigma\|_{2}^{2}>\sigma^{-2}M_{n}q\log\left(\frac{n}{q}\right)\C \Y\right)\\
\le & \E_{\Y}\frac{\E_{\bm{\mu},\bm{w}}\left[\|\wt\b/\sigma-\b^{0}/\sigma\|_{2}^{2}\C \Y\right]}{\sigma^{-2}M_{n}q\log\left(\frac{n}{q}\right)}\rightarrow0.
\end{align*}

\subsection{Proof of Corollary \ref{normal_mean_dirichlet}}

\label{sec:appendix:normal_means_corollary_proof} 

The condition (1) of
Theorem \ref{normal_means_weight} is satisfied. With $\alpha\geq2$,
the condition (2) also holds because when $w_{i}\sim\frac{1}{\alpha}\text{Gamma}(\alpha,1)$,
$\frac{1}{w_{i}}\sim\alpha\times\text{Inverse-Gamma}(\alpha,1)$ and
when $\bm{w}\sim n\text{Dir}(\alpha,\cdots,\alpha)$, $w_{i}\sim n\times\text{Beta}(\alpha,(n-1)\alpha)$.
Both of them satisfy condition (2). So we only need to check condition
(3).

When $\bm{w}\sim n\text{Dir}(\alpha,\cdots,\alpha)$, the following equation holds for any
$t\geq\frac{\alpha+1}{\alpha}$,
\begin{align*}
& \P_{w_{i}}(w_{i}>t)=\P\left(\text{Beta}(\alpha,n\alpha-\alpha)>\frac{t}{n}\right)=\frac{\int_{t/n}^{1}v^{\alpha-1}(1-v)^{n\alpha-\alpha-1}dv}{\text{Beta}(\alpha,n\alpha-\alpha)}\\
& \stackrel{z=\alpha nv}{=}\frac{1}{\text{Beta}(\alpha,n\alpha-\alpha)}\int_{\alpha t}^{\alpha n}\left(\frac{z}{\alpha n}\right)^{\alpha-1}\left(1-\frac{z}{\alpha n}\right)^{n\alpha-\alpha-1}d\frac{z}{\alpha n}\\
& =\frac{1}{\text{Beta}(\alpha,n\alpha-\alpha)(\alpha n)^{\alpha}}\int_{\alpha t}^{\alpha n}z^{\alpha-1}\left[\left(1-\frac{1}{\alpha n/z}\right){}^{-\alpha n/z+1}\right]{}^{\frac{n\alpha-\alpha-1}{1-\alpha n/z}}dz\\
& \stackrel{(a)}{\leq}\frac{1}{\text{Beta}(\alpha,n\alpha-\alpha)(\alpha n)^{\alpha}}\int_{\alpha t}^{\alpha n}z^{\alpha-1}e^{\frac{n\alpha-\alpha-1}{z-\alpha n}z}dz\\
& \leq\frac{1}{\text{Beta}(\alpha,n\alpha-\alpha)(\alpha n)^{\alpha}}\int_{\alpha t}^{\alpha n}z^{\alpha-1}e^{\frac{n\alpha-\alpha-1}{\alpha t-\alpha n}z}dz\\
& \stackrel{(b)}{\leq}\frac{1}{\text{Beta}(\alpha,n\alpha-\alpha)(\alpha n)^{\alpha}}\int_{\alpha t}^{\alpha n}z^{\alpha-1}e^{-z}dz=\frac{\Gamma(n\alpha)}{\Gamma(n\alpha-\alpha)(\alpha n)^{\alpha}\Gamma(\alpha)}\int_{\alpha t}^{\alpha n}z^{\alpha-1}e^{-z}dz\\
& \leq\frac{\Gamma(n\alpha)}{\Gamma(n\alpha-\alpha)(\alpha n)^{\alpha}}\P(\text{Gamma}(\alpha,1)>\alpha t)\leq\P(\text{Gamma}(\alpha,1)>\alpha t)\\
& \stackrel{(c)}{\leq}\frac{\E_{v\sim\text{Gamma}(\alpha,1)}e^{xv}}{e^{\alpha tx}}=\frac{(1-x)^{-\alpha}}{e^{\alpha tx}}=e^{\alpha(-\log(1-x)-tx)},
\end{align*}
where (a) uses the fact $(1-\frac{1}{x})^{-x+1}\leq e$ for any $x>0$.
Inequality (b) uses $t\geq\frac{\alpha+1}{\alpha}$. Inequality (c)
uses Chernoff bound and we need $x\in(0,1)$ for the above inequality
to hold.
Notice that when $w_{i}\sim\frac{1}{\alpha}\text{Gamma}(\alpha,1)$,
we can directly get 
\[
\P_{w_{i}}(w_{i}>t)=\P(\text{Gamma}(\alpha,1)>\alpha t)\leq e^{\alpha(-\log(1-x)-tx)}.
\]
Setting $x=1-\frac{1}{t}$ and %$\alpha\geq \max(\frac{\frac{\Delta^2}{7.9}+log(2)-log(\Delta)}{2-log(4)}, \frac{\frac{(\Delta^*)^2}{7.9}+log(2)-log(\Delta^*)}{2-log(4)})$ where $\Delta^*=\frac{t}{\eta+\gamma}\Delta^2$, then 
$\alpha=\frac{\log[\frac{(1-\theta)\lambda_{0}}{\theta\lambda_{1}}]}{(t-1-\log t)(\eta+\gamma)}$,
we have
\begin{equation}
\begin{split} 
& \P_{w_{i}}(w_{i}>t)\leq e^{\alpha(-\log(1-x)-tx)}=e^{\alpha(\log(t)+1-t)}\asymp\frac{q}{n}.
\end{split}
\label{alpha_rate}
\end{equation}
Setting $t=\eta+\gamma$, we notice that $t=\eta+\gamma\ge\frac{\alpha+1}{\alpha}$
when $n$ is sufficiently large and thus condition (3) always holds
when $n$ is sufficiently large. We can get Lemma \ref{normal_mean_dirichlet}
by applying Theorem \ref{normal_means_weight}.

Now let us consider the case when   $\alpha$ depends on $\sigma^{2}$.
From the proof of Theorem \ref{normal_means_weight}, we know that for
any fixed $\sigma>0$,
\begin{align*}
& \E_{\Y}\P_{\bm{\mu},\bm{w}}\left(\|\wt\b-\b^{0}\|_{2}^{2}>M_{n}q\log\left(\frac{n}{q}\right)\C \Y\right)=\E_{\Y}\P_{\bm{\mu},\bm{w}}\left(\|\wt\b/\sigma-\b^{0}/\sigma\|_{2}^{2}>\sigma^{-2}M_{n}q\log\left(\frac{n}{q}\right)\C \Y\right)\\
\le & \E_{\Y}\frac{\E_{\bm{\mu},\bm{w}}\left[\|\wt\b/\sigma-\b^{0}/\sigma\|_{2}^{2}\C \Y\right]}{\sigma^{-2}M_{n}q\log\left(\frac{n}{q}\right)}
\stackrel{(a)}{\lesssim} \frac{\sigma^{2}}{M_{n}q\log\frac{n}{q}}q\log\left(\frac{n}{q}\right)\E_{w_{i}}\frac{1}{w_{i}},
\end{align*}
where $(a)$ follows from the fact that the dominating term in the upper
bound for risk $\E_{\Y}\P_{\bm{\mu},\bm{w}}\left(\|\wt\b/\sigma-\b^{0}/\sigma\|_{2}^{2}>M_{n}q\log\left(\frac{n}{q}\right)\C \Y\right)$
comes from $4q[\Delta^U(\lambda_0,\lambda_1)]^2\E_{w_{i}}\frac{1}{w_{i}}$. When $w_{i}\sim\frac{1}{\alpha}\text{Gamma}(\alpha,1),$we
have $\E_{w_{i}}\frac{1}{w_{i}}=1+\frac{1}{\alpha-1}$. Thus 
\[
\E_{\Y}\P_{\bm{\mu},\bm{w}}\left(\|\wt\b-\b^{0}\|_{2}^{2}>M_{n}q\log\left(\frac{n}{q}\right)\C \Y\right)\lesssim\frac{1}{M_{n}}\left(\sigma^{2}+\frac{\sigma^{2}}{\alpha-1}\right).
\]
An ideal choice of $\alpha$ should thus satisfy $\alpha-1\propto\sigma^{2}$.
We thereby suggest choosing $\alpha\gtrsim\sigma^{2}\log[\frac{(1-\theta)\lambda_{0}}{\theta\lambda_{1}}]$
when $\bm{w}\sim n\times\text{Dir}(\alpha,\cdots,\alpha)$, observing
that $w_{i}\sim n\times\text{Beta}(\alpha,(n-1)\alpha)$ and thus
\[
\E_{w_{i}}\frac{1}{w_{i}}=\frac{1}{n}\frac{\alpha+(n-1)\alpha-1}{\alpha-1}=1+\frac{n-1}{n(\alpha-1)}.
\]
%{\color{blue} is the following sentence redundant?}Similar arguments as above leads to $\alpha\gtrsim\sigma^{2}\log[\frac{(1-\theta)\lambda_{0}}{\theta\lambda_{1}}]$.

\subsection{Explanation of Remarks below  Corollary \ref{normal_means_weight}}
\label{sec:appendix_normal_means_remark}

We utilize the notation introduced in Section \ref{sec_appendix:normal_means_notations}.
Here we want to show that for $\bm{w}\sim n\times\text{Dir}(\alpha,\cdots,\alpha)$
where $0<\alpha<2$ is a fixed constant, the risk for BB-SSL arising
from \eqref{eq: normal_means_model} can be arbitrarily large.

From the proof of Theorem \ref{normal_means_weight}, we know that
for active coordinate $i$, when $n$ is sufficiently large, 
\begin{equation}\label{eq_appendix:remark}
\begin{split}
& \E_{\bm{Y}}\E_{\bm{\mu,w}}[\|\widetilde{\bm{\beta}}-\bm{Y}\|_{2}^{2}\mid\bm{Y}]\\
\ge & \E_{Y_{i}}\left[\E_{\mu_{i},w_{i}}\left[(\widetilde{\beta_{i}}-Y_{i})^{2}\1\left(\sqrt{w_{i}}(Y_{i}-\mu_{i})>\Delta(\lambda_{0}w_{i}^{-1/2},\lambda_{1}w_{i}^{-1/2})\right)\1\left(|Y_{i}-\mu_{i}|\ge w_{i}^{-1}\lambda^{*}(\tilde{\beta}_{i}-\mu_{i})\right)\mid Y_{i}\right]\right]\\
\stackrel{(a)}{=} & \E_{Y_{i}}\left[\E_{\mu_{i},w_{i}}\left[\left(w_{i}^{-1}\lambda^{*}(\tilde{\beta}_{i}-\mu_{i})\right){}^{2}\1\left(\sqrt{w_{i}}(Y_{i}-\mu_{i})>\Delta(\lambda_{0}w_{i}^{-1/2},\lambda_{1}w_{i}^{-1/2})\right)\right.\right.
\\&\quad\quad\left.\left.\1\left(|Y_{i}-\mu_{i}|\ge w_{i}^{-1}\lambda^{*}(\tilde{\beta}_{i}-\mu_{i})\right)\mid Y_{i}\right]\right]\\
\ge & \lambda_{1}^{2}\E_{Y_{i},\mu_{i}}\left[\E_{w_{i}}\left[\frac{1}{{w_{i}^{2}}}
\1\left({w_{i}}>\frac{\Delta^{2}(\lambda_{0}w_{i}^{-1/2},\lambda_{1}w_{i}^{-1/2})}{(Y_{i}-\mu_{i})^{2}}\right)\1\left(w_{i}\ge\frac{\lambda_{0}}{|Y_{i}-\mu_{i}|}\right)
\right.\right.
\\&\quad\quad\quad\quad\left.\left.\1\left(g(0;\lambda_0w_i^{-1/2},\lambda_1w_i^{-1/2})>0,\,\lambda_0w_i^{-1/2}-\lambda_1w_i^{-1/2}>2\right)\mid Y_{i},\mu_{i}\right]\right]\\
\stackrel{(b)}{\ge} & \lambda_{1}^{2}\E_{Y_{i},\mu_{i}}\left[\E_{w_{i}}\left[\frac{1}{{w_{i}^{2}}}
\1\left({w_{i}}>\frac{[\Delta^{U}(\lambda_{0},\lambda_{1})+\lambda_{1}(w_i^{-1/2}-1)]^{2}}{(Y_{i}-\mu_{i})^{2}}\right)\1\left(w_{i}\ge\frac{\lambda_{0}}{|Y_{i}-\mu_{i}|}\right)
\1\left(w_{i}<B_{n}\right)\mid Y_{i},\mu_{i}\right]\right]\\
\ge & \lambda_{1}^{2}\E_{Y_{i},\mu_{i}}\left[\E_{w_{i}}\left[\frac{1}{{w_{i}^{2}}}
\1\left({w_{i}}>\frac{[\Delta^{U}(\lambda_{0},\lambda_{1})+\lambda_{1}(\sqrt{|Y_i-\mu_i|/\lambda_0}-1)]^{2}}{(Y_{i}-\mu_{i})^{2}}\right)\right.\right.
\\&\quad\quad\quad\quad\left.\left.\1\left(w_{i}\ge\frac{\lambda_{0}}{|Y_{i}-\mu_{i}|}\right)\1\left(w_{i}<B_{n}\right)\mid Y_{i},\mu_{i}\right]\right]\\
\ge & \lambda_{1}^{2}\E_{Y_{i},\mu_{i}}\left[\E_{w_{i}}\left[\frac{1}{{w_{i}^{2}}}\1\left({w_{i}}>C(Y_{i},\mu_{i})\right)\1\left(w_{i}<B_{n}\right)\mid Y_{i},\mu_{i}\right]\right],
\end{split}
\end{equation}
where $(a)$ follows from Equation \eqref{eq_appendix:NN_solution},
$(b)$ follows from Lemma \ref{lemma_appendix:Delta_bound}, and we denote
\[
B_{n}=\min\left\{\frac{(\lambda_{0}-\lambda_{1})^{2}(1-p^{*}(0;\lambda_{0},\lambda_{1}))^{2}}{2\log[1/p^{*}(0;\lambda_{0},\lambda_{1})]},\frac{(\lambda_0-\lambda_1)^2}{4}\right\},
\]
\[
C(Y_{i},\mu_{i})=\max\left\{ \frac{[\Delta^{U}(\lambda_{0},\lambda_{1})+\lambda_{1}(\sqrt{|Y_i-\mu_i|/\lambda_0}-1)]^{2}}{(Y_{i}-\mu_{i})^{2}},
\frac{\lambda_{0}}{|Y_{i}-\mu_{i}|}\right\}.
\]
Notice that $n\text{Beta}(\alpha,n\alpha-\alpha)\xrightarrow{\text{d}}\frac{1}{\alpha}\text{Gamma}(\alpha,1)$
and as $n\rightarrow\infty$,
\[
\E_{w_{i}}\left[\frac{1}{{w_{i}^{2}}}\1\left(B_{n}>{w_{i}}>C(Y_{i},\mu_{i})\right)\mid Y_{i},\mu_{i}\right]\rightarrow\E_{v\sim\text{Gamma}(\alpha,1)}\left[\frac{\alpha^{2}}{{v^{2}}}\1\left(\alpha B_{n}>{v}>\alpha C(Y_{i},\mu_{i})\right)\mid Y_{i},\mu_{i}\right],
\]
\[
\E_{v\sim\text{Gamma}(\alpha,1)}\left[\frac{\alpha^{2}}{{v^{2}}}\1\left(\alpha B_{n}>{v}>\alpha C(Y_{i},\mu_{i})\right)\mid Y_{i},\mu_{i}\right]\le\frac{1}{C^{2}(Y_{i},\mu_{i})},
\]
\[
\frac{1}{C^{2}(Y_{i},\mu_{i})}\le\frac{|Y_{i}-\mu_{i}|^{2}}{\lambda_0^{2}},
\]
where $\frac{|Y_{i}-\mu_{i}|^{2}}{\lambda_{0}^{2}}$ is integrable.
Thus, direct application of the Dominated Convergence Theorem shows
that as $n\rightarrow\infty$,
\begin{align*}
& \E_{Y_{i},\mu_{i}}\left[\E_{w_{i}}\left[\frac{1}{{w_{i}^{2}}}\1\left({w_{i}}>C(Y_{i},\mu_{i})\right)\1\left(w_{i}<B_{n}\right)\mid Y_{i},\mu_{i}\right]\right]\\
\rightarrow & \E_{Y_{i},\mu_{i}}\left[\E_{v\sim\text{Gamma}(\alpha,1)}\left[\frac{\alpha^{2}}{{v^{2}}}\1\left(\alpha B_{n}>{v}>\alpha C(Y_{i},\mu_{i})\right)\mid Y_{i},\mu_{i}\right]\right].
\end{align*}
Plugging into Equation \eqref{eq_appendix:remark}, we know that when $n$ is sufficiently
large,
\begin{align*}
& \E_{\bm{Y}}\E_{\bm{\mu,w}}[\|\widetilde{\bm{\beta}}-\bm{Y}\|_{2}^{2}\mid\bm{Y}]\\
\ge & \frac{1}{2}\lambda_{1}^{2}\E_{Y_{i},\mu_{i}}\left[\E_{v\sim\text{Gamma}(\alpha,1)}\left[\frac{\alpha^{2}}{{v^{2}}}\1\left(\alpha B_{n}>{v}>\alpha C(Y_{i},\mu_{i})\right)\mid Y_{i},\mu_{i}\right]\right]\\
= & \frac{1}{2}\lambda_{1}^{2}\alpha^{2}\E_{Y_{i},\mu_{i}}\left[\E_{z\sim\text{Inv-Gamma}(\alpha,1)}\left[z^{2}\1(\frac{1}{\alpha B_{n}}<{z}<\frac{1}{C(Y_{i},\mu_{i})\alpha})\mid Y_{i},\mu_{i}\right]\right].
\end{align*}

For all constants $M,m$ which satisfy $M>2m>0$, we have
\begin{align*}
& \E_{z\sim\text{Inv-Gamma}(\alpha,1)}z^{2}\1(m<{z}<M)=\frac{1}{\Gamma(\alpha)}\int_{m}^{M}z^{2}\frac{1}{z^{\alpha+1}}e^{-1/z}dz=\frac{1}{\Gamma(\alpha)}\int_{m}^{M}z^{1-\alpha}e^{-1/z}dz\\
\stackrel{x=1/z}{=} & \frac{1}{\Gamma(\alpha)}\int_{1/M}^{1/m}x^{\alpha-3}e^{-x}dx>\frac{1}{\Gamma(\alpha)}\int_{1/M}^{2/M}x^{\alpha-3}e^{-x}dx
=\frac{M^{2-\alpha}\exp\left\{ -2/M\right\} }{\Gamma(\alpha)2^{3-\alpha}}.
\end{align*}
We set
\[
M=\frac{1}{C(Y_{i},\mu_{i})\alpha},\quad m=\frac{1}{\alpha B_{n}}.
\]
Then, suppose $\beta_{i}^{0}>2$, when $n$ is sufficiently large,
we have
\begin{align*}
& \E_{\bm{Y}}\E_{\bm{\mu,w}}[\|\widetilde{\bm{\beta}}-\bm{Y}\|_{2}^{2}\mid\bm{Y}]\\
\ge & \frac{1}{2}\lambda_{1}^{2}\alpha^{2}\E_{Y_{i},\mu_{i}}\left[\E_{z\sim\text{Inv-Gamma}(\alpha,1)}\left[z^{2}\1(m<{z}<M)\1\left(M>2m\right)\mid Y_{i},\mu_{i}\right]\right]\\
\ge & \frac{1}{2}\lambda_{1}^{2}\alpha^{2}\E_{Y_{i},\mu_{i}}\left[\frac{M^{2-\alpha}\exp\left\{ -2/M\right\} }{\Gamma(\alpha)2^{3-\alpha}}\1\left(M>2m\right)\right]\\
= & \frac{\lambda_{1}^{2}\alpha^{\alpha}}{2^{4-\alpha}\Gamma(\alpha)}\E_{Y_{i},\mu_{i}}\left[\left(\frac{1}{C(Y_{i},\mu_{i})}\right)^{2-\alpha}\exp\left\{ -2C(Y_{i},\mu_{i})\alpha\right\} \1\left(C(Y_{i},\mu_{i})<B_{n}/2\right)\right]\\
\ge & \frac{\lambda_{1}^{2}\alpha^{\alpha}}{2^{4-\alpha}\Gamma(\alpha)}\E_{Y_{i},\mu_{i}}\left[\left(\frac{1}{C(Y_{i},\mu_{i})}\right)^{2-\alpha}\exp\left\{ -2C(Y_{i},\mu_{i})\alpha\right\} \1\left(C(Y_{i},\mu_{i})<B_{n}/2\right)\1\left(|Y_{i}-\beta_{i}^{0}|\le1\right)\1(|\mu_{i}|<1)\right]\\
\ge & \frac{\lambda_{1}^{2}\alpha^{\alpha}}{2^{4-\alpha}\Gamma(\alpha)}\E_{Y_{i},\mu_{i}}\left[\left(\frac{1}{C_{n}}\right)^{2-\alpha}\exp\left\{ -2C_{n}\alpha\right\} \1\left(C_{n}<B_{n}/2\right)\1\left(|Y_{i}-\beta_{i}^{0}|\le1\right)\1(|\mu_{i}|<1)\right]\\
= & \frac{\lambda_{1}^{2}\alpha^{\alpha}}{2^{4-\alpha}\Gamma(\alpha)}\left(1/C_{n}\right)^{2-\alpha}\exp\left\{ -2C_{n}\alpha\right\} \1\left(C_{n}<B_{n}/2\right)\P\left(|N(0,1)|\le1\right)\P(|\mu_{i}|<1)\\
= & C\frac{\lambda_{1}^{2}\alpha^{\alpha}}{2^{4-\alpha}\Gamma(\alpha)}\left(1/C_{n}\right)^{2-\alpha}\exp\left\{ -2C_{n}\alpha\right\} \1\left(C_{n}<B_{n}/2\right),
\end{align*}
where 
\[
C_{n}=\max\left\{ \frac{[\Delta^{U}(\lambda_{0},\lambda_{1})+\lambda_{1}(\sqrt{(\beta_{i}^{0}+2)/\lambda_{0}}-1)]^{2}}{(\beta_{i}^{0}-2)^{2}},\frac{\lambda_{0}}{|\beta_{i}^{0}-2|}\right\} .
\]
Notice that $C_{n}<B_{n}/2$ is always satisfied when $n$ is sufficiently
large, so we know that when $n$ is sufficiently large,
\[
\E_{\bm{Y}}\E_{\bm{\mu,w}}[\|\widetilde{\bm{\beta}}-\bm{Y}\|_{2}^{2}\mid\bm{Y}]\ge C\frac{\lambda_{1}^{2}\alpha^{\alpha}}{2^{4-\alpha}\Gamma(\alpha)}\left(1/C_{n}\right)^{2-\alpha}\exp\left\{ -2C_{n}\alpha\right\},
\]
where the lower bound depends on $\beta_{i}^{0}$ through $C_n$. For any fixed,
sufficiently large $n$, when $\beta_{i}^{0}$ becomes larger and
larger, $C_{n}$ becomes smaller and smaller, and as $\beta_{i}^{0}\rightarrow\infty$,
$C_{n}\rightarrow0$, and the lower bound goes to infinity, which
implies that $\E_{\bm{Y}}\E_{\bm{\mu,w}}[\|\widetilde{\bm{\beta}}-\bm{Y}\|_{2}^{2}\mid\bm{Y}]\rightarrow\infty$ as $\beta_i^0\rightarrow\infty$.

Thus, we conclude that the risk $\E_{\bm{Y}}\E_{\bm{\mu,w}}[\|\widetilde{\bm{\beta}}-\bm{Y}\|_{2}^{2}\mid\bm{Y}]$
depends on the truth $\bm{\beta^{0}}$ and for any fixed, sufficiently
large $n$, when one of the coordinates $\beta_{i}^{0}$ becomes arbitrarily
large, the risk can also be arbitrarily large.

\subsection{Proof of Theorem \ref{thm: regression_weight_selected} and Theorem \ref{thm:regression_weight}}
\label{sec_appendix:regression}

\subsubsection{Definitions and Lemmas used in Theorem \ref{thm: regression_weight_selected} and Theorem \ref{thm:regression_weight}}\label{sec:appendix_regression_proof}

We rewrite model \eqref{eq:model} in a matrix form $\Y=\X\b+\bm{\epsilon}$ and   denote %{\color{blue} below you use $\pen$ but elsewhere you use $pen$. I would choose the second one. It does not matter as long as the notation is the same. Also, sometimes you write $pen(\cdot\C\theta)$ and sometimes $pen(\cdot)$. Please fix.}
$$
\pen(\b\mid\theta)=\log\left[\frac{\pi(\b\C\theta)}{\pi(\bm{0}_{p}\C\theta)}\right],
$$ 
where $\bm 0_p\in\mathbb{R}^p$ is a vector of all 0's.
We write $\bm{W}=\text{diag}(\sqrt{w_{1}},\cdots,\sqrt{w_{n}})\in\mathbb{R}^{n\times n}$ and   denote with
$$
Q(\b)=-\frac{1}{2\sigma^2}\|\bm{W}\Y-\bm{W}\X\bm{\mu}-\bm{W}\X\b\|_2^{2}+\pen(\b\C\theta)
$$
and with
\begin{equation}\label{eq:bhat}
\wh\b=\arg\max_{\b}Q(\b).
\end{equation}
Notice that  the BB-SSL solution $\wt\b$ satisfies  $\wt\b=\wh\b+\bm{\mu}$. The matrix norm $\|\cdot\|_{a}$ is defined
as $\|\X\|_{a}=\sup_{\b}\frac{\|\X\b\|_{a}}{\|\b\|_{a}}$ where $\|\cdot\|_{a}$
is the vector $a$-norm. Write $\Th=\wh\b-\b_{0}$. We use $\|\cdot\|_{1}$ to denote the vector 1-norm. We define 
\begin{equation}\label{def_appendix:delta_regression}
\Delta=\inf_{t>0}[nt/2-\sigma^2\rho(t\C\theta)/t],
\end{equation}
with $\rho(t\C\theta)$ defined as 
$$
\rho(t\C\theta)=-\lambda_{1}|t|+\log[p^{*}(0)/p^{*}(t)],
$$ 
where $p^*(t)$ is equal to the $p^*(t;\lambda_0,\lambda_1)$ defined in Equation \eqref{eq:pstar}. We fix $\lambda_0,\lambda_1$ in $p^*(t;\lambda_0,\lambda_1)$ and thus simply write $p^*(t)$. %{\color{blue}The notation is different, I like the notation $p^*(t)$ as opposed to $p^*(t;\lambda_0,\lambda_1)$. Is there a reason to use the second one in the proof of Theorem 4.1? It is OK to have different notations if there is a reason}. 
Notice that 
$$
\pen(\b\mid\theta)=\sum_{j=1}^p\rho(\beta_j\mid\theta).
$$
The proofs
below use similar ideas and techniques as \citet{rovckova2018spike}. We first outline auxiliary Lemmas and then prove them later in this section.

\begin{definition}\label{def:eta-NC} Let $\tilde{\eta}\in(0,1]$.
	We say that $\X$  satisfies the $\tilde{\eta}$-null
	consistency ($\tilde{\eta}$-NC) condition with a penalty function $\text{pen}(\b\mid\theta)$ if 
	\[
	\arg\max_{\b\in\R^{p}}\left\{-\frac{1}{2\sigma^2}\|\bm{\epsilon}/\tilde{\eta}-\X\b\|_2^{2}+\pen(\b\mid\theta)\right\}=\bm{0}_{p}.
	\]
\end{definition}

%	\begin{lemma}\label{cone_lemma}
%		If $m\X$ satisfies $\frac{\eta^*}{M}$-NC condition, where $m$ and $M$ are defined in Theorem \ref{thm: regression_weight_selected}, i.e., $\arg\max_{\b\in\R^p} \{-\frac{1}{2}||M\bm\epsilon/\eta^*-m\X\b||^2+pen(\b\C\theta) \}=\bm 0_p$, then $\arg\max_{\b\in\R^p} \{-\frac{1}{2}||\bm W\bm \epsilon/\eta^*-\bm W\X\b||^2+pen(\b\C\theta) \}=\bm 0_p$ with high probability.
%	\end{lemma}

\begin{lemma}\label{BBLASSO_identifiability_Zhanglemma} Under Condition
	(5) in Theorem \ref{thm: regression_weight_selected}, we have 
	\[
	\lim_{n\rightarrow\infty}\P_{\b_0}\left(\arg\max_{\b\in\R^{p}}\left\{-\frac{1}{2\sigma^2}\|\bm{\epsilon}/\tilde{\eta}-\X\b\|_2^{2}+\pen(\b\C\theta)\right\}=\bm{0}_{p}\right)=1,
	\]
	i.e. $\X$ satisfies the $\tilde{\eta}$-NC condition with probability
	approaching 1. \end{lemma}

\begin{lemma}\label{BBLASSO_identifiability_lemma} Under Conditions
	(1)-(5) in Theorem \ref{thm: regression_weight_selected}, given 
	that $\|\ep\|_{\infty}\lesssim\sqrt{\log n}$ and $\X$ satisfies $\tilde{\eta}$-NC
	condition, we have 
	\[
	\lim_{n\rightarrow\infty}\P_{\m,\bm{w}}\left(\arg\max_{\b\in\R^{p}}\left\{ -\frac{1}{2\sigma^2}\|\bm{W}(\bm{\epsilon}-\X\bm{\mu})/\eta^{*}-\bm{W}\X\b\|_2^{2}+pen(\b\C\theta)\right\} =\bm{0}_{p}\,\big|\,\ep\right)=1,
	\]
	where $\eta^{*}=\max\left\{ \wt\eta+C_{n}\frac{\|\X\|}{\lambda_{1}},\frac{\wt\eta}{m}\right\} $
	and $C_{n}$ is any sequence that satisfies $C_{n}\rightarrow\infty$.
\end{lemma}

\begin{lemma}\label{BBLASSO_cone_lemma} If, for $\eta^{*}\in(0,1]$, 
	\begin{equation}\label{eq:assumption_cone}
	\arg\max_{\b\in\R^{p}}\left\{ -\frac{1}{2\sigma^2}\|\bm{W}(\bm{\epsilon}-\X\bm{\mu})/\eta^{*}-\bm{W}\X\b\|_2^{2}+\pen(\b\C\theta)\right\} =\bm{0}_{p},
	\end{equation}
	then $\wh\b$ defined in \eqref{eq:bhat} lies inside a
	cone 
	$$
	C(\eta^{*};\b)=\{\Th\in\R^{p}:(\eta^{*}+1)\pen(\Th_{S}\C\theta)\leq(1-\eta^{*})\pen(\Th_{S^{C}}\C\theta)\}
	$$
	with  high probability, where $S$ is the active set of $\beta_{j}^0$'s, $S^C=\{1,2,\cdots,p\}\backslash S$ and
	$$\pen(\Th_S\mid \theta)=\sum_{j\in S}\rho(\hat\beta_j-\beta^0_j\mid\theta) \quad\text{and}\quad \pen(\Th_{S^C}\mid \theta)=\sum_{j\in S^C}\rho(\hat\beta_j-\beta^0_j\mid\theta),$$
	with $\hat\beta_j$ the $j$-th dimension of $\hat\b$.
	%{\color{blue} what is $\hat\beta_j$? Let's refer to \eqref{eq:bhat}}
\end{lemma}

\begin{lemma}\label{BBLASSO_inf_norm} 
	If \eqref{eq:assumption_cone} holds
	and $\max\limits_{1\leq i\leq n} w_{i}\leq M$, then 
	$$
	\|\X^{T}\W^{2}(\ep-\X\m)\|_{\infty}\leq \sqrt{M}\Delta^U\eta^{*},
	$$ 
	where $\Delta^U=\sqrt{2n\sigma^{2}\log[1/p^{*}(0)]}+\sigma^{2}\lambda_{1}$.
\end{lemma}

\begin{definition} The minimal restricted eigenvalue is defined as
	\[
	c(\eta^{*};\b)=\inf_{\Th\in\R^{p}}\left\{ \frac{\|\X\Th\|_2}{\|\X\|_2\|\Th\|_2}:\Th\in C(\eta^{*};\b)\right\}. 
	\]
\end{definition}

\begin{definition} The compatibility number $\phi(C)$ of vectors
	in cone $C\subset\R^{p}$ is defined as 
	\[
	\phi(C)=\inf_{\Th\in\R^{p}}\left\{ \frac{\|\X\Th\|_2\|\Th\|_{0}^{1/2}}{\|\X\|_2\|\Th\|_{1}}:\Th\in C(\eta^{*};\b)\right\}. 
	\]
\end{definition}

\subsubsection{Proof of Lemma \ref{BBLASSO_identifiability_Zhanglemma}}
\label{sec_appendix:BBLASSO_identifiability_Zhanglemma}

This Lemma is a direct consequence of Proposition 3 of \citet{zhang2012general}.
This Proposition says the following. In a regression
model \eqref{eq:model}, suppose $\delta\in(0,1]$ and $\xi_{0}>0$,
and 
\begin{equation}
-\frac{\sigma^2}{n}\rho(t\mid\theta)\ge\min\left\{ \frac{\Delta^{2}}{2n^{2}},\frac{\Delta|t|}{n}\right\} \quad\text{with}\quad\frac{\Delta}{n}\ge(1+\xi_{0})\frac{\sigma}{\widetilde{\eta}}n^{-1/2}\left(1+\sqrt{2\log\left(2p/\delta\right)}\right).\label{eq:appendix_lemma_identifiability_delta}
\end{equation}
Then, the $\widetilde{\eta}$-NC condition is satisfied with 
probability at least $2-e^{\delta/2}-\exp\left\{ -n(1-1/\sqrt{2})^{2}\right\} $,
provided that

\begin{align}\label{eq:appendix_lemma_identifiability_matrix}
&\max\left\{ \lambda_{\max}^{1/2}\left(\frac{\X_{B}^{\top}\bm{P}_{A}\X_{B}}{n}\right):B\cap A=\emptyset,|A|=\text{rank}(\bm{P}_{A})=|B|=k,\right.\nonumber
\\&\quad\quad\left.k(1+\xi_{0})^{2}\left(1+\sqrt{2\log(2p/\delta)}\right)^{2}\le2n\right\} \le\xi_{0}.
\end{align}

From the Condition (5), we know that \eqref{eq:appendix_lemma_identifiability_matrix}
holds with $\delta=2p^{-1/4}$. Thus, we only need to show \eqref{eq:appendix_lemma_identifiability_delta}
holds. Denote 
\begin{align*}
f_{1}(t) & =\frac{1}{t}\log\frac{p^{*}(t)}{p^{*}(0)}-\frac{1}{\sigma}\sqrt{2n\log[1/p^{*}(0)]},\quad t>0,\\
f_{2}(t) & =\lambda_{1}t+\log p^{*}(t)-\frac{\sigma^{2}\lambda_{1}^{2}}{2n}-\lambda_{1}\sqrt{2\sigma^{2}/n\log[1/p^{*}(0)]}.
\end{align*}
Notice that $\Delta\le\sqrt{2n\sigma^{2}\log[1/p^{*}(0)]}+\sigma^{2}\lambda_{1}$
from Proposition 5 in \cite{moran2018variance} and \eqref{eq:appendix_lemma_identifiability_delta}
holds trivially when $t=0$. Thus,  in order to show \eqref{eq:appendix_lemma_identifiability_delta},
we only need to show that 
\[
\max\left\{ f_{1}(t),f_{2}(t)\right\} \ge0,\quad\forall t>0,
\]
\[
\Delta\ge(1+\xi_{0})\frac{\sigma}{\widetilde{\eta}}n^{1/2}\left(1+\sqrt{2.5\log p}\right).
\]
Notice that
\begin{align*}
&\lim_{t\rightarrow0+}\frac{\partial f_{1}(t)}{\partial t}
=\lim_{t\rightarrow0+}\left[-\frac{1}{t^{2}}\log\frac{p^{*}(t)}{p^{*}(0)}+\frac{1}{t}\frac{\partial\log p^{*}(t)}{\partial t}\right]
\\&\stackrel{(a)}{=}\lim_{t\rightarrow0+}\left[-\frac{1}{2t}\frac{\partial\log p^{*}(t)}{\partial t}+\frac{1}{t}\frac{\partial\log p^{*}(t)}{\partial t}\right]
=\lim_{t\rightarrow0+}\frac{1}{2t}\frac{\partial\log p^{*}(t)}{\partial t}>0
\end{align*}
where $(a)$ follows from L'Hopital's rule. It implies that there
exists a positive constant $t_{0}>0$ such that, for all $t\in(0,t_{0}]$,
$\frac{\partial f_{1}(t)}{\partial t}>0$ and thus 
we have
\begin{align*}
f_{1}(t)	&>\lim_{t\rightarrow0+}f_{1}(t)=\lim_{t\rightarrow0+}\left[\frac{1}{t}\log\frac{p^{*}(t)}{p^{*}(0)}-\frac{1}{\sigma}\sqrt{2n\log[1/p^{*}(0)]}\right]
\\&\stackrel{(a)}{=}\lim_{t\rightarrow0+}\frac{\partial\log p^{*}(t)}{\partial t}-\frac{1}{\sigma}\sqrt{2n\log[1/p^{*}(0)]}
\\&=\frac{\lambda_{0}-\lambda_{1}}{\theta\lambda_{1}/[(1-\theta)\lambda_{0}]+1}-\frac{1}{\sigma}\sqrt{2n\log[1+(1-\theta)\lambda_{0}/(\theta\lambda_{1})]}
\\&\gtrsim p-\sqrt{n\log p}
\end{align*}
where $(a)$ follows from L'Hopital's rule. Thus, $f_{1}(t)>0$ for all $t\in(0,t_{0}]$, when $n$ is sufficiently large.
For $t\ge t_{0}$, 
\[
f_{2}(t)\ge f_{2}(t_{0})=\lambda_{1}t_{0}+\log p^{*}(t_{0})-\frac{\sigma^{2}\lambda_{1}^{2}}{2n}-\lambda_{1}\sqrt{2\sigma^{2}/n\log[1/p^{*}(0)]}>0
\]
is always satisfied when $n$ is sufficiently large.
Combining the above, we know that (when $n$ is sufficiently large)
\[
\max\left\{ f_{1}(t),f_{2}(t)\right\} \ge0,\quad\forall t>0.
\]
Recall that Proposition 5 of \cite{moran2018variance} implies that 
\[
\lim_{n\rightarrow\infty}\frac{\Delta}{\sqrt{2n\sigma^{2}(\eta+\gamma)\log p}}=1,
\]
\[
\lim_{n\rightarrow\infty}\frac{(1+\xi_{0})\frac{\sigma}{\widetilde{\eta}}n^{1/2}\left(1+\sqrt{2.5\log p}\right)}{\sqrt{2n\sigma^{2}(\eta+\gamma)\log p}}=\frac{1+\xi_{0}}{\widetilde{\eta}}\sqrt{\frac{5}{4(\eta+\gamma)}}.
\]
Thus, if $(1+\xi_{0})/\widetilde{\eta}<\sqrt{4(\eta+\gamma)/5}$,
when $n$ is sufficiently large, we know that $\Delta\ge(1+\xi_{0})\frac{\sigma}{\widetilde{\eta}}n^{1/2}\left(1+\sqrt{2.5\log p}\right)$.
To sum up, if $(1+\xi_{0})/\widetilde{\eta}<\sqrt{4(\eta+\gamma)/5}$
then (for $n$ is sufficiently large) 
\[
\P_{\b_0}\left(\arg\max_{\b\in\R^{p}}\left\{-\frac{1}{2}\|\bm{\epsilon}/\tilde{\eta}-\X\b\|_2^{2}+pen(\b\C\theta)\right\}=\bm{0}_{p}\right)
\geq2-e^{1/p^{1/4}-e^{-n(1-1/\sqrt{2})^{2}}},
\]
and thus 
\[
\lim_{n\rightarrow\infty}\P_{\b_0}\left(\arg\max_{\b\in\R^{p}}\left\{-\frac{1}{2}\|\bm{\epsilon}/\tilde{\eta}-\X\b\|_2^{2}+pen(\b\C\theta)\right\}=\bm{0}_{p}\right)=1.\qedhere
\]

\subsubsection{Proof of Lemma \ref{BBLASSO_identifiability_lemma}}

On the event  $\|\ep\|_{\infty}\lesssim\sqrt{\log n}$, since
$\E w_{i}=1$, we have 
\[
\begin{split} & \text{Var}_{\m,\bm{w}}\left[(\ep-\X\m)^{T}\W^{2}\X\b\C\ep\right]=\text{Var}_{\m,\bm{w}}\left(\sum_{i}(\epsilon_{i}-\x_{i}^{T}\m)w_{i}\x_{i}^{T}\b\,\big|\,\ep\right)\\
= & \E_{\m,\bm{w}}\left[\left(\sum_{i}(w_{i}-1)\epsilon_{i}\x_{i}^{T}\b-\sum_{i}w_{i}\x_{i}^{T}\m\x_{i}^{T}\b\right)^{2}\,\big|\,\ep\right]\\
= & \E_{\bm{w}}\left[\left(\sum_{i}(w_{i}-1)\epsilon_{i}\x_{i}^{T}\b\right)^{2}\,\big|\,\ep\right]+\E_{\m,\bm{w}}\left[\sum_{i}w_{i}\x_{i}^{T}\m\x_{i}^{T}\b\right]^{2}\\
= & \sum_{i}\text{Var}(w_{i})(\epsilon_{i}\x_{i}^{T}\b)^{2}+\sum_{i\neq j}\text{Cov}(w_{i},w_{j})\epsilon_{i}\x_{i}^{T}\b\epsilon_{j}\x_{j}^{T}\b+\sum_{i,j}\x_{i}^{T}\b\x_{j}^{T}\b\x_{i}^{T}\E(\m\m^{T})\x_{j}\E\left(w_{i}w_{j}\right)\\
\stackrel{(a)}{\leq} & \frac{C_{1}}{\log n}\sum_{i}(\epsilon_{i}\x_{i}^{T}\b)^{2}+\frac{C_{2}}{n\log n}\sum_{i\neq j}\frac{1}{2}\left[(\epsilon_{i}\x_{i}^{T}\b)^{2}+(\epsilon_{j}\x_{j}^{T}\b)^{2}\right]+\E\left(w_{i}w_{j}\right)\frac{2}{\lambda_{0}^{2}}\sum_{i,j}\x_{i}^{T}\b\x_{j}^{T}\b(\x_{i}^{T}\x_{j})\\
\leq & \frac{C_{1}}{\log n}\sum_{i}(\epsilon_{i}\x_{i}^{T}\b)^{2}+\frac{C_{2}}{n\log n}(n-1)\sum_{i}(\epsilon_{i}\x_{i}^{T}\b)^{2}+C\E\left(w_{i}w_{j}\right)\frac{2\max_{i\neq j}|\x_{i}^{T}\x_{j}|}{\lambda_{0}^{2}}(\sum_{i}|\x_{i}^{T}\b|)^{2}\\
\stackrel{(b)}{\leq} & \frac{\wt C_{1}}{\log n}\|\ep\|_{\infty}^{2}\|\X\b\|_2^{2}+C_{3}\frac{2n\max_{i\neq j}|\x_{i}^{T}\x_{j}|}{\lambda_{0}^{2}}\sum_{i}(\x_{i}^{T}\b)^{2}\\
\stackrel{(d)}{\leq} & \wt C_{3}\|\X\b\|_2^{2},
\end{split}
\]
where (a) uses the assumption (4) in Theorem \ref{thm: regression_weight_selected},
the fact $ab\leq\frac{1}{2}(a^{2}+b^{2})$ and $\mu\sim$ Spike. The inequality (b) follows
from $\|\bm{\epsilon}^{T}\X\b\|_{2}^{2}\le\|\bm{\epsilon}\|_{\infty}^{2}\|\X\b\|_{2}^{2}$
and the Cauchy-Schwarz inequality $(\sum_{i}|\x_{i}^{T}\b|)^{2}\leq n\sum_{i}(\x_{i}^{T}\b)^{2}$.
The inequality (d) follows from the fact that $\lambda_{0}\asymp p^{\gamma}$
where $\gamma\geq1$, $\|\ep\|_{\infty}\lesssim\sqrt{\log n}$ and $\max_{i\neq j}|\x_{i}^{T}\x_{j}|\lesssim \lambda_0^2/n$.
Thus, from the Markov's inequality, on the event that $\|\ep\|_{\infty}\lesssim\sqrt{\log n}$,
we have for any $t>0$,
\[
\begin{split}
\P_{\m,\bm{w}}\left(|(\ep-\X\m)^{T}\W^{2}\X\b-\ep^{T}\X\b|>t\C\ep\right)\leq\frac{\text{Var}_{\m,\bm{w}}\left[(\ep-\X\m)^{T}\W^{2}\X\b\mid\ep\right]}{t^{2}}\leq\frac{\wt C_{3}\|\X\b\|_2^{2}}{t^{2}}.
\end{split}
\]
Set $t=C_{n}\|\X\b\|_2$ where $C_{n}\rightarrow\infty$, we have 
\begin{equation}
\lim_{n\rightarrow\infty}\P_{\m,\bm{w}}\left(|(\ep-\X\m)^{T}\W^{2}\X\b-\ep^{T}\X\b|>C_{n}\|\X\b\|_2\,\C\ep\right)=0.
\label{eq:appendix_lemma_NC0}
\end{equation}
When $|(\ep-\X\m)^{T}\W^{2}\X\b-\ep^{T}\X\b|\leq C_{n}\|\X\b\|_2$,
we have 
\begin{equation}
(\ep-\X\m)^{T}\W^{2}\X\b\leq|\ep^{T}\X\b|+C_{n}\|\X\b\|_2.\label{eq:appendix_lemma_NC1}
\end{equation}
Notice that 
\[
\|\X\b\|_2+\frac{\|\X\|_2}{\lambda_{1}}pen(\b\C\theta)\stackrel{(e)}{\leq}\|\X\|_2\times\|\b\|_2+\frac{\|\X\|_2}{\lambda_{1}}(-\lambda_{1}\|\b\|_1)\le0,
\]
where (e) follows from $pen(\b\C\theta)=-\lambda_{1}\|\b\|_1+\sum_{j}\log\frac{p^{*}(0)}{p^{*}(\beta_{j})}\leq-\lambda_{1}\|\b\|_1$.
Plugging this into \eqref{eq:appendix_lemma_NC1}, we have 
\begin{equation}
(\ep-\X\m)^{T}\W^{2}\X\b\leq|\ep^{T}\X\b|-C_{n}\frac{\|\X\|_2}{\lambda_{1}}pen(\b\C\theta)\label{eq:appendix_lemma_NC2}.
\end{equation}
Since the $\tilde{\eta}$-NC condition holds, we have 
\begin{equation}
-\frac{\tilde{\eta}}{2\sigma^2}\|\X\b\|_2^{2}+\ep^{T}\X\b+\tilde{\eta}\times pen(\b\C\theta)\leq0,\,\forall\b\label{eq:appendix_lemma_NC4}.
\end{equation}
Thus, on condition that $m<\min_{i}w_{i}\le\max_{i}w_{i}<M$, if we
choose $\eta^{*}=\max\left\{ \tilde{\eta}+C_{n}\frac{\|\X\|}{\lambda_{1}},\frac{\tilde{\eta}}{m}\right\} $,
we have $\forall\b$, 
\begin{equation}
\begin{split} 
& -\frac{\eta^{*}}{2\sigma^2}\|\W\X\b\|_2^{2}+(\ep-\X\m)^{T}\W^{2}\X\b+\eta^{*}pen(\b\C\theta)\\
\stackrel{(f)}{\leq} & -\frac{\eta^{*}m}{2\sigma^2}\|\X\b\|_2^{2}+|\ep^{T}\X\b|+\left(\eta^{*}-C_{n}\frac{\|\X\|}{\lambda_{1}}\right)pen(\b\C\theta)\\
\stackrel{(g)}{\leq} & -\frac{\tilde{\eta}}{2\sigma^2}\|\X\z\|_2^{2}+\ep^{T}\X\z+\tilde{\eta}\times pen(\z\C\theta)\leq0,
\end{split}
\label{eq:appendix_lemma_NC3}
\end{equation}
where (f) follows from Equation \eqref{eq:appendix_lemma_NC2}, (g)
follows from the definition of $\eta^{*}$ and the fact that $pen(\b\C\theta)\leq0$
for any $\b$. We set $\z=\b$ if $\ep^{T}\X\b\geq0$ and $\z=-\b$
if $\ep^{T}\X\b<0$. The last inequality directly follows from the $\tilde{\eta}$-NC
condition \eqref{eq:appendix_lemma_NC4}.

Previous analysis implies that under conditions (1), (4) and (5) in Theorem
\ref{thm: regression_weight_selected} (and assuming $m<\min_{i}w_{i}\le\max_{i}w_{i}<M$,
$\|\ep\|_{\infty}\lesssim\sqrt{\log n}$ and that the $\tilde{\eta}$-NC condition
holds), whenever $|(\ep-\X\m)^{T}\W^{2}\X\b-\ep^{T}\X\b|\leq C_{n}\|\X\b\|_2$
holds, \eqref{eq:appendix_lemma_NC3} holds. Thus, when $\|\ep\|_{\infty}\lesssim\sqrt{\log n}$
and $\X$ satisfies $\tilde{\eta}$-NC condition,
\[
\begin{split} & \P_{\m,\bm{w}}\left(|(\ep-\X\m)^{T}\W^{2}\X\b-\ep^{T}\X\b|\leq C_{n}\|\X\b\|_2\,\C\ep,m<\min_{i}w_{i}\le\max_{i}w_{i}<M\right)\\
\leq & \P_{\m,\bm{w}}\left(-\frac{\eta^{*}}{2\sigma^2}\|\W\X\b\|_2^{2}+(\ep-\X\m)^{T}\W^{2}\X\b+\eta^{*}pen(\b\C\theta)\leq0\C\ep,m<\min_{i}w_{i}\le\max_{i}w_{i}<M\right).
\end{split}
\]
Notice that the conditions (2) and (3) in Theorem \ref{thm: regression_weight_selected}
say
\[
\P\left(m<\min_{i}w_{i}\le\max_{i}w_{i}<M\right)\rightarrow1.
\]
Thus, as $n\rightarrow\infty$,
\begin{align*}
& \P_{\m,\bm{w}}\left(|(\ep-\X\m)^{T}\W^{2}\X\b-\ep^{T}\X\b|\leq C_{n}\|\X\b\|_2\,\C\ep,m<\min_{i}w_{i}\le\max_{i}w_{i}<M\right)\\
\rightarrow & \P_{\m,\bm{w}}\left(|(\ep-\X\m)^{T}\W^{2}\X\b-\ep^{T}\X\b|\leq C_{n}\|\X\b\|_2\,\C\ep\right),
\end{align*}
\begin{align*}
& \P_{\m,\bm{w}}\left(-\frac{\eta^{*}}{2\sigma^1}\|\W\X\b\|_2^{2}+(\ep-\X\m)^{T}\W^{2}\X\b+\eta^{*}pen(\b\C\theta)\leq0\C\ep,m<\min_{i}w_{i}\le\max_{i}w_{i}<M\right)\\
\rightarrow & \P_{\m,\bm{w}}\left(-\frac{\eta^{*}}{2\sigma^2}\|\W\X\b\|_2^{2}+(\ep-\X\m)^{T}\W^{2}\X\b+\eta^{*}pen(\b\C\theta)\leq0\C\ep\right).
\end{align*}
So given that $\|\ep\|_{\infty}\lesssim\sqrt{\log n}$ and that $\X$ satisfies
$\tilde{\eta}$-NC condition, we have 
\[
\begin{split} & \lim_{n\rightarrow\infty}\P_{\m,\bm{w}}\left(|(\ep-\X\m)^{T}\W^{2}\X\b-\ep^{T}\X\b|\leq C_{n}\|\X\b\|_2\,\C\ep\right)\\
\leq & \lim_{n\rightarrow\infty}\P_{\m,\bm{w}}\left(-\frac{\eta^{*}}{2\sigma^2}\|\W\X\b\|_2^{2}+(\ep-\X\m)^{T}\W^{2}\X\b+\eta^{*}pen(\b\C\theta)\leq0\C\ep\right).
\end{split}
\]
Combined with Equation \eqref{eq:appendix_lemma_NC0}, we know (given
that $\|\ep\|_{\infty}\lesssim\sqrt{\log n}$ and $\X$ satisfies $\tilde{\eta}$-NC
condition), the following holds 
\[
\lim_{n\rightarrow\infty}\P_{\m,\bm{w}}\left(-\frac{\eta^{*}}{2\sigma^2}\|\W\X\b\|_2^{2}+(\ep-\X\m)^{T}\W^{2}\X\b+\eta^{*}pen(\b\C\theta)\leq0\C\ep\right)=1,
\]
which is equivalent to the conclusion in Lemma \ref{BBLASSO_identifiability_lemma}.

\subsubsection{Proof of Lemma \ref{BBLASSO_cone_lemma}}
Starting from basic inequality $Q(\wh\b)\geq Q(\b_{0})$, we get 
\begin{equation}
\|\bm{W}\X\bm{\Theta}\|_2^{2}-2(\W\ep-\W\X\m)^{T}\W\X\Th+2\sigma^{2}pen(\b_{0}\C\theta)-2\sigma^{2}pen(\wh\b\C\theta)\leq0\label{eq:BBLASSO_cone_lemma1}.
\end{equation}
From 
$$
\arg\max_{\b\in\R^{p}}\left\{-\frac{1}{2\sigma^{2}}\|\bm{W}\bm{(}\bm{\epsilon}-\X\bm{\mu})/\eta^{*}-\bm{W}\X\b\|_2^{2}+pen(\b\C\theta)\right\}=\bm{0}_{p},
$$
we have, for all $\Th\in\mathbb{R}^p$,
\begin{equation}
-2\Th^{T}\X^{T}\W^{2}(\ep-\X\m)\geq-\eta^{*}\|\W\X\bm{\Theta}\|_2^{2}+2\eta^{*}\sigma^{2}pen(\bm{\Theta}\C\theta).
\label{eq:BBLASSO_cone_lemma2}
\end{equation}
Notice that $pen(\cdot\C\theta)$ is super-additive ($pen(\bm{a}\mid\theta)+pen(\bm{b}\mid\theta)\le pen(\bm{a}+\bm{b}\mid\theta)$
for all $\bm{a,b}$). Plugging \eqref{eq:BBLASSO_cone_lemma2} into \eqref{eq:BBLASSO_cone_lemma1} (denoting $\b=\b_{0}$), we get 
\[
\begin{split} 
& (1-\eta^{*})\|\bm{W}\X\Th\|_2^{2}\leq-2\eta^{*}\sigma^{2}pen(\bm{\Theta}\C\theta)-2\sigma^{2}pen(\b\C\theta)+2\sigma^{2}pen(\bm{\Theta}+\b\C\theta)\\
& =-2\eta^{*}\sigma^{2}pen(\bm{\Theta}_{S}\C\theta)-2\eta^{*}\sigma^{2}pen(\bm{\Theta}_{S^{C}}\C\theta)-2\sigma^{2}pen(\b_{S}\C\theta)+2\sigma^{2}pen(\bm{\Theta}_{S}+\b_{S}\C\theta)+2\sigma^{2}pen(\bm{\Theta}_{S^{C}}\C\theta)\\
& \stackrel{(a)}{\le}-2\eta^{*}\sigma^{2}pen(\bm{\Theta}_{S}\C\theta)-2\eta^{*}\sigma^{2}pen(\bm{\Theta}_{S^{C}}\C\theta)-2\sigma^{2}pen(\bm{\Theta}_{S}\C\Theta)+2\sigma^{2}pen(\bm{\Theta}_{S^{C}}\C\theta)\\
& =-2(\eta^{*}+1)\sigma^{2}pen(\bm{\Theta}_{S}\C\theta)-2(\eta^{*}-1)\sigma^{2}pen(\bm{\Theta}_{S^{C}}\C\theta),
\end{split}
\]
where $(a)$ utilizes the fact that 
\[
pen(\bm{\Theta}_{S}+\b_{S}\C\theta)+pen(\bm{\Theta}_{S}\C\Theta)=pen(\bm{\Theta}_{S}+\b_{S}\C\theta)+pen(-\bm{\Theta}_{S}\C\Theta)\le pen(\b_{S}\C\theta).
\]
Since $(1-\eta^{*})\|\bm{W}\X\bm{\theta}\|_2^{2}\geq0$, we get the
desired conclusion.

\subsubsection{Proof of Lemma \ref{BBLASSO_inf_norm}}
The proof follows from the proof of Lemma 1 in the Appendix of \citet{zhang2012general}.
For any $j\in{1,2,\cdots,p}$, denote with $\bm{1}_{j}\in\mathbb{R}^{p}$
the vector where the $j$-th element is 1 and all the other elements
are 0. Denote 
$$
G(\b)=-\frac{1}{2}\|\bm{W}(\bm{\epsilon}-\X\bm{\mu})/\eta^{*}-\bm{W}\X\b\|_2^{2}+pen(\b\C\theta).
$$
Notice that $pen(\b\mid\theta)=\sum_{j=1}^{p}\rho(\beta_{j}\mid\theta)$.
For any $t>0$, $G(\bm{0}_{p})\geq G(t\bm{1}_{j})$ leads to
\[
t\left(\W\X_{j}\right)^{T}\W(\bm{\epsilon}-\X\bm{\mu})/\eta^{*}\le\frac{t^{2}}{2}\|\W\X_{j}\|_2^{2}-\sigma^{2}pen\left(t\bm{1}_{j}\mid\theta\right)\stackrel{(a)}{\le}\frac{t^{2}Mn}{2}-\sigma^{2}\rho\left(t\mid\theta\right),
\]
where $(a)$ follows from the fact that $\max_{i}w_{i}\le M$. Thus,
for any $t>0$, 
\[
\left(\W\X_{j}\right)^{T}\W(\bm{\epsilon}-\X\bm{\mu})/\eta^{*}\leq M\left(\frac{1}{2}nt-\frac{\sigma^{2}}{M}\frac{\rho(t\mid\theta)}{t}\right).
\]
Again, from the definition of $\Delta$ \eqref{def_appendix:delta_regression} and Proposition 5 in \cite{moran2018variance}, we have
\begin{align*}
& \|\X^{T}\W^{2}\left(\ep-\X\m\right)\|_{\infty}\leq\eta^{*}M\inf_{t>0}\left(\frac{1}{2}nt-\frac{\sigma^{2}}{M}\rho(t\mid\theta)/t\right)
\\\le&\eta^{*}M\left[\sqrt{2n\sigma^{2}/M\log[1/p^{*}(0)]}+\sigma^{2}\lambda_{1}/M\right]\\
\le&\eta^{*}\sqrt{M}\left[\sqrt{2n\sigma^{2}\log[1/p^{*}(0)]}+\sigma^{2}\lambda_{1}\right]=\eta^{*}\sqrt{M}\Delta^{U}.
\end{align*}

\subsubsection{Proof of Theorem \ref{thm: regression_weight_selected}}\label{sec:appendix_thm_regression_selected_proof}

The proof follows from the proof of Theorem 7 in \citet{rovckova2018spike}.
First, we assume that (i) $\arg\max_{\b\in\R^{p}}\{-\frac{1}{2\sigma^{2}}\|\bm{W}(\bm{\epsilon}-\X\bm{\mu})/\eta^{*}-\bm{W}\X\b\|_2^{2}+pen(\b\C\theta)\}=\bm{0}_{p}$;
(ii) $\max w_{i}\leq M$, (iii) $\min w_{i}\geq m$, (iv) $\|\ep\|_{\infty}\lesssim\sqrt{\log n}$
and (v) $\Th\in C(\eta^{*};\b)$ holds. Then, from Lemma \ref{BBLASSO_inf_norm},
we know that
\[
\begin{split}\|\X^{T}\W^{2}(\ep-\W\m)\|_{\infty}\leq\sqrt{M}\eta^{*}\Delta^{U}\end{split},
\]
where $\Delta^{U}$ is defined in Lemma \ref{BBLASSO_inf_norm}.
Following \citet{rovckova2018spike}, we denote $c_{+}^{w}=0.5(1+\sqrt{1-\frac{4\|\bm{W}\X_{j}\|_2^{2}}{\sigma^{2}(\lambda_{0}-\lambda_{1})^{2}}})$. Notice that $\delta_{c_{+}^{w}}=\frac{1}{\lambda_{0}-\lambda_{1}}\log(\frac{1-\theta}{\theta}\frac{\lambda_{0}}{\lambda_{1}}\frac{c_{+}^{w}}{1-c_{+}^{w}})$ is the inflection point of $Q(\b)$ in the
$j$-th direction while keeping the other coordinates fixed. Since $\|\bm{W}\X_{j}\|_2^{2}\le Mn\leq\sigma^{2}(\lambda_{0}-\lambda_{1})^{2}$
when $n$ is sufficiently large, $c_{+}^{w}$ is well-defined
when $n$ is sufficiently large. Denote $\hat{q}=\|\hat{\b}\|_{0}$
and $q=\|\b_{0}\|_{0}$. From the basic inequality $0\geq Q(\b_{0})-Q(\wh\b)$,
we get 
\begin{equation}
\begin{split}0 & \geq\|\bm{W}\X\bm{\Theta}\|_2^{2}-2(\bm{W}\bm{\epsilon}-\W\X\m)^{T}\bm{W}\X\bm{\Theta}+2\sigma^{2}\log\frac{\pi(\b_{0}\C\theta)}{\pi(\wh\b\C\theta)}\\
& \geq\|\bm{W}\X\bm{\Theta}\|_2^{2}-2(\bm{W}\bm{\epsilon}-\W\X\m)^{T}\bm{W}\X\bm{\Theta}+2\sigma^{2}\left[-\lambda_{1}\|\b_{0}-\wh\b\|_1+\sum_{j=1}^{p}\log\frac{p^{*}(\hat{\beta}_{j})}{p^{*}(0)}+\sum_{j=1}^{p}\log\frac{p^{*}(0)}{p^{*}(\beta^0_{j})}\right]\\
& \stackrel{(a)}{\ge}\|\bm{W}\X\bm{\Theta}\|_2^{2}-2(\bm{W}\bm{\epsilon}-\W\X\m)^{T}\bm{W}\X\bm{\Theta}+2\sigma^{2}\left[-\lambda_{1}\|\b_{0}-\wh\b\|_1+\hat{q}b^{w}+(\hat{q}-q)\log\frac{1}{p^{*}(0)}\right]\\
& \stackrel{(b)}{\ge}\|\bm{W}\X\bm{\Theta}\|_2^{2}-2\|(\bm{W}\bm{\epsilon}-\W\X\m)^{T}\bm{W}\X\|_{\infty}\times\|\bm{\Theta}\|_{1}-2\sigma^{2}\lambda_{1}\|\b_{0}-\wh\b\|_{1}\\
& \quad\quad+2\sigma^{2}\hat{q}b^{w}+2(\hat{q}-q)\sigma^{2}\log[1/p^{*}(0)],
\end{split}
\label{eq:thm_regress1}
\end{equation}
where $(a)$ follows from the fact that $p^{*}(\hat{\beta}_{j})>c_{+}^{w}$
when $\hat{\beta}_{j}\neq0$ and we denote $0>b^{w}=\log c_{+}^{w}>\log0.5$,
$(b)$ uses Holder inequality. So, from Equation \eqref{eq:thm_regress1} and Lemma \ref{BBLASSO_cone_lemma}, we have
\[
\begin{split}
0 & \geq\|\bm{W}\X\bm{\Theta}\|_2^{2}-2(\sqrt{M}\eta^{*}\Delta^{U}+\sigma^{2}\lambda_{1})\|\bm{\Theta}\|_1+2\sigma^{2}\hat{q}b^{w}+2(\hat{q}-q)\sigma^{2}\log[1/p^{*}(0)]\\
& \stackrel{(a)}{\geq} mc^{2}(\eta^{*};\b)\|\bm{\Theta}\|_2^{2}\|\X\|_2^{2}-2(\sqrt{M}\eta^{*}\Delta^{U}+\sigma^{2}\lambda_{1})\|\bm{\Theta}\|_1+2\sigma^{2}\hat{q}b^{w}+2\sigma^{2}(\hat{q}-q)\log[1/p^{*}(0)]\\
& \geq mc^{2}(\eta^{*};\b)\|\bm{\Theta}\|_2^{2}\|\X\|_2^{2}-2(\sqrt{M}\eta^{*}\Delta^{U}+\sigma^{2}\lambda_{1})\|\bm{\Theta}\|_{2}\|\bm{\Theta}\|_{0}^{1/2}+2\sigma^{2}\hat{q}b^{w}+2(\hat{q}-q)\sigma^{2}\log[1/p^{*}(0)],
\end{split}
\]
where (a) follows from the definition of $c(\eta^{*};\b)$.
This is equivalent to
\[
\begin{split}
\left(\sqrt{m}c(\eta^{*};\b)\|\bm{\Theta}\|_2\|\X\|_2-\frac{\sqrt{M}\eta^{*}\Delta^{U}+\sigma^{2}\lambda_{1}}{\sqrt{m}c(\eta^{*};\b)\|\X\|_2}\|\bm{\Theta}\|_{0}^{1/2}\right)^{2}-\frac{(\sqrt{M}\eta^{*}\Delta^{U}+\sigma^{2}\lambda_{1})^{2}}{mc^{2}(\eta^{*};\b)\|\X\|_2^{2}}\|\bm{\Theta}\|_{0}\\
+2\sigma^{2}\hat{q}b^{w}+2(\hat{q}-q)\sigma^{2}\log\frac{1}{p^{*}(0)}\le0,
\end{split}
\]
and consequently,
\[
\begin{split}
(\hat{q}-q)\log\frac{1}{p^{*}(0)}+\hat{q}b^{w}\leq\frac{(\sqrt{M}\eta^{*}\Delta^{U}+\sigma^{2}\lambda_{1})^{2}}{2m\sigma^{2}c^{2}(\eta^{*};\b)\|\X\|_2^{2}}\|\bm{\Theta}\|_{0}\stackrel{(a)}{\le}\frac{(\sqrt{M}\eta^{*}\Delta^{U}+\sigma^{2}\lambda_{1})^{2}}{2m\sigma^{2}c^{2}(\eta^{*};\b)n}(\hat{q}+q),
\end{split}
\]
where $(a)$ follows from $\|\X\|_2^{2}\ge n$. Thus, 
\[
\begin{split}
\hat{q}\leq q\frac{A+B}{B+b^{w}-A}=q\left(1+\frac{2A-b^{w}}{B+b^{w}-A}\right)\leq q\left(1+\frac{2r}{1-r}\right),
\end{split}
\]
where 
\[
A=\frac{(\sqrt{M}\eta^{*}\Delta^{U}+\sigma^{2}\lambda_{1})^{2}}{2m\sigma^{2}c^{2}(\eta^{*};\b)n},\,B=\log[1/p^{*}(0)],\,b^{w}=\log c_{+}^{w}\in(\log0.5,0),\,r=\frac{A}{B}.
\]
For simplicity assume that $\frac{1-\theta}{\theta}=C_1p^{\eta},\lambda_{0}=C_{2}p^{\gamma}$
with $C_{1}C_{2}=4$. Then $B=\log(1+\frac{1-\theta}{\theta}\frac{\lambda_{0}}{\lambda_{1}})>(\eta+\gamma-1)\log p$,
so 
\[
\begin{split}
r&=\frac{A}{B}=\left(\frac{\sqrt{M}\eta^{*}\Delta^{U}}{\sigma c(\eta^{*};\b)\sqrt{2nmB}}+\frac{\sigma\lambda_{1}}{c(\eta^{*};\b)\sqrt{2mnB}}\right)^{2}
\\&<\left(\frac{\eta^{*}}{c}\sqrt{\frac{M}{m}}+\frac{\sigma\lambda_{1}\sqrt{M}}{c\sqrt{2mn(\eta+\gamma-1)\log p}}\right)^{2}=D.
\end{split}
\]

Thus, the desired conclusion holds when conditions (i)-(v) hold. Now,
we only need to verify that $\lim_{n\rightarrow\infty}\P_{\b,\m,\bm{w}}\left(\text{condition (i), (ii), (iii), (iv), (v) all holds}\right)=1$.
Notice that (i) implies (v) from Lemma \ref{BBLASSO_cone_lemma} and
\[
\lim_{n\rightarrow\infty}\P_{\b_0}\left(\text{condition (iv) holds}\right)\ge\lim_{n\rightarrow\infty}\P_{\b_0}\left(\|\ep\|_{\infty}\leq\sqrt{2\log n}\right)=1.
\]
Combined with Lemma \ref{BBLASSO_identifiability_Zhanglemma} and Lemma \ref{BBLASSO_identifiability_lemma},
we have  %{\color{blue}I am confused about why do we condition on $\X$ and not on $\Y$ below?}
\[
\begin{split} & \lim_{n\rightarrow\infty}\P_{\b_0,\m,\bm{w}}\left(\arg\max_{\b\in\R^{p}}\{-\frac{1}{2\sigma^2}\|\bm{W}(\bm{\epsilon}-\X\bm{\mu})/\eta^{*}-\bm{W}\X\b\|_2^{2}+pen(\b\C\theta)\}=\bm{0}_{p}\right)\\
= & \lim_{n\rightarrow\infty}\E_{\b_0}\P_{\m,\bm{w}}\left(\arg\max_{\b\in\R^{p}}\{-\frac{1}{2\sigma^2}\|\bm{W}(\bm{\epsilon}-\X\bm{\mu})/\eta^{*}-\bm{W}\X\b\|_2^{2}+pen(\b\C\theta)\}=\bm{0}_{p}\C\ep\right)=1.
\end{split}
\]
This means $\lim_{n\rightarrow\infty}\P_{\b_0,\m,\bm{w}}\left(\text{condition (i) holds}\right)=1$.
In addition, assumptions (2) and  (3) in Theorem \ref{thm: regression_weight_selected}
say that 
$$
\lim_{n\rightarrow\infty}\P_{\b_0,\m,\bm{w}}\left(\text{condition (ii) holds}\right)=1\quad\text{and}\quad \lim_{n\rightarrow\infty}\P_{\b_0,\m,\bm{w}}\left(\text{condition (iii) holds}\right)=1.$$
Thereby, from the union bound, $\lim_{n\rightarrow\infty}\P_{\b_0,\m,\bm{w}}\left(\text{condition (i), (ii), (iii), (iv), (v) all holds}\right)=1$.
Since 
\[
\begin{split}\lim_{n\rightarrow\infty}\E_{\b_0}\P_{\m,\bm{w}}\left(\hat{q}\leq q(1+\frac{2D}{1-D})\C\Y^{(n)}\right) & =\lim_{n\rightarrow\infty}\P_{\b_0,\m,\bm{w}}\left(\hat{q}\leq q(1+\frac{2D}{1-D})\right)\\
& \geq\lim_{n\rightarrow\infty}\P_{\b_0,\m,\bm{w}}\left(\text{conditions (i) - (v) all hold}\right).
\end{split}
\]
We have 
\[
\lim_{n\rightarrow\infty}\E_{\b_0}\P_{\m,\bm{w}}\left(\hat{q}\leq q(1+\frac{2D}{1-D})\C\Y^{(n)}\right)=1.
\]

\subsubsection{Proof of Theorem \ref{thm:regression_weight}}\label{sec:appendix_thm_regression_proof}
The proof follows from
the proof of Theorem 8 in \citet{rovckova2018spike}. Notice that 
\[
\|\widetilde{\b}-\b_{0}\|_2\le\|\wt\b-\wh\b\|_2+\|\wh\b-\b_{0}\|_2.
\]
First, we prove the high probability bound for $\|\Th\|_2=\|\widehat{\b}-\b_{0}\|_2$.
Since $Q(\wh\b)\geq Q(\b_{0})$, we know if (i) $\arg\max_{\b\in\R^{p}}\{-\frac{1}{2\sigma^{2}}\|\bm{W}(\bm{\epsilon}-\X\bm{\mu})/\eta^{*}-\bm{W}\X\b\|_2^{2}+pen(\b\C\theta)\}=\bm{0}_{p}$;
(ii) $\max w_{i}\leq M$, then 
\begin{equation}
\begin{split}0 & \geq\|\bm{W}\X\bm{\Theta}\|_2^{2}-2(\bm{W}\bm{\epsilon}-\W\X\m)^{T}\bm{W}\X\bm{\Theta}+2\log\frac{\pi(\b_{0}\C\theta)}{\pi(\wh\b\C\theta)}\\
& \stackrel{(a)}{\ge}\|\bm{W}\X\bm{\Theta}\|_2^{2}-2\|(\bm{W}\bm{\epsilon}-\W\X\m)^{T}\bm{W}\X\|_{\infty}\|\bm{\Theta}\|_1-2\lambda_{1}\|\bm{\Theta}\|_1+2q\log p^{*}(0)\\
& \stackrel{(b)}{\ge}\|\bm{W}\X\bm{\Theta}\|_2^{2}-2\left(\sqrt{M}\eta^{*}\Delta^{U}+\lambda_{1}\right)\|\bm{\Theta}\|_1+2q\log p^{*}(0),
\end{split}
\label{pf: regression_weight}
\end{equation}
where $(a)$ follows from $\log\frac{\pi(\b_{0}\C\theta)}{\pi(\b\C\theta)}\geq-\lambda_{1}\|\bm{\Theta}\|_1+q\log p^{*}(0)$
and Holder Inequality, $(b)$ follows from Lemma \ref{BBLASSO_inf_norm}.
From Theorem \ref{thm: regression_weight_selected}, $\|\bm{\Theta}\|_{0}\leq(1+K)q$,
and using $4ab\leq a^{2}+4b^{2}$, we have 
\[
\begin{split}
2(\sqrt{M}\eta^{*}\Delta^{U}+\lambda_{1})\|\bm{\Theta}\|_1 
& \leq3(\sqrt{M}\eta^{*}\Delta^{U}+\lambda_{1})\frac{\|\X\bm{\Theta}\|_2\sqrt{(K+1)q}}{\|\X\|_2\phi}-(\sqrt{M}\eta^{*}\Delta^{U}+\lambda_{1})\|\bm{\Theta}\|_1\\
& \leq\frac{m\|\X\bm{\Theta}\|_2^{2}}{2}+\frac{5(K+1)q(\sqrt{M}\eta^{*}\Delta^{U}+\lambda_{1})^{2}}{m\|\X\|_2^{2}\phi^{2}}-(\sqrt{M}\eta^{*}\Delta^{U}+\lambda_{1})\|\bm{\Theta}\|_1.
\end{split}
\]
Plugging into \eqref{pf: regression_weight}, we know (i), (ii), plus
(iii) $\min w_{i}\geq m$ and (iv) $\|\ep\|_{\infty}\lesssim\sqrt{\log n}$
implies the following: 
\[
\begin{split}0 & \geq\|\bm{W}\X\bm{\Theta}\|_2^{2}-\frac{m\|\X\bm{\Theta}\|_2^{2}}{2}-\frac{5(K+1)q(\sqrt{M}\eta^{*}\Delta^{U}+\lambda_{1})^{2}}{m\|\X\|_2^{2}\phi^{2}}+(\sqrt{M}\eta^{*}\Delta^{U}+\lambda_{1})\|\bm{\Theta}\|_1+2q\log p^{*}(0)\\
& \geq\frac{m}{2}\|\X\bm{\Theta}\|_2^{2}-\frac{5(K+1)q(\sqrt{M}\eta^{*}\Delta^{U}+\lambda_{1})^{2}}{m\|\X\|_2^{2}\phi^{2}}+(\sqrt{M}\eta^{*}\Delta^{U}+\lambda_{1})\|\bm{\Theta}\|_1+2q\log p^{*}(0).
\end{split}
\]
Thus, whenever (i)-(iv) holds, in view that $\|\X\|_2^2\ge n$,
we have
\[
\begin{split}
\frac{m}{2}\|\X\bm{\Theta}\|_2^{2}+(\sqrt{M}\eta^{*}\Delta^{U}+\lambda_{1})\|\bm{\Theta}\|_1 
& \leq\frac{5(K+1)q(C_{3}\sqrt{M}\eta^{*}\sqrt{n\log p})^{2}}{mn\phi^{2}}+2qC_{4}\log p\\
& <\frac{C_{5}^{2}M(\eta^{*})^{2}}{m\phi^{2}}q(1+K)\log p.
\end{split}
\]
Thus, whenever (i)-(iv) holds, 
\[
\|\X\bm{\Theta}\|_2\leq\frac{C_{5}\eta^{*}\sqrt{M}}{\sqrt{m}\,\phi}\sqrt{q(1+K)\log p}.
\]
It follows from definition of $c$ that 
\[
\lim_{n\rightarrow\infty}
\P_{\bm{w},\m,\b_{0}}\left(\|\bm{\Theta}\|_2\leq\frac{C_{5}\eta^{*}\sqrt{M}}{\sqrt{m}\,\phi\,c}\sqrt{q(1+K)\frac{\log p}{n}}\right)\geq\lim_{n\rightarrow\infty}\P\left(\text{condition (i)-(iv) holds}\right)=1.
\]
Notice that the difference between $\wt\b$ and $\wh\b$ only depends
on $\m$ and satisfies 
\[
\begin{split}
\P_{\m}\left(\|\wt\b-\wh\b\|_2^{2}>t\C\Y^{(n)}\right)
\leq\frac{\E\sum_{j=1}^{p}(\hat{\beta}_{j}-\tilde{\beta}_{j})^{2}}{t}
=\frac{1}{t}\sum_{j=1}^{p}\left(\frac{1}{\lambda_{0}^{2}}+\frac{2}{\lambda_{0}^{2}}\right)
=\frac{1}{t}\frac{3p}{\lambda_{0}^{2}}.\end{split}
\]
Set $t=\frac{C_{5}^{2}(\eta^{*})^{2}M}{m\,\phi^{2}\,c^{2}}{q(1+K)\frac{\log p}{n}}$,
then $\frac{1}{t}\frac{3p}{\lambda_{0}^{2}}\rightarrow0$ when $n,p\rightarrow\infty$.
Thus, from triangle inequality, 
\[
\begin{split} & \E_{\b_{0}}\P_{\bm{w},\m}\left(\|\wt\b-\b_{0}\|_2>\frac{C_{5}\eta^{*}\sqrt{M}}{\sqrt{m}\,\phi\,c}\sqrt{q(1+K)\frac{\log p}{n}}\C\Y^{(n)}\right)\\
& \leq\E_{\b_{0}}\P_{\bm{w},\m}\left(\|\wt\b-\wh\b\|_2^{2}>\frac{C_{5}\eta^{*}\sqrt{M}}{2\sqrt{m}\,\phi\,c}\sqrt{q(1+K)\frac{\log p}{n}}\C\Y^{(n)}\right)\\
& +\E_{\b_{0}}\P_{\bm{w},\m}\left(\|\bm{\Theta}\|_2>\frac{C_{5}\eta^{*}\sqrt{M}}{2\sqrt{m}\,\phi\,c}\sqrt{q(1+K)\frac{\log p}{n}}\C\Y^{(n)}\right)\\
& \leq\frac{1}{t}\frac{3p}{\lambda_{0}^{2}}+\P_{\b_{0},\m,\bm{w}}\left(\|\bm{\Theta}\|_2>\frac{C_{5}\eta^{*}\sqrt{M}}{2\sqrt{m}\,\phi\,c}\sqrt{q(1+K)\frac{\log p}{n}}\C\Y^{(n)}\right),
\end{split}
\]
where the right hand side  goes to zero as $n$ goes to infinity. 
Thus, we have 
\[
\lim_{n\rightarrow\infty}\E_{\b}\P_{\bm{w},\m}\left(\|\wt\b-\b_{0}\|_2>\frac{C_{5}\eta^{*}\sqrt{M}}{\sqrt{m}\,\phi\,c}\sqrt{q(1+K)\frac{\log p}{n}}\C\Y^{(n)}\right)=0.
\]

\subsection{Proof of Corollary \ref{regression_dirichlet}}\label{sec:appendix_regression_corollary_proof}
We notice that if
$w_{i}\sim\frac{1}{\alpha}\text{Gamma}(\alpha,1)$, we have $\E w_{i}=1$,
$Var(w_{i})=\frac{1}{\alpha}\lesssim\frac{1}{\log n}$, $Cov(w_{i},w_{j})=0$.
So we only need to prove the two high probability bounds for the order
statistics for $w_{i}$. %Set $\alpha\geq max(\frac{\frac{\Delta^2}{7.9}+log(2)-log(\Delta)}{2-log(4)}, \frac{\frac{(\Delta^*)^2}{7.9}+log(2)-log(\Delta^*)}{2-log(4)})$ where $\Delta^*=\Delta\sqrt{\frac{t+1}{\eta+\gamma}},\Delta=\sqrt{2log(\frac{(1-\theta)\lambda_0}{\theta \lambda_1})}$.
Let $\alpha=2(\eta+\gamma)\log p$, from the union bound we have
\[
\begin{split} & \P\left(\min w_{i}<\frac{1}{e}\right)\leq n\P\left(w_{i}<\frac{1}{e}\right)=n\P_{v\sim\text{Gamma}(\alpha,1)}(v<\alpha/e)=ne^{-\alpha/e}\sum_{i=\alpha}^{\infty}\frac{(\alpha/e)^{i}}{i!}\\
& \leq ne^{-\alpha/e}\sum_{i=\alpha}^{\infty}\frac{(\alpha/e)^{i}}{\sqrt{2\pi}i^{i+1/2}e^{-i}}=n\frac{e^{-\alpha/e}}{\sqrt{2\pi}}\sum_{i=\alpha}^{\infty}\left(\frac{\alpha}{i}\right)^{i}\leq n\frac{e^{-\alpha/e}}{\sqrt{2\pi}}[\sum_{i=\alpha+1}^{\infty}\left(\frac{\alpha}{i}\right)^{i}+1]\\
& \leq n\frac{e^{-\alpha/e}}{\sqrt{2\pi}}\left[\sum_{i=\alpha+1}^{\infty}\left(\frac{\alpha}{\alpha+1}\right)^{i}+1\right]=n\frac{e^{-\alpha/e}}{\sqrt{2\pi}}\left[\left(\frac{\alpha}{\alpha+1}\right)^{\alpha+1}(\alpha+1)+1\right]\leq C_{0}n\frac{e^{-\alpha/e}}{\sqrt{2\pi}}\left[\frac{\alpha+1}{e}+1\right]\\
& \leq C_{1}(\log p)\,n\,p^{-\frac{2(\eta+\gamma)}{e}}.
\end{split}
\]
Since $2(\eta+\gamma)/e\ge4/e>1$, we know $\lim_{n\rightarrow\infty}\P(\min w_{i}<\frac{1}{e})=0$.
From the proof of Corollary \ref{normal_mean_dirichlet} and the union
bound, for any $t\ge\frac{\alpha+1}{\alpha}$,
\[
\begin{split} & \P(\max w_{i}>t)=ne^{\alpha(log(t)+1-t}\leq Cn\left(p\right)^{-\frac{\eta+\gamma}{t}}.\end{split}
\]
Set $t=\frac{2}{3}(\eta+\gamma)>\frac{4}{3}\ge\frac{\alpha+1}{\alpha}$
when $n$ is sufficiently large, then $\lim_{n\rightarrow\infty}\P(\max w_{i}>t)=0$.
If $w_{i}\sim n\text{Dir}(\alpha,\cdots,\alpha)$, we have $\E w_{i}=1$,
$Var(w_{i})=n^{2}[\frac{1/n(1-1/n)}{n\alpha+1}]\lesssim\frac{1}{\log n}$,
$Cov(w_{i},w_{j})=-n^{2}[\frac{1/n^{2}}{n\alpha+1}]\lesssim\frac{1}{n\log n}$.
Using the fact that $w_{i}\xrightarrow{d}\frac{1}{\alpha}Gamma(\alpha,1)$,
we can prove 
\[
\begin{split} & \lim_{n\rightarrow\infty}\P\left(\min w_{i}<\frac{1}{e}\right)\end{split}
=0,
\]
\[
\begin{split} & \lim_{n\rightarrow\infty}\P(\max w_{i}>2(\eta+\gamma)/3)=0.\end{split}
\]
Thus, conditions (1)-(4) hold for both $\bm{w}\sim n\text{Dir}(\alpha,\cdots,\alpha)$
and $w_{i}\sim\frac{1}{\alpha}Gamma(\alpha,1)$ where $\alpha\gtrsim\log p$.

\subsection{Derivation of Observations in Section \ref{sec:motivation_noncenter}}
\label{sec_appendix:motivation_noncenter_proof}
In this section, we theoretically justify our  statements from Section \ref{sec:motivation_noncenter}.

\subsubsection{Notation}
Define 
\begin{align*}
c_0^{(-)}&=(1-\theta)\lambda_0e^{-y_i\lambda_0+\lambda_0^2/2n},  &c_0^{(+)}&=(1-\theta)\lambda_0e^{y_i\lambda_0+\lambda_0^2/2n}. 
\end{align*}
Next, define
\begin{align*}
\phi_0^{(-)}(x)&=\phi(x;y_i-\lambda_0/n,1/n), 
&\phi_0^{(+)}(x)&=\phi(x;y_i+\lambda_0/n,1/n), 
\end{align*}
where $\phi(x; \mu, \sigma^2)$ is the Gaussian density with mean $\mu$ and variance $\sigma^2$. The quantities $c_1^{(-)},\,c_1^{(+)},\,\phi_1^{(-)}(x),\,\phi_1^{(+)}(x)$ are defined in a similar way in Section \ref{sec:motivation_noncenter} in the main paper, with $\lambda_0$ replaced by $\lambda_1$. 
\iffalse
Define $c_+=0.5\left(1+\sqrt{1-4/(\lambda_0-\lambda_1)^2}\right)$ and $\delta_{c_+}=1/(\lambda_0-\lambda_1)\log\left[\frac{1-\theta}{\theta}\frac{\lambda_0}{\lambda_1}\frac{c_+}{1-c_+}\right]$. Also $p^*(x)$ and $\lambda^*(x)$ are as defined in \eqref{eq: p_star}.

Consider Gaussian sequence model \eqref{eq:Gaussian_sequence}
If multiplying both sides of (\ref{eq:Gaussian_sequence}) by $\sqrt{n}$, we get the linear model with noise variance equal to 1 and $||\X_j||=\sqrt{n}$. So here we have $\Delta^L<\Delta_j=\inf_{t>0}\left[nt/2-\rho(t\C\theta)/t\right]\doteq\Delta<\Delta^U$ where $\Delta^L=\sqrt{2n\log1/p^*(0)-d}+\lambda_1$ and $\Delta^U=\sqrt{2n\log1/p^*(0)}+\lambda_1$, $d=-[\lambda^*(\delta_{c_+})-\lambda_1]^2-2n\log p^*(\delta_{c_+})$.  \fi 
Throughout this section, we assume that the parameters $\lambda_0$ and  $\lambda_1$ satisfy $(1-\theta)/\theta\asymp n^a$, $\lambda_0\asymp n^d$ where $a,d\geq 2$ and $1/\sqrt{n}<\lambda_1\leq c_0$ with some constant $c_0$.

\subsubsection{Active Coordinates}\label{sec:appendix_motivation_active}

\begin{prop}\label{prop:1}
	Assume the true posterior defined in \eqref{eq:true_posterior} with an active coordinate. Conditioning on the event  $|y_i|>|\beta_i^0|/2$, we have 
	$$
	w_0=\pi(\gamma_i=0\C y_i)\rightarrow 0.
	$$
\end{prop}
\begin{proof}
	Without loss of generality, we assume $y_i>\beta_i/2>0$.
	\begin{equation}\label{eq: w1}
	\begin{split}
	w_1&=\pi(\gamma_i=1\C y_i)=\frac{\pi(y_i\C \gamma_i=1)\pi(\gamma_i=1)}{\pi(y_i\C \gamma_i=1)\pi(\gamma_i=1)+\pi(y_i\C \gamma_i=0)\pi(\gamma_i=0)}
	\\&=\frac{\int_{\beta_i}\pi(y_i\C \beta_i)\pi(\beta_i\C \gamma_i=1)\pi(\gamma_i=1)d\beta_i}{\int_{\beta_i}\pi(y_i\C \beta_i)\pi(\beta_i\C \gamma_i=1)\pi(\gamma_i=1)d\beta_i+\int_{\beta_i}\pi(y_i\C \beta_i)\pi(\beta_i\C \gamma_i=0)\pi(\gamma_i=0)d\beta_i}
	\\&=\frac{\int_{0}^{\infty}c_1^{(-)}\phi_1^{(-)}(\beta_i)d\beta_i +\int_{-\infty}^{0}c_1^{(+)}\phi_1^{(+)}(\beta_i)d\beta_i}
	{\int_{0}^{\infty}c_1^{(-)}\phi_1^{(-)}(\beta_i)d\beta_i+\int_{0}^{\infty}c_0^{(-)}\phi_0^{(-)}(\beta_i)d\beta_i +\int_{-\infty}^{0}c_1^{(+)}\phi_1^{(+)}(\beta_i)d\beta_i+\int_{-\infty}^{0}c_0^{(+)}\phi_0^{(+)}(\beta_i)d\beta_i}
	\end{split}
	\end{equation}
	We consider the four terms in the denominator separately. It is helpful to divide each of them by $\theta\lambda_1$. Regarding the first term, we have
	\begin{equation}\label{eq:motiv_denom1}
	\begin{split}
	\frac{1}{\theta\lambda_1}\int_{0}^{\infty}c_1^{(-)}\phi_1^{(-)}(\beta_i)d\beta_i
	=e^{-y_i\lambda_1+\lambda_1^2/2n}\left(1-\Phi\left(-\sqrt{n}(y_i-\frac{\lambda_1}{n})\right)\right)\rightarrow e^{-y_i\lambda_1}.
	\end{split}
	\end{equation}
	Regarding the second term in the denominator, from the Mills ratio we have
	\begin{equation*}
	\begin{split}
	&\frac{1}{\theta\lambda_1}\int_{0}^{\infty}c_0^{(-)}\phi_0^{(-)}(\beta_i)d\beta_i
	= \frac{1}{\theta\lambda_1}c_0^{(-)}\Phi\left(-\sqrt{n}\left(y_i-\frac{\lambda_0}{n}\right)\right)
	\\&\leq C\frac{(1-\theta)\lambda_0}{\theta\lambda_1}e^{-y_i\lambda_0+\lambda_0^2/2n}\frac{\phi\left(\sqrt{n}(y_i-\frac{\lambda_0}{n}\right)}{\sqrt{n}(y_i-\frac{\lambda_0}{n})}
	=C\frac{(1-\theta)\lambda_0}{\sqrt{n}(y_i-\frac{\lambda_0}{n})\theta\lambda_1}\frac{1}{\sqrt{2\pi}}e^{-\frac{n}{2}y_i^2}
	\end{split}
	\end{equation*}
	where $\frac{(1-\theta)\lambda_0}{\sqrt{n}(y_i-\frac{\lambda_0}{n})\theta\lambda_1}\frac{1}{\sqrt{2\pi}}e^{-\frac{n}{2}y_i^2}\rightarrow 0$. Thus, 
	\begin{equation}\label{eq:motiv_denom2}
	\frac{1}{\theta\lambda_1}\int_{0}^{\infty}c_0^{(-)}\phi_0^{(-)}(\beta_i)d\beta_i\rightarrow 0.
	\end{equation}
	For the third term in the denominator, we write
	\begin{equation}\label{eq:motiv_denom3}
	\begin{split}
	\frac{1}{\theta\lambda_1}\int_{-\infty}^{0}c_1^{(+)}\phi_1^{(+)}(\beta_i)d\beta_i
	=e^{y_i\lambda_1+\lambda_1^2/2n}\Phi\left(-\sqrt{n}(y_i+\frac{\lambda_1}{n})\right)\rightarrow 0.
	\end{split}
	\end{equation}
	For the fourth term, we then have
	\begin{equation*}
	\begin{split}
	\frac{1}{\theta\lambda_1}\int_{-\infty}^{0}c_0^{(+)}\phi_0^{(+)}(\beta_i)d\beta_i
	&=\frac{1}{\theta\lambda_1}\int_{-\infty}^{0}(1-\theta)\lambda_0\sqrt{\frac{n}{2\pi}}e^{-\frac{n}{2}(\beta_i-y_i)^2+\beta_i\lambda_0}d\beta_i
	\\&\leq \frac{1}{\theta\lambda_1}\int_{-\infty}^{0}(1-\theta)\lambda_0\sqrt{\frac{n}{2\pi}}e^{-\frac{n}{2}y_i^2+\beta_i\lambda_0}d\beta_i
	=\frac{1}{\theta\lambda_1}(1-\theta)\sqrt{\frac{n}{2\pi}}e^{-\frac{n}{2}y_i^2}
	\end{split}
	\end{equation*}
	where $\frac{1}{\theta\lambda_1}(1-\theta)\sqrt{\frac{n}{2\pi}}e^{-\frac{n}{2}y_i^2}\rightarrow 0$. Thus, 
	\begin{equation}\label{eq:motiv_denom4}
	\frac{1}{\theta\lambda_1}\int_{-\infty}^{0}c_0^{(+)}\phi_0^{(+)}(\beta_i)d\beta_i\rightarrow 0.
	\end{equation}
	From \eqref{eq:motiv_denom1}, \eqref{eq:motiv_denom2}, \eqref{eq:motiv_denom3} and \eqref{eq:motiv_denom4}, we know that 
	$w_1\rightarrow \frac{e^{-y_i\lambda_1}}{e^{-y_i\lambda_1}}=1$.
	Thus,
	$w_0=1-w_1\rightarrow 0$.
\end{proof}

\begin{prop}
	For the true posterior defined in \eqref{eq:true_posterior}, we have 
	$$
	\pi\left( \sqrt{n}(\beta_i-y_i)\C y_i,\gamma_i=1  \right)\rightarrow \phi(\sqrt{n}(\beta_i-y_i);0,1).
	$$
\end{prop}
\begin{proof}
	Setting $u_i=\sqrt{n}(\beta_i-y_i)$, we have
	\begin{equation}\label{eq:true_posterior_positive}
	\begin{split}
	\pi\left( u_i\C y_i,\gamma_i=1  \right)&=
	\frac{1}{\sqrt{n}}\frac{\1(u_{i}\geq -\sqrt{n}y_i)c_1^{(-)}\phi_1^{(-)}(u_i/\sqrt{n}+y_i) +\1(u_{i}< -\sqrt{n}y_i)c_1^{(+)}\phi_1^{(+)}(u_i/\sqrt{n}+y_i)}
	{\int_{0}^{\infty}c_1^{(-)}\phi_1^{(-)}(\beta_i)d\beta_i+ \int_{-\infty}^{0}c_1^{(+)}\phi_1^{(+)}(\beta_i)d\beta_i}.
	\end{split}
	\end{equation}
	Notice that both 
	$$
	\frac{1}{\sqrt{n}}\phi_1^{(-)}(u_i/\sqrt{n}+y_i) \rightarrow \phi(u_i;0,1)
	\quad\text{and}\quad
	\frac{1}{\sqrt{n}}\phi_1^{(+)}(u_i/\sqrt{n}+y_i) \rightarrow \phi(u_i;0,1).
	$$
	For any $u_i$, only one of $\1(u_{i}< -\sqrt{n}y_i)$  and $\1(u_{i}\geq -\sqrt{n}y_i)$ holds, and the denominator in \eqref{eq:true_posterior_positive} does not depend on $u_i$, so
	$
	\pi\left( u_i\C y_i,\gamma_i=1  \right)\rightarrow \phi(u_i;0,1)
	$.
\end{proof}

\begin{prop}
	Consider the fixed WBB estimator $\wh\beta_i$. Conditioning on $w_i$ and $y_i$, when $\wh\beta_i\neq 0$ and  $|y_i|>\frac{|\beta_i^0|}{2}$,  we have %{\color{blue} I do not understand the conditioning below? Is it a distribution that collapses to a point?}
	$$
	n\left(\wh\beta_i-y_i\right)\rightarrow -\frac{1}{w_i}\lambda_1.
	$$
\end{prop}
\begin{proof}
	Fixed WBB estimate $\wh\beta_i$ satisfies equation \eqref{eq: WBB_solution} in the main text. Without loss of generality, we assume $y_i>\frac{\beta_i^0}{2}>0$.
	Conditioning on $y_i, w_i$, when $\wh{\beta_i}\neq0$,
	\begin{equation*}
	\begin{split}
	&n\left(\wh{\beta_i}-y_i\right)
	=-\frac{1}{w_i}\lambda_1-\frac{1}{w_i}(\lambda_0-\lambda_1)(1-p^*(\wh{\beta_i}))  
	\\&=-\frac{1}{w_i}\lambda_1-\frac{1}{w_i}\frac{\frac{(1-\theta)\lambda_0}{\theta\lambda_1}(\lambda_0-\lambda_1)e^{-\C\wh{\beta_i}\C(\lambda_0-\lambda_1)}}{1+\frac{(1-\theta)\lambda_0}{\theta\lambda_1}e^{-\C\wh{\beta_i}\C(\lambda_0-\lambda_1)}}
	\rightarrow -\frac{1}{w_i}\lambda_1.\qedhere
	\end{split}
	\end{equation*}
\end{proof}

\begin{prop}
	Consider the fixed WBB estimator $\wh\beta_i$. Conditioning on the event that $|y_i|>\frac{|\beta_i^0|}{2}$, we have 
	$$
	\P_{w_i}\left( \wh\beta_i=0  \C y_i\right)\rightarrow 0.
	$$
\end{prop}

\begin{proof}
	\iffalse
	Consider multiplying both sides of the Gaussian sequence model by $\sqrt{n}$.
	If set {\color{blue} I do not understand why $\x_i$ is a vector? In the sequence model this should be just a number?Why should it converge to some constant $x_i$'?}
	$\sqrt{n}\x_{i}\rightarrow\x_{i}'$, $y_{i}\sqrt{nw_{i}}\rightarrow y_{i}'$,
	$\beta_{i}/\sqrt{w_{i}}\rightarrow\beta_{i}'$, $\lambda_{0}/\sqrt{w_{i}}\rightarrow\lambda_{0}'$,
	$\lambda_{1}\sqrt{w_{i}}\rightarrow\lambda_{1}'$, we get {\color{blue} I am confused because this is not a regression. Also the notation $N(0,1)$ is somewhat confusing.}
	\begin{align}\label{eq_appendix:rescale}
	y_{i}'=\left(\x_{i}'\right)^{\top}\b+N(0,1)
	\end{align}
	which is the regression model and the equation \eqref{eq: WBB_solution} follows
	directly. 
	\fi
	The fixed WBB estimator $\hat\beta_i$ satisfies Equation \eqref{eq: WBB_solution}.
	When $n$ is sufficiently large, from Proposition 5 in \cite{moran2018variance}, we have
	\begin{equation}
	\Delta_{w_{i}}\le \sqrt{2/(nw_i)\log[1/p^*(0)]}+\lambda_1/(nw_i).
	\label{eq_appendix:delta_wi}
	\end{equation}
	Thus, from the Markov's inequality, we find %{\color{blue} below we have two notations $\mathbb I$ and $\bm 1$, please fix}
	\begin{equation}
	\begin{split} 
	& \P_{w_{i}}\left(\wh\beta_{i}=0\C y_{i}\right) 
	\leq\P\left(|\wh\beta_{i}-y_{i}|^{1/2}\ge |y_{i}|^{1/2}\mid y_i\right)
	\leq\frac{\E\left[|\wh\beta_{i}-y_{i}|^{1/2}\mid y_i\right]}{|y_{i}|^{1/2}}\\
	& \stackrel{(a)}{=}\frac{1}{|y_{i}|^{1/2}}\left\{ \E\left[|y_{i}|^{1/2}\1\left(|y_{i}|\leq\Delta_{w_{i}}\right)\mid y_i\right]
	+\E\left[\left(\frac{1}{w_{i}n}\lambda^{*}(\wh{\beta_{i}})\right)^{1/2}\1\left(|y_{i}|>\Delta_{w_{i}}\right)\mid y_i\right]\right\} \\
	& \stackrel{(b)}{\le} \P\left(|\beta_i^0|/2\le \sqrt{2/(nw_i)\log[1/p^*(0)]}+\lambda_1/(nw_i)\right)
	+\frac{1}{|y_{i}|^{1/2}}\E\left(\frac{1}{nw_{i}}\lambda^{*}(\wh{\beta_{i}})\right)^{1/2}
	\\& \leq \P\left(w_i^{1/2}<\frac{1}{|\beta^0_i|}\left[\sqrt{\frac{2}{n}\log[1/p^*(0)]+\frac{2\lambda_1}{n}|\beta_i^0|}+\sqrt{\frac{2}{n}\log[1/p^*(0)]}\right]\right)
	+\frac{1}{|y_{i}|^{1/2}}\E\left(\frac{1}{nw_{i}}\lambda^{*}(\wh{\beta_{i}})\right)^{1/2}.
	\end{split}
	\label{eq: active_WBB_bound}
	\end{equation}
	where (a) follows from Equation \eqref{eq: WBB_solution}, (b) follows from the condition $|y_i|>|\beta_i^0|/2$ and \eqref{eq_appendix:delta_wi}.
	Notice that 
	$$
	\P\left(w_i^{1/2}<\frac{1}{|\beta^0_i|}\left[\sqrt{\frac{2}{n}\log[1/p^*(0)]+\frac{2\lambda_1}{n}|\beta_i^0|}+\sqrt{\frac{2}{n}\log[1/p^*(0)]}\right]\right)
	\rightarrow 0.
	$$
	In order to bound $\E\left(\frac{1}{nw_i}\lambda^{*}(\wh{\beta_{i}})\right)^{1/2}$,
	notice that 
	\[
	\begin{split}
	\wh\beta_i=&\arg\max_{\beta_i\in\R}\left\{- \frac{w_i n}{2}(y_i-\beta_i)^2+\log \pi(\beta_i\C\theta) \right\}=\arg\max_{\beta_i\in\R}\left\{- \frac{1}{2}(\sqrt{w_in}y_i-\sqrt{w_in}\beta_i)^2+\log \pi(\beta_i\C\theta) \right\}.
	\end{split}
	\]
	Thus, following the same analysis as that for Equation \eqref{eq:appendix_normalmeans_active3}, \eqref{eq:appendix_normalmeans_active4} in Section \ref{sec:appendix_normal_means_proof},
	we know that
	\[
	\begin{split}
	&\E\left(\frac{1}{nw_{i}}\lambda^{*}(\wh{\beta_{i}})\right)^{1/2}
	\le \E \left(\frac{\sqrt{nw_i}+\lambda_1}{nw_i}\right)^{1/2}
	+\E \left(\frac{2\sqrt{nw_i}+\lambda_1}{nw_i}\right)^{1/2}+\frac{1}{\lambda_0}
	\\\le&
	(1+\sqrt{2})\E \frac{1}{(nw_i)^{1/4}}+2\E \left(\frac{\lambda_1}{nw_i}\right)^{1/2}
	+\frac{1}{\lambda_0},
	\end{split}
	\]
	and thus, the right hand size of \eqref{eq: active_WBB_bound} satisfies
	\[
	\P\left(w_i^{1/2}<\frac{1}{|\beta^0_i|}\left[\sqrt{\frac{2}{n}\log[1/p^*(0)]+\frac{2\lambda_1}{n}|\beta_i^0|}+\sqrt{\frac{2}{n}\log[1/p^*(0)]}\right]\right)
	+
	\frac{1}{y_{i}^{2}}\E\left(\frac{1}{nw_{i}}\lambda^{*}(\wh{\beta_{i}})\right)^{2}\rightarrow0.
	\]
	Thus $\P_{w_{i}}\left(\wh\beta_{i}=0\C y_{i}\right)\rightarrow0$
	holds. 
\end{proof}

\subsubsection{Inactive Coordinates}\label{sec:appendix_motivation_inactive}

\begin{prop}
	For $w_0$ and $w_1$ defined in \eqref{eq:true_posterior}, when conditioning on $y_i\asymp \frac{1}{\sqrt{n}}$, we have
	$w_1\rightarrow 0$ and $w_0\rightarrow 1$.
\end{prop}

\begin{proof}
	The expression for $w_1$ is given in \eqref{eq: w1}. Again, we consider the four terms in the denominator separately and divide each one of them by $\theta\lambda_1$. For the first term, we write
	\begin{equation}\label{eq:motiv_inactive_denom1}
	\begin{split}
	\frac{1}{\theta\lambda_1}\int_{0}^{\infty}c_1^{(-)}\phi_1^{(-)}(\beta_i)d\beta_i
	=e^{-y_i\lambda_1+\lambda_1^2/2n}\left(1-\Phi\left(-\sqrt{n}(y_i-\frac{\lambda_1}{n})\right)\right)\rightarrow 1-\Phi\left(-\sqrt{n}y_i\right).
	\end{split}
	\end{equation}
	For the second term, for any fixed $\epsilon>0$, we have
	\begin{equation*}
	\begin{split}
	\frac{1}{\theta\lambda_1}\int_{0}^{\infty}c_0^{(-)}\phi_0^{(-)}(\beta_i)d\beta_i
	&=\frac{1}{\theta\lambda_1}\int_{0}^{\infty}(1-\theta)\lambda_0\sqrt{\frac{n}{2\pi}}e^{-\frac{n}{2}(\beta_i-y_i)^2-\beta_i\lambda_0}d\beta_i
	\\&
	\geq \frac{1}{\theta\lambda_1}\int_{0}^{\epsilon y_i}(1-\theta)\lambda_0\sqrt{\frac{n}{2\pi}}e^{-\frac{n(\epsilon-1)^2}{2}y_i^2-\beta_i\lambda_0}d\beta_i
	\\&=\frac{1}{\theta\lambda_1}(1-\theta)\sqrt{\frac{n}{2\pi}}e^{-\frac{n(\epsilon-1)^2}{2}y_i^2}\left( 1-e^{-\epsilon y_i\lambda_0}  \right)
	\end{split}
	\end{equation*}
	where $e^{-\epsilon y_i\lambda_0}    \rightarrow 0$. 
	Since the right-hand size term
	$$
	\frac{1}{\theta\lambda_1}(1-\theta)\sqrt{\frac{n}{2\pi}}e^{-\frac{n(\epsilon-1)^2}{2}y_i^2}\left( 1-e^{-\epsilon y_i\lambda_0}  \right)\rightarrow \infty
	$$
	because $y_i\asymp n^{-1/2}$ we know
	\begin{equation}\label{eq:motiv_inactive_denom2}
	\frac{1}{\theta\lambda_1}\int_{0}^{\infty}c_0^{(-)}\phi_0^{(-)}(\beta_i)d\beta_i \rightarrow \infty.
	\end{equation}
	\iffalse On the other side,
	\begin{equation*}
	\begin{split}
	\frac{1}{\theta\lambda_1}\int_{0}^{\infty}c_0^{(-)}\phi_0^{(-)}(\beta_i)d\beta_i
	&=\frac{1}{\theta\lambda_1}\int_{0}^{\infty}(1-\theta)\lambda_0\sqrt{\frac{n}{2\pi}}e^{-\frac{n}{2}(\beta_i-y_i)^2-\beta_i\lambda_0}d\beta_i
	\\&\leq \frac{1}{\theta\lambda_1}\int_{0}^{\infty}(1-\theta)\lambda_0\sqrt{\frac{n}{2\pi}}e^{-\frac{n}{2}y_i^2-\beta_i\lambda_0}d\beta_i
	=\frac{1}{\theta\lambda_1}(1-\theta)\sqrt{\frac{n}{2\pi}}e^{-\frac{n}{2}y_i^2}
	\end{split}
	\end{equation*}
	Thus, 
	$$\frac{1}{\theta\lambda_1}\int_{0}^{\infty}c_0^{(-)}\phi_0^{(-)}(\beta_i)d\beta_i\approx \frac{1}{\theta\lambda_1}(1-\theta)\sqrt{\frac{n}{2\pi}}e^{-\frac{n}{2}y_i^2}$$
	\fi
	Regarding the third term in denominator,  we obtain
	\begin{equation}\label{eq:motiv_inactive_denom3}
	\begin{split}
	\frac{1}{\theta\lambda_1}\int_{-\infty}^{0}c_1^{(+)}\phi_1^{(+)}(\beta_i)d\beta_i
	=e^{y_i\lambda_1+\lambda_1^2/2n}\Phi\left(-\sqrt{n}(y_i+\frac{\lambda_1}{n})\right).
	\rightarrow \Phi\left(-\sqrt{n}y_i\right)
	\end{split}
	\end{equation}
	Finally, for  the fourth term, for any fixed $\epsilon>0$, we have
	\begin{equation*}
	\begin{split}
	\frac{1}{\theta\lambda_1}\int_{-\infty}^{0}c_0^{(+)}\phi_0^{(+)}(\beta_i)d\beta_i
	&=\frac{1}{\theta\lambda_1}\int_{-\infty}^{0}(1-\theta)\lambda_0\sqrt{\frac{n}{2\pi}}e^{-\frac{n}{2}(\beta_i-y_i)^2+\beta_i\lambda_0}d\beta_i
	\\&\geq \frac{1}{\theta\lambda_1}\int_{-\epsilon y_i}^{0}(1-\theta)\lambda_0\sqrt{\frac{n}{2\pi}}e^{-\frac{n(1+\epsilon)^2}{2}y_i^2+\beta_i\lambda_0}d\beta_i
	\\&=\frac{1}{\theta\lambda_1}(1-\theta)\sqrt{\frac{n}{2\pi}}e^{-\frac{n(1+\epsilon)^2}{2}y_i^2}\left( 1-e^{-\epsilon y_i\lambda_0}  \right)
	\end{split}
	\end{equation*}
	where $e^{-\epsilon y_i\lambda_0}    \rightarrow 0$. 
	\iffalse On the other side,
	\begin{equation*}
	\begin{split}
	\frac{1}{\theta\lambda_1}\int_{-\infty}^{0}c_0^{(+)}\phi_0^{(+)}(\beta_i)d\beta_i
	&=\frac{1}{\theta\lambda_1}\int_{-\infty}^{0}(1-\theta)\lambda_0\sqrt{\frac{n}{2\pi}}e^{-\frac{n}{2}(\beta_i-y_i)^2+\beta_i\lambda_0}d\beta_i
	\\&\leq \frac{1}{\theta\lambda_1}\int_{-\infty}^{0}(1-\theta)\lambda_0\sqrt{\frac{n}{2\pi}}e^{-\frac{n}{2}y_i^2+\beta_i\lambda_0}d\beta_i
	=\frac{1}{\theta\lambda_1}(1-\theta)\sqrt{\frac{n}{2\pi}}e^{-\frac{n}{2}y_i^2}
	\end{split}
	\end{equation*}
	\fi
	Since the right hand side term satisfies
	$$
	\frac{1}{\theta\lambda_1}(1-\theta)\sqrt{\frac{n}{2\pi}}e^{-\frac{n(1+\epsilon)^2}{2}y_i^2}\left( 1-e^{-\epsilon y_i\lambda_0}  \right) \rightarrow \infty
	$$
	we know
	\begin{equation}\label{eq:motiv_inactive_denom4}
	\frac{1}{\theta\lambda_1}\int_{-\infty}^{0}c_0^{(+)}\phi_0^{(+)}(\beta_i)d\beta_i \rightarrow \infty.
	\end{equation}
	%Thus, 
	%$$\frac{1}{\theta\lambda_1}\int_{-\infty}^{0}c_0^{(+)}\phi_0^{(+)}(\beta_i)d\beta_i\approx \frac{1}{\theta\lambda_1}(1-\theta)\sqrt{\frac{n}{2\pi}}e^{-\frac{n}{2}y_i^2}$$
	Combining \eqref{eq:motiv_inactive_denom1}, \eqref{eq:motiv_inactive_denom2}, \eqref{eq:motiv_inactive_denom3} and \eqref{eq:motiv_inactive_denom4}, we know that in \eqref{eq: w1},
	the $\text{denominator}\,\rightarrow \infty$ and thus the  $\text{numerator}\,\rightarrow 1$.
	Thus
	$w_1\rightarrow 0$ and $w_0=1-w_1\rightarrow 1$.
	\iffalse satisfies
	\begin{equation}\label{eq: inactive_pos_denom}
	\begin{split}
	&\frac{1}{\theta\lambda_1}\left[\int_{0}^{\infty}c_1^{(-)}\phi_1^{(-)}(\beta_i)d\beta_i+\int_{0}^{\infty}c_0^{(-)}\phi_0^{(-)}(\beta_i)d\beta_i +\int_{-\infty}^{0}c_1^{(+)}\phi_1^{(+)}(\beta_i)d\beta_i+\int_{-\infty}^{0}c_0^{(+)}\phi_0^{(+)}(\beta_i)d\beta_i\right]
	\\&\approx \frac{2}{\theta\lambda_1}(1-\theta)\sqrt{\frac{n}{2\pi}}e^{-\frac{n}{2}y_i^2}+1
	\end{split}
	\end{equation}
	\fi 
\end{proof}

We denote 
$$
d_{\text{TV}}(P,Q)=\sup_{A\in\mathcal{F}}|P(A)-Q(A)|,
$$
which is the total variation distance between two probability measures $P$ and $Q$ on a sigma-algebra $\mathcal{F}$ of subsets of the sample space.
We have the following result:
\begin{prop}
	When $n$ is sufficiently large, conditioning on  the event %{\color{blue} I do not understand the notation $\asymp O$ should it be either $O$ or $\asymp$? } 
	$|y_i|\asymp \frac{1}{\sqrt{n}}$, we have 
	\iffalse
	$$
	\pi\left(\lambda_0\beta_{i}\C y_i,\gamma_i=0\right) \rightarrow
	\frac{1}{2}e^{-|\lambda_0\beta_i|}.
	$$
	\fi
	$$
	d_{\text{TV}}\left(\pi\left(\lambda_0\beta_{i}\C y_i,\gamma_i=0\right),\,\text{Laplace}(1)\right)\rightarrow 0,
	$$
	where $\text{Laplace}(1)$ is the Laplace distribution whose density is $f(t)=\frac{1}{2}e^{-|t|}$.
\end{prop}
\begin{proof}
	Notice that 
	\begin{equation*}
	\begin{split}
	\pi\left(\beta_i\C y_i,\gamma_i=0\right)
	&\propto \pi\left(y_i\C \beta_i,\gamma_i=0\right) \pi\left(\beta_i\C \gamma_i=0\right)
	\propto
	e^{-\frac{n}{2}(\beta_i-y_i)^2}e^{-|\beta_i|\lambda_0}.
	\end{split}
	\end{equation*}
	Thus, letting $\beta_i'=\lambda_0\beta_i$, since $\lambda_0\asymp p^d$ with $d\ge 2$ we have for any fixed $\beta_i'$ 
	\begin{equation*}
	\begin{split}
	\pi\left(\beta_i'\C y_i,\gamma_i=0\right)
	\propto
	e^{-\frac{n}{2}(\beta_i'/\lambda_0-y_i)^2}e^{-|\beta_i'|}
	\rightarrow 
	e^{-\frac{n}{2}(y_i)^2}e^{-|\beta_i'|}
	\end{split}
	\end{equation*}
	which implies that 
	$$
	\pi\left(\beta_i'\C y_i,\gamma_i=0\right)
	\rightarrow 
	e^{-|\beta_i'|}.
	$$
	From \cite{scheffe1947useful}, we know that $\beta_i'$ converges in total variation to Laplace(1).
\end{proof}

\begin{prop} 
	Consider the  fixed WBB estimator $\wh\beta_{i}$. Conditioning
	on $|y_{i}|\asymp n^{-1/2}$, we have 
	$$
	\P_{w_{i}}\left(\wh{\beta_{i}}=0\C y_{i}\right)\rightarrow1.
	$$
\end{prop}

\begin{proof} The definition of the fixed WBB sample $\wh{\beta_{i}}$ is in
	\eqref{eq: WBB_solution}. Notice that
	\[
	\Delta_{w_{i}}=\inf_{t>0}\left[t/2-\rho(t\mid\theta)/(nw_{i}t)\right]=\frac{1}{nw_{i}}\inf_{t>0}\left[nw_{i}t/2-\rho(t\mid\theta)/t\right]\asymp\frac{\sqrt{2nw_{i}\log[1/p^{*}(0)]}}{nw_{i}}.
	\]
	Then, when $n$ is sufficiently large,
	\begin{align*}
	&\P_{w_{i}}\left(\wh{\beta_{i}}=0\C y_{i}\right)\ge\P_{w_{i}}\left(|y_{i}|\le\Delta_{w_{i}}\C y_{i}\right)  \\&\ge\P_{w_{i}}\left(|y_{i}|\le\frac{1}{2}\frac{\sqrt{2nw_{i}\log[1/p^{*}(0)]}}{nw_{i}}\C y_{i}\right)
	\ge\P_{w_{i}}\left(w_{i}\le\sqrt{\frac{\log[1/p^{*}(0)]}{2n|y_{i}|^{2}}}\C y_{i}\right),
	\end{align*}
	where the right-hand size 
	\[
	\P_{w_{i}}\left(w_{i}\le\sqrt{\frac{\log[1/p^{*}(0)]}{2n|y_{i}|^{2}}}\C y_{i}\right)\rightarrow1
	\]
	since $y_{i}\asymp n^{-1/2}$ and $\log[1/p^{*}(0)]\rightarrow\infty$.
	Thus,
	\[
	\P_{w_{i}}\left(\wh{\beta_{i}}=0\C y_{i}\right)\rightarrow1.\qedhere
	\]
\end{proof}

\begin{remark}
	The random WBB is equivalent to the fixed WBB by setting weights to be $\bm w/w_p$ where $w_p$ is the weight assigned to the  prior term. Thus, using exactly the same arguments as fixed WBB, we can prove:  $\P_{w_i}\left(\wh{\beta}_i^{\text{random}}=0\C y_i\right)\rightarrow 1$.
\end{remark}

\section{Details of Connection to NPL in Section \ref{sec:connection}}\label{sec:appendix_connection_explanation}

In Algorithm \ref{alg: Fong}, we define the loss function $l$ as %{\color{blue} what is $\tilde \beta_j$?}
\begin{equation}\label{eq: loss-NPL}
l(\x_i,y_i,\b)=-\frac{1}{2\sigma^2}(y_i-\x_i^T\b)^2+\frac{1}{n}\log \left[\int_\theta\prod_{j=1}^p\pi(\beta_j\C\theta)d\pi(\theta)\right].
\end{equation}

\paragraph{Motivation for the Prior.}
For paired data $(\x_i,y_i)$, \cite{fong2019scalable} uses independent prior which assumes that $y_i$ does not depend on $\x_i$:
$$
\text{Prior 1: }\wt\x_k\sim \hat F_n(\x)=\frac{1}{n}\sum_{i=1}^{n}\delta(\x_i),\quad \wt y_k\C\wt \x_k\sim N(0,\sigma^2).
$$
%$$\text{Prior 1: }f_\pi(y,x)=f_\pi(y\C x)f_\pi(x), \,f_\pi(x) =\frac{1}{n}\sum_{i=1}^{n}\delta_{x_i}(x), \,f_\pi(y|x) = \epsilon\sim N(0,1)$$
However, this choice of $F_\pi$ might be problematic when the sample size $n$ is small. When $n$ is small, a well-specified prior can help us better estimate $\b$ but this independent prior shrinks all coeffcients towards zero and will result in bias (see Figure \ref{fig: Fong_compare1}).

%There are two questions that arise automatically from the above observations: (1) whether we can improve by choosing a more reasonable $F_\pi$; (2) what is the link between NPL posterior with our BB-SSL posterior since \cite{fong2019scalable}'s posterior bootstrap sampling algorithm is built on very similar ideas as BB-SSL. 

One possible solution is to use $y=\x^T\wh{\b}+\epsilon$ where $\wh{\b}$ is the MAP of $\b$ under SSL penalty. This choice of $f_\pi(y|x)$ has an Empirical Bayes flavor (\cite{martin2014asymptotically}). 
However, it includes information only from the posterior mode, ignoring shape information contained in the posterior variance. We could consider incorporating such information by adding noise $\bm\mu$ to $\wh{\b}$. We would want $\bm\mu$ to be centered at the origin but  not too far away from the origin. One choice that comes to mind is the spike distribution. So the prior $F_\pi$ becomes 
$$\wt\x_k\sim \hat F_n(\x)=\frac{1}{n}\sum_{i=1}^{n}\delta(\x_i),\quad \wt y_k\C\wt \x_k=\wt\x_k^T(\wh{\b}+\bm\mu)+\epsilon$$
where $\m\sim\text{Spike and }\epsilon\sim N(0,\sigma^2)$. %We also use loss defined in \ref{eq: loss-NPL}.
If $\wh{\b}$ is close enough to the truth, we have $\wt\x_k^T\wh{\b}\approx \wt\x_k^T\b_0$. Since $y_i=\x_i^T\b_0+\epsilon_i$ and $\epsilon_i\stackrel{d}{=}\epsilon$, we can set $\wt y_k\C\wt \x_k=y_i+\x_i^T\bm\mu$ where $i$ satisfies $\wt\x_k=\x_i$. Then the above prior becomes
$$
\text{Prior 2: }\wt\x_k\sim \hat F_n(\x)=\frac{1}{n}\sum_{i=1}^{n}\delta(\x_i),\quad \wt y_k\C\wt \x_k=y_i+\x_i^T\bm\mu\, \,\text{where}\,\,i\,\,\text{satisfies}\,\,\wt \x_k=\x_i.
$$

\paragraph{Derivation for Equation \eqref{eq: fong}}
When choosing $m=n$ in Algorithm \ref{alg: Fong}, the NPL posterior of \cite{fong2019scalable} using Prior 2 becomes
\begin{align*}
\wt\b^t&=\arg\max_{\b\in\R^p}\left\{-\frac{1}{2\sigma^2}\sum_{i=1}^nw_i(y_i-\x_i^T\b)^2-\frac{1}{2\sigma^2}\sum_{i=1}^n\tilde w_i(y_i+\x_i'\bm{\mu}-\x_i'\b)^2+ \frac{1}{n}\log \left[\int_\theta\prod_{j=1}^p\pi(\beta_j\C\theta)d\pi(\theta)\right]
\right\}
\\&=\arg\max_{\b\in\R^p}\left\{-\frac{1}{2\sigma^2}\sum_{i=1}^n(w_i+\tilde w_i)(y_i+\frac{\tilde w_i}{w_i+\tilde{w_i}}\x_i^T\bm\mu-\x_i^T\b)^2+ \frac{1}{n}\log \left[\int_\theta\prod_{j=1}^p\pi(\beta_j\C\theta)d\pi(\theta)\right]
\right\}
\\&\approx\arg\max_{\b\in\R^p}\left\{-\frac{1}{2\sigma^2}\sum_{i=1}^n(w_i+\tilde w_i)(y_i+\frac{c}{c+n}\x_i^T\bm\mu-\x_i^T\b)^2+ \frac{1}{n}\log \left[\int_\theta\prod_{j=1}^p\pi(\beta_j\C\theta)d\pi(\theta)\right]
\right\}
\\&\stackrel{\tilde\b=\b-c/(c+n)\bm\mu}{=}\arg\max_{\tilde\b\in\R^p}\left\{-\frac{1}{2\sigma^2}\sum_{i=1}^nw_i^*(y_i-\x_i^T\tilde\b)^2+ \log \left[\int_\theta\prod_{j=1}^p\pi(\tilde\beta_j+\frac{c}{c+n}\mu_j\C\theta)d\pi(\theta)\right]
\right\}+\frac{c}{c+n}\bm{\mu},
\end{align*}
where $w_i^*=n(w_i+\tilde w_i)$. Since $\bm\mu$ and $-\bm\mu$ follow the same distribution, define $\bm\mu^*=-\bm\mu$ and
\begin{equation*}
\tilde\b^t\text{\ensuremath{\overset{\text{D}}{=}}}\arg\max_{\tilde\b\in\R^p}\left\{-\frac{1}{2\sigma^2}\sum_{i=1}^nw_i^*(y_i-\x_i^T\tilde\b)^2+ \log\left[\int_\theta\prod_{j=1}^p\pi(\tilde\beta_j-\frac{c}{c+n}\mu_j^*\C\theta)d\pi(\theta)\right],
\right\}-\frac{c}{c+n}\bm\mu^*
\end{equation*}
where $(\bm w_{1:n},\tilde w_{1:n})\sim  \text{Dir}(1,\cdots,1,c/n, c/n, \cdots,c/n)$ and thus $(w_1^*,w_2^*,\cdots,w_n^*)\sim n\text{Dir}(1+c/n, \cdots,1+c/n)$. 
%They look similar with two differences: (1) in Equation (\ref{eq: fong}), although shrinking towards a random location, they add it back after calculating MAP, so the estimated variance will be smallar - this underestimation effect is even more severe for inactive $\beta_i$'s - we will get exactly 0, as illustrated in Figure \ref{fig: Fong_compare1}. This will not be the case for BB-SSL. (2) the role of concentration paramerters are different. In Equation (\ref{eq: fong}), we can control the concentration parameter of Dirichlet distribution by tuning $c$. The larger $c$ is, the more weight we put on the prior and also the more close to 1 our weights $w_i^*$ will become. In BB-SSL, we put constant weight on prior and the tuning parameter $\alpha$ only controls the degree of concentration of our weights.

\begin{sidewaysfigure}
	\begin{subfigure}{0.5\hsize}\centering
		\includegraphics[width=\hsize, height=5in]{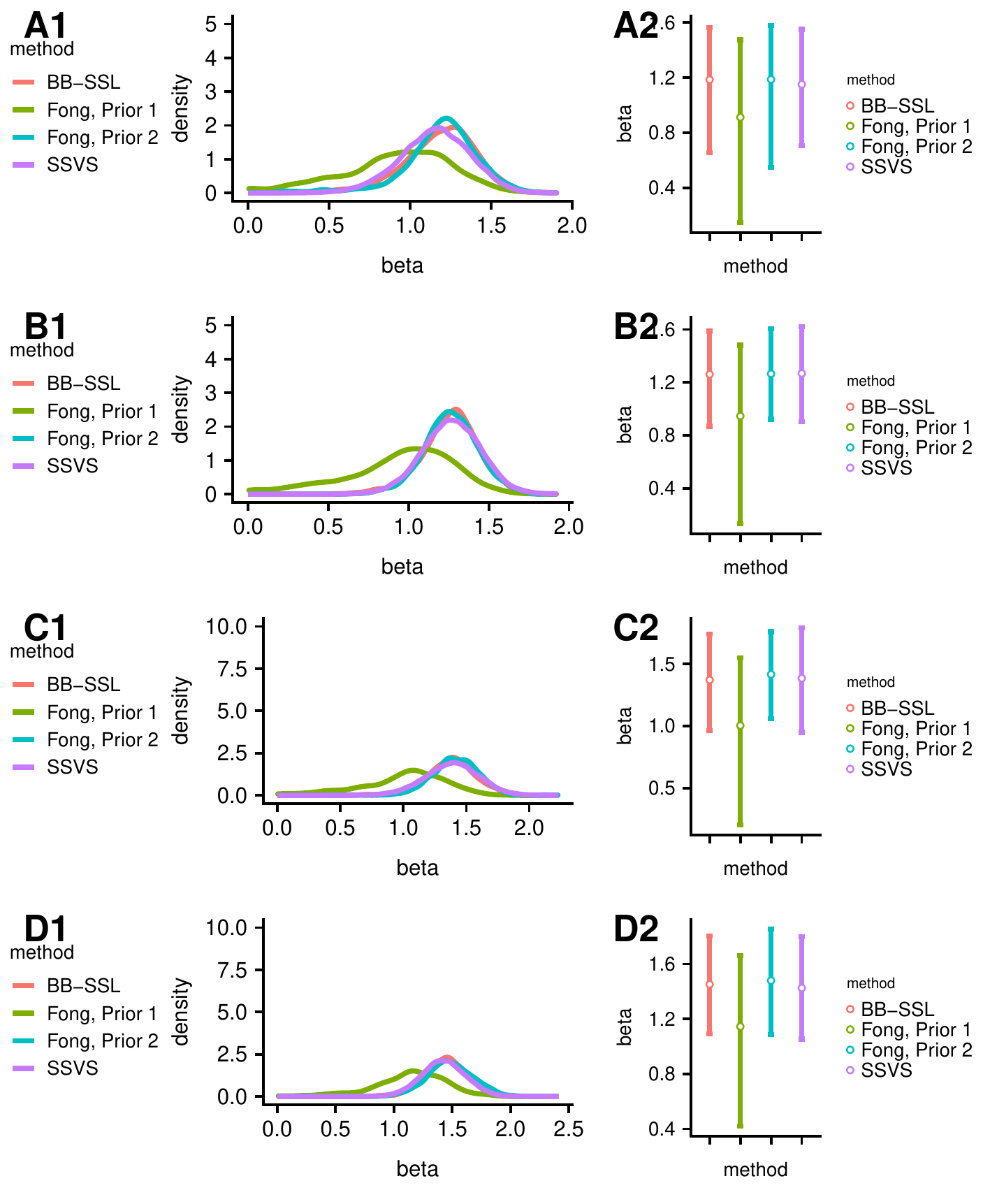}
		\caption{\small Active predictors, from top to bottom: $\beta_1,\beta_{4},\beta_{7},\beta_{10}$}
	\end{subfigure}%
	%\hfill <-- it is superfluous 
	\begin{subfigure}{0.5\hsize}\centering
		\includegraphics[width=\hsize, height=5in]{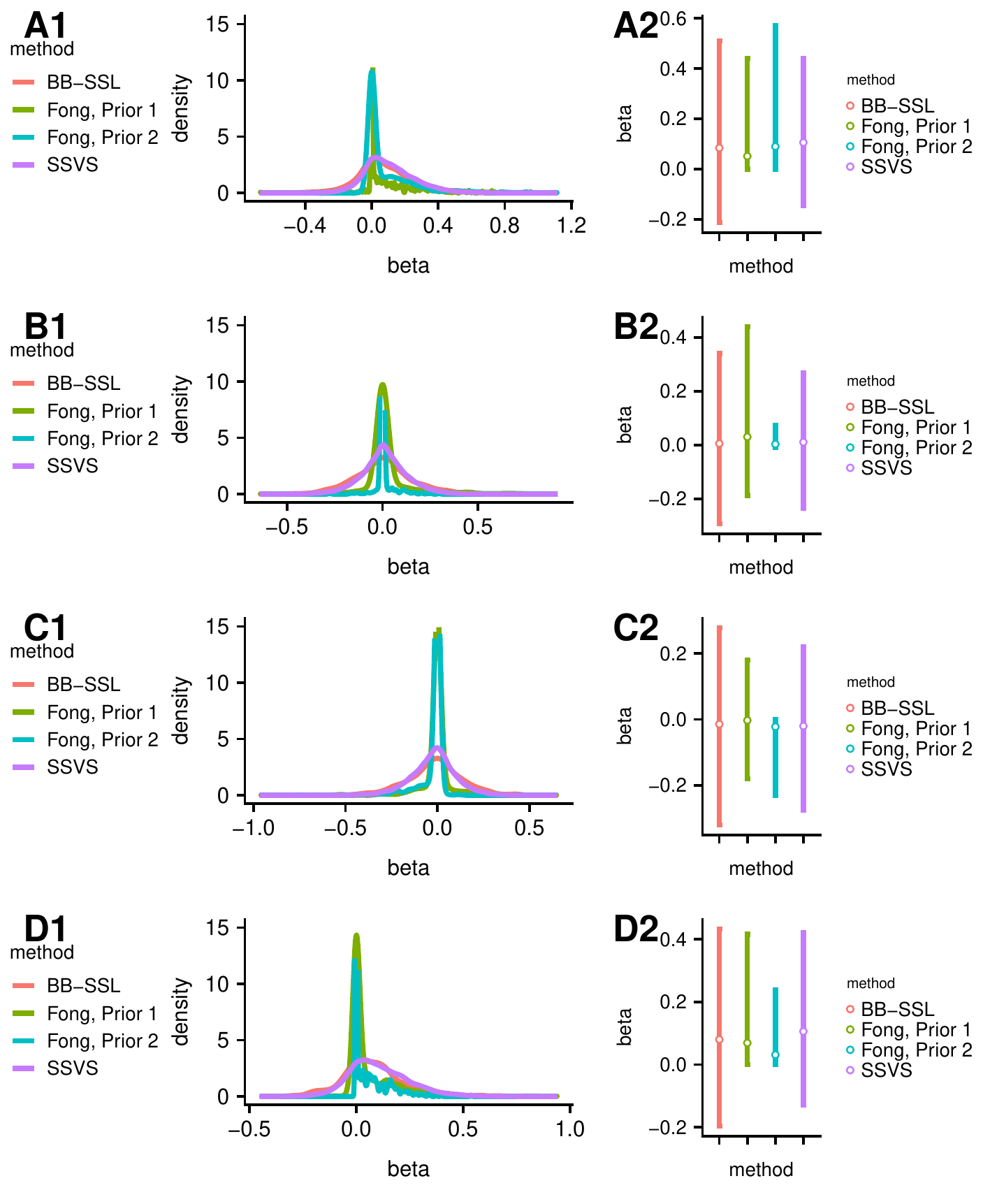}
		\caption{\small Inactive predictors, from top to bottom: $\beta_2,\beta_{5},\beta_{8},\beta_{11}$}
	\end{subfigure}
	\caption{\small Comparison of \cite{fong2019scalable}'s algorithm 2 with different choice of priors and loss function versus BB-SSL and SSVS under simulated dataset. We use $n=50, p=12$, $\b=(1.3,0,0,1.3,0,0,1.3,0,0,1.3,0,0)^T$, predictors are grouped into 4 blockes with correlation coefficient $\rho=0.6$, $\alpha=2$, $m=50$, $c=10$, $\lambda_0=7,\lambda_1=0.15$.}\label{fig: Fong_compare1}
\end{sidewaysfigure}

\section{Another Gibbs Sampler with Complexity $O(np)$}\label{sec_appendix:MCMC_np}
A referee suggested a one-site Gibbs sampler which can be implemented with our continuous Spike-and-Slab LASSO prior and which can potentially improve  computational efficiency. The  sampler is reminiscent of   \cite{geweke1996variable} who proposed a variant for point mass spikes to avoid getting stuck after generating $\beta_i=0$. However, although   \cite{george1997approaches} mentioned that this implementation might also be useful under continuous spike-and-slab priors (where the variance ratio of slab to spike is large), we find this implementation to be rarely used for continuous priors. One-site Gibbs samplers are expected to converge slower and yield higher autocorrelation, which is one of the  reasons  why block-Gibbs samplers have been  preferred. We nevertheless explore this Gibbs sampler further. 

The algorithm to implement \cite{geweke1996variable}'s Gibbs sampler for the Spike-and-Slab LASSO prior is outlined in Algorithm \ref{alg:Gibbs2} and we further refer to it as Gibbs II. We calculate the vector $\bm r_j$ through $\bm r_j=(\y-\X\b) +\X_j\beta_j$ and $\y-\X\b$ is calculated once before running MCMC and updated after each update on $\beta_j$. When $p>n$, Algorithm \ref{alg:Gibbs2} has complexity $O(np)$ compared to $O(p^3)$ for SSVS1 and $O(n^2p)$ for SSVS2. However, generating from the Truncated Normal distribution typically requires an accept-reject algorithm  \citep{geweke1991efficient} which creates some computational bottlenecks. We utilize the R package `truncnorm' \citep{trautmann2015package} which implements this accept-reject algorithm in C. The performance of Algorithm \ref{alg:Gibbs2}  against other methods is summarized in Table \ref{table:Gibbs2}. We see that the per-iteration running time of Gibbs II is similar to SSVS2 (using the trick in \cite{bhattacharya2016fast}), yet the serial correlation is higher than for both implementations of SSVS. Thus, contrary to the point-mass spike, Algorithm \ref{alg:Gibbs2} (which is based on \cite{geweke1996variable}) trades  computational efficiency for serial correlation.  In general, the time used for generating the same number of effective samples is slightly longer than SSVS2, but improves significantly over SSVS1. From this observation, we believe Gibbs II is yet another useful method to speed up the standard SSVS1 under the Spike-and-Slab LASSO prior.

\begin{table}
	\centering
	\scalebox{0.7}{
		\begin{tabular}{l|l|l|l|l} 
			\hline\hline
			method                  & time per 100 iter (sec)  & ESS per 100 iter & lag-one auto correlation on $\beta_j$'s &  time per 100 ESS (sec)\\ 
			\hline
			BB-SSL                  & 0.7                  & 100        &  -0.002   &  0.7\\ 
			\hline
			SSVS1   & 40             &98    &     0.006     & 40.82      \\ 
			\hline
			SSVS2   & 3.2          &98        &0.006      & 3.27         \\ 
			\hline
			Gibbs II & 3.1       & 83        &0.100     & 3.73       \\
			\hline\hline
	\end{tabular}}
	\caption{\small Performance comparison of Algorithm \ref{alg:Gibbs2} (referred to as Gibbs II) against other methods. We set $n=100, p=1000$, variables are equi-correlated with $\rho=0.9$, and the active coefficients are $\beta_{active}=(2,4,-4,6)'$.}
	\label{table:Gibbs2}
\end{table}

A less relevant, but perhaps interesting, observation is that since the Spike-and-Slab LASSO prior is a mixture of Laplace distributions, SSVS1 (which augments the data with $\tau_j$'s) can be seen as a generalization of the MCMC algorithm in \cite{park2008bayesian} while Gibbs II can be seen as a spike-and-slab version of the MCMC algorithm in \cite{hans2009bayesian}. Our conclusion from this exercise is that, for the Spike-and-Slab LASSO prior,  BB-SSL is actually doing better (both  in terms of serial correlation and timing) compared with Gibbs II.

{\scriptsize
	\begin{algorithm}[t]\small
		%		\KwData{ Data ($Y_i$, $\x_i$) for $1\leq i\leq n, \x_i\in \R ^p$}
		%		\KwResult{$({\bm\beta}^1,\dots,\b^T)$}
		\spacingset{1.1}
		\textbf{Set}: $\lambda_0\gg\lambda_1,$  $a,b>0$, $T$ (number of MCMC iterations), $B$ (number of samples to discard as burn-in).
		\\
		\textbf{Initialize}: $\bm \beta^0$ (e.g. LASSO solution after 10-fold cross validation) and $\bm\tau^0$.
		\\
		\For{$t=1,2,\cdots, T$}{
			%{\vspace{-0.3cm} \hspace{5cm}\color{white} \tiny ahoj}\\
			\For{$j=1,2,\cdots, p$}{
				\textbf{(a)} Sample $\gamma_{j}\sim\pi\left(\gamma_{j}\vert\theta,\b_{-j},\bg_{-j},\y\right)=\text{Bernoulli}\left(\frac{\pi_{1}}{\pi_{1}+\pi_{0}}\right)$, where 
				\begin{align*}
				&\pi_{1}=\theta c_{1},\quad\pi_{0}=(1-\theta)c_{0}, 
				\\&c_{1}=\lambda_{1}\left[
				e^{(\mu^1_+)^2/(2s_j)}\mathbb{P}\left(N(\mu_{+}^{1},s_{j})\ge0\right)
				+e^{(\mu^1_-)^2/(2s_j)}\mathbb{P}\left(N(\mu_{-}^{1},s_{j})\le0\right)\right],
				\\
				&c_{0}=\lambda_{0}\left[e^{(\mu^0_+)^2/(2s_j)}
				\mathbb{P}\left(N(\mu_{+}^{0},s_{j})\ge0\right)+e^{(\mu^0_-)^2/(2s_j)}\mathbb{P}\left(N(\mu_{-}^{0},s_{j})\le0\right)\right],
				\\
				&\mu_{-}^{1}=\frac{\bm r_{j}^{T}\X_{j}+\sigma^{2}\lambda_{1}}{\|\X_{j}\|_2^2},
				\quad
				\mu_{+}^{1}=\frac{\bm r_{j}^{T}\X_{j}-\sigma^{2}\lambda_{1}}{\|\X_{j}\|_2^2},
				\quad
				\bm r_{j}=\y-\sum_{k\neq j}\X_{k}\beta_{k},
				\\&\mu_{-}^{0}=\frac{\bm r_{j}^{T}\X_{j}+\sigma^{2}\lambda_{0}}{\|\X_{j}\|_2^2},
				\quad
				\mu_{+}^{0}=\frac{\bm r_{j}^{T}\X_{j}-\sigma^{2}\lambda_{0}}{\|\X_{j}\|_2^2},
				\quad
				s_{j}=\sigma^{2}\|\X_{j}\|_2^{-2}.
				\end{align*}
				\textbf{(b)} Sample $\beta_{j}\sim\pi(\beta_{j}\vert\gamma_{j},\theta,\b_{-j},\bm\gamma_{-j},\y)=u_{0}\text{TruncNormal}^{-}\left(\mu_{-}^{\gamma_{j}},s_{j}\right)+(1-u_{0})\text{TruncNormal}^{+}\left(\mu_{+}^{\gamma_{j}},s_{j}\right)$ where
				\begin{align*}
				&u_{0}=\frac{\mathbb{P}\left(N(\mu_{-}^{\gamma_{j}},s_{j})\le 0\right)}{\mathbb{P}\left(N(\mu_{-}^{\gamma_{j}},s_{j})\le 0\right)+\mathbb{P}\left(N(\mu_{+}^{\gamma_{j}},s_{j})\ge 0\right)\exp\left\{ -\frac{2\bm r_{j}^{T}\X_{j}\lambda_{\gamma_{j}}}{\|\X_{j}\|^2_2}\right\} },
				\\&
				\text{TruncNormal}^-(\mu,s)=\P\left(X\mid X\leq 0\right)\quad
				\text{where}\quad
				X\sim N(\mu,s),
				\\&
				\text{TruncNormal}^+(\mu,s)=\P\left(X\mid X\geq 0\right)\quad
				\text{where}\quad
				X\sim N(\mu,s).
				\end{align*}
			}
			\textbf{(c)} Sample $\theta\sim\mathrm{Beta}(\sum_{j=1}^{p} \gamma_j+a,p-\sum_{j=1}^{p} \gamma_j+b)$.\\	
		}
		\textbf{Return}: $\b^t,\bm\gamma^t,\theta^t$ where $t=B+1,B+2,\cdots,T$.
		\caption{\bf Gibbs II}\label{alg:Gibbs2}
	\end{algorithm}
}

\clearpage
\section{Details on Computational Complexity Analysis}\label{sec:appendix_computation}
Below are details for the computational complexity analysis as shown in Table \ref{table: computation} of the main text.

\paragraph{BB-SSL, WBB1, WBB2.} BB-SSL uses R package \texttt{SSLASSO} (\cite{rockova2019package}) to implement the coordinate-descent algorithm in \cite{rovckova2018spike}. It iteratively updates $\bm \beta$ until convergence. %In practice, it divides all coordinates into three classes: active coordinates, inactive coordinates and candidate coordinates. Then updating proceeds in three loops: in the inner loop, iteratively updates those active coordinates, if convergence has been reached, break the inner loop; in the intermediate loop, update all candidate coordinates, if some candidate coordinate enters eligible set, iteratively update all candidate coordinates; in the outer loop, update all inactive coordinates, if some inactive coordinate enters the eligible set, iteratively update all inactive coordinates. 
At each iteration, we first update the active coordinates, then the candidate coordinates, and finally the inactive coordinates. The total number of iterations is limited to a pre-defined number. There are two ways to update each coordinate, one way is to keep track of a residual vector and for each coordinate we compute the inner product between the residual vector and $\X_j$ -- this takes $O(n)$; another way is to pre-compute the Gram matrix and for each coordinate, we calculate the inner product between $\X^T\X_j$ and $\hat\b$ -- this takes $O(p)$. For both ways, we update $\theta$ every $c$ iterations where each update is $O(p)$. So for a single value of $\lambda_{0}$, SSLASSO is
$O(\min(\text{maxiter}\times p(n+\frac{p}{c_1}),\, (n+\text{maxiter})\times p^2))$. For a sequence of $\lambda_{0}$'s, complexity is $O(L\times\min(\text{maxiter}\times p(n+\frac{p}{c_1}),\, (n+\text{maxiter})\times p^2))$ where $L$ is the length of $\lambda_{0}$'s. 
Usually if the biggest $\lambda_{0}$ is large enough, we would expect that the larger $\lambda_{0}$'s can reach convergence quickly using the estimated $\b$ from previous $\lambda_{0}$'s, so usually it takes less than $O\left(L\times\min\left(\text{maxiter}\times p(n+\frac{p}{c_1}),\, (n+\text{maxiter})\times p^2\right)\right)$.

\iffalse
When choosing $\lambda_{0}$, we suggest plotting the regularization path from SSL and choose a large enough $\lambda_{0}$ after which further increasing $\lambda_{0}$ does not affect the chosen model. This is also reasonable from the goal of variable selection. %Under this case, we find that it usually takes

the larger $\lambda_{0}$'s can reach convergence quickly using the estimated $\beta$ from previous $\lambda_{0}$'s, so usually it takes less than $O(L\times\min(\text{maxiter}\times p(n+\frac{p}{c_1}),\, (n+\text{maxiter})\times p^2))$.

Empirically, when $\lambda_{0}$ is small and does not provide enough shrinkage, SSL takes more iterations to converge. When choosing $\lambda_{0}$, we suggest plotting the regularization path from SSL and choose a large enough $\lambda_{0}$ after which further increasing $\lambda_{0}$ does not affect the chosen model. This is also reasonable from the goal of variable selection. When $\lambda_{0}$ is extremely small, a single value fitted SSL might have difficulty converging so we have to use a sequence of $\lambda_{0}$'s to burn-in. It usually takes much longer time than a single value fitted SSL. Note that the computational complexity of other methods - SSVS1, SSVS2 and Skinny Gibb - does not depend on $\lambda_{0}$. 
\fi

\paragraph{SSVS1.} The computational complexity for Algorithm \ref{SSVS} is $O(p^3)$ per iteration when $p>n$. Under this setting, we use the Woodbury matrix identity to simplify the matrix multiplication $(\bm X^T\bm X+\bm D_{\tau}^{-1})^{-1}=\bm D_{\tau}-\bm D_{\tau}\bm X^T(\bm{I}_n+\bm X\bm D_{\tau}\bm X^T)^{-1}\bm X\bm D_{\tau}$, whose complexity is $O(p^2n)$. %When $p\le n$, we use the usual way to invert matrices and the computational complexity is also $O(p^2n)$.
%Another computational bottleneck is to generate a $p$-dimensional multivariate normal vector $\bm \beta$ because it involves factorizing the $p$-by-$p$ covariance matrix, which is of complexity $O(p^3)$. In practice we use R package ``mvnfast" (\cite{fasiolo2016introduction}, written in C++) to generate multivatiate normal vector, which is much faster than the typically used function ``mvrnorm". For example, when $n=100,p=2000$, it is almost 10 times faster than mvrnorm.

\paragraph{SSVS2.} Using \cite{bhattacharya2016fast}'s matrix inversion formula to generate the $p$-dimensional multivatiate Gaussian takes $O(n^2p+n^3)=O(n^2\max(n,p))$.

\paragraph{Skinny Gibbs.} We modified Skinny Gibbs to sample from the posterior using SSL prior. Theoretically it is of complexity $O(np)$. However, sometimes when $n$ is relatively small, we observe that the running time of Skinny Gibbs is slower than Bhattacharya's method. This is because Skinny Gibbs involves an $O(n)$ matrix product for each coordinate and the update for each coordinate is implemented via for-loop, whereas in Bhattacharya's method the $O(n^2p)$ operation is one matrix product which is very efficiently optimized in R. We will see that as $n$ increases, this problem diminishes and Skinny Gibbs becomes faster than Bhattacharya's method.

\paragraph{WLB.} Generating each WLB sample involves solving a least square problem whose complexity is $O(p^2n)$ when $p\leq n$. WLB is not applicable when $p>n$.

%\paragraph{VB} For VB, we first iteratively update variational parameters, and each iterations is $O(np^2)$. After fitting the parameters, we sample from the variational posterior, where each iteration only involves sampling from a $p$-dimensional multivariate Gaussian. If we  pre-compute the Cholesky decomposition of the covariance matrix,  each iteration is $O(p^2)$.

\section{Additional Experimental Results}\label{sec_appendix:simulation}
\subsection{Low Dimensional Setting}\label{sec:simu_low_ind}
When all covariates are mutually independent, we first investigate the marginal density of $\beta_i$'s, which is shown in Figure \ref{fig:low_ind_beta}. We find that all methods perform well for active $\beta_i$'s. WBB1 and WBB2 are doing poorly for inactive $\beta_i$'s.
For the marginal mean of $\gamma_i$'s, as shown in Figure \ref{fig:low_ind_gamma}, all methods perform  well. 
All methods can detect over 95\% of models. 
In Table \ref{table:approx_low}, we quantify the performance of each method in the low-dimensional setting. In the settings we tried, BB-SSL consistently has the best performance.

\begin{sidewaysfigure}
	\begin{subfigure}{0.5\hsize}\centering
		\includegraphics[width=\hsize, height=5in]{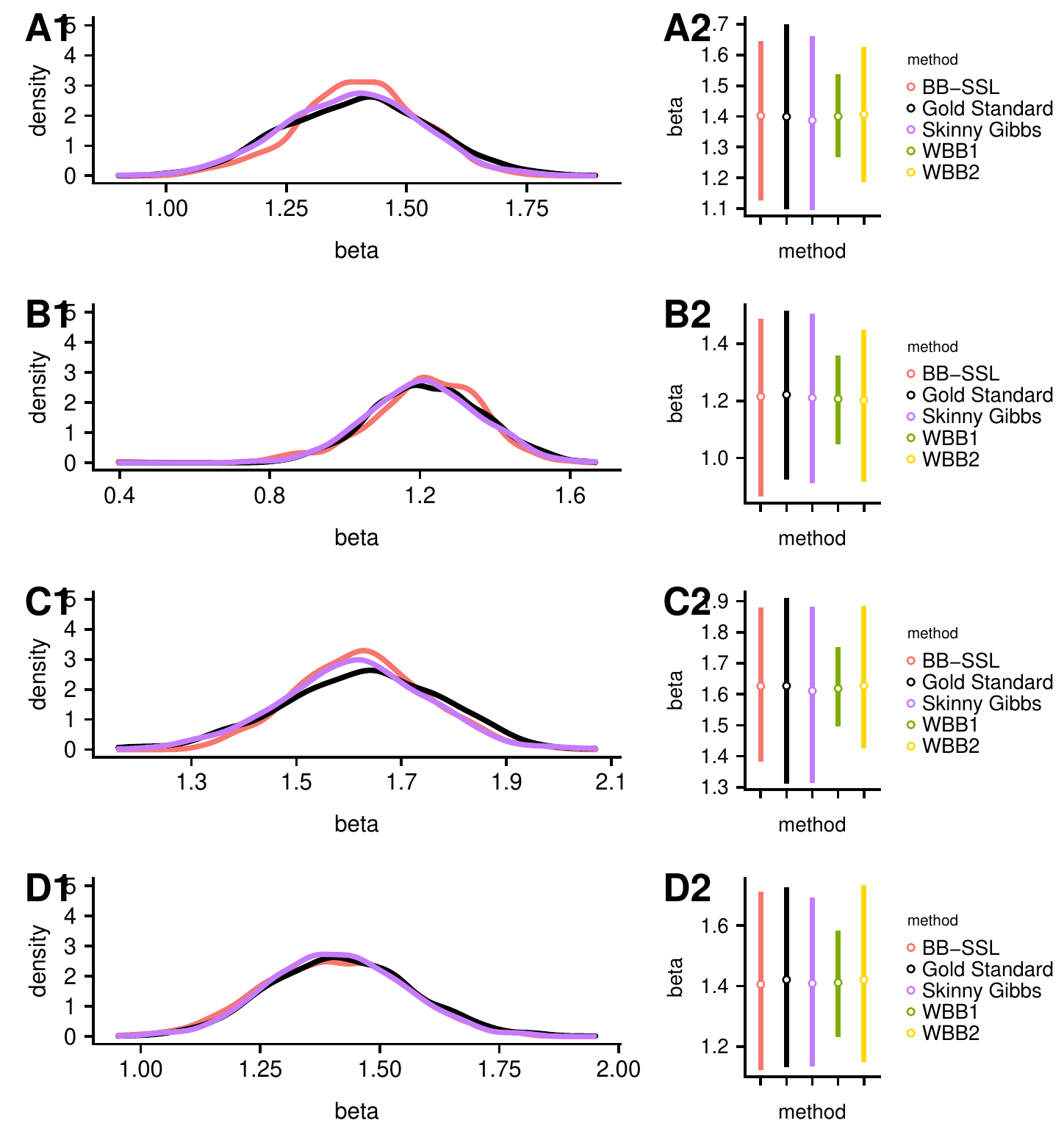}
		\caption{\scriptsize Active predictors, from top to bottom: $\beta_1,\beta_{4},\beta_{7},\beta_{10}$}
	\end{subfigure}%
	%\hfill <-- it is superfluous 
	\begin{subfigure}{0.5\hsize}\centering
		\includegraphics[width=\hsize, height=5in]{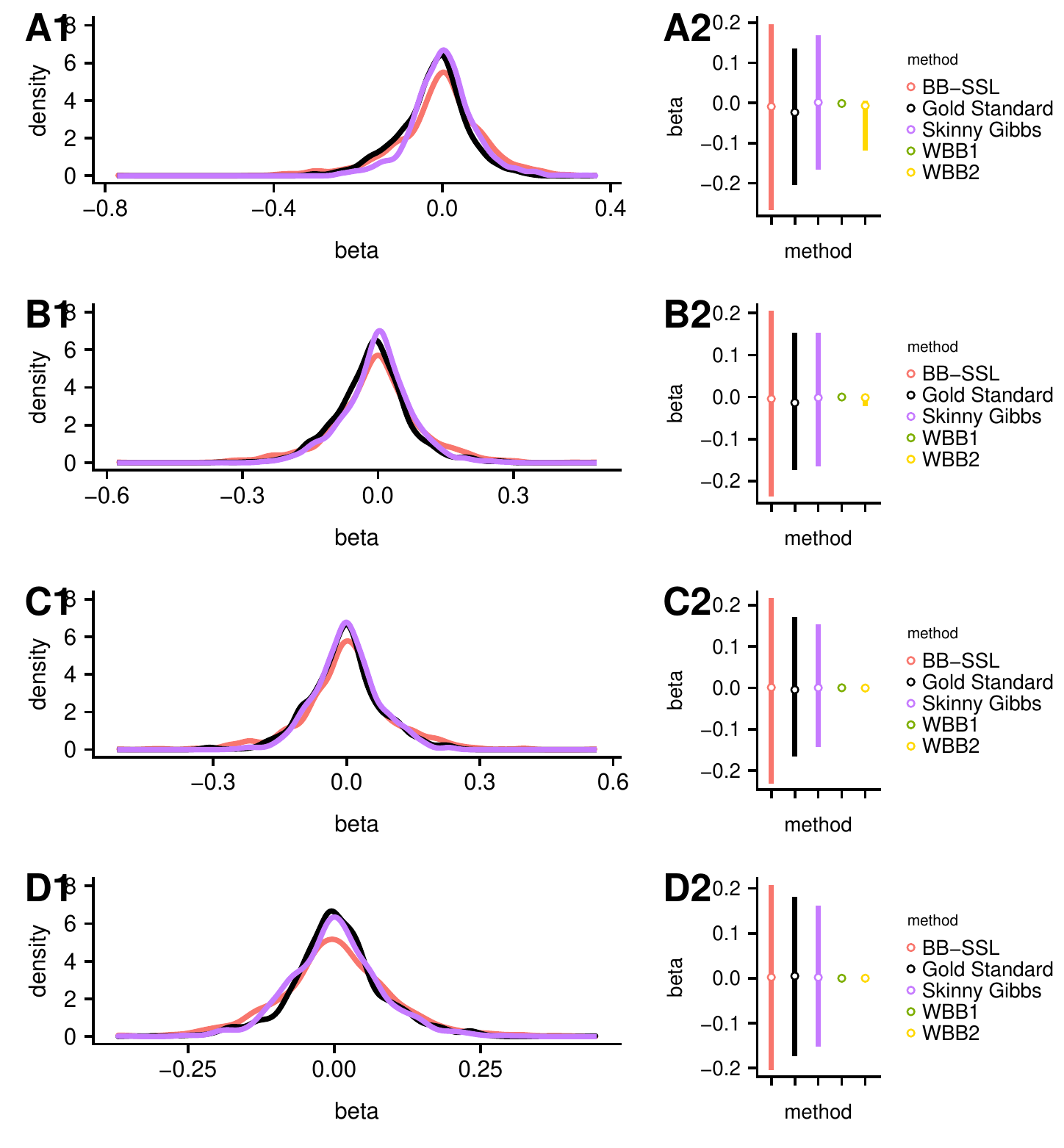}
		\caption{\scriptsize Inactive predictors, from top to bottom: $\beta_2,\beta_{5},\beta_{8},\beta_{11}$}
	\end{subfigure}
	\caption{\small Estimated posterior density (left panel) and credible intervals (right panel) of $\beta_i$'s in the low-dimensional independent case. We have $n=50, p=12, \beta_{active}=(1.3,1.3,1.3,1.3)',\lambda_0=7,\lambda_1=0.15, \rho=0$. Each method has %\footnote{do we need thinning for BB-SSL?} 
		$5\,000$ sample points (after thinning for SSVS and Skinny Gibbs). BB-SSL is fitted using a single value $\lambda_0=7$. Since WBB1 and WBB2 produce a point mass at zero, we exclude them from density comparisons. }
	\label{fig:low_ind_beta}
\end{sidewaysfigure}

\begin{figure}
	\begin{center}
		\includegraphics[width = .7\textwidth]{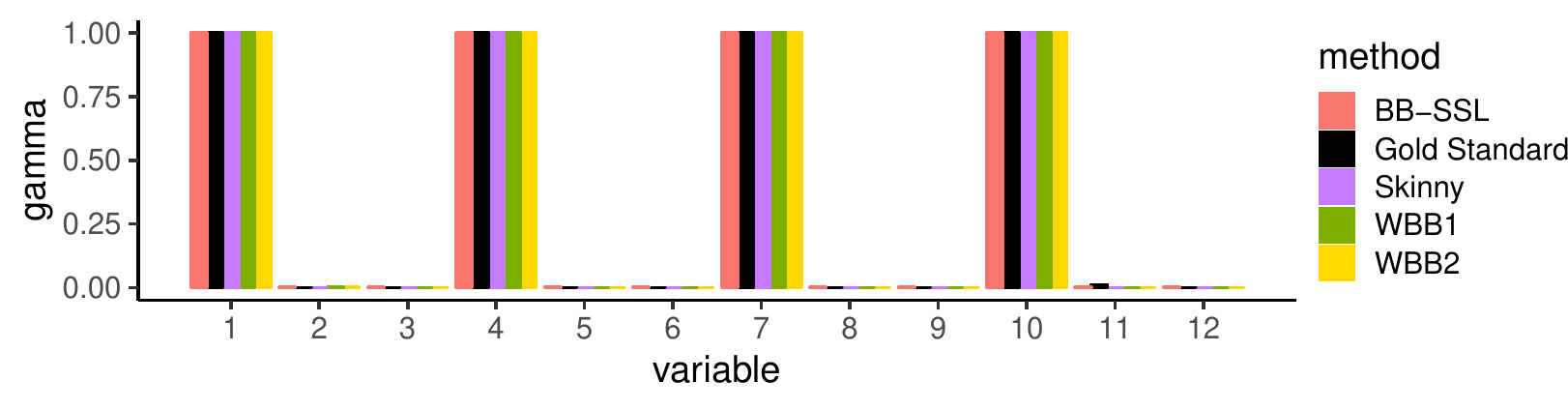}
	\end{center}
	\caption{\small Mean of $\gamma_i, i = 1,2,\cdots,12$ in low-dimensional, independent case. $n=50, p=12, \beta_{active}=(1.3,1.3,1.3,1.3)$ and predictors mutually independent. Each method is (thinned to) 1,000 sample points. $\lambda_0=13,\lambda_1=0.05$. \texttt{SSLASSO} is fitted using a single $\lambda_0$.  }
	\label{fig:low_ind_gamma}
\end{figure}

\begin{table}
	\begin{subtable}[!t]{\textwidth}
		\centering
		\scalebox{0.58}{
			\begin{tabular}{l|l|l|l|l|l|l|l|l|l||l|l|l|l|l|l|l|l|l} 
				\hline\hline
				\bf \large Setting&\multicolumn{9}{l||}{  \cellcolor[gray]{0.8}\large \bf Block-wise $\rho=0.6$, $\beta_{active}=(1.3,1.3,1.3,1.3)'$}   &\multicolumn{9}{l}{  \large\cellcolor[gray]{0.8} \bf Block-wise $\rho=0.9$, $\beta_{active}=(1.3,1.3,1.3,1.3)'$}   \\\hline
				\multirow{3}{*}{Metric}& \multicolumn{6}{l|}{\quad\quad\quad\quad\quad\quad\quad\quad$\beta_j$'s}                                                                                         & \multicolumn{2}{l|}{\quad$\gamma_j$'s }              & Model 
				& \multicolumn{6}{l|}{\quad\quad\quad\quad\quad\quad\quad\quad$\beta_j$'s}                                                                                         & \multicolumn{2}{l|}{\quad$\gamma_j$'s }              & Model 
				\\ \cline{2-19}
				Metric& \multicolumn{2}{l|}{KL} & \multicolumn{2}{l|}{JD of 90\% CI} & \multicolumn{2}{l|}{`Bias'} & \multicolumn{2}{l|}{`Bias'}    & HD &
				\multicolumn{2}{l|}{KL } & \multicolumn{2}{l|}{JD of 90\% CI} & \multicolumn{2}{l|}{`Bias'} & \multicolumn{2}{l|}{`Bias'}    & HD 
				\\ \cline{2-19}
				& $+$ & $-$           & $+$ & $-$         & $+$ & $-$       & $+$ & $-$        & all    & $+$ & $-$   & $+$ & $-$       & $+$ & $-$      & $+$ & $-$         & all    \\ 
				\hline
				Skinny Gibbs & 0.24  & 0.25                & 0.29   & 0.34                               & \bf{0.08} &0.07       & \bf{0.02}&0.0007& \bf{0} 
				& 0.68  & 0.33                     & 0.37   & 0.39                                  & 0.15 &0.08       & \bf{0.04}&0.009& \bf{0.2}   \\ 
				WBB1        & 0.22  & 1.66        &0.28& 0.62                                   & 0.09 & 0.04    & \bf{0.02}&0.0004& \bf{0}  
				& 0.67  & 2.73                 & 0.41 & 0.89                                       & 0.14 & \bf{0.08}    & 0.05&0.003& \bf{0.2}     \\ 
				WBB2        & 0.20 & 1.69       &0.28 & 0.62         & \bf{0.08} & {0.04}        &\bf{0.02}&0.0004& \bf{0}
				& 0.56 & 2.69       &0.33 & 0.87          & 0.14 & \bf{0.07}        &0.05&0.003& \bf{0.2}  \\ 
				BB-SSL      	& \bf{0.15}   & \bf{0.07}                     & \bf{0.21}   & \bf{0.24}                                    &{0.09}  & \bf{0.03}            &\bf{0.02}&\bf{0.0003}& \bf{0}   
				& \bf{0.14}   & \bf{0.09}                     & \bf{0.21}   & \bf{0.24}                                    &\bf{0.13}  & \bf{0.07}            &\bf{0.04}&\bf{0.002}& \bf{0.2}   
				\\
				\hline\hline
				\bf \large Setting&\multicolumn{9}{l||}{  \cellcolor[gray]{0.8} \bf \large Equi-correlation $\rho=0.6$, $\beta_{active}=(1.3,1.3,1.3,1.3)'$}   &\multicolumn{9}{l}{  \large \cellcolor[gray]{0.8} \bf Equi-correlation $\rho=0.9$, $\beta_{active}=(1,1.5,-1.5,2)'$}   \\\hline
				\multirow{3}{*}{Metric}& \multicolumn{6}{l|}{\quad\quad\quad\quad\quad\quad\quad\quad$\beta_j$'s}                                                                                         & \multicolumn{2}{l|}{\quad$\gamma_j$'s }              & Model 
				& \multicolumn{6}{l|}{\quad\quad\quad\quad\quad\quad\quad\quad$\beta_j$'s}                                                                                         & \multicolumn{2}{l|}{\quad$\gamma_j$'s }              & Model 
				\\ \cline{2-19}
				Metric& \multicolumn{2}{l|}{KL} & \multicolumn{2}{l|}{JD of 90\% CI} & \multicolumn{2}{l|}{`Bias'} & \multicolumn{2}{l|}{`Bias'}    & HD &
				\multicolumn{2}{l|}{KL } & \multicolumn{2}{l|}{JD of 90\% CI} & \multicolumn{2}{l|}{`Bias'} & \multicolumn{2}{l|}{`Bias'}    & HD 
				\\ \cline{2-19}
				& $+$ & $-$           & $+$ & $-$         & $+$ & $-$       & $+$ & $-$        & all    & $+$ & $-$   & $+$ & $-$       & $+$ & $-$      & $+$ & $-$         & all    \\ 
				\hline
				Skinny Gibbs & 0.11  & 0.20                     & 0.20   & 0.31                                  & \bf{0.06} &0.06       & {0.05}&0.002& \bf{0.2} 
				& 0.66   & 0.29                      & 0.27   & 0.40                     & 0.30 & {0.06}            & 0.01&0.016& 1  \\ 
				WBB1        & 0.15 & 1.97            & 0.22 & 0.75                                       & 0.09 & 0.05    & 0.02&\bf{0.001}& \bf{0.2}  
				& 1.18  &2.72             &0.37  & 0.88                                    & {0.14}  & {0.06}            &\bf{0.01}&\bf{0.001}& \bf{0.2}    \\ 
				WBB2        & 0.14 & 1.98       &0.23 & 0.75          & 0.09 & {0.05}        &0.02&\bf{0.001}& \bf{0.2}
				& 1.20 & \bf{2.70}              & 0.37   & 0.87     & {0.14}  & \bf{0.05}            & \bf{0.01}&\bf{0.001}& \bf{0.2}  \\ 
				BB-SSL      	& \bf{0.10}   & \bf{0.07}                     & \bf{0.19}   & \bf{0.24}                                    &{0.08}  & \bf{0.03}            &\bf{0.02}&{0.002}& \bf{0.2}   
				& \bf{0.23}   & \bf{0.07}                     & {0.51}   & \bf{0.20}                                    & \bf{0.12}  & \bf{0.05}            &\bf{0.01}&\bf{0.001}& \bf{0.2}   
				\\\hline\hline 
		\end{tabular}}
		%\caption{\small $n=100,p=1000,\rho=0.6$ equi-correlation structure with $\beta_{active}=(2,3,-3,4)'$.}
	\end{subtable}
	\caption{ \small Evaluation of approximation properties (relative to SSVS) in the low-dimensional setting with $n=50$ and $p=12$ based on 10 independent runs. We set $\lambda_0=7,\lambda_1=0.15$. The best performance is marked in bold font. %More results on approximation performance in different settings are in the Appendix (Section \ref{sec_appendix:simulation}). 
		KL is the Kullback-Leibler divergence, JD is the Jaccard distance  of credible intervals (CI), HD is the Hamming distance of the median models.  `Bias'   refers to the $l_1$ distance of estimated posterior means. We denote with $*$ all numbers smaller than 0.0001, with $+$ an average over active coordinates, and with $-$ an average over inactive coordinates. }
	\label{table:approx_low}
\end{table}

\subsection{High Dimensional, Block-wise Correlation Structure}\label{sec:appendix_simu_high.99}

Figure  \ref {fig:high_ind_beta} shows the posterior for $\beta_i$'s when $\rho=0$, Figure \ref{fig:high_high_beta} is for $\rho=0.6$ and Figure \ref{fig:high_moderate_beta} is for $\rho=0.9$ in the block-wise correlated setting. Skinny Gibbs slightly underestimates the variance for active coordinates. WBB1 and WBB2 does poorly for inactive coordinates. In general BB-SSL does  well. Figure \ref{fig:high_dim_gamma_6} shows the marginal inclusion probabilities (MIP) when $\rho=0.6$ ($\rho=0,0.9$ are included in the main text) and we see that all methods perform well except that Skinny Gibbs sometimes overestimates MIP.

\begin{sidewaysfigure}
	\begin{subfigure}{0.5\hsize}\centering
		\includegraphics[width=\hsize, height=5in]{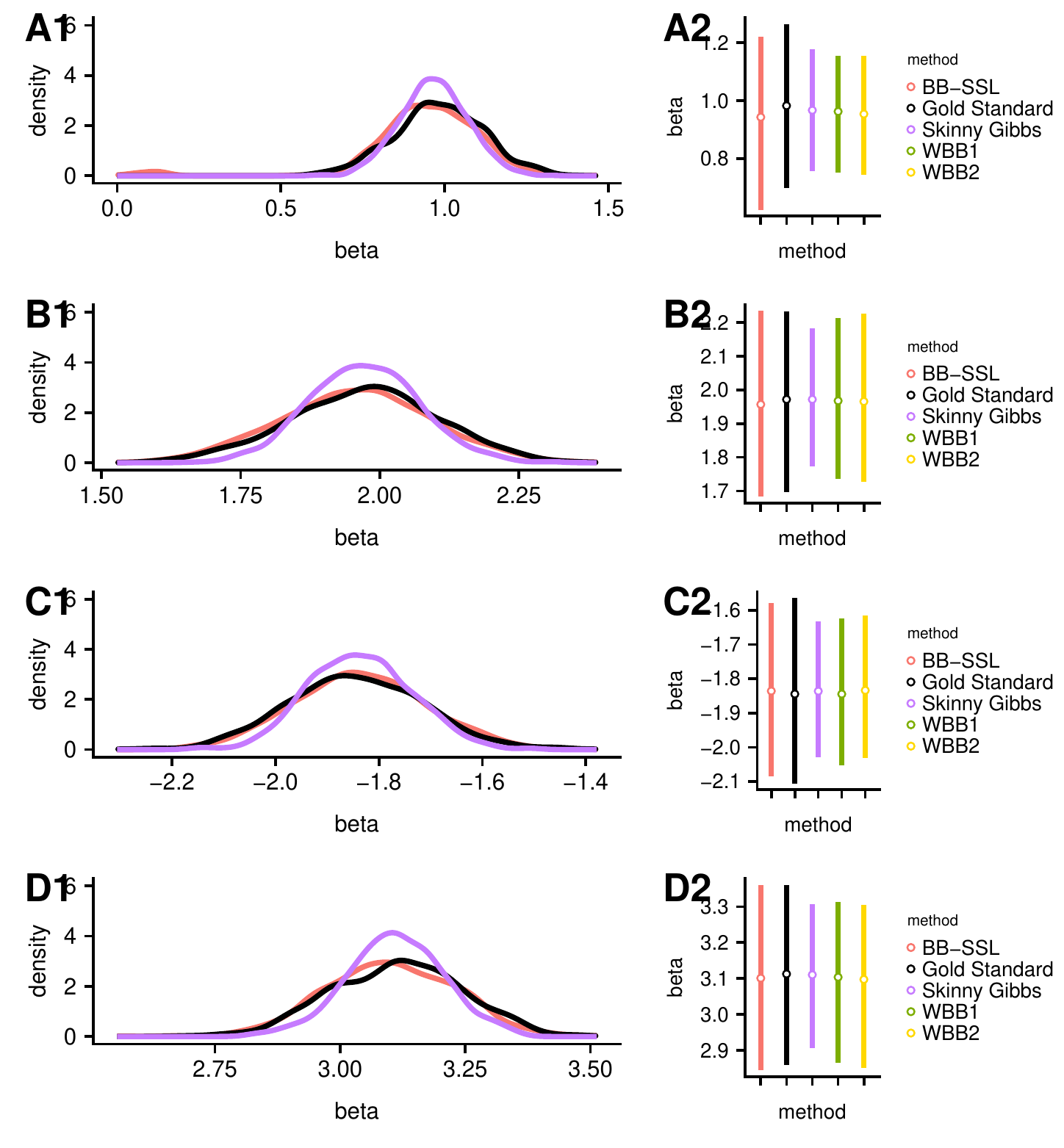}
		\caption{\small Active predictors, from top to bottom: $\beta_1,\beta_{11},\beta_{21},\beta_{31}$}
	\end{subfigure}%
	%\hfill <-- it is superfluous 
	\begin{subfigure}{0.5\hsize}\centering
		\includegraphics[width=\hsize, height=5in]{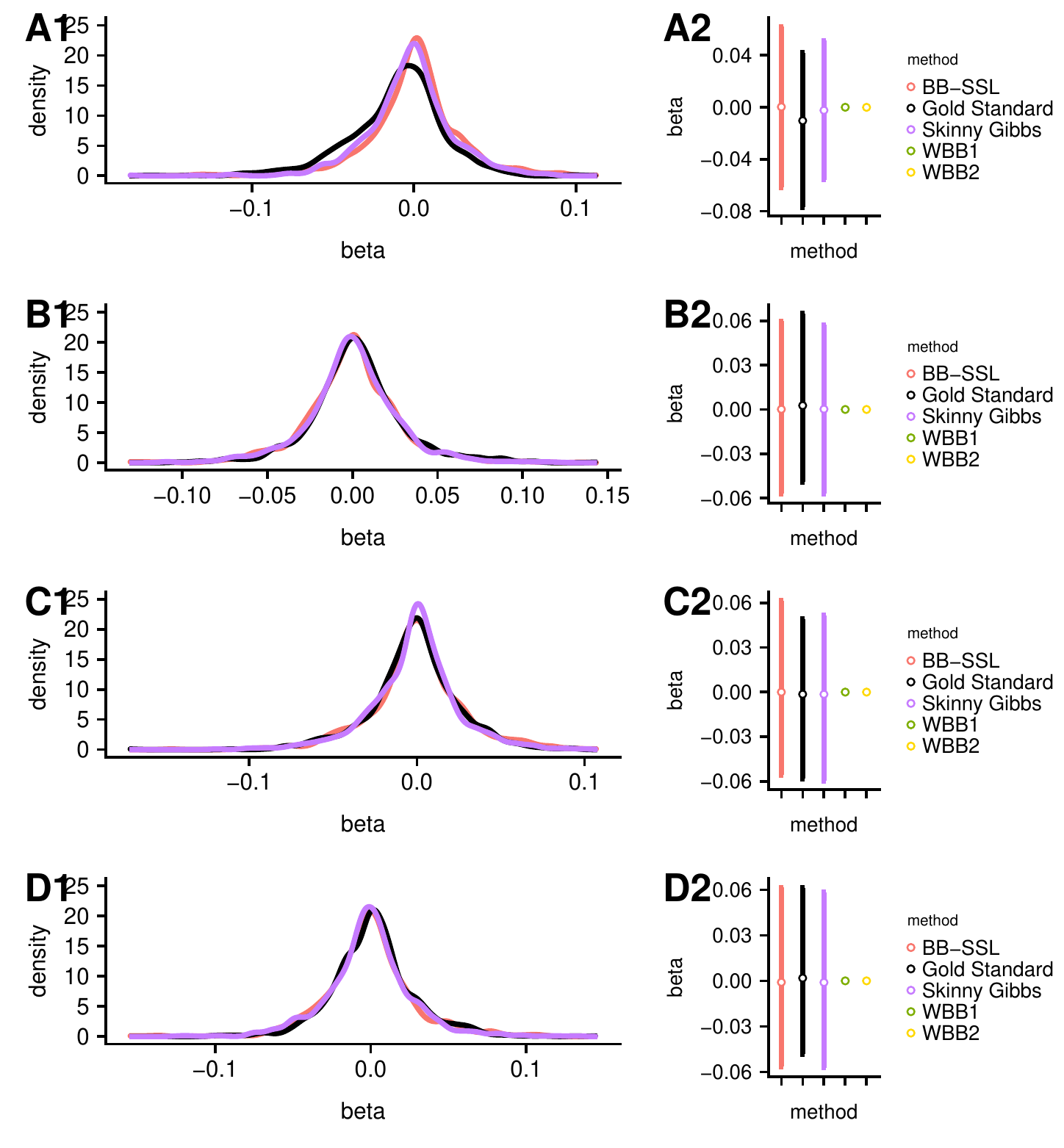}
		\caption{\small Inactive predictors, from top to bottom: $\beta_2,\beta_{12},\beta_{22},\beta_{32}$}
	\end{subfigure}
	\caption{\small Estimated posterior density (left panel) and credible intervals (right panel) of $\beta_i$'s in the high-dimensional independent case ($\rho=0$). We have $n=100, p=1000, \beta_{active}=(1,2,-2,3)',\lambda_0=50,\lambda_1=0.05$. Each method has %\footnote{do we need thinning for BB-SSL?} 
		$5\,000$ sample points (after thinning for SSVS and Skinny Gibbs). BB-SSL is fitted using a single $\lambda_0$ and initialized at \texttt{SSLASSO} solution on the original $\X,\y$. Since WBB1 and WBB2 produce a point mass at zero, we exclude them from density comparisons.}
	\label{fig:high_ind_beta}
\end{sidewaysfigure}

\begin{sidewaysfigure}
	\begin{subfigure}{0.5\hsize}\centering
		\includegraphics[width=0.9\hsize, height=5in]{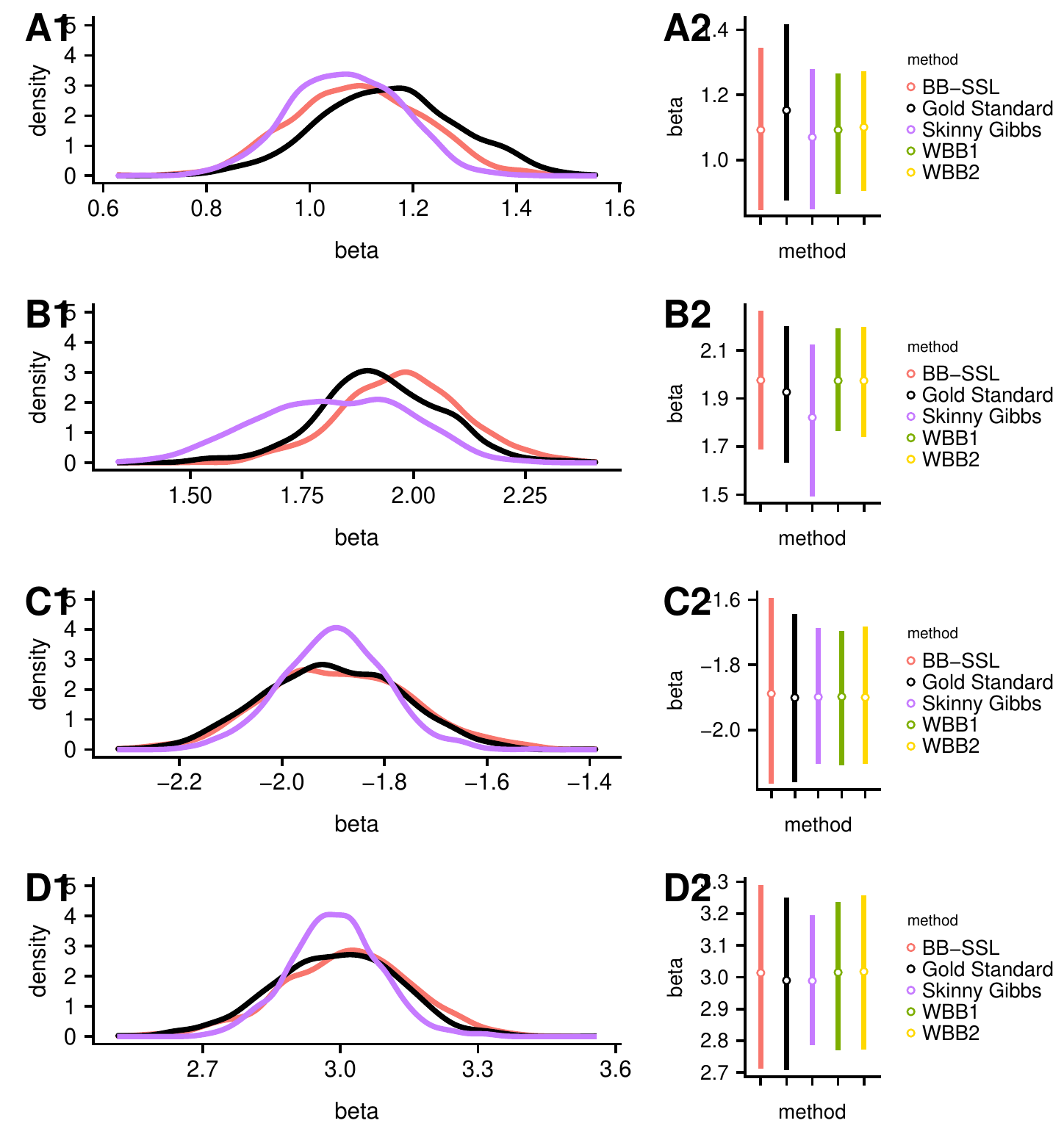}
		\caption{\scriptsize Active predictors, from top to bottom: $\beta_1,\beta_{11},\beta_{21},\beta_{31}$}
	\end{subfigure}%
	%\hfill <-- it is superfluous 
	\begin{subfigure}{0.5\hsize}\centering
		\includegraphics[width=0.9\hsize, height=5in]{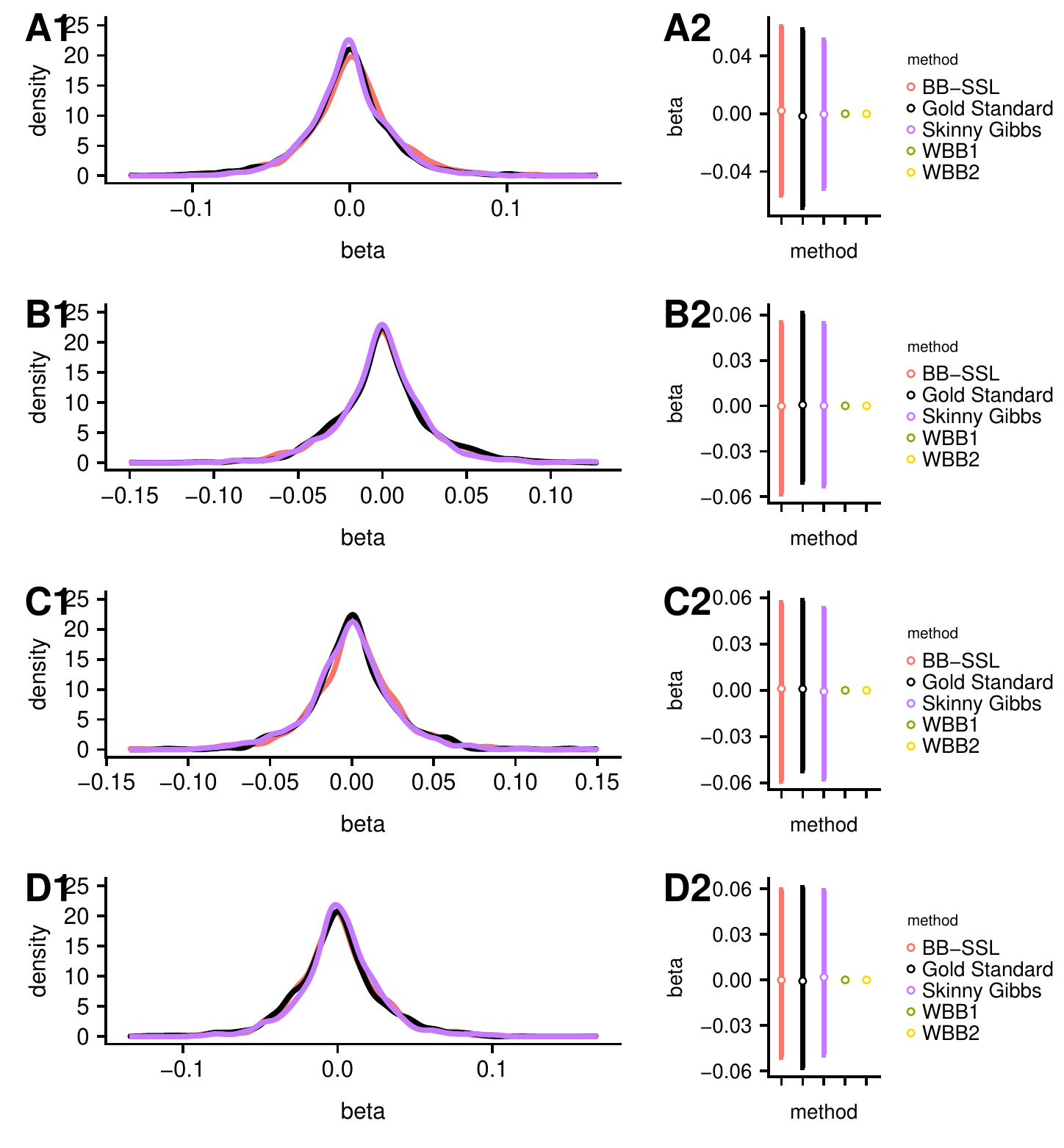}
		\caption{\scriptsize Inactive predictors, from top to bottom: $\beta_2,\beta_{12},\beta_{22},\beta_{32}$}
	\end{subfigure}
	\caption{\small Estimated posterior density (left panel) and credible intervals (right panel) of $\beta_i$'s in the high-dimensional block-wise correlated case ($\rho=0.6$). We have $n=100, p=1000, \beta_{active}=(1,2,-2,3)',\lambda_0=50,\lambda_1=0.05$. Each method has %\footnote{do we need thinning for BB-SSL?} 
		$5\,000$ sample points (after thinning for SSVS and Skinny Gibbs). BB-SSL is fitted using a single $\lambda_0$ and initialized at \texttt{SSLASSO} solution on the original $\X,\y$. Since WBB1 and WBB2 produce a point mass at zero, we exclude them from density comparisons.}
	\label{fig:high_high_beta}
\end{sidewaysfigure}

\begin{sidewaysfigure}
	\begin{subfigure}{0.5\hsize}\centering
		\includegraphics[width=0.9\hsize, height=5in]{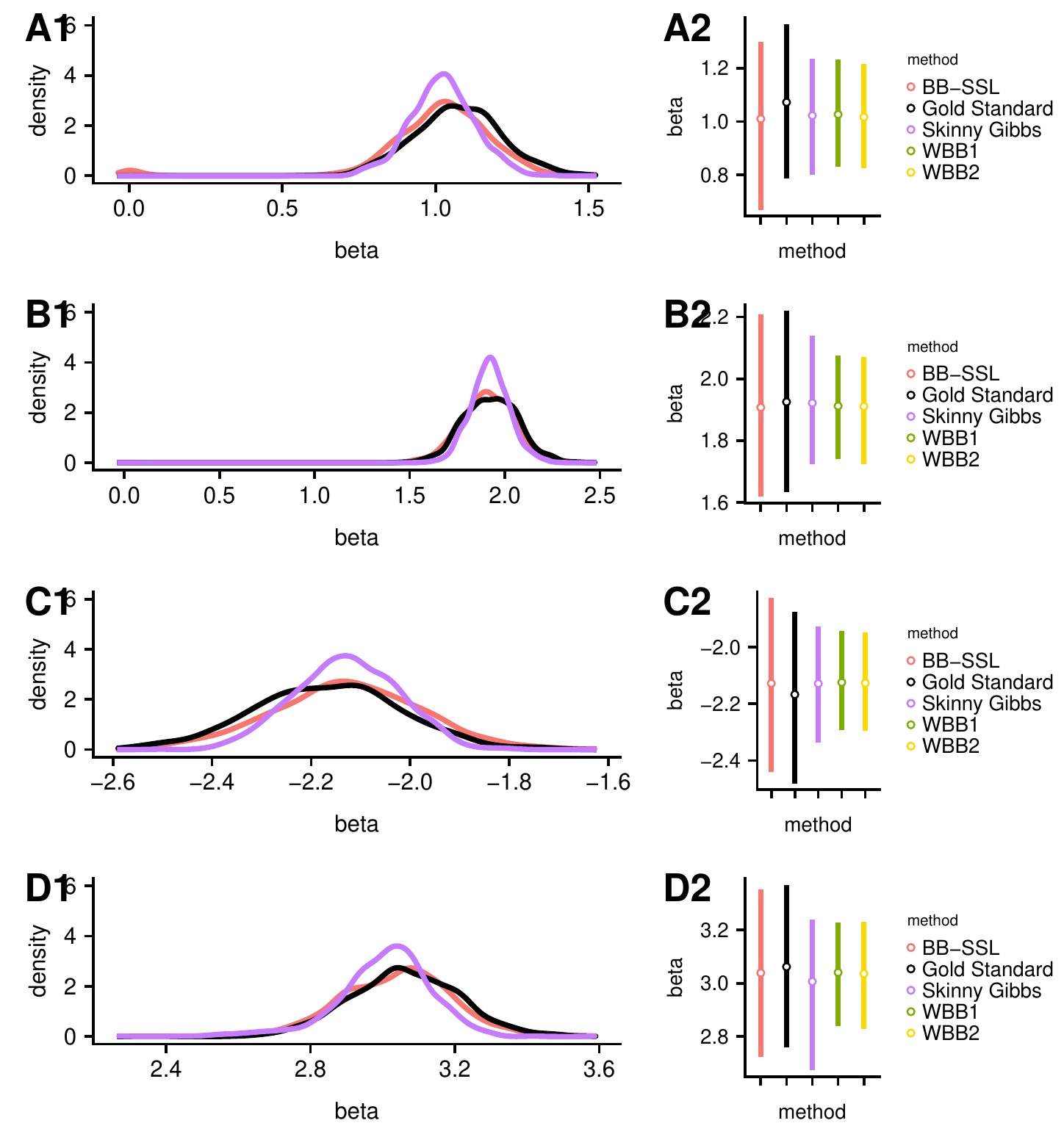}
		\caption{\small Active predictors, from top to bottom: $\beta_1,\beta_{11},\beta_{21},\beta_{31}$}
	\end{subfigure}%
	%\hfill <-- it is superfluous 
	\begin{subfigure}{0.5\hsize}\centering
		\includegraphics[width=0.9\hsize, height=5in]{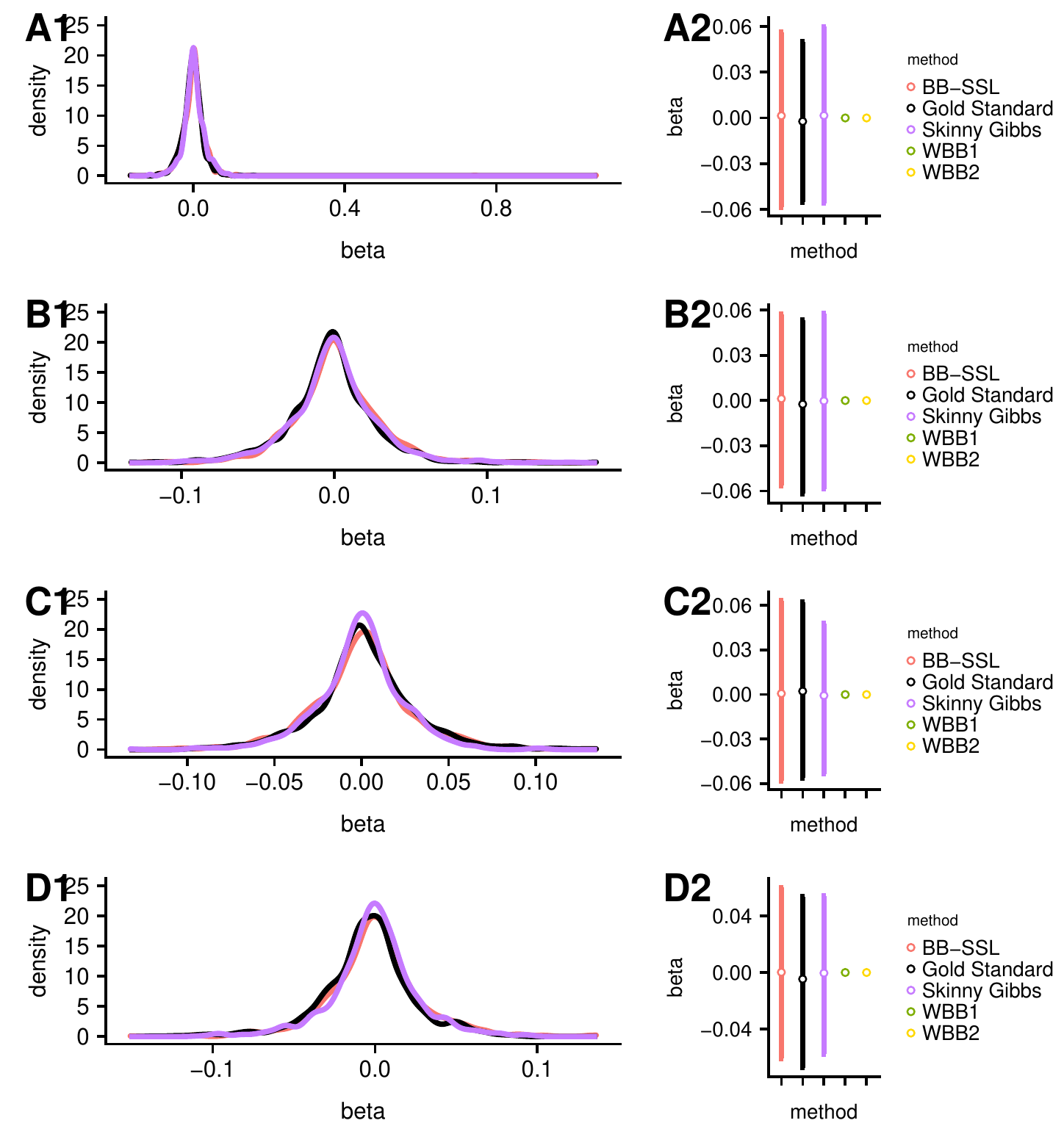}
		\caption{\small Inactive predictors, from top to bottom: $\beta_2,\beta_{12},\beta_{22},\beta_{32}$}
	\end{subfigure}
	\caption{\small  Estimated posterior density (left panel) and credible intervals (right panel) of $\beta_i$'s in the high-dimensional, block-wise correlated case ($\rho=0.9$). We have $n=100, p=1000, \beta_{active}=(1,2,-2,3)',\lambda_0=50,\lambda_1=0.05$. Each method has %\footnote{do we need thinning for BB-SSL?} 
		$5\,000$ sample points (after thinning for SSVS and Skinny Gibbs). BB-SSL is fitted using a single $\lambda_0$ and initialized at \texttt{SSLASSO} solution on the original $\X,\y$. Since WBB1 and WBB2 produce a point mass at zero, we exclude them from density comparisons.}
	\label{fig:high_moderate_beta}
\end{sidewaysfigure}

\begin{figure}
	\centering
	\begin{subfigure}{.4\textwidth}
		\centering
		\includegraphics[width=.85\linewidth,height=1.65in]{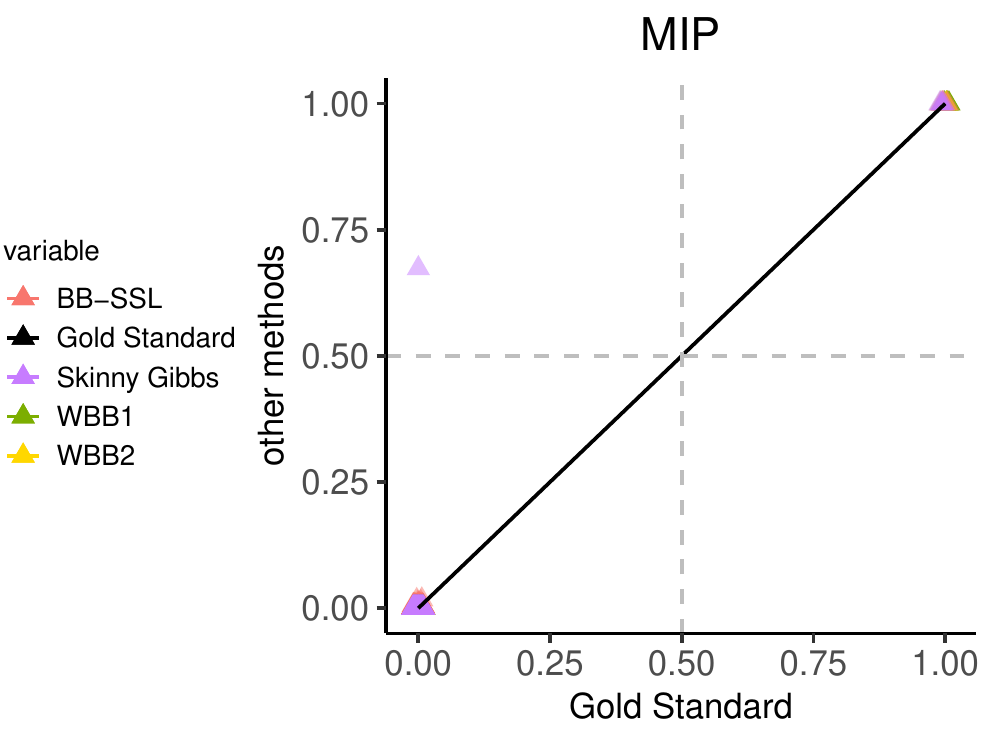}
		\caption{\small Block-wise $\rho=0.6$}
		\label{fig:high_dim_gamma_6}
	\end{subfigure}%
	\begin{subfigure}{.3\textwidth}
		\centering
		\includegraphics[width=.85\linewidth,height=1.65in]{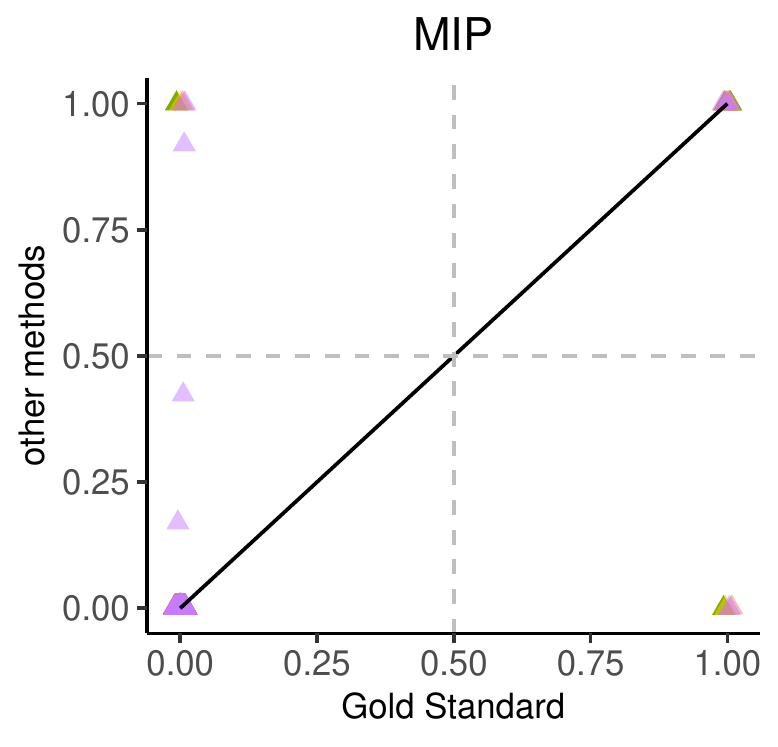}
		\caption{\small Block-wise $\rho=0.99$}
		\label{fig:high_dim_gamma_99}
	\end{subfigure}%
	\begin{subfigure}{.3\textwidth}
		\centering
		\includegraphics[width=.85\linewidth,height=1.65in]{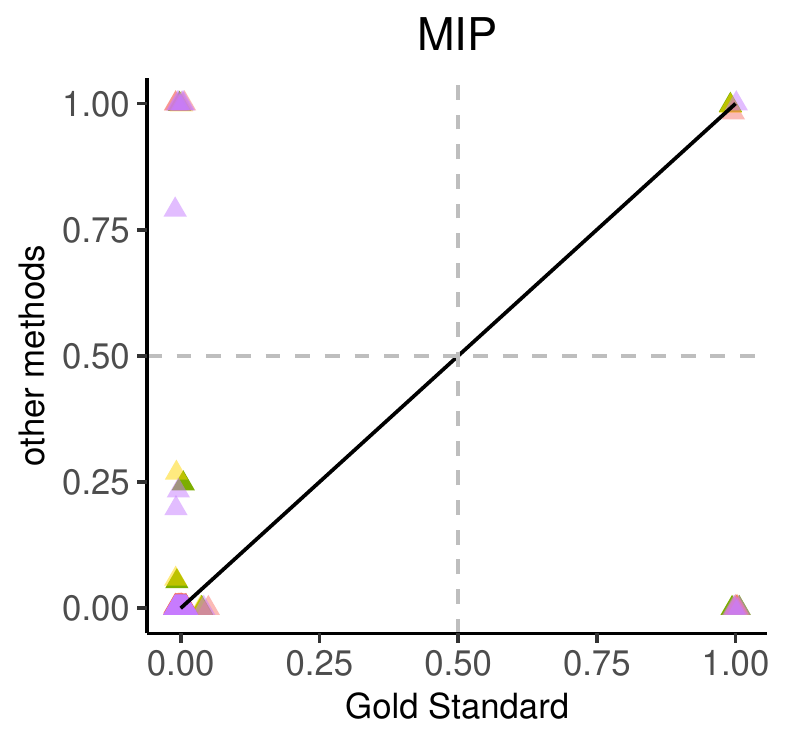}
		\caption{\small Equi-correlation $\rho=0.9$}
		\label{fig:high_dim_gamma_9_all}
	\end{subfigure}%
	\caption{\small Posterior means of $\gamma_i$'s (i.e. a marginal inclusion probabilities) in high-dimensional settings with $n=100, p=1000$. We set $\lambda_0=50,\lambda_1=0.05$. }
\end{figure}

In the more extreme case where $\rho=0.99$, SSVS suffers severely from local entrapment with the signal $\beta_{active}=(1,2,-2,3)'$, so we set a stronger signal with $\beta_{active}=(2,4,-4,6)'$. All the other settings are the same as before. We run two SSVS chains, one initialized at the truth and the other at origin. In Figure \ref{fig:high_dim_cor99_density}, we observe discrepancies in the posterior approximation for the two chains. This is why SSVS may not be a reliable gold standard to make comparisons with. We nevertheless observe that BB-SSL is close to the SSVS approximation obtained by initialization at the truth. Figure \ref{fig:high_dim_gamma_99} shows the MIP and no method performs perfect in this extreme setting. %and even the gold standard (SSVS) falls into a local region. %BB-SSL estimates one mode accurately but sometimes fails to detect the other mode.  %BB-SSL is good at detecting multi-modality but sometimes overshoot (for $\beta_2$ and $\beta_{12}$ it has a false mode at the correlated truth). Skinny Gibbs also falls into local trap as SSVS and has some problems with the variance. WBB1 and WBB2 does not work for inactive coordinates.

\begin{sidewaysfigure}
	\begin{subfigure}{0.5\hsize}\centering
		\includegraphics[width=0.9\hsize, height=5in]{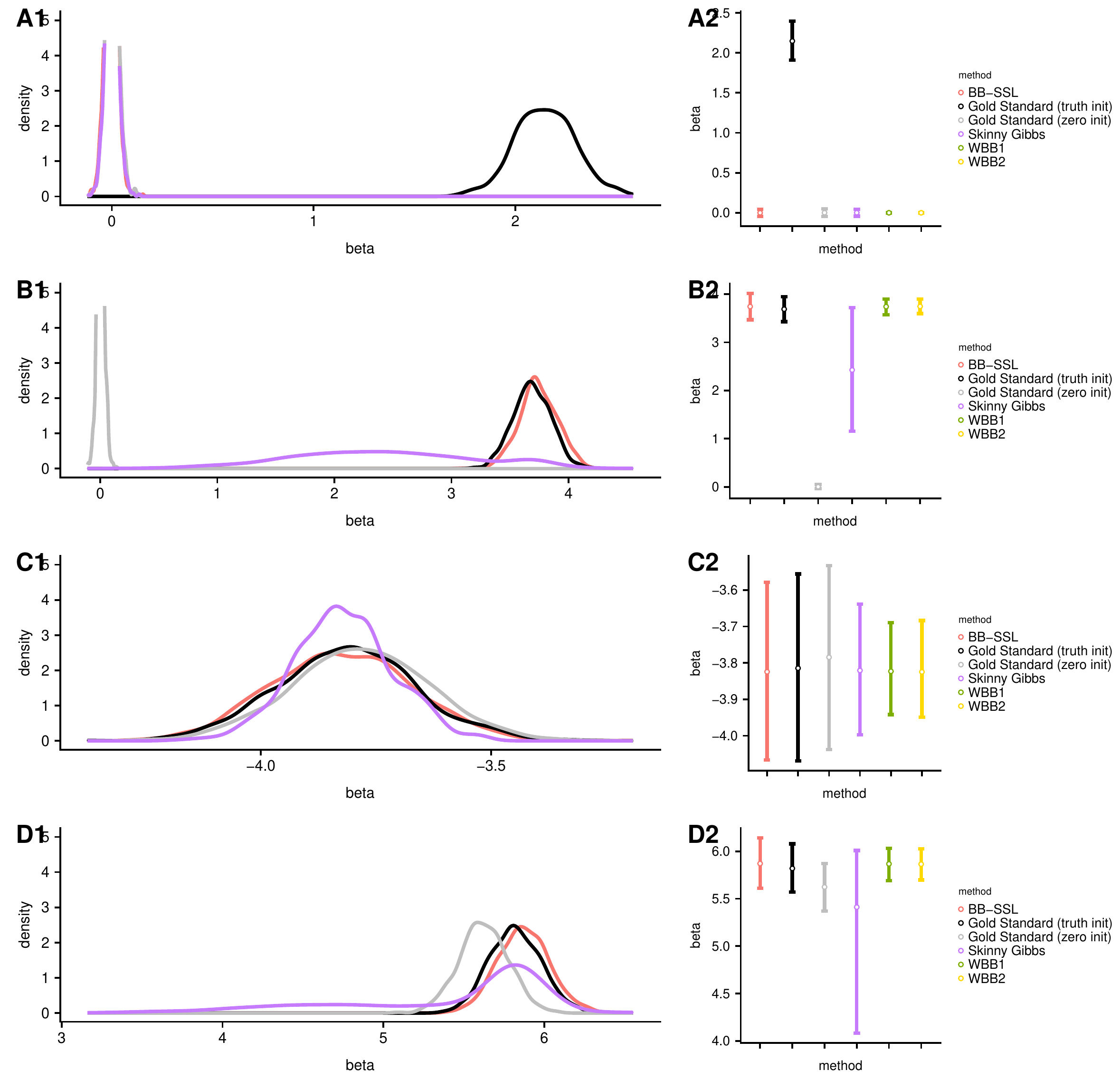}
		\caption{\scriptsize Active predictors, from top to bottom: $\beta_1,\beta_{11},\beta_{21},\beta_{31}$}
	\end{subfigure}%
	%\hfill <-- it is superfluous 
	\begin{subfigure}{0.5\hsize}\centering
		\includegraphics[width=0.9\hsize, height=5in]{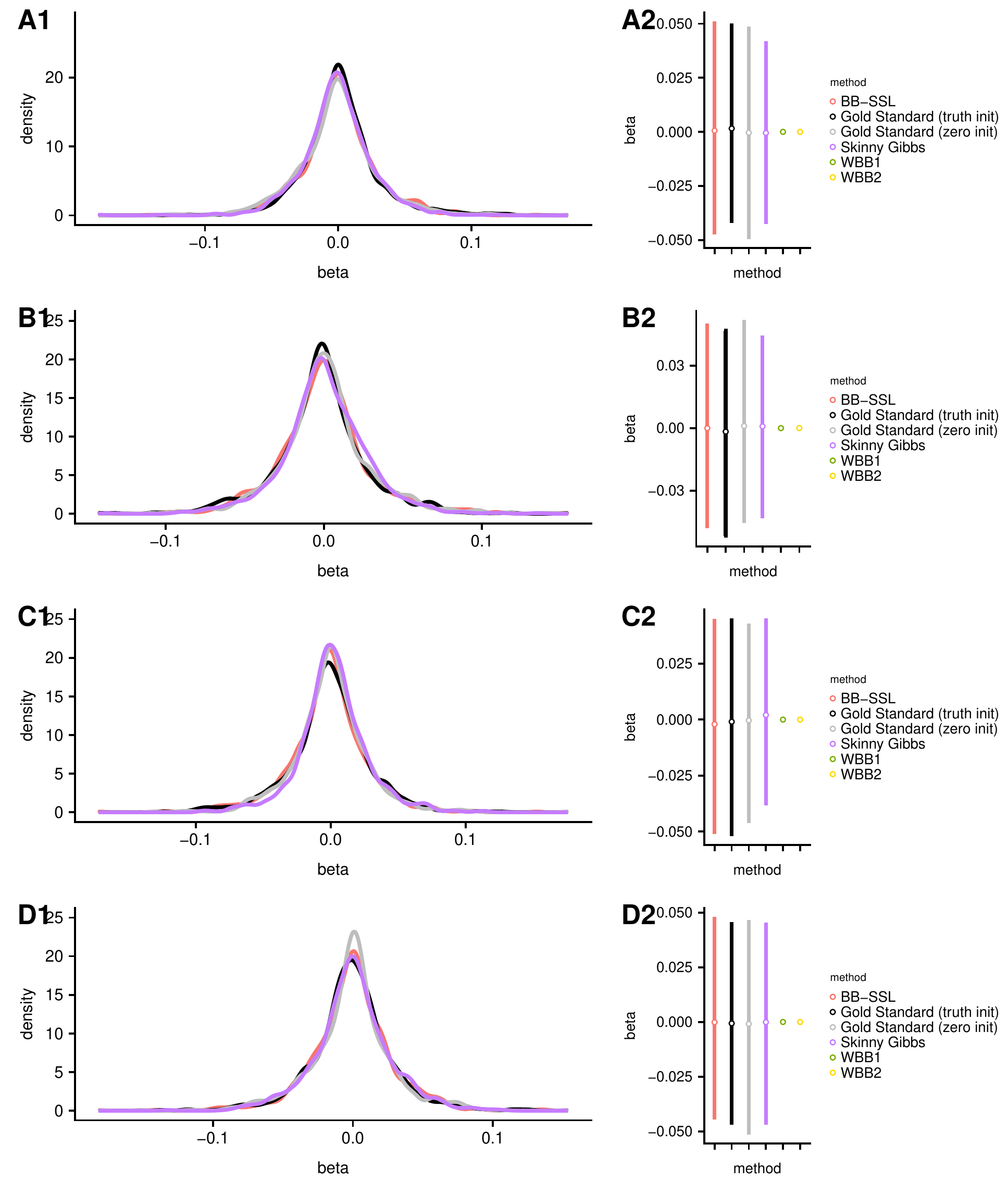}
		\caption{\scriptsize Inactive predictors, from top to bottom: $\beta_2,\beta_{12},\beta_{22},\beta_{32}$}
	\end{subfigure}
	\caption{\small Estimated posterior density (left panel) and credible intervals (right panel) of $\beta_i$'s in the high-dimensional block-wise correlated case ($\rho=0.99$). We have $n=100, p=1000, \beta_{active}=(2,4,-4,6)',\lambda_0=50,\lambda_1=0.05$. Each method has %\footnote{do we need thinning for BB-SSL?} 
		$5\,000$ sample points (after thinning for SSVS and Skinny Gibbs). BB-SSL is fitted using a single $\lambda_0$ and initialized at \texttt{SSLASSO} solution on the original $\X,\y$. Since WBB1 and WBB2 produce a point mass at zero, we exclude them from density comparisons.}
	\label{fig:high_dim_cor99_density}
\end{sidewaysfigure}

\iffalse
\begin{figure}
	\begin{center}
		\includegraphics[width=.3\textwidth]{High_dim_cor_point99_gamma_updated_1.pdf}
	\end{center}
	\caption{\scriptsize Mean of $\gamma_i, i = 1,2,\cdots,1000$ in high-dimensional, block-wise correlated case. $n=100, p=1000, \beta_{active}=(1,2,-2,3),\lambda_0=50,\lambda_1=0.05$ and predictors are grouped into blocks of size 10 with $\rho=0.99$. SSL is fitted using a sequence of $\lambda_0$'s. The gold standard is initialized at truth.}
	\label{fig:high_dim_cor99_MIP}
\end{figure}
\fi

\subsection{High Dimensional, Equi-correlation Structure}
When all covariates are equi-correlated with $\rho=0.9$, SSVS appears to suffer severely from a local trap when $\beta_0=2,\beta_2=3,\beta_3=-3$ and $\beta_4=4$. 
Thus, we also try a setting with larger signals $\beta_0=2,\beta_2=4,\beta_3=-4$ and $\beta_4=6$. Results for the smaller signal $\beta_{active}=(2,3,-3,4)'$ are summarized in Figure \ref{fig:r9} and results for the  larger signal $\beta_{active}=(2,4,-4,6)'$  are in Figure \ref{fig:r9_large}. We show results for both SSVS initializations (one at original and the other at the truth) and observe discrepancies in the posterior approximation. This is why SSVS may not be a reliable gold standard to make comparisons with. We nevertheless observe that BB-SSL is close to the SSVS approximation obtained by initialization at the truth. %The marginal inclusion probabilities from each method are shown in Figure \ref{fig:high_dim_gamma_9_all}. %We observe that all methods have some difficulty in distinguishing 

\begin{sidewaysfigure}
	\begin{subfigure}{0.5\hsize}\centering
		\includegraphics[width=\hsize, height=5in]{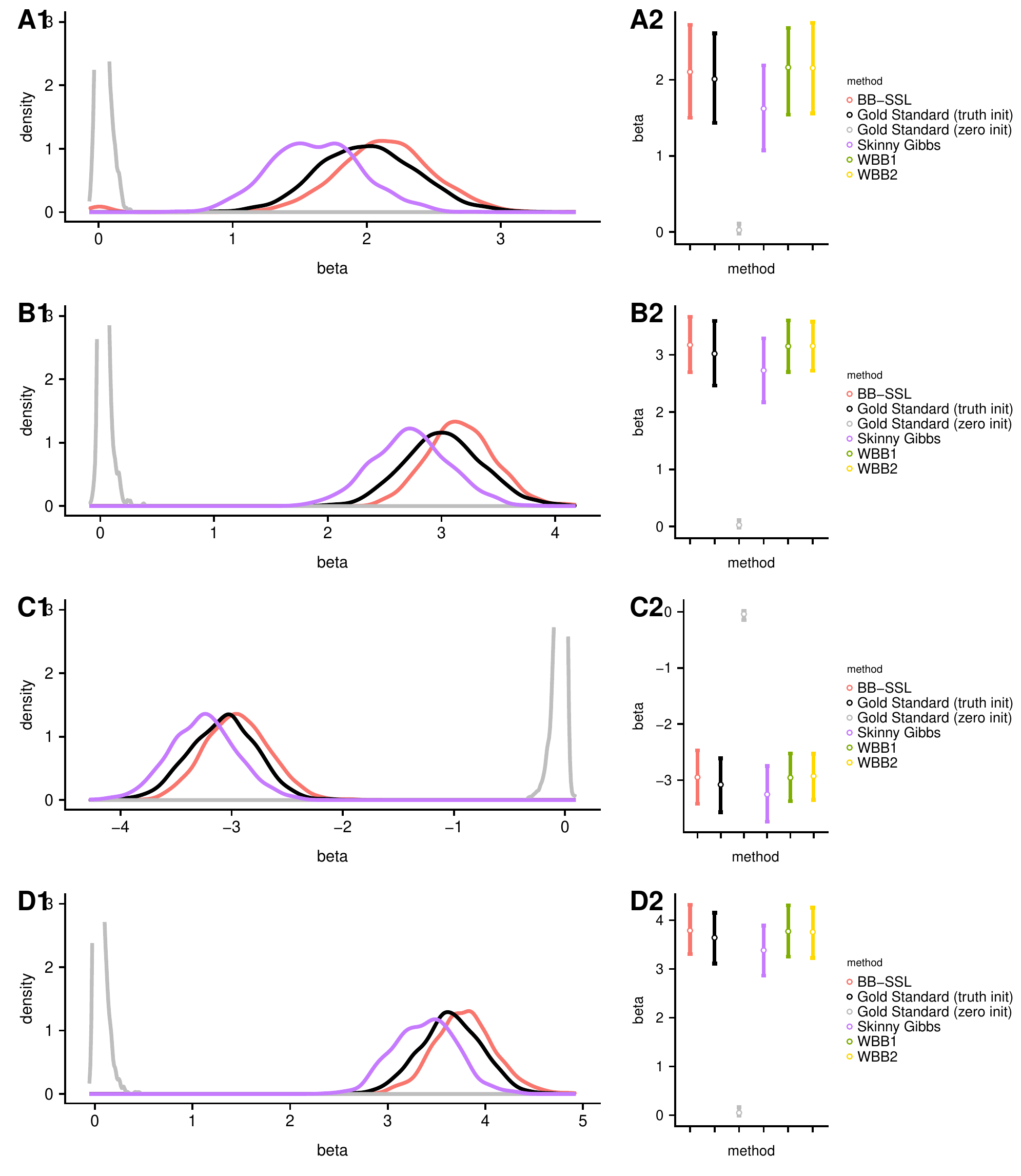}
		\caption{\small Active predictors, from top to bottom: $\beta_1,\beta_{2},\beta_{3},\beta_{4}$}
	\end{subfigure}%
	%\hfill <-- it is superfluous 
	\begin{subfigure}{0.5\hsize}\centering
		\includegraphics[width=\hsize, height=5in]{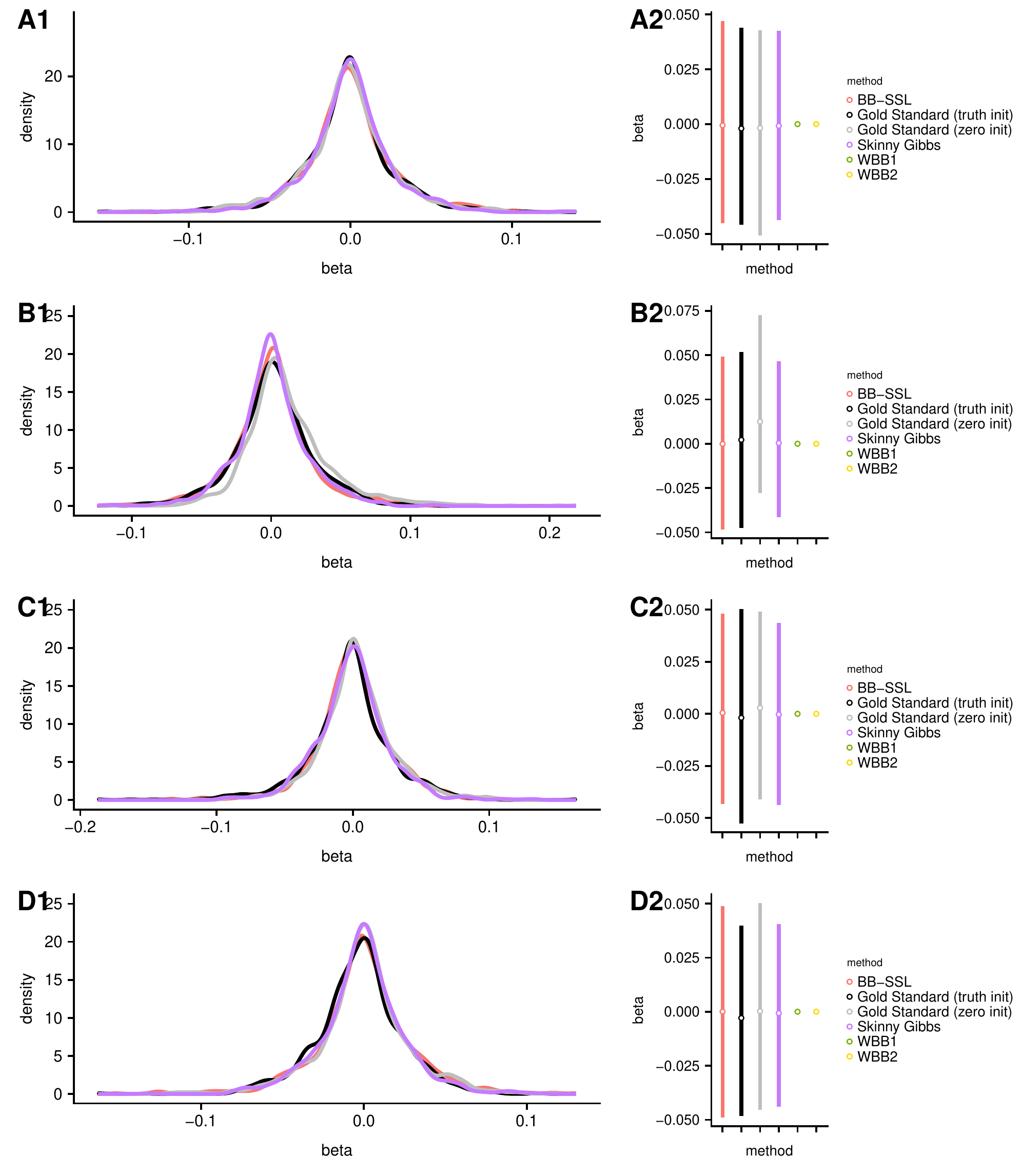}
		\caption{\small Inactive predictors, from top to bottom: $\beta_5,\beta_{6},\beta_{7},\beta_{8}$}
	\end{subfigure}
	\caption{\small  Estimated posterior density (left panel) and $90\%$ credible intervals (right panel) of $\beta_i$'s when all covariatess are correlated with $\rho=0.9$. }
	\label{fig:r9}
\end{sidewaysfigure}	

\begin{sidewaysfigure}
	\begin{subfigure}{0.5\hsize}\centering
		\includegraphics[width=\hsize, height=5in]{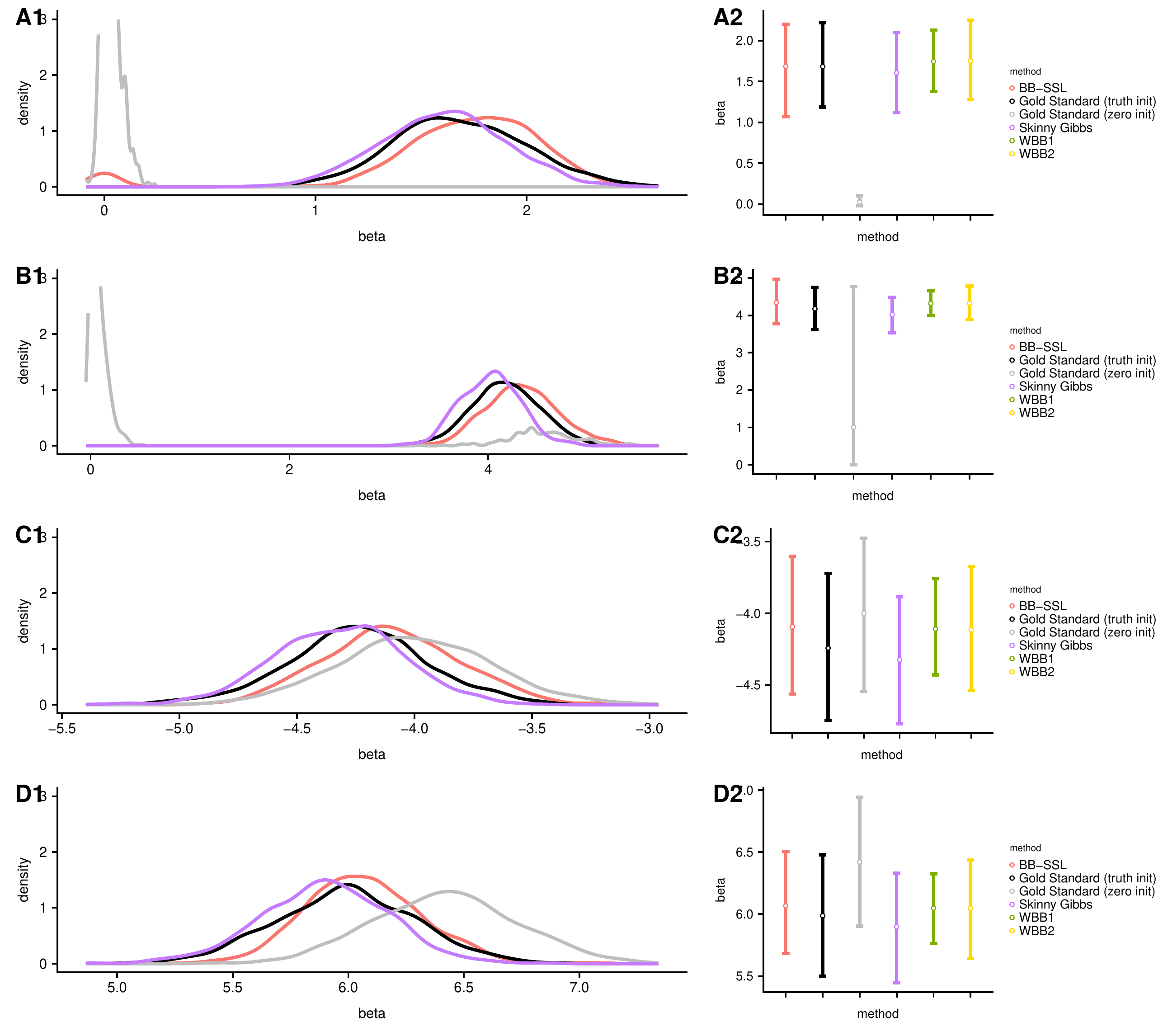}
		\caption{\small Active predictors, from top to bottom: $\beta_1,\beta_{2},\beta_{3},\beta_{4}$}
	\end{subfigure}%
	%\hfill <-- it is superfluous 
	\begin{subfigure}{0.5\hsize}\centering
		\includegraphics[width=\hsize, height=5in]{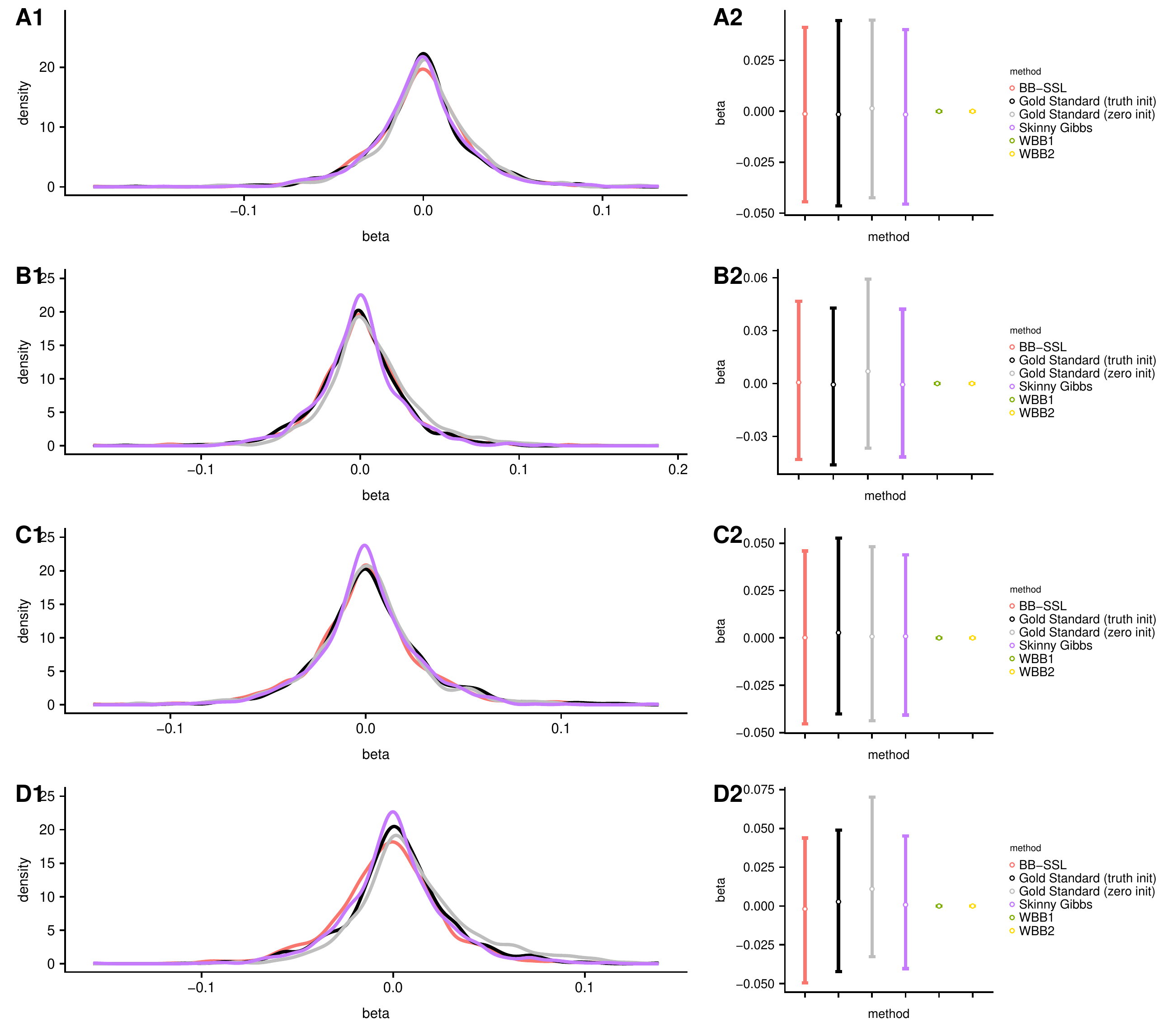}
		\caption{\small Inactive predictors, from top to bottom: $\beta_5,\beta_{6},\beta_{7},\beta_{8}$}
	\end{subfigure}
	\caption{\small  Estimated posterior density (left panel) and $90\%$ credible intervals (right panel) of $\beta_i$'s when all covariatess are correlated with $\rho=0.9$. }
	\label{fig:r9_large}
\end{sidewaysfigure}

\subsection{Influence of $\alpha$ On The Posterior}
In this section we investigate the influence of $\alpha$ on BB-SSL posterior under (a) high-dimensional, block-wise correlated ($\rho=0.6$) setting, as shown in Figure \ref{fig:appendix_highdim_6_alpha}, and (b) Durable Goods Marketing Data Set, as shown in Figure \ref{fig:appendix_causal_trace}.
In both datasets, as we increase $\alpha$, the change in posterior variance reduces for each unit increase in $\alpha$.

\begin{figure}\centering
	\begin{subfigure}{.5\textwidth}\centering
		\includegraphics[width=1\textwidth,height=0.4\textheight]{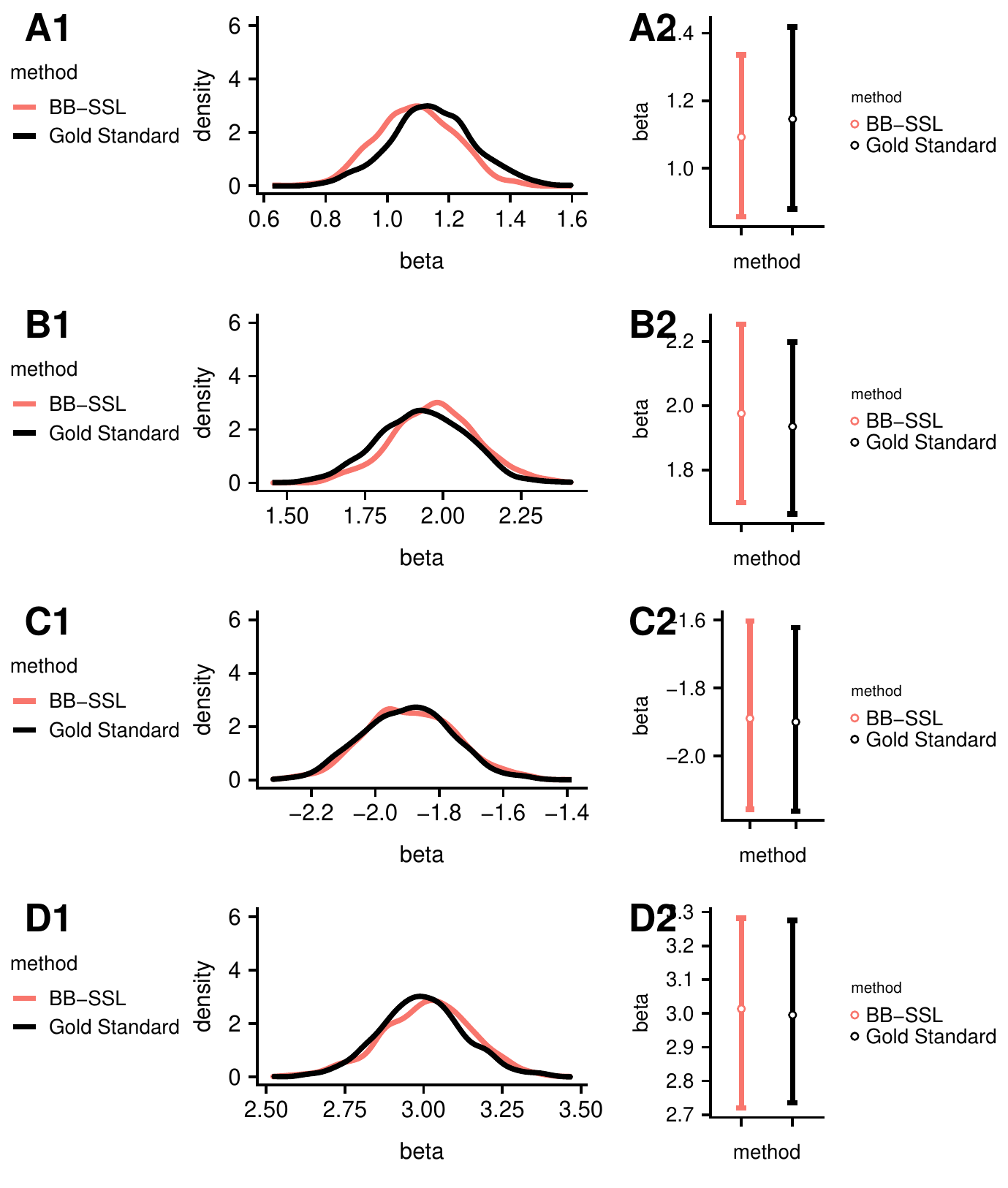}
		\caption{\scriptsize $\alpha=2$.}
	\end{subfigure}%
	\begin{subfigure}{.5\textwidth}\centering
		\includegraphics[width=1\textwidth,height=0.4\textheight]{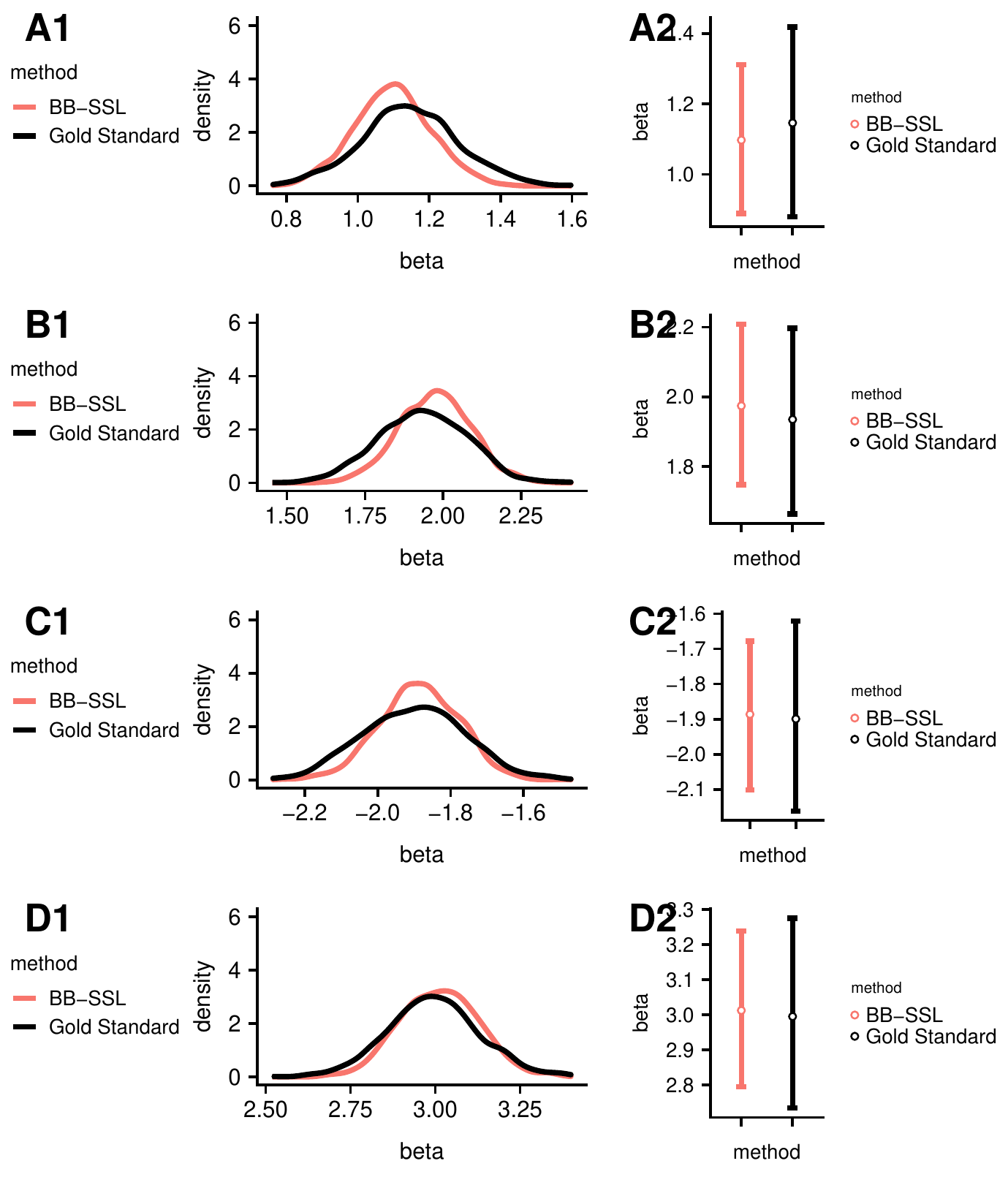}
		\caption{\scriptsize $\alpha=10$.}
	\end{subfigure}%
	
	\begin{subfigure}{.5\textwidth}\centering
		\includegraphics[width=1\textwidth,height=0.4\textheight]{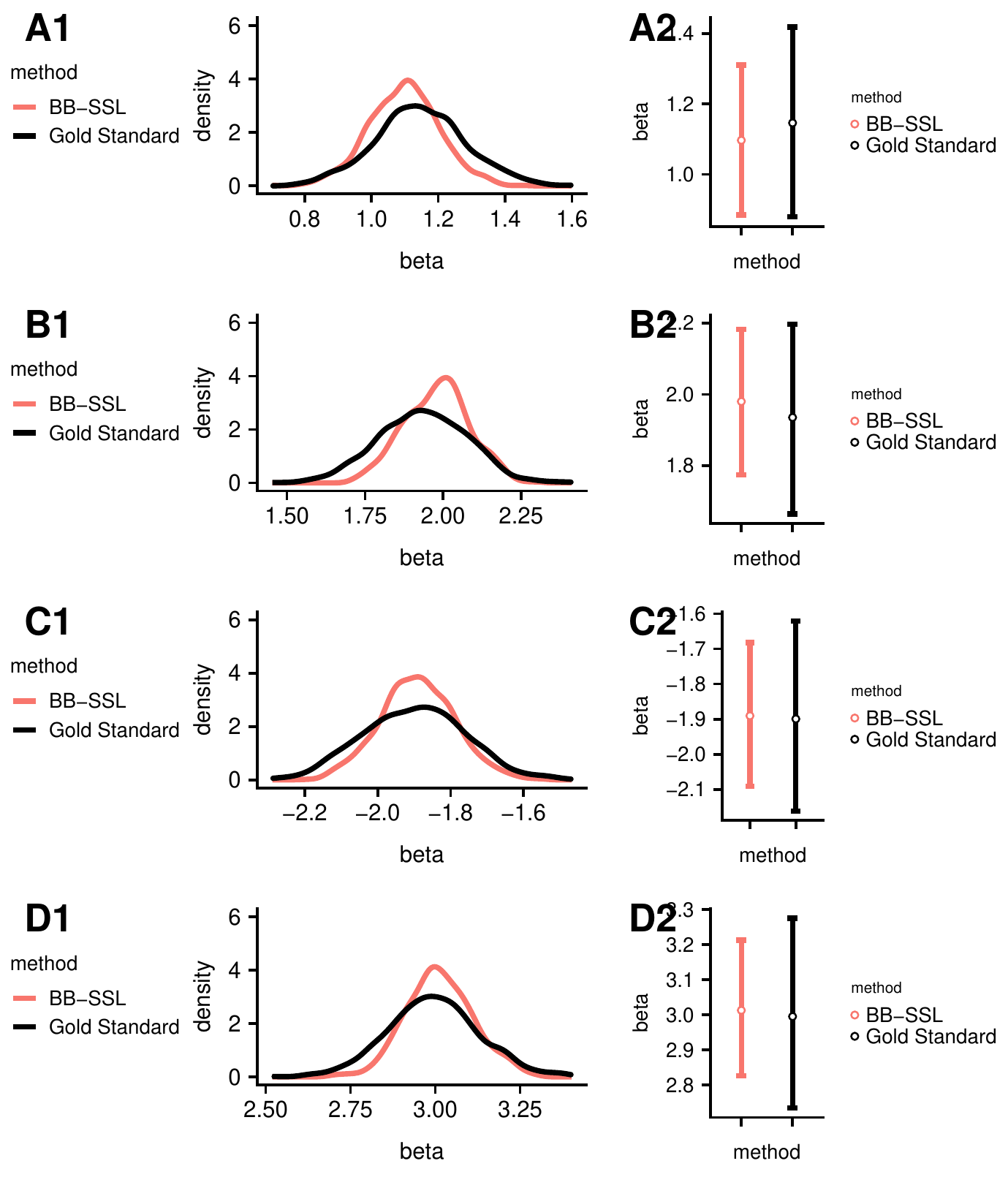}
		\caption{\scriptsize $\alpha=100$.}
	\end{subfigure}%
	\begin{subfigure}{.5\textwidth}
		\includegraphics[width=1\textwidth,height=0.4\textheight]{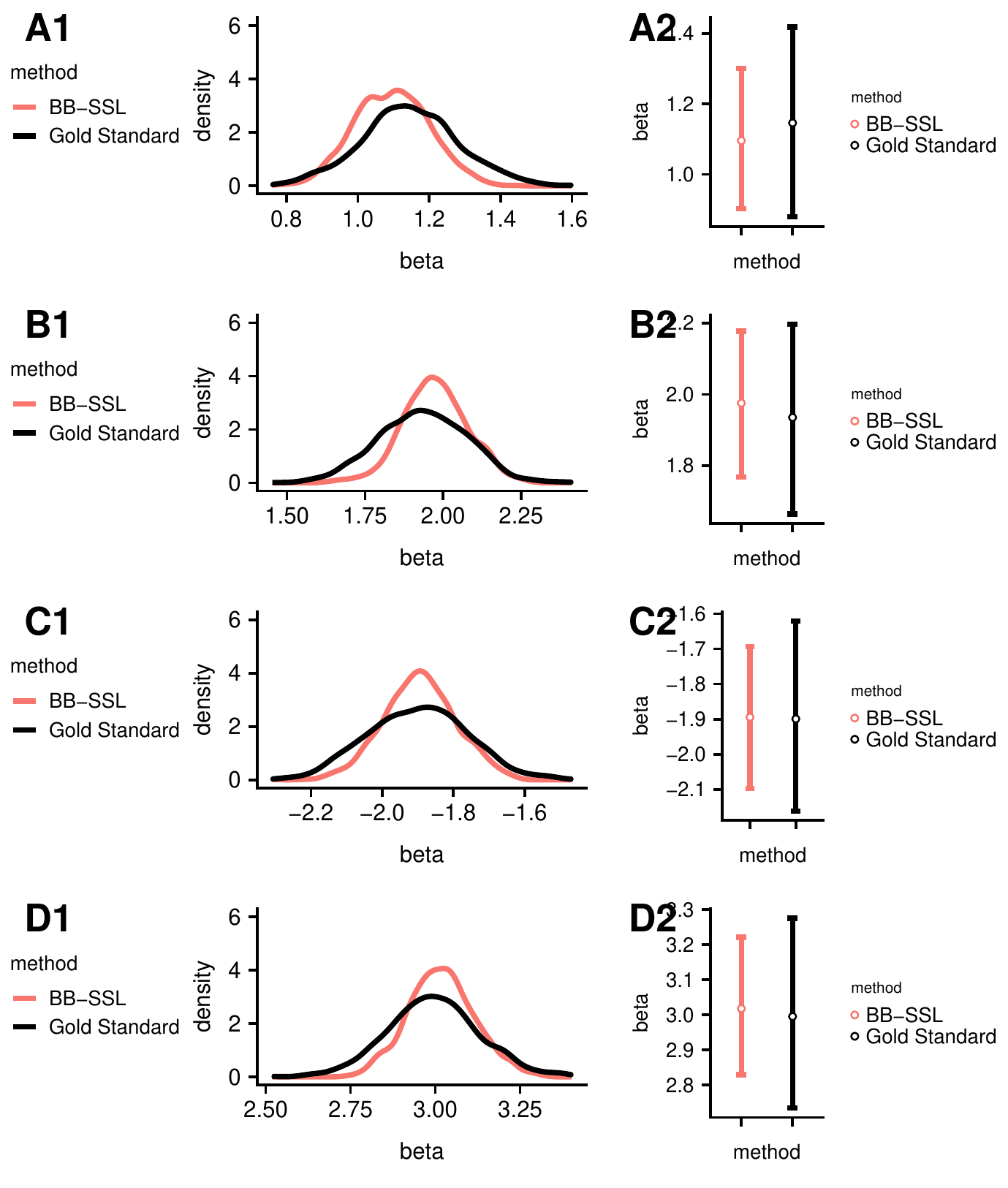}
		\caption{\scriptsize $\alpha=1000$.}
	\end{subfigure}
	
	\caption{\small Comparison of posterior density for active $\beta_i$'s when choosing different $\alpha$ in high-dimensional, block-wise correlated setting ($\rho=0.6$). \texttt{SSLASSO} is fitted with $\lambda_0$ being an equal difference sequence of length 10 starting at 0.05 and ending at 50. We set $\lambda_1=0.05$. Since WBB1 and WBB2 produce a point mass and do not fit into the $y$-axis, we exclude it in the density plot.}
	\label{fig:appendix_highdim_6_alpha}
\end{figure}

\begin{figure}\centering
	\includegraphics[width=.8\textwidth]{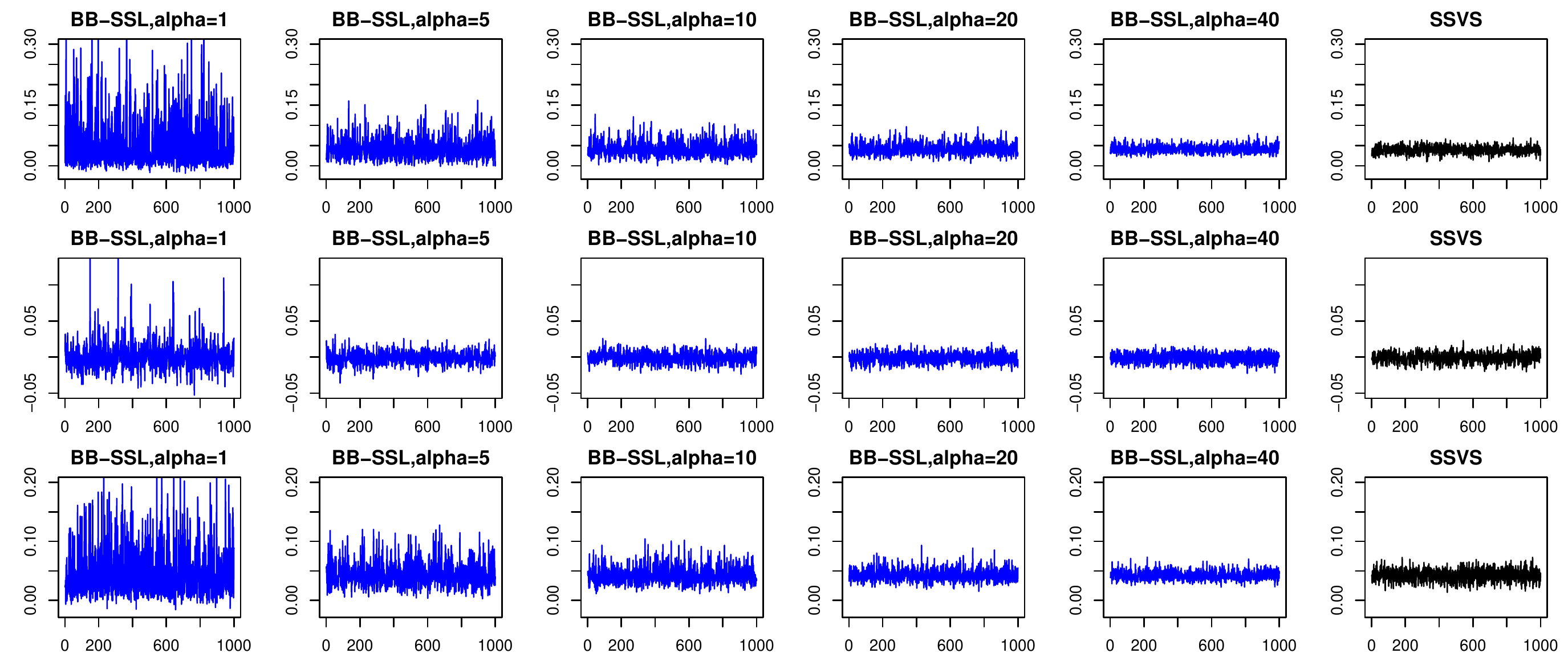}
	
	\includegraphics[width=.8\textwidth]{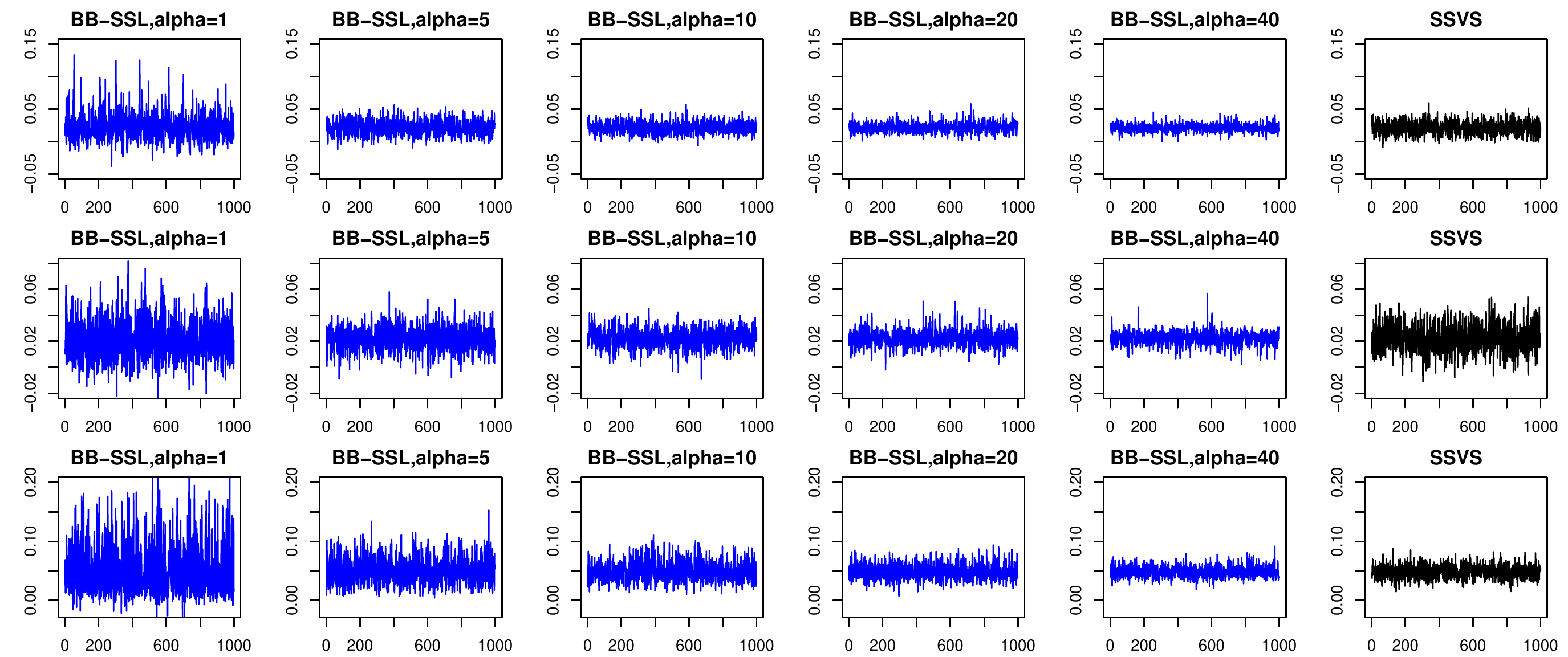}
	
	\includegraphics[width=.8\textwidth]{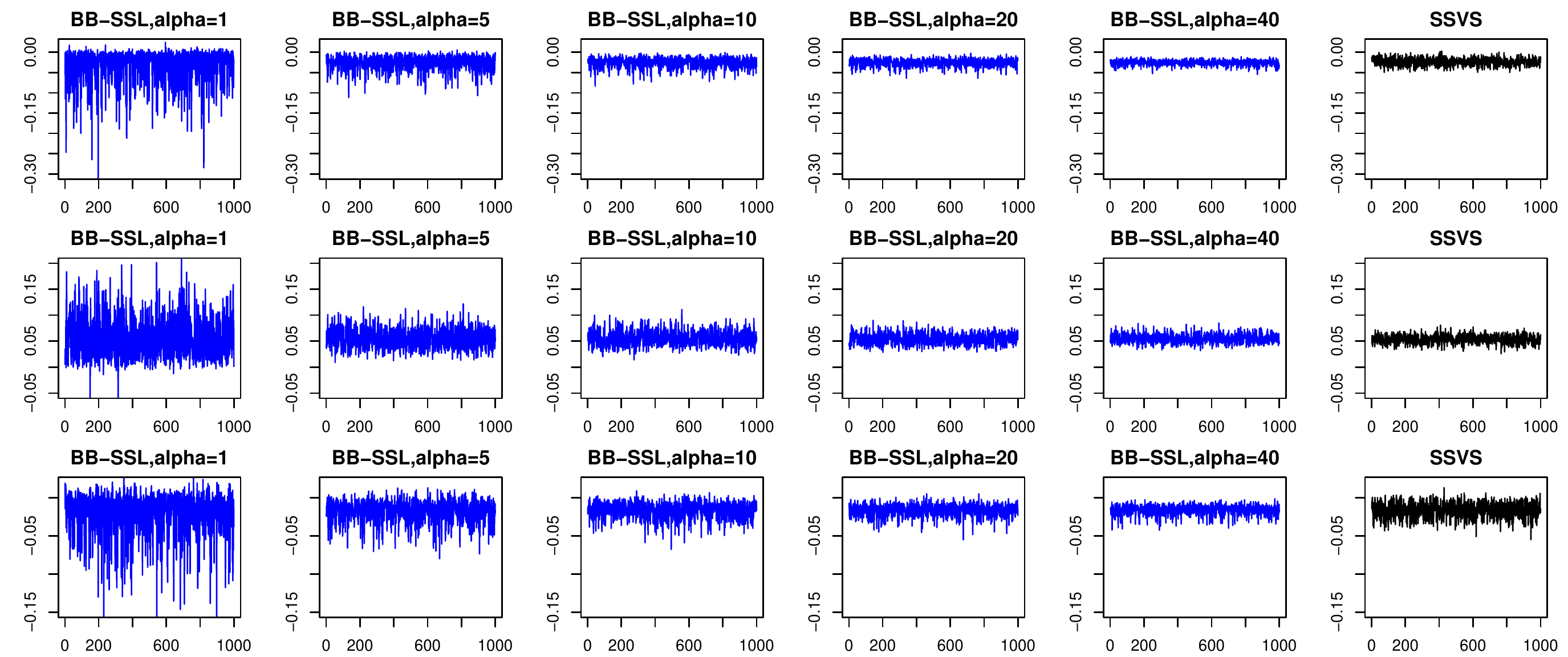}
	
	\includegraphics[width=.8\textwidth]{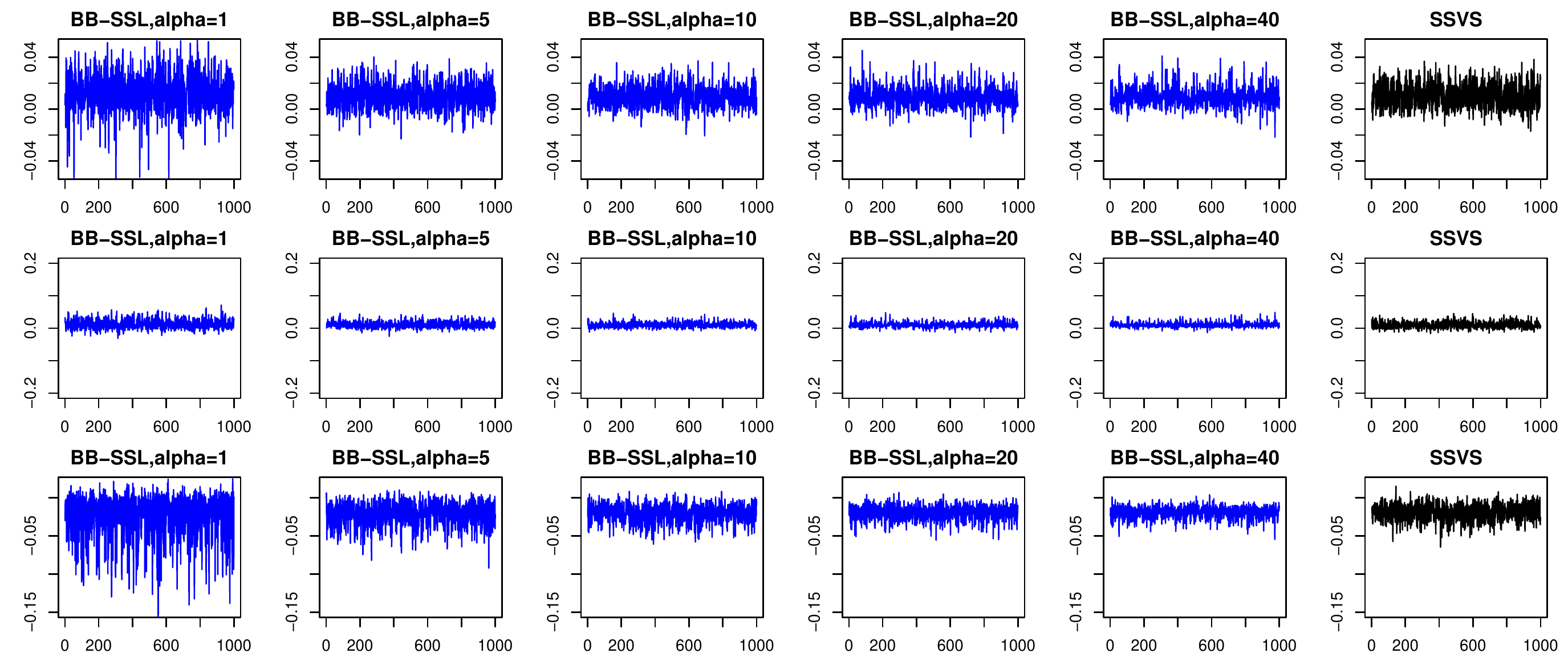}
	\caption{\small Trace plots of $\beta_i,i,1,2,\cdots,12$ under varying $\alpha$'s in model \eqref{eq:real_data_model} in the Durable Goods Marketing Data Set.}
	\label{fig:appendix_causal_trace}
\end{figure}

\iffalse
\begin{figure}\centering
	%\begin{subfigure}{.5\textwidth}\centering
	\includegraphics[height=0.4\textheight]{DurData_interaction_density_additional.pdf}
	%\caption{\scriptsize $\beta_i,i=5,6,7,8$}
	%\end{subfigure}%
	%\begin{subfigure}{.5\textwidth}
	%	\includegraphics[width=1\textwidth,height=0.4\textheight]{DurData_interaction_density31.pdf}
	%	\caption{\scriptsize $\beta_i,i=9,10,11,12$}
	%\end{subfigure}
	\caption{\small Comparison of posterior density for ``S-U-CAT-NBR-24MO", ``S-SAL-TOT12M", ``S-U-CAT-NBR-24MO $\times$ interaction", ``S-SAL-TOT12M $\times$ interaction" in interaction model when choosing $\alpha=20$.}
	\label{fig:appendix_causal__density}
\end{figure}
\fi

\end{document}